\documentclass[11pt]{article}

\usepackage[letterpaper, margin=1in]{geometry}
\usepackage{amsfonts,amssymb,amsmath,amsthm}
\usepackage[utf8]{inputenc}
\usepackage[T1]{fontenc}
\usepackage{newtxtext}
\usepackage[hidelinks]{hyperref}
\usepackage[shortlabels]{enumitem}
\usepackage{color}
\usepackage{tikz}
\usetikzlibrary{arrows,matrix,positioning}
\usepackage{tikz-cd}
\usepackage{graphicx,  fullpage, url}
\usepackage{verbatim}
\usepackage{mathrsfs}
\usepackage{mathtools}
\usepackage{array}
\usepackage{thm-restate}
\usepackage{mdframed}
\usepackage{bbm}
\usepackage{boxedminipage}
\usepackage{float}
\usepackage{authblk}

\theoremstyle{plain}
\newtheorem{theorem}{Theorem}[section]
\newtheorem{lemma}[theorem]{Lemma}
\newtheorem{definition}[theorem]{Definition}
\newtheorem{corollary}[theorem]{Corollary}
\newtheorem{claim}[theorem]{Claim}

\newtheorem{remark}[theorem]{Remark}
\newtheorem{fact}[theorem]{Fact}

\newcommand{\N}{\mathbb{N}}
\newcommand{\Z}{\mathbb{Z}}
\newcommand{\F}{\mathbb{F}}
\newcommand{\E}{\mathbb{E}}

\newcommand{\Var}{\mathrm{Var}}

\newcommand{\calL}{\mathcal{L}}

\newcommand{\bi}{\mathbf{i}}
\newcommand{\bj}{\mathbf{j}}
\newcommand{\ba}{\mathbf{a}}
\newcommand{\bb}{\mathbf{b}}
\newcommand{\bc}{\mathbf{c}}
\newcommand{\bu}{\mathbf{u}}

\newcommand{\bX}{\mathbf{X}}
\newcommand{\bZ}{\mathbf{Z}}

\newcommand{\dist}{\delta}
\newcommand{\MDS}{\mathrm{MDS}}

\newcommand{\spn}{\ensuremath{\operatorname{span}}}
\newcommand{\img}{\ensuremath{\operatorname{Image}}}

\renewcommand{\char}{\textsf{char}}
\newcommand{\wt}{\mathsf{wt}}
\newcommand{\mult}{\mathsf{mult}}
\newcommand{\GZP}{\text{\normalfont GZP}}
\newcommand{\LCC}{\mathrm{LCC}}
\newcommand{\LLDC}{\mathrm{LLDC}}
\newcommand{\ALLDC}{\mathrm{ALLDC}}

\newcommand{\poly}{\mathrm{poly}}
\newcommand{\polylog}{\mathsf{polylog}}

\newcommand{\RS}{\mathrm{RS}}
\newcommand{\RM}{\mathrm{RM}}

\newcommand{\MULT}{\mathsf{MULT}}

\newcommand{\PRUNE}{\mathsf{Prune}}

\title{Advances in List Decoding of Polynomial Codes}

\author[1]{Mrinal Kumar}
\author[2]{Noga Ron-Zewi}
\affil[1]{School of Technology and Computer Science, Tata Institute of Fundamental Research, \newline Email: \texttt{mrinal.kumar@tifr.res.in}}
\affil[2]{Department of Computer Science, University of Haifa. \newline Email:  \texttt{noga@cs.haifa.ac.il}}

\date{}

\begin{document}

\maketitle

\begin{abstract}
 Error-correcting codes are a method for representing data, so that one can recover the original information even if some parts of it were corrupted. The basic idea, which dates back to the revolutionary work of Shannon and Hamming about a century ago, is to encode the data into a redundant form, so that the original information can be decoded  from the redundant encoding even in the presence of some noise or corruption. 
One prominent family of error-correcting codes are Reed-Solomon Codes which encode the data using evaluations of low-degree polynomials. Nearly six decades after they were introduced, Reed-Solomon Codes, as well as some related families of polynomial-based codes, continue to be widely studied, both from a theoretical perspective and from the point of view of applications. 

Besides their obvious use in communication, error-correcting codes such as Reed-Solomon Codes are also useful for various applications in theoretical computer science. These applications often require the ability to cope with many errors,
much more than what is possible information-theoretically. List-decodable codes are a special class of error-correcting codes that enable correction from more errors than is traditionally possible by allowing a small list of candidate decodings. These codes have turned out to be extremely useful in various applications across theoretical computer science and coding theory. 

In recent years, there have been significant advances in list decoding of Reed-Solomon Codes and related families of polynomial-based codes. This includes efficient list decoding of such codes up to the information-theoretic capacity, with optimal list-size, and using fast nearly-linear time, and even sublinear-time, algorithms. In this book, we survey these developments.

\end{abstract}

\newpage
\setcounter{tocdepth}{2}
\thispagestyle{empty}
{\small \tableofcontents}

\newpage

\section{Introduction}\label{sec:intro}

The problems of communicating and storing data, robustly and efficiently, in the presence of errors or noise is a fundamental problem of interest at the intersection of mathematics and engineering. The theory of error correcting codes (or coding theory) was developed in the landmark works of Shannon \cite{shannon48} and Hamming \cite{hamming50} in the 1940s/50s as an attempt towards formalizing and understanding these problems. 
The central notion in this theory is that of an \emph{error correcting code} (or simply a \emph{code}).
    An error correcting code
    $C$ of \emph{block length} $n$ over a finite \emph{alphabet} $\Sigma$ is a subset of $\Sigma^n$. The words in $C$ are referred to as \emph{codewords}, and typically, we also associate a set $M$ of strings (called \emph{messages}) with the code $C$ such that there is a bijection $\Phi$ (called an \emph{encoding map}) from  $M$ to $C$. 
    
    The \emph{rate} of a code $C$ is defined as $R(C) := \frac{\log (|C|)} {n\log (|\Sigma|)
    }$, and it measures the amount of the \emph{redundancy} in the code. Specifically, to use a code to communicate in the presence of noise, the idea is that if we want to communicate a message $m \in M$ over a noisy channel, we instead \emph{encode} it as $\Phi(m)$ and send $\Phi(m)$ over the channel. The rate of the code then measures the number of possible messages one can encode using the code, out of all possible length $n$ strings over $\Sigma$. 
   The \emph{minimum (Hamming) distance} of a code $C$ is defined as $\Delta(C) := \min_{c \neq c' \in C} \Delta(c, c')$, where $\Delta(c,c')$, referred to as the \emph{Hamming distance} between $c$ and $c'$, is the number of entries on which $c$ and $c'$ differ.
The minimum distance of a code is a measure of how \emph{far} in Hamming distance any two distinct codewords in $C$ are from each other, and consequently, how noise tolerant the code might be. Specifically, the channel might corrupt some of the entries of $\Phi(m)$ during transmission and the receiver might receive a string $w \in \Sigma^n$. However, if $w$ and $\Phi(m)$ differ on less than $\frac {\Delta(C)} 2$ entries, then $\Phi(m)$ (and hence $m$) can be recovered uniquely by the decoder by searching for the unique codeword in $C$ that is \emph{closest} to the \emph{received word} $w$. 

The above discussion also highlights the innate tension between the rate and distance of codes - a larger rate appears to force a lower distance (and hence, lower error tolerance). Understanding the precise tradeoff between rate and distance of codes continuous to be a fundamental open problem at the heart of coding theory. In particular, the classical \emph{Singleton Bound} from the 60's states that for a linear code $C \subseteq \Sigma^n$ of size $|C|=\Sigma^k$ it must hold that $\Delta(C)\leq n-k+1$. This in particular implies that
$\delta(C) \leq 1- R(C)$ for any code $C$, where $\delta(C):=\frac {\Delta(C)} {n}$ is the \emph{relative distance} of $C$. Codes satisfying the Singleton bound with equality are called \emph{maximum distance separable} (MDS) codes, and by now we know of various families of codes satisfying this property. 

The aforementioned framework of using codes for communication in the presence of noise clearly also highlights \emph{computational questions} of great interest in coding theory, namely that of explicitly constructing families of codes with non-trivial rate and distance and designing efficient {encoding} and {decoding} algorithms for these codes. Moreover, such computational aspects are also of great interest to computer science,  
because of the numerous surprising connections and applications that error-correcting codes have had to areas such as complexity theory, cryptography, pseudorandomness, and algorithm design, which often require explicit and efficiently decodable codes. 

\subsection{List decoding}

If a code $C$ has distance $\Delta(C)$, then clearly, for any word $w \in \Sigma^n$, the number of codewords at Hamming distance less than $\frac{\Delta(C)} 2$ from $w$ is at most one. 
Consequently, $\frac{\Delta(C)} 2$ is called the \emph{unique decoding radius} of the code.
The notion of \emph{list decoding}, introduced in the works of Elias \cite{Elias57} and Wozencraft \cite{wozencraft58} relaxes this error bound to beyond $\frac {\Delta(C)} 2$. At this distance, the codeword closest to $w$ need not be unique, but we can still hope that the \emph{list} of such codewords is small. 
More formally, 
    for $\alpha > 0$ and a positive integer $L$, a code $C \subseteq \Sigma^n$ is said to be $(\alpha, L)$-\emph{list decodable} if for every $w \in \Sigma^n$, there are at most $L$ different codewords $c \in C$ such that the Hamming distance $\Delta(w, c)$ is at most $\alpha n$. 

One of the ways to see that this notion of list decoding offers something more than the notion of unique decoding is via a theorem of Johnson, which (informally) shows that \emph{every} code $C$ is $(\alpha, L)$-list decodable with \emph{constant} list size $L$ as long as the parameter $\alpha$ is less than a function $J(\Delta)$ of the distance $\Delta$ of the code. $J(\Delta)$ is referred to as the \emph{Johnson Bound} (or the Johnson Radius) of the code, and it is always as large as $\frac \Delta 2$, but can be much larger. 
Furthermore, an extension of the aforementioned Singleton Bound to the setting of list decoding shows that any $(\alpha, L)$-list decodable code must satisfy that $\alpha \leq \frac{L} {L+1}(1-R)$. Note that the $L=1$ case corresponds to the classical Singleton Bound which implies that $\alpha \leq \frac{\delta} {2} \leq \frac{1-R} {2}$, but for a growing $L$, the bound comes close to twice as large. We say that a code achieves \emph{list-decoding capacity} if it approaches this bound, i.e., if it is $(1-R-\epsilon, O_\epsilon(1))$-list decodable for any $\epsilon >0$ (We will say that the code achieves list-decoding capacity even if the list is bounded as a function of the block length). Once more, by now we know of several families of capacity-achieving list-decodable codes, which by allowing a small (constant-size) list, can tolerate almost twice as many errors than what is possible in the unique decoding setting!

In the communication setting, list decoding may be useful when the receiver has some \emph{side information} about the transmitted message which enables the receiver to eliminate some of the codewords in the list, and list decodable codes are also a useful building block in the construction of uniquely decodable codes. Furthermore, list decodable codes serve as a \emph{bridge} between the \emph{adversarial error model} described above (pioneered by Hamming \cite{hamming50}) in which an adversary may corrupt an arbitrary subset of the codeword entries, and the \emph{random error model} (pioneered by Shannon \cite{shannon48}) which assumes that each codeword entry is corrupted independently with some fixed probability. It is well known that one can tolerate more (roughly twice as many) errors in the random error model than in the adversarial model, and thus using list-decodable codes, one can handle more errors than in the adversarial unique decoding setting in the presence of both random errors (in which case the list will typically contain a single codeword) and adversarial errors (by allowing a small list). List decodable codes also turned to be highly useful for applications in computer science, like hardness results in complexity theory and explicit construction of various pseudorandom objects like expanders, extractors and pseudorandom generators. 

The applications described above highlight two main fundamental questions about list-decodable codes. The first is a \emph{combinatorial question}, where the goal is to bound the list size of various codes for a given decoding radius. The second is a \emph{computational question} about designing explicit list-decodable codes with efficient list decoding algorithms which given a received word,  output the list of close-by codewords.
A substantial fraction of such results were obtained for families of algebraic error correcting codes based on \emph{low-degree polynomials} (also motivated by the applications described above). 
Our goal in this survey is to discuss some of these developments, with the focus being on polynomial-based families of codes that we discuss next.

\subsection{Polynomials in coding theory}

Algebraic techniques, and in particular, the properties of low degree polynomials over finite fields have played an outsized role in the developments in coding theory since the early years of this research area. Indeed, the families of error correcting codes based on the evaluation of low degree polynomials (univariate or multivariate) are among the most well studied families of codes.

The flag bearer of such a code family is undoubtedly the family of \emph{Reed-Solomon Codes}. The codewords of this code correspond to the evaluations of all univariate polynomials of degree less than $k$ with coefficients in a finite field $\F_q$ on a specified set $E \subseteq \F_q$ of $n$ field elements. The distance of these codes stems from the basic algebraic fact that two distinct univariates of degree less than $k$ cannot agree on more than $k$ inputs. It follows that Reed-Solomon Codes are MDS codes. 

The mathematical properties of low degree polynomials at the core of Reed-Solomon Codes have found vast generalizations and have led to numerous other fundamental and well studied code families. Some classical examples are \emph{Reed-Muller Codes} whose codewords correspond to evaluations of low degree \emph{multivariate} polynomials over a product set, 
BCH Codes where the codewords correspond to low-degree univariate polynomials whose evaluations lie in a small subfield, and \emph{Algebraic Geometry} (AG) \emph{codes} where the codewords correspond to evaluations of certain functions on rational points of some algebraic curve. 

Besides their elegance and their fundamental mathematical properties, polynomial-based codes as above in many cases achieve the best possible tradeoff between their rate and error-correction radius, in various communication scenarios, and their algebraic structure also lends itself to 
 fast encoding and decoding algorithms. Furthermore, polynomial-based codes such as Reed-Solomon Codes are also widely used in practice due to their good concrete parameters.
The algebraic properties of these codes are also often useful for computer science applications. For example, one such useful property is the \emph{multiplication property}, which guarantees that the point-wise multiplication of a pair of codewords (for example, evaluations of a pair of degree $k$ polynomials) is a codeword in another related code (corresponding to evaluations of degree $2k$ polynomials).

In this survey, we will focus on the list-decoding properties of polynomial-based codes, which we discuss next.

\subsection{List decoding with polynomial codes}

The celebrated work of Sudan \cite{Sudan97} and Guruswami-Sudan \cite{GS-list-dec} from the 1990s gave efficient (polynomial-time) algorithms for list decoding Reed-Solomon Codes beyond the unique decoding radius, \emph{up to the Johnson Bound}. Note that up to this bound the list size is guaranteed to be constant, so the question is only \emph{computational}, namely, designing efficient list decoding algorithms which given a received word,  output the list of close-by codewords.
This result led to significant interest in understanding the list decodability of Reed-Solomon Codes \emph{beyond the Johnson Bound}. Beyond this radius, the question is both \emph{combinatorial} and \emph{computational}, namely, showing upper bounds on the list size, and if such upper bounds exist, then also searching for efficient list decoding algorithms. 
While this continues to be an extremely active area of research with many open research directions, the interest in this problem has been an important driving force behind many other significant developments in this area. 

This includes the discovery of variations of Reed-Solomon Codes that were shown to achieve \emph{list decoding capacity}, with a list size that matches the \emph{generalized singleton Bound}, together with efficient list decoding algorithms up to capacity for these codes. 
Some prominent examples of such codes are \emph{Folded Reed-Solomon Codes}, where one bundles \emph{correlated evaluations}  of a low-degree polynomial into a single codeword entry, 
and \emph{multiplicity codes}, where the encoding also includes evaluations of \emph{derivatives} of a low-degree polynomial.

 Recent work in this area has also led to significant advances in the understanding of the \emph{combinatorial} list decodability of Reed-Solomon Codes beyond the Johnson Bound.  Specifically, while it was shown that certain Reed-Solomon Codes are \emph{not} list-decodable well beyond the Johnson Bound with a list size that is polynomial in the block length (let alone, a constant list size), it was also shown that Reed-Solomon Codes over \emph{random} evaluation points achieve \emph{list-decoding capacity}, with  a list size that matches the \emph{generalized singleton Bound}, with high probability. Finding explicit evaluation points satisfying this property, together with an efficient list-decoding algorithm, constitutes a major open problem in this area by the time of writing of this survey. 

A distadvantage of all the aforementioned polynomial-based codes is that their alpahbet is very large (polynomial in the block length). The alphabet size can be brought down to a \emph{constant} by resorting to appropriate variants of AG codes. 
With regard to efficiency, by now we know of fast \emph{near-linear time} implementation of the list-decoding algorithms mentioned above. Moreover, using \emph{Reed-Muller Codes} and \emph{multivariate multiplicity codes} one can also obtain \emph{local list-decoding} algorithms that are able to list-decode individual codeword entries in \emph{sublinear time}.

\subsection{Organization and scope of this survey}\label{subsec:intro_scope}

The survey is organized as follows: 
\begin{itemize}
\item In Section \ref{sec:prelim}, we begin by presenting formal definitions and  basic facts about list-decodable codes and low-degree polynomials, and introducing some families of polynomial-based codes that will be the focus of this survey.

\item In Section \ref{sec:johnson}, we present efficient (polynomial-time) algorithms  for list decoding Reed-Solomon Codes up to the Johnson Bound, and we also present combinatorial lower bounds on the list size of Reed-Solomon Codes beyond the Johnson bound. 

\item In Section \ref{sec:capacity}, we present efficient (polynomial-time)  algorithms for list decoding multiplicity codes up to capacity, and we also show that a similar algorithm can be used to list decode Reed-Solomon codes over subfield evaluation points  in slightly non-trivial sub-exponential time. These algorithms in particular show that the list size of multiplicity codes at capacity is at most polynomial in the block length, while the list-size of Reed-Solomon codes over subfield evaluation points is at most sub-exponential in the block length.  

\item In  Section \ref{sec:list}, we present combinatorial upper bounds on the list size of polynomial-based codes at capacity.  Specifically, we show that the list size of Reed-Solomon codes over random evaluation points and multiplicity codes  is in fact a constant, independent of the block length (and even matches the generalized Singleton bound),  and that the list-size of  Reed-Solomon codes over subfield evaluation points  can also be reduced to a constant (and the running time of the list decoding algorithm to a polynomial) by passing to an appropriate subcode. 

\item In Section \ref{sec:linear}, we present faster near-linear time implementations for some of the list-decoding algorithms presented in the previous sections. 
\item In Section 
\ref{sec:local}, we present local list decoding algorithms for multivariate polynomial-based codes that are able to decode individual codeword entries in sublinear time. 

\end{itemize}
Some of the material presented in this survey is also covered by other surveys. 
Specifically, the unique decoding algorithm for Reed-Solomon Codes presented in Section \ref{subsec:rs_unique}
is also covered in the online textbook \cite[Section 12.1]{GRS_survey} about error-correcting codes, and the list-decoding algorithms for Reed-Solomon Codes presented in Sections \ref{subsec:rs_list} and \ref{subsec:rs_johnson}
are also covered in 
 \cite[Section 12.2]{GRS_survey}, in Guruswami's survey on algorithmic list decoding
\cite[Section 4]{Gur-survey} and in Sudan's recent survey \cite[Sections 4 and 5]{Sudan-survey-2025} on the applications of algebra in algorithmic coding theory.
The algorithm presented in Section \ref{subsubsec:mult_no_mult} for list decoding of multiplicity codes beyond unique decoding radius follows the presentation in \cite[Section 17]{GRS_survey} (also see \cite[Section 6]{Sudan-survey-2025}) for the related family of Folded Reed-Solomon Codes.  
The local correction algorithm for Reed-Muller Codes presented in Section \ref{subsec:rm_unique} is also covered in Yekhanin's survey on locally decodable codes \cite[Section 2]{Yekhanin_survey}.

This survey focuses solely on list-decoding properties of polynomial-based codes. Parallel to the time of writing of this survey, it was discovered that certain families of \emph{expander-based codes} have list-decoding properties competitive to those of polynomial-based codes, in particular they can be list-decoded up to capacity with optimal list-size over a constant-size alphabet, and in nearly-linear time \cite{st25, JMST25}.

Finally, we do not discuss in this survey any of the applications of list-decodable codes. 
Some applications in coding theory are surveyed in \cite[Section 11]{Guruswami-Thesis}, 
and  applications in theoretical computer science are surveyed in  \cite{Sudan_list_dec_survey}, \cite[Section 4]{Trevisan_survey}, \cite[Section 12]{Guruswami-Thesis}, 
\cite{Guruswami_list_dec_app_survey}, \cite[Section 5]{Vadhan_survey}, and \cite[Sections 20 and 24]{GRS_survey}.
List recovery, a certain generalization of the notion  
of list-decoding that is useful for applications in coding theory and theoretical computer science, is surveyed in \cite{RV_survey}.

\section{Notation and definitions}\label{sec:prelim}

In this section, we first present some basic notations, and then we provide formal definitions and state some basic facts about list-decodable codes and low-degree polynomials in Sections \ref{subsec:prelim_list} and \ref{subsec:prelim_poly}, respectively. Finally, in Section \ref{subsec:prelim_poly_codes} we introduce some families of polynomial-based codes that will be the focus of this survey.

\paragraph{Notations.} For a pair of strings $u,v \in\Sigma^{n}$, the \textsf{(Hamming) distance} between $u$ and $v$
is the number of entries on which $u$ and $v$ defer, and is denoted
 $\Delta(u,v):= \left|\left\{ i\in\left[n\right]: u_i\ne v_i \right\} \right|$, and the
\textsf{relative (Hamming) 
distance} between $u$ and $v$ is the \emph{fraction} of entries on which $u$ and $v$
differ, and is denoted by $\dist(u,v):=\frac {\Delta(u,v)} {n}$.  We say that a string $u\in \Sigma^n$ is \textsf{$\alpha$-close} (\textsf{$\alpha$-far}, respectively) to a string $v \in \Sigma^n$ if $\dist(v,u)\leq \alpha$ ($\dist(v,u) > \alpha$, respectively).
For $w \in \Sigma^n$ and $\alpha \in (0,1)$, we define the \textsf{Hamming ball of radius $\alpha$ centered at $w$}
 as the set $B(w,\alpha):=\{u \in \Sigma^n \mid \dist(u,w) \leq \alpha \}$. 
 
 For any prime power $q$,  we denote by $\F_{q}$ the finite field of $q$ elements. 
For a  field $\F$ and $u \in \F^n$, the \textsf{(Hamming) weight}
 $|u|$ of $u$ is the number of non-zero entries of $u$, that is, $|u|:= \left|\left\{ i\in\left[n\right]:u_i \ne 0\right\} \right|$. 
Throughout, we let $\N:=\{0,1,2, \ldots \}$ denote the set of non-negative integers. 
For a vector $\bi =(i_1, \ldots , i_m) \in \N^m$, the \textsf{weight} $\wt(\bi)$ of $\bi$, equals $\wt(\bi)=\sum_{j=1}^m i_j$, and for $\bi=(i_1, \ldots, i_m), \bj=(j_1, \ldots, j_m)\in \N^m$, we let ${\bi \choose \bj}:= \prod_{k=1}^m {i_k \choose j_k}$.

\subsection{List-decodable codes}\label{subsec:prelim_list}

Before formally defining list-decodable codes, we start with some basic notation and definitions about error-correcting codes. 

An \textsf{(Error-correcting) code} is a subset $C\subseteq\Sigma^{n}$.  We call $\Sigma$ and $n$ the \textsf{alphabet} and the
\textsf{block length} of the code, respectively, and the elements of $C$ are called \textsf{codewords}. 
 The \textsf{rate}
of a code $C \subseteq \Sigma^n$ is the ratio $R:=\frac{\log(|C|)}{n \cdot \log(|\Sigma|)}$, the \textsf{(Hamming) distance} of $C$ is $\Delta(C):= \min_{c \neq c' \in C} \Delta(c,c')$, and its \textsf{relative distance} is $\dist(C):= \frac{\Delta(C) } {n}$.

We say that a code $C\subseteq\Sigma^{n}$ is \textsf{$\F$-linear} if $\Sigma = \F^s$ 
for some  field $\F$ and a positive integer $s$, and $C$ is an $\F$-linear subspace of $\Sigma^n$.
If $\Sigma = \F$, and $C$ is an $\F$-linear code,  
 then we simply say that $C$ is \textsf{linear}. For an $\F$-linear code $C\subseteq\Sigma^{n}$, the rate of $C$ equals $\frac{\dim_{\F}(C)}{n\cdot\dim_{\F}(\Sigma)}$, and the distance of $C$ equals $\min_{0 \neq c \in C} |c|$. 
A \textsf{generator matrix} for a linear code $C$ of dimension $k$ is a (full-rank) matrix $G \in \F^{n \times k}$ so that $\img(G) = C$, and a \textsf{parity-check matrix} for $C$ is a (full-rank) matrix $H \in \F^{ (n-k) \times n}$ so that $\ker(H)= C$. The \textsf{dual code} of $C$ is the code $C^\perp \subseteq \F^n$ containing all strings $c' \in \F^{n}$ satisfying that $\sum_{i=1}^n c'_i \cdot c_i = 0$ for all $c \in C$. 
It follows by definition that $(C^{\perp})^{\perp} = C$, and that $H$ is a parity-check matrix for $C$ if and only if $H^T$ is a generator matrix for $C^{\perp}$. 

For a code $C \subseteq \Sigma^n$ of relative
distance $\delta$, and a given $\alpha<\frac \delta 2$,  the \textsf{problem of (unique) decoding from an $\alpha$-fraction
of errors} is the task of finding, given a string $w \in \Sigma^n$, the unique $c\in C$ (if any) which
satisfies $\dist(c,w)\leq\alpha$.

List decoding is a paradigm that allows one to correct more than a $\frac \delta 2$ fraction of errors by returning a small list of close-by codewords. 
More formally, for $\alpha \in (0,1)$ and a positive integer $L$ we say that a code $C \subseteq \Sigma^n$ is \textsf{$(\alpha,L)$-list decodable} if for any $w \in \Sigma^n$ there are at most $L$ different codewords $c \in C$ which satisfy that $\dist(c,w)\leq \alpha$. The \textsf{problem of list-decoding from an $\alpha$-fraction
of errors} is the task of finding, given a string $w \in \Sigma^n$, the list of codewords $c \in C$ which
satisfy that $\dist(c,w)\leq\alpha$. 

\paragraph{The Johnson Bound.} The Johnson Bound states that \emph{any} code of relative distance $\delta$ is list-decodable from a fraction of 
$ \alpha < 1 - \sqrt{1-\delta}$ errors with a list size which depends on the proximity of $\alpha$ to $1 - \sqrt{1-\delta}$. Note that $ 1- \sqrt{1-\delta} \in (\frac \delta 2, \delta)$ for any $\delta \in (0,1)$, so this bound improves on the unique decoding bound. 

\begin{theorem}[Johnson Bound]\label{thm:johnson}
Let $C \subseteq \Sigma^{n}$ be a code of relative distance $\delta$.
Then for any $\alpha<1-\sqrt{1-\delta}$, $C$ is $(\alpha,L)$-list decodable  with list size 
$L = \frac{\delta - \alpha}{(1-\alpha)^{2}-(1-\delta)}$.
 \end{theorem}
 
\begin{proof}

Let $w \in \Sigma^n$ be a string, and let $\mathcal{L}  =\{ c\in C \mid \dist(c,w) \leq \alpha\}$. Our goal will be to show that 
$ L:=\mid \mathcal{L} \mid \leq \frac {\delta -\alpha} {(1-\alpha)^2 -  (1-\delta)}$. The proof follows by providing upper and lower bounds on the  average agreement between a pair of distinct codewords in the list, given by  
$\Pr_{c \neq c' \in \mathcal{L}, i \in [n]} \left[c_i = c'_i \right],$
and comparing the two bounds. 

For the upper bound, note that since $C$ has relative distance $\delta$, we have that $\dist(c,c') \geq \delta$ for any distinct $c, c' \in \mathcal{L}$.  Consequently, we have that
\begin{equation}\label{eq:johnson1}
\Pr_{ c \neq c' \in \mathcal{L}, i \in [n]} \left[c_i = c'_i\right] \leq 1-\delta.
\end{equation}

 Next we show a lower bound on the above expression. For $i \in [n]$, let $t_i$ be the number of codewords in $\mathcal{L}$ which agree with $w$ on the $i$-th entry, that is,
$ t_i : = | \{  c \in \mathcal{L} \mid    c_i = w_i   \} |.$ Note that since $\dist(c, w)\leq \alpha$ for any $c \in \mathcal{L}$, we have that $\E_{i \in [n]} [t_i] \geq L (1- \alpha)$. Consequently, we have that:
\begin{eqnarray}\label{eq:johnson2}
 \Pr_{c \neq c' \in \mathcal{L}, i \in [n]} \left[c_i = c'_i \right] & \geq & \Pr_{c \neq c' \in \mathcal{L}, i \in [n]} \left[c_i= c'_i = w_i\right] \\
& \geq  & \frac{\E_{i \in [n]} {t_i\choose 2} } { {L \choose 2}} 
\geq  \frac{ {\E_{i \in [n]} t_i \choose 2} } { {L \choose 2}} 
 \geq  \frac{ {L (1-\alpha) \choose 2}} { { {L \choose 2}} }, \nonumber
\end{eqnarray}
where the one before last inequality follows  by convexity.

The combination of the bounds given in (\ref{eq:johnson1}) and (\ref{eq:johnson2}) gives the following inequality:
$$ {L (1-\alpha) \choose 2}   \leq  {L \choose 2} \cdot (1- \delta),$$
which rearranging gives the desired bound of $L \leq \frac {\delta -\alpha} {(1-\alpha)^2 -  (1-\delta)}$.

	\end{proof}

\paragraph{The generalized Singleton Bound.}

The classical \textsf{Singleton Bound} implies that any code of rate $R$ and relative distance $\delta$ must satisfy that $\delta \leq 1-R$.
The \textsf{generalized Singleton Bound} extends this fact by showing that any $(\alpha,L)$-list decodable code of rate $R$  must satisfy that $\alpha \lesssim \frac{L} {L+1} \cdot (1- R)$. Note that the classical Singleton Bound for unique decoding  corresponds to the special case of $L=1$, in which case $\alpha \leq \frac \delta 2 \leq \frac {1-R} 2 $. 

\begin{theorem}[Generalized Singleton Bound]\label{thm:gen_singleton}
Let $C \subseteq \Sigma^{n}$ be a code of rate $R$ that is $(\alpha,L)$-list decodable. Then
$$
\alpha \leq 
 \frac{L} {L+1} \cdot \left(1- R + o_L(1) \right),
 $$
 where for a fixed $L$, the term $o_L(1)$ tends to zero as the block length $n$ tends to infinity.
\end{theorem}

\begin{proof}
Let $b:= \frac{L+1} {L} \cdot \alpha n$, for simplicity assume that $b$ is an integer.

We first claim that for any string $u \in \Sigma^{n-b}$, there are at most $L$ distinct codewords in $C$ whose prefix is $u$. To see this, suppose on the contrary that there exists a string $u \in \Sigma^{n-b}$ and $L+1$ distinct codewords $c_0,c_1, \ldots, c_L \in C$ whose prefix is $u$. Let $w \in \Sigma^n$ be a string whose first $n-b$ entries agree with $u$, and for $i=0,1, \ldots, L$, $w$ and $c_i$ agree on the $i$-th subsequent block of length $\frac {b} {L+1}$ (once more, assume that $\frac {b} {L+1}$  is an integer for simplicity). Then by construction for any $i \in \{0,1,\ldots, L\}$, 
$$\dist(w, c_i) \leq \frac{L} {L+1} \cdot \frac b n= \alpha,$$
which contradicts the assumption that $C$ is $(\alpha,L)$-list decodable.

Now, since for any string $u \in \Sigma^{n-b}$, there are at most $L$ distinct codewords in $C$ whose prefix is $u$, we have that $|C| \leq \Sigma^{n-b} \cdot L$, 
which rearranging and taking logarithms gives
$$  \frac b n \leq 1 -  \frac{\log(|C|)} {n \cdot  \log(|\Sigma|)}  + \frac{\log (L)} {n \cdot \log(|\Sigma|)}.$$
Plugging into the above inequality the values $R = \frac{\log(|C|)} {n \cdot  \log(|\Sigma|)}$ and $b = \frac{L+1} {L} \cdot \alpha n$ gives
 $$
\alpha \leq 
\frac{L} {L+1} \cdot \left(1- R + \frac{\log (L)} {n \cdot \log(|\Sigma|)} \right) =
 \frac{L} {L+1} \cdot \left(1- R + o_L(1) \right),
 $$
 which completes the proof of the theorem.
\end{proof}

The above generalized Singleton Bound in particular implies that $(\alpha,L)$-list decodable codes must have $\alpha \leq 1 - R- \epsilon$ for some $\epsilon = \epsilon (L)$, which approaches zero as $L$ grows (specifically, $\epsilon = O(\frac 1L)$, or equivalently, $L=O(\frac 1 \epsilon)$).  We say that a code $C \subseteq \Sigma^n$ attains  \textsf{list-decoding capacity} if it matches this coarser bound, i.e., it is $(1-R- \epsilon, O_\epsilon(1))$-list decodable for any $\epsilon >0$\footnote{We will sometimes abuse notation and say that the code achieves list-decoding capacity even if the list is bounded as a function of the block length.}

\subsection{Polynomials over finite fields}\label{subsec:prelim_poly}

This survey is mainly concerned with algebraic codes, based on low-degree polynomials over finite fields. In what follows, we introduce some definitions and gather several known facts about polynomials over finite fields that will be used for defining these codes and analyzing their list decoding properties.

\subsubsection{Univariate polynomials}

For a field $\F$, we let $\F[X]$ denote the \emph{ring} of univariate polynomials over $\F$, and we let $\F(X)$ denote the \emph{field} of rational functions over $\F$ in the indeterminate $X$.  
The  \textsf{degree} $\deg(f)$ of a polynomial $f(X) = \sum_{i\in \N} f_i X^i \in  \F[X]$ is the maximum $i$ so that $f_i \neq 0$. 
 We let $\F_{<d}[X]$ ($\F_{\leq d}[X]$, respectively) denote the linear space of  all univariate polynomials over $\F$ of degree smaller than  $d$ (at most  $d$, respectively). A \textsf{root} of a polynomial $f(X) \in \F[X]$ is a point $a\in \F$ so that $f(a)=0$. 
 It is well-known that a  non-zero polynomial $f(X) \in \F[X]$ of degree $d$ has at most $d$ roots in $\F$. 

It will sometimes be useful for us to also count roots with \emph{multiplicities}. To define this notion, we first need to define the notion of a \textsf{(Hasse) derivative}, which is a variant of derivatives that is more suitable for finite fields. Let $f(X) \in  \F[X]$ be a non-zero univariate polynomial over a field $\F$. 
For a non-negative integer $i$,  the \textsf{$i$'th order (Hasse) derivative} $f^{(i)}(X)$ of $f$ is defined as the coefficient of $Z^i$ in the expansion
\[ f(X + Z) = \sum_{i \in \N} f^{(i)}(X) Z^{i}. \] 
In particular, for any $a \in \F$,
substituting $a, X-a$ into $X, Z$ respectively in the above expression gives the expansion
\begin{equation}\label{eq:prelim_hasse}
 f(X) = \sum_{i \in \N} f^{(i)}(a) (X-a)^{i}.
 \end{equation}
The above expansion should be compared with the more familiar Taylor Expansion with respect to classical derivatives, where the coefficient of $(X-a)^{i}$ is divided by an additional factor of $i!$. 
Eliminating this additional factor of $i!$, which could be zero over a finite field, turns out to be convenient, and tends to make theorem statements cleaner. 

The \textsf{multiplicity} of a polynomial $f(X) \in \F[X]$ at a point $a \in \F$, denoted $\mult(f, a)$, is the largest integer $s$ such that for any non-negative integer $i <s$, we have that $f^{(i)} (a) =0$. Note that $\mult(f, a) \geq 0$ for any $a \in \F$, and that $\mult(f, a) > 0$ for any root $a$ of $f$.
The following fact, which is a direct consequence of the expansion given in (\ref{eq:prelim_hasse}), generalizes the well-known bound on the number of roots of a low-degree univariate polynomial, taking into account also multiplicities. 

\begin{fact}\label{fact:prelim_hasse_sz_univ}
Let $f(X) \in  \F[X]$ be a non-zero univariate polynomial of degree $d$. Then
$$\sum_{a \in \F} \mult(f, a) \leq d .$$
In particular, for any positive integer $s$, $\mult(f, \ba) \geq s$ for at most a
 $\frac{d} {s|\F|}$-fraction of the  points $a \in \F$.
\end{fact}
Note that the $s=1$ case of the 'in particular' part of the above fact corresponds to the standard bound on the number of roots of a low-degree univariate polynomial.

We note the following basic properties of the Hasse Derivative, which follow directly from the definition of the Hasse derivative.
\begin{lemma}\label{lem:prelim_hasse_properties_univ}
The following holds for any field $\F$, univariate polynomials $f(X), g(X) \in \F[X]$, and non-negative integers $i,j$:
 \begin{enumerate}
\item (Linearity) $(f(X)+ g(X))^{(i)} = f^{(i)}(X) + g^{(i)}(X)$.
\item (Product rule) $(f(X) \cdot g(X))^{(i)} = \sum_{k=0}^{i}  f^{(k)}(X) \cdot g^{(i-k)}(X)$.
\item (Iterative application) $(f^{(i)}(X))^{(j)} = {i+j \choose i} \cdot f^{(i+j)}(X) $.
\end{enumerate}
\end{lemma}
Note that, as opposed to classical derivatives, the Hasse derivative is only iterative up to multiplication by an additional binomial coefficient, but
on the other hand, the product rule does not require multiplication by an additional binomial coefficient. 

Finally, for univariate polynomials $f(X), h(X)\in \F[X]$, we use $\left(f(X) \mod  h(X) \right)$ to denote the unique remainder obtained by dividing $f(X)$ by $h(X)$.  We also say that $f(X)$ and $g(X)$ are \textsf{congruent modulo $h(X)$} (denoted by $(f(X) \equiv g(X) \mod h(X))$) if  $(f(X)-g(X))$ is divisible by $h(X)$.

\subsubsection{Multivariate polynomials}

The notions defined above can also be generalized to the setting of multivariate polynomials. 
Specifically, for a field $\F$ and a positive integer $m$, we let $\F[X_1, \ldots, X_m]$ denote the \emph{ring} of $m$-variate polynomials over $\F$, and we let $\F(X_1, \ldots, X_m)$ denote the \emph{field} of rational functions over $\F$ in the indeterminates $X_1, \ldots, X_m$. For ease of notation, we let $\bX:=(X_1, X_2, \ldots)$, and throughout the survey we also use the convention that bold letters $\ba, \bb, ...$ denote vectors over $\F$. For $\bi =(i_1, \ldots, i_m) \in \N^m$, we let $\bX^{\bi}$ denote the monomial $\bX^{\bi}:=X_1^{i_1} \cdots X_m^{i_m}$. 

The \textsf{(total) degree} $\deg(f)$ of a polynomial $f(\bX) = \sum_{\bi} f_{\bi} \bX^{\bi} \in  \F[\bX]$ is the maximum weight $\wt(\bi) = \sum_j {\bi}_j$ of $\bi$ so that $f_\bi \neq 0$.
As before, we let $\F_{<d}[X_1, \ldots,X_m]$ ($\F_{\leq d}[X_1, \ldots, X_m]$, respectively) denote the linear space of  all $m$-variate polynomials over $\F$ of (total) degree smaller than  $d$ (at most  $d$, respectively). In the multivariate setting, it will sometimes be useful for us to also consider  a notion of weighted degree. Specifically, for a vector $\bb \in \N^m$, 
the \textsf{$\bb$-weighted degree $\deg_{\bb}(f)$} of a polynomial $f(\bX) = \sum_{\bi} f_{\bi} \bX^{\bi} \in  \F[\bX]$  is the maximum weighted sum $\sum_j \bb_j \cdot {\bi}_j$ of $\bi$ so that $f_\bi \neq 0$. 

As in the univariate setting, we define a \textsf{root} of a polynomial $f(\bX) \in \F[\bX]$ as a point $\ba$ so that $f(\ba)=0$. The following well-known \textsf{Schwartz-Zippel Lemma} generalizes the well-known bound on the number of roots of a low-degree univariate polynomial to the multivariate setting. 

\begin{lemma}[\cite{Schwartz80, Zippel79,DL78}]\label{lem:prelim_sz}
Let $f(X_1, \ldots, X_m) \in  \F[X_1, \ldots, X_m]$ be a non-zero $m$-variate polynomial over a finite field $\F$ of (total) degree $d$. Then $f$ has at most $d \cdot |\F|^{m-1}$ roots in $\F^m$.
\end{lemma}

The notion of a Hasse derivative can also be extended to the multivariate setting as follows. Let $f(\bX) \in  \F[\bX]$ be a non-zero $m$-variate polynomial over a field $\F$.
For a vector $\bi \in \N^m$,  we define the \textsf{$\bi$'th (Hasse) derivative} $f^{(\bi)}(\bX)$ of $f$ as the coefficient of $\bZ^\bi$ in the expansion
\[ f(\bX + \bZ) = \sum_{\bi \in \N^m} f^{(\bi)}(\bX) \bZ^{\bi}. \] 
We define the \textsf{multiplicity} of $f$ at a point $\ba \in \F^m$, denoted $\mult(f, \ba)$, to be the largest integer $s$ such that for any $\bi \in \N^m$ with $\wt(\bi) <s$, we have that $f^{(\bi)} (\ba) =0$. 

The following theorem extends the above Lemma \ref{lem:prelim_sz} to also account  for multiplicities of roots.

\begin{theorem}[\cite{DKSS13}, Lemma 2.7]\label{thm:prelim_hasse_sz}
Let $f(\bX) \in  \F[\bX]$ be a non-zero $m$-variate polynomial over a finite field $\F$ of (total) degree $d$. Then
$$\sum_{\ba \in \F^m} \mult(f, \ba) \leq d \cdot |\F|^{m-1}.$$
In particular, for any positive integer $s$, $\mult(f, \ba) \geq s$ for at most a
 $\frac{d} {s|\F|}$-fraction of the  points $\ba \in \F^m$.
\end{theorem}
Note that the $m=1$ case of the above theorem corresponds to the above Fact \ref{fact:prelim_hasse_sz_univ}, while the $s=1$ case of the 'in particular' part corresponds to the above Lemma \ref{lem:prelim_sz}.

We also have the following extension of Lemma \ref{lem:prelim_hasse_properties_univ}. 
\begin{lemma}\label{lem:prelim_hasse_properties}
The following holds for any field $\F$, $m$-variate polynomials $f(\bX), g(\bX) \in \F[\bX]$, and $\bi, \bj \in \N^m$:
 \begin{enumerate}
\item (Linearity) $(f(\bX)+ g(\bX))^{(\bi)} = f^{(\bi)}(\bX) + g^{(\bi)}(\bX)$.
\item (Product rule) $(f(\bX) \cdot g(\bX))^{(\bi)} = \sum_{\mathbf{r}, \mathbf{k} \;:\; \mathbf{r}+ \mathbf{k} = \bi}  f^{(\mathbf{r})}(\bX) \cdot g^{(\mathbf{k})}(\bX)$.
\item (Iterative application) $(f^{(\bi)}(\bX))^{(\bj)} = {\bi+\bj \choose \bi} \cdot f^{(\bi+\bj)}(\bX) $.
\end{enumerate}
\end{lemma}

\subsection{Some families of polynomial codes}\label{subsec:prelim_poly_codes}

Next we introduce some families of polynomial codes that will be the focus of this survey. All codes are defined over a finite field $\F_q$ of $q$ elements, where $q$ is some prime power. 
We begin with the most basic family of Reed-Solomon Codes. 

\paragraph{Reed-Solomon (RS) Codes.}
Let $a_1, \ldots, a_n$ be distinct points in $\F_q$ (called \textsf{evaluation points}), and let $k < n$ be a positive integer (the \textsf{degree parameter}).  
The \textsf{Reed-Solomon (RS) Code} $\RS_q(a_1,\ldots,a_n;k)$ is a code 
which associates with any univariate polynomial $f(X) \in \F_q[X]$ of degree smaller than $k$ a codeword $(f(a_1), \ldots, f(a_n)) \in \F_q^n.$ In the special case where $n =q$, we let  $\RS_q(k):=\RS_q(a_1,\ldots,a_n;k)$.

It follows by definition that the Reed-Solomon Code $\RS_q(a_1,\ldots,a_n;k)$ is a linear code of rate $\frac k n$. Furthermore, since any two distinct univariate polynomials of degree smaller than $k$ can agree on at most $k-1$ points in $\F_q$ (since their difference is a non-zero univariate polynomial of degree smaller than $k$, and so has at most $k-1$ roots in $\F_q$), it follows that $\RS_q(a_1,\ldots,a_n;k)$ has distance at least $n-k+1$, and relative distance at least $1- \frac k n$. So the Reed-Solomon Code attains the Singleton Bound (and so is an MDS code). 

\medskip

Another family of polynomial codes that will be at the focus of this survey are \emph{multiplicity codes}, which extend Reed-Solomon Codes by also evaluating  \emph{derivatives} of polynomials.

\paragraph{Multiplicity codes.}

As before, let $a_1, \ldots, a_n$ be distinct evaluation points in $\F_q$, and let $s$ be a positive integer (the \textsf{multiplicity parameter}). Setting the multiplicity parameter $s$ to be large allows us to choose a degree parameter that is larger than the block length, specifically, we can choose the degree parameter $k$ to be any positive integer so that $k<sn$. 
The \textsf{(univariate) multiplicity code} $\MULT_q^{(s)}(a_1,\ldots,a_n;k)$ is defined similarly to the corresponding Reed-Solomon Code $\RS_q(a_1,\ldots,a_n;k)$, except that  each codeword of the form $(f(a_1), \ldots, f(a_n))$ is now replaced with the codeword $\left(f^{(<s)}(a_1),  \ldots, f^{(<s)}(a_n)\right)$, where for $a \in \F_q$, we use $f^{(<s)}(a) \in \F_q^s$  
to denote the vector which includes all (Hasse) derivatives of $f$ at $a$ of order smaller than $s$, that is, $f^{(<s)}(a) = (f(a), f^{(1)}(a), \ldots, f^{(s-1)}(a))$. 

It follows by definition that the multiplicity code $\MULT_q^{(s)}(a_1,\ldots,a_n;k)$ is an $\F_q$-linear code over the alphabet $\F_q^s$ of  rate $\frac{k} {sn}$. Furthermore, by Fact \ref{fact:prelim_hasse_sz_univ}, $\MULT_q^{(s)}(a_1,\ldots,a_n;k)$ has distance at least $n-\frac {k-1} s$, and relative distance at least $1-\frac {k} {sn}$. Finally, note that Reed-Solomon Codes correspond to the special case of multiplicity parameter $s=1$.

\medskip

Next we present multivariate extensions of the above codes. The first such extension are the \emph{Reed-Muller  Codes} which extend Reed-Solomon codes to the setting of multivariate polynomials. 

\paragraph{Reed-Muller (RM) Codes.} 
For simplicity, we only present here the definition of Reed-Muller codes evaluated over the whole field. 
More specifically, as in the case of Reed-Solomon Codes, let $k$ be a degree parameter so that $k<q$, and let $m$ be a positive integer. 
The \textsf{Reed-Muller ($\RM$) Code}  $\RM_{q,m}(k)$ is a code 
which associates with any $m$-variate polynomial $f(X_1, \ldots, X_m) \in \F_q[X_1, \ldots, X_m]$ of degree smaller than $k$ a codeword  $(f(\ba))_{\ba \in \F_q^m}.$

It follows by definition that the Reed-Muller Code $\RM_{q,m}(k)$ is a linear code of block length $q^m$. Its dimension is the number of monomials in $m$ variables of degree smaller than $k$ which equals ${m+k-1 \choose m}$, and so this code has rate $\frac{{m+k-1 \choose m}} {q^m}$. Furthermore, by the Schwartz-Zippel Lemma (Lemma \ref{lem:prelim_sz}), $\RM_{q,m}(k)$ has  relative distance at least $1-\frac {k} {q}$. 
Further note that the $m=1$ case of the Reed-Muller Code defined above corresponds to the Reed-Solomon code $\RS_{q}(k)$, evaluated over the whole field.
 
 \bigskip
 
Finally, we also present an extension of multiplicity codes to the multivariate setting.

\paragraph{Multivariate multiplicity codes.} Once more, for simplicity, we only present here the definition of multivariate multiplicity codes evaluated over the whole field. 
As in the case of univariate multiplicity codes, let $k$ be a degree parameter and let  $s$ be a multiplicity parameter so that $k<s \cdot q$, and as in the case of Reed-Muller codes, let $m$ be a positive integer. The \textsf{multivariate multiplicity code} $\MULT_{q,m}^{(s)}(k)$ is defined similarly to the corresponding Reed-Muller Code $\RM_{q,m}(k)$, except that  each codeword of the form $(f(\ba))_{\ba \in \F_q^m}$ is now replaced with the codeword $(f^{(<s)}(\ba))_{\ba \in \F_q^m}$, where similarly to the univariate multiplicity code case, we let $f^{(<s)}(\ba)$ 
denote the vector which includes all $m$-variate (Hasse) derivatives of $f$ at $\ba$ of order smaller than $s$, that is, $f^{(<s)}(\ba) = (f^{(\bi)}(\ba))_{\bi \in \N^m \;:\; \wt(\bi)<s}$. Note that the length of this vector is the number of vectors $\bi \in \N^m$ of weight smaller than $s$ which equals ${m+s-1 \choose m}$. 

It follows by definition that the multivariate multiplicity code $\MULT^{(s)}_{q,m}(k)$ is an $\F_q$-linear code over the alphabet $\F_q^{ {m+s-1 \choose m}}$ and of dimension ${m+k-1 \choose m}$, and so has rate $\frac{ {m+k-1 \choose m} } { {m+s-1 \choose m} \cdot q^m}$. Furthermore, by Theorem \ref{thm:prelim_hasse_sz}, $\MULT^{(s)}_{q,m}(k)$ has relative distance at least $1-\frac {k} {sq}$. 

\subsection{Bibliographic notes}

\paragraph{List decoding.}
The model of \textsf{list decoding} was first introduced by Elias \cite{Elias57} and Wozencraft \cite{wozencraft58}. 
The \textsf{Johnson Bound} is a classical bound in coding theory, and a bound in a similar form to the one stated in Theorem \ref{thm:johnson} appeared in the context of list decoding in 
\cite[Section 4.1]{GRS00}. Our proof of Theorem \ref{thm:johnson} follows the proof of  \cite[Theorem 3.1]{Gur-survey} which is attributed there to Radhakrishnan. The \textsf{Singleton Bound} is another classical bound in coding theory that is attributed to Singleton \cite{Singleton64}, but also previously appeared  in
\cite{Joshi58}. The \textsf{Generalized Singleton Bound} was stated and proved much more recently by Shangguan and Tamo \cite{ST20}, and further improvements and extensions of this bound were given by Roth \cite{Roth22} and by Goldberg, Shangguan and Tamo \cite{GST24}.

 In this section we only stated and proved simplified \emph{alphabet-independent}  versions of both the Johnson Bound and the generalized Singleton Bound, as our focus in this survey is on polynomial-based codes that are defined over a large alphabet, and achieve the best possible list-decoding radius. Alphabet-dependent versions of the Johnson Bound can be found in \cite[Section 3.1]{Gur-survey}. It is also known that over a $q$-ary alphabet, the \textsf{list-decoding capacity} (i.e., the best possible list-decoding radius for codes of rate $R$) is $H_q^{-1}(1- R)$, where
$$
H_q(\alpha)=\alpha\log_q(q-1)+\alpha\log_q\left(\frac 1 \alpha\right) +(1-\alpha)\log_q\left(\frac 1 {1-\alpha}\right)
$$
 denotes the \emph{$q$-ary entropy function} (see e.g., \cite[Section 3.2]{Gur-survey}). It can be checked that  $H_q^{-1}(1- R)$ approaches $1-R$ for a growing alphabet.

\paragraph{Polynomial-based codes.}
  Theorem \ref{thm:prelim_hasse_sz} which bounds the number of roots of a low-degree polynomial, \emph{counted with multiplicities}, was stated and proved by Dvir, Kopparty, Saraf, and Sudan \cite{DKSS13}. This bound extends the classical bound for the number of roots without multiplicities, stated in Lemma \ref{lem:prelim_sz}, that was proven independently by Schwartz \cite{Schwartz80}, Zippel \cite{Zippel79}, and DeMillo and Lipton \cite{DL78}. 

   \textsf{Reed-Solomon Codes} are classical codes that were introduced by Reed and Solomon \cite{RS60}, while \textsf{Reed-Muller Codes} were introduced by Muller \cite{Mul54}, and a fast decoder for these codes was presented by Reed \cite{Reed54}. The univariate version of \textsf{multiplicity codes} was introduced by Rosenbloom and Tsfasman \cite{RT97}, while their multivariate extension was introduced more recently by Kopparty, Saraf, and Yekhanin \cite{KSY14}.

\section{Algorithmic list decoding up to Johnson Bound}\label{sec:johnson}

In this section we focus on the most basic family of \emph{Reed-Solomon $(\RS)$  Codes}, and show how to efficiently list-decode these codes beyond the unique decoding radius, up to the \emph{Johnson Bound}  (cf.,  Theorem \ref{thm:johnson}). Note that up to this bound the list size is guaranteed to be constant, so the question is only \emph{computational}, namely, designing efficient list decoding algorithms which given a received word,  output the list of close-by codewords.

As a warmup, we first present in Section \ref{subsec:rs_unique} below an efficient (polynomial-time) algorithm for \emph{unique decoding} of Reed-Solomon Codes up to half their minimum distance. Then, in Section \ref{subsec:rs_list} we present an efficient algorithm for list decoding of Reed-Solomon Codes \emph{beyond} half their minimum distance, while in Section \ref{subsec:rs_johnson} we show how to extend this algorithm up to the \emph{Johnson Bound}.  
Finally, in Section we present combinatorial lower bounds on the list size of Reed-Solomon Codes beyond the Johnson bound, which in particular imply that there is no efficient algorithm which list decodes general Reed-Solomon codes well beyond the Johnson bound. 

\subsection{Unique decoding of Reed-Solomon Codes}\label{subsec:rs_unique}

In this section, we first present, as a warmup, an efficient (polynomial time) algorithm for \emph{unique decoding} of Reed-Solomon Codes up to half their minimum distance.

More specifically, fix  a prime power $q$, distinct evaluation points $a_1, \ldots, a_n \in \F_q$, and a positive integer $k<n$. Recall that the  \emph{Reed-Solomon $(\RS)$  Code} $\RS_{q}(a_1, \ldots, a_n; k)$ is the code which associates with each polynomial $f(X) \in \F_q[X]$ of degree smaller than $k$ a codeword $(f(a_1), \ldots, f(a_n)) \in \F_q^n.$
In order to avoid annoying rounding issues, in what follows we shall assume for simplicity that $k,n$ have the same parities.
Suppose that $w \in \F_q^n$ is a received word, and suppose that there exists a (unique) polynomial $f(X) \in \F_q[X]$ of degree smaller than $k$ so that $f(a_i) \neq w_i$ for at most $e :=\frac{n-k} {2}$ indices $i \in [n]$. That is, $f(a_i) = w_i$ for at least $t: =n-e= \frac{n+k} {2}$ indices $i \in [n]$. Below we shall present an algorithm which finds $f(X)$ (and so also the associated codeword $(f(a_1), f(a_2), \ldots, f(a_n))$). 
In particular, if we let $R:=\frac{k} {n}$ denote that rate of the $\RS_{q}(a_1, \ldots, a_n; k)$ code, then this algorithm will be able to decode from $\frac{1-R}{2}$ fraction of errors. 

\paragraph{Overview.}
As will typically be the case, the decoding algorithm is divided into two main steps. The first step is an \emph{interpolation} step in which we find a non-zero low-degree bivariate polynomial 
$Q(X,Y) \in \F_q[X,Y]$ so that $Q(a_i, w_i)=0$ for any $i \in [n]$. 
The second step is a \emph{root finding} step in which we find a univariate polynomial $g(X) \in \F_q[X]$ satisfying that $Q(X, g(X)) = 0$. To show correctness, we need to prove that such a non-zero low-degree polynomial $Q$ exists, and that $g(X) = f(X)$.
We also need to ensure that both steps can be performed efficiently (in time polynomial in $n$ and $\log (q)$). In what follows, we describe separately each of these steps, and analyze their correctness and efficiency.

\paragraph{Step 1: Interpolation.}

We start by describing the first step of interpolation, in which we would like to find a non-zero low-degree bivariate polynomial 
$Q(X,Y) \in \F_q[X,Y]$ so that $Q(a_i, w_i)=0$ for any $i \in [n]$. 

How can such a polynomial $Q$ be found? Let 
$$E(X) :=  \prod_{i : f(a_i) \neq w_i} \left(X-a_i\right)$$ 
denote the 
\textsf{error-locating polynomial} whose roots are all the error-locations, i.e., all evaluation points $a_i$ for which $f(a_i) \neq w_i$. 

The main observation is that if we let
$$
Q(X,Y) = E(X) \cdot f(X) - E(X) \cdot Y,
$$
then for any $i \in [n]$ we have that 
$$
Q(a_i, w_i) = E(a_i) \cdot f(a_i) -  E(a_i) \cdot w_i =0.
$$
To see that the above indeed holds, note that if $f(a_i) = w_i$, then we clearly have that $Q(a_i, w_i) =0$, and if  $f(a_i) \neq w_i$, then $E(a_i)=0$ by the definition of $E$, and so we also clearly have that $Q(a_i,w_i)=0$ in this case.

Motivated by the above, we would like to search for a non-zero bivariate polynomial $Q(X,Y)$ over $\F_q$ of the form
$$Q(X,Y) = A (X) +  B (X) \cdot Y ,$$ 
where $A(X)$ has degree at most $e+k-1= \frac{n+k} {2} -1$ and $B(X)$ has degree at most $e=\frac{n-k} {2}$, and which satisfies that $Q(a_i, w_i)=0$ for any $i \in [n]$. 
By the above, such a non-zero polynomial $Q(X,Y)$  exists by the choice of $A(X) = E(X) \cdot f(X) $ and $B(X) = -E(X)$. Next we give an alternative linear-algebraic argument for the existence of $Q$ that will later be useful when attempting to extend the unique decoding algorithm to the setting of list decoding.

\begin{lemma}\label{lem:rs_unique_interpolation}
There exists a non-zero bivariate polynomial $Q(X,Y)$ over $\F_q$ of the form
$$Q(X,Y) = A (X) +  B (X) \cdot Y ,$$ 
where $A(X)$ has degree at most $\frac{n+k} {2}-1$ and $B(X)$ has degree at most $\frac{n-k} {2}$, and which satisfies that $Q(a_i, w_i)=0$ for any $i \in [n]$. 

Furthermore, such a polynomial $Q(X,Y)$ can be found by solving a system of linear equations which is obtained by viewing the  coefficients of the monomials in $A(X)$ and $B(X)$ as unknowns, and viewing each requirement of the form $Q(a_i,w_i)= A(a_i)+ B(a_i) \cdot w_i=0$ as a homogeneous linear constraint on these unknowns.       
\end{lemma}

\begin{proof}
If we view the coefficients of the monomials in $A(X)$ and $B(X)$ as unknowns, then each requirement of the form $Q(a_i,w_i)= A(a_i)+ B(a_i) \cdot w_i=0$ imposes a homogeneous linear constraint on these unknowns. Further note that the number of unknowns is 
$(\deg(A)+1) + (\deg(B) +1) =n+1$,
while the number of linear constraints is $n$, and consequently this linear system has a non-zero solution.
\end{proof}

\paragraph{Step 2: Root finding.} In this step, we would like to find a univariate polynomial $g(X) \in \F_q[X]$ satisfying that $Q(X, g(X)) = 0$, and we would like to show that $g(X)=f(X)$. To this end, we first show that $Q(X,f(X)) =0$. This will follow by showing that the \emph{univariate} polynomial $P_f(X):=Q(X,f(X))$ has more roots than its degree.

\begin{lemma}\label{lem:rs_unique_P_is_root}
Suppose that $Q(X,Y)$ is a non-zero bivariate polynomial over $\F_q$ of the form
$$Q(X,Y) = A (X) +  B (X) \cdot Y ,$$ 
where $A(X)$ has degree at most $\frac{n+k} {2} - 1$ and $B(X)$ has degree at most $\frac{n-k} {2}$, and which satisfies that $Q(a_i, w_i)=0$ for any $i \in [n]$. 
Let $f(X) \in \F_q[X]$ be a univariate polynomial  of degree smaller than $k$ so that $f(a_i) = w_i$ for at least $\frac{n+k} {2}$ indices $i \in [n]$. Then $Q(X,f(X))=0$. 
\end{lemma}

\begin{proof}
The polynomial $P_f(X):=Q(X,f(X)) \in \F_q[X]$ is a univariate polynomial of degree at most
$ \frac{n+k} {2}-1$, which satisfies that 
$P_f(a_i)=Q(a_i, f(a_i)) = Q(a_i, w_i)=0$ for any $i \in [n]$ for which $f(a_i) = w_i$.  Thus $P_f(X)$ is a univariate polynomial of degree at most $\frac {n+k} {2}-1$ with at least $ \frac{n+k} {2}$ roots, and so it must be the zero polynomial. 
\end{proof}

The next lemma shows that $f(X)$ is the unique polynomial satisfying that $Q(X,f(X))=0$. 

\begin{lemma}\label{lem:rs_unique_solution_size}
Suppose that $Q(X,Y)$ is a non-zero bivariate polynomial over $\F_q$ of the form
$$Q(X,Y) = A (X) +  B (X) \cdot Y.$$
Then there exists at most one univariate polynomial $f(X) \in \F_q[X]$ so that  $Q(X,f(X))=0$. 
Furthermore, if $Q(X,f(X))=0$, then $B(X) \neq 0$ and $f(X) = -\frac {A(X)} {B(X)}$.
\end{lemma}

\begin{proof}
Assume that $Q(X,f(X))=0$ for some univariate polynomial  $f(X) \in \F_q[X]$.
We first show that  $B(X) \neq 0$.
Suppose on the contrary that $B(X) = 0$. Then by our assumption, we also have that $0=Q(X,f(X))=A(X)$, which implies in turn that $Q(X,Y)= A(X)+B(X) \cdot Y =0$, contradicting the assumption that $Q(X,Y)$ is non-zero. By assumption that $0=Q(X,f(X))= A(X)+B(X) \cdot f(X)$, this implies in turn that $f(X)=-\frac {A(X)} {B(X)}$, and so $f(X)$ is uniquely determined.
In particular, there exists at most one univariate polynomial $f(X) \in \F_q[X]$ so that  $Q(X,f(X))=0$.
\end{proof}

\medskip

The full description of the algorithm appears  in Figure \ref{fig:rs_unique} below, followed by correctness and efficiency analysis.

\begin{figure}[h]
  \begin{boxedminipage}{\textwidth} \small \medskip \noindent
    $\;$

    \underline{\textbf{Unique decoding of $\RS_{q}(a_1, \ldots, a_n;k)$:}}
    
    \medskip
 $\triangleright$ For simplicity assume that $n,k$ have the same parities.
\begin{itemize}
\item \textbf{INPUT:} $w\in \F_q^n$.
\item \textbf{OUTPUT:} A polynomial $f(X) \in \F_q[X]$ of degree smaller than $k$ so that $f(a_i) = w_i$ for at least $\frac{n+k} {2}$ indices $i \in [n]$ if such a polynomial exists, or $\bot$ otherwise.

\end{itemize}

    \begin{enumerate}
\item \label{step:rs_unique_interpolation} Find a non-zero bivariate polynomial $Q(X,Y)$ over $\F_q$ of the form
$$Q(X,Y) = A (X) +  B (X) \cdot Y ,$$ 
where $A(X)$ has degree at most $\frac{n+k} {2}-1$ and $B(X)$ has degree at most $\frac{n-k} {2}$, and which satisfies that $Q(a_i, w_i)=0$ for any $i \in [n]$.

Such a polynomial $Q(X,Y)$ can be found by solving a system of linear equations which is obtained by viewing the  coefficients of the monomials in $A(X)$ and $B(X)$ as unknowns, and viewing each requirement of the form $Q(a_i,w_i)= A(a_i)+ B(a_i) \cdot w_i=0$ as a homogeneous linear constraint on these unknowns.       

\item  \label{step:rs_unique_root_finding}
If $B(X) = 0$ return $\bot$; Otherwise, find  a univariate polynomial $f(X)$ so that $Q(X, f(X))=0$ by letting $f(X):=-\frac{A(X)} {B(X)}$.
\item \label{step:rs_unique_check} If $f(X)$ is a polynomial of degree smaller than $k$ so that $f(a_i) = w_i$ for at least $\frac{n+k} {2}$ indices $i \in [n]$,
return $(f(a_1),f(a_2),\ldots, f(a_n))$; Otherwise, return $\bot$. 

\end{enumerate}

  \medskip

  \end{boxedminipage}

\caption{Unique decoding of Reed-Solomon Codes}
\label{fig:rs_unique}
\end{figure}

\paragraph{Correctness.}  Suppose that $f(X) \in \F_q[X]$ is a univariate polynomial of degree smaller than $k$ so that $f(a_i) = w_i$ for at least $t :=\frac{n+k} {2}$ indices $i \in [n]$. By Lemma \ref{lem:rs_unique_interpolation}, there exists a non-zero bivariate polynomial $Q(X,Y)$ satisfying the requirements on Step \ref{step:rs_unique_interpolation}. By Lemma \ref{lem:rs_unique_P_is_root}, we have that $Q(X,f(X))=0$, and by Lemma \ref{lem:rs_unique_solution_size} this implies in turn that $B(X) \neq 0$ and $f(X)=- \frac{A(X)} {B(X)}$, and so the algorithm will compute the correct $f(X)$ on Step \ref{step:rs_unique_root_finding}. Clearly, $f(X)$ will also pass the checks on Step \ref{step:rs_unique_check}, and so the algorithm will return $f(X)$ as required.

\paragraph{Efficiency.} Step \ref{step:rs_unique_interpolation} can be performed in time $\poly(n, \log (q))$ by solving a system of $n$ linear equations in $n+1$ variables using Gaussian elimination. Step \ref{step:rs_unique_root_finding} can be performed in time $\poly(n,\log (q))$ by performing polynomial division using Euclid's algorithm. Step \ref{step:rs_unique_check} can also clearly be performed in time $\poly(n, \log (q))$. 

\subsection{List decoding Reed-Solomon Codes beyond unique decoding radius}\label{subsec:rs_list}

In this section, we present an efficient algorithm for list decoding of Reed-Solomon Codes \emph{beyond} half their minimum distance.

As before, fix  a prime power $q$, distinct evaluation points $a_1, \ldots, a_n \in \F_q$, and a positive integer $k<n$, and let $\RS_q(a_1, \ldots, a_n;k)$ be the corresponding Reed-Solomon Code. 
Suppose that $w\in \F_q^n$ is a received word. 
Our goal is to compute the \emph{list} of all polynomials $f(X) \in \F_q[X]$ of degree smaller than $k$ so that $f(a_i) \neq w_i$ for at most $e: = n-\sqrt{2 k n}-1$ indices $i \in [n]$. That is, $f(a_i) = w_i$ for at least $t: =n-e= \sqrt{2 k n}+1$ indices $i \in [n]$. 

In particular, if we let $R:=\frac{k} {n}$ denote the rate of the $\RS_{q}(a_1, \ldots, a_n;k)$ code, then this algorithm will be able to list-decodable from up to roughly $(1- \sqrt{2R})$-fraction of errors. It can be verified that this is strictly larger than the unique decoding radius of $\frac{1 -R} {2}$ whenever the rate $R$ is smaller than $0.17$.

\paragraph{Overview.}
Recall that in the unique decoding setting, the decoding algorithm was divided into two main steps: The interpolation step in which we searched  for a non-zero low-degree bivariate polynomial $Q(X,Y) \in \F_q[X,Y]$
of the form $Q(X,Y) = A (X) +  B (X) \cdot Y$  
so that $Q(a_i, w_i)=0$ for any $i \in [n]$, and the root finding step in which we searched for the (unique) univariate polynomial $f(X) \in \F_q[X]$ satisfying that $Q(X, f(X)) = 0$. 

Inspecting the proofs of  Lemmas \ref{lem:rs_unique_interpolation} and  \ref{lem:rs_unique_P_is_root}, we observe that the main properties we needed out of $Q$ 
were that the number of unknown coefficients in $Q$ is greater than the number $n$ of constraints of the form $Q(a_i,w_i)=0$, and that the $(1,k)$-\emph{weighted} $\deg_{(1,k)}(Q)$ of $Q$ is smaller than the \emph{agreement parameter} $t$.\footnote{Recall that for a vector $\bb \in \N^m$, 
the \emph{$\bb$-weighted degree $\deg_{\bb}(f)$} of a polynomial $f(\bX) = \sum_{\bi} f_{\bi} \bX^{\bi} \in  \F[\bX]$  is the maximum weighted sum $\sum_j \bb_j \cdot {\bi}_j$ of $\bi$ so that $f_\bi \neq 0$. } 
The first property guarantees the existence of a non-zero $Q$. The second property  guarantees that for any univariate polynomial $f(X) \in \F_q[X]$ of degree smaller than $k$ so that $f(a_i)=w_i$ for at least $t$ of the indices $i \in [n]$, it holds that
 $P_f(X):=Q(X,f(X))$ is a univariate polynomial of degree smaller than $t$ that has at least $t$ roots, and consequently $P_f(X)=Q(X,f(X))=0$.

In the list-decoding setting, the agreement parameter $t$ is smaller, which means that the $(1,k)$-weighted degree of $Q$ is also required to be smaller. Consequently, $Q$ may potentially have less than $n$ monomials, and so we may not be able to argue the existence of a non-zero polynomial $Q$. To cope with this, we now no longer require that  $Q$ has the special form 
$Q(X,Y) = A (X) +  B (X) \cdot Y$, but we only require that $Q$ has a low $(1,k)$-weighted degree, which can potentially increase the number of monomials in $Q$.
It turns out that this change indeed suffices for setting the agreement parameter $t$ to be significantly smaller than in the unique decoding setting. 
We also need to ensure that the root finding step can be performed efficiently when $Q$ has this general form. 

In what follows, we describe separately the changes made in the interpolation and root finding steps, and analyze their correctness and efficiency.

\paragraph{Step 1: Interpolation.}

In this step we would now like to find a non-zero bivariate polynomial 
$Q(X,Y) \in \F_q[X,Y]$ with low $(1,k)$-weighted degree, so that $Q(a_i, w_i)=0$ for any $i \in [n]$. As opposed to the unique decoding setting, we now do not impose any additional structure on $Q$ beyond the requirement on its weighted degree. The following lemma shows the existence of such a non-zero polynomial $Q$ with $(1,k)$-weighted degree that is significantly smaller than  in the unique decoding setting.

\begin{lemma}\label{lem:rs_list_interpolation}
There exists a non-zero bivariate polynomial $Q(X,Y)$ over $\F_q$ with $(1,k)$-weighted degree at most $\sqrt{2 kn}$, and which satisfies that $Q(a_i, w_i)=0$ for any $i \in [n]$. 
Furthermore, such a polynomial $Q(X,Y)$ can be found by solving a system of linear equations which is obtained by viewing the  coefficients of the monomials in $Q(X,Y)$ as unknowns, and viewing each requirement of the form $Q(a_i,w_i)=0$ as a homogeneous linear constraint on these unknowns.      
\end{lemma}

\begin{proof}
If we view the coefficients of the monomials in $Q$ as unknowns, then each requirement of the form $Q(a_i,w_i)=0$ imposes a homogeneous linear constraint on these unknowns. 
It thus suffices to show that the number of monomials in $Q$ is greater than the number of linear constraints $n$, and so this linear system has a non-zero solution.

Let $M$ denote the number of monomials in $Q(X,Y)$, and assume for simplicity that $\ell := \sqrt {2 n /k} $ is an integer. First note that by our requirement that $\deg_{(1,k)}(Q) \leq \ell \cdot k$,
 the degree of the $Y$ variable in $Q(X,Y)$ can be any $j \in \{0,1, \ldots, \ell \}$, and  for any monomial in which the degree of the $Y$ variable is $j$, the degree of the $X$ variable can be any $i \in \{0,1,\ldots,  k(\ell-j) \} $. 
Thus we have that
\begin{equation*}
M= \sum_{j=0}^{\ell} (k \ell - kj+1) 
=  \left(\frac{ k \ell } {2} +1\right) \cdot (\ell+1) 
 >   \frac {k \ell^2} {2} =n.
\end{equation*}

We conclude that the number of unknown coefficients in $Q$ is greater than the number of linear constraints $n$, and so there exists a non-zero polynomial $Q$ satisfying the requirements of the lemma.
\end{proof}

\paragraph{Step 2: Root finding.} In this step, we would like to find all univariate polynomials $f(X) \in \F_q[X]$ satisfying that $Q(X, f(X)) = 0$, and we would like to show that any 
univariate polynomial
  $f(X) \in \F_q[X]$ of degree smaller than $k$ so that  $f(a_i) = w_i$ for at least $t=\sqrt{ 2 k n}+1$ indices $i \in [n]$ satisfies that 
  $Q(X,f(X))=0$.

\begin{lemma}\label{lem:rs_list_P_is_root}
Suppose that $Q(X,Y)$ is a non-zero bivariate polynomial over $\F_q$ which satisfies that $\deg_{(1,k)}(Q) \leq \sqrt{ 2 kn}$ and $Q(a_i, w_i)=0$ for any $i \in [n]$. 
Let $f(X) \in \F_q[X]$ be a univariate polynomial  of degree smaller than $k$ so that $f(a_i) = w_i$ for at least $\sqrt{2 kn} +1$ indices $i \in [n]$. Then $Q(X,f(X))=0$. 
\end{lemma}

\begin{proof}
By assumption that $\deg_{(1,k)}(Q) \leq \sqrt{ 2 kn}$, the
polynomial $P_f(X):=Q(X,f(X)) \in \F_q[X]$ is a univariate polynomial of degree at most
$ \sqrt{2 kn}$. Additionally, it satisfies that 
$Q(a_i, f(a_i)) = Q(a_i, w_i)=0$ for any $i \in [n]$ for which $f(a_i) = w_i$.  Thus $P_f(X)$ is a univariate polynomial of degree at most $\sqrt{2 kn}$ with at least $\sqrt{2 kn}+1$ roots, and so it must be the zero polynomial. 
\end{proof}

In order to ensure that the root finding step can be performed efficiently when $Q$ has a general form,  we need to show that there are not too many polynomials $f(X)$ of degree smaller than $k$ satisfying that $Q(X,f(X))=0$, and that these polynomials can be found efficiently.  To this end we observe the following.

\begin{lemma}\label{lem:rs_list_solution_size}
Suppose that $Q(X,Y)$ is a non-zero bivariate polynomial over $\F_q$, and  $f(X) \in \F_q[X]$ is a univariate polynomial so that  $Q(X,f(X))=0$. Then $Y-f(X)$ divides $Q(X,Y)$ in the bivariate polynomial ring $\F_q[X,Y]$. 
\end{lemma}

\begin{proof}
Let $D$ be the degree of $Y$ in $Q$. So, we can write $Q$ as $Q(X,Y) = \sum_{i = 0}^{D} Q_i(X) Y^i$. Clearly, $Q(X,f(X))$ equals $\sum_{i = 0}^{D} Q_i(X) (f(X))^i$. Therefore, 
\[
Q(X,Y) - Q(X,f(X)) = \sum_{i = 1}^D Q_i(X) (Y^i - f(X)^i) \, .
\]
For every positive integer $i$, we note that $(Y^i - f(X)^i)$ factors as $(Y-f(X))(Y^{i-1} + Y^{i-2}f(X) + \dots + Yf(X)^{i-2} + f(X)^{i-1})$. In other words, every term in the sum $\sum_{i = 1}^D Q_i(X) (Y^i - f(X)^i)$ is divisible by $(Y-f(X))$, and furthermore, the quotient is a polynomial in $\F_q[X,Y]$. Thus, $Q(X,Y) - Q(X,f(X))$ is divisble by $(Y-f(X))$ over $\F_q[X,Y]$. So, if $f(X)$ is such that $Q(X,f(X))$ is zero, then we get that $Q(X,Y)$ is divisble by $(Y-f(X))$ over $\F_q[X,Y]$. 
\end{proof}

Thus, in order to find all univariate polynomials $f(X) \in \F_q[X]$ satisfying that $Q(X,f(X))=0$, it suffices to factorize $Q(X,Y)$ which can be performed in polynomial time by well known algorithms for bivariate/multivariate polynomial factorization. We summarise the result we need in the following theorem and refer to the cited references and \cite[Section 4.5]{Gur-survey} for a description of the algorithm. 

\begin{theorem}[Bivariate factorization \cite{Kaltofen1985, Lenstra1985}]\label{thm:bivariate-factorization}
    Let $\F$ be any finite field. Then, there is a randomized algorithm that takes as input a bivariate polynomial $Q(X,Y) \in \F[X,Y]$ and outputs all its factors of the form $Y-f(X)$ where $f \in \F[X]$ in time $\poly(\log (|\F|), \deg(Q))$. 
\end{theorem}

\medskip

The full description of the algorithm appears  in Figure \ref{fig:rs_list} below, followed by correctness and efficiency analysis.

\begin{figure}[h]
  \begin{boxedminipage}{\textwidth} \small \medskip \noindent
    $\;$

    \underline{\textbf{List decoding of $\RS_{q}(a_1, \ldots, a_n;k)$:}}
    
    \medskip

\begin{itemize}
\item \textbf{INPUT:} $ w\in \F_q^n$.
\item \textbf{OUTPUT:} The list $\calL$ of all polynomials $f(X) \in \F_q[X]$ of degree smaller than $k$ so that $f(a_i)= w_i$ for at least $\sqrt{2k n} +1$ indices $i \in [n]$.

\end{itemize}

    \begin{enumerate}
\item \label{step:rs_list_interpolation} Find a non-zero bivariate polynomial $Q(X,Y)$ over $\F_q$ with $(1,k)$-weighted degree at most $\sqrt{2nk}$, and which satisfies that $Q(a_i, w_i)=0$ for any $i \in [n]$.

Such a polynomial $Q(X,Y)$ can be found by solving a system of linear equations which is obtained by viewing the  coefficients of the monomials in $Q(X,Y)$ as unknowns, and viewing each requirement of the form $Q(a_i,w_i)=0$ as a homogeneous linear constraint on these unknowns.       

\item  \label{step:rs_list_root_finding}
Find the list $\calL'$ of all univariate polynomials $f(X)$ so that $Q(X, f(X))=0$ by factoring $Q(X,Y)$ using the algorithm in \autoref{thm:bivariate-factorization}, and including in $\calL'$ all univariate polynomials $f(X)$ so that $Y-f(X)$ is a factor of $Q(X,Y)$.
\item \label{step:rs_list_check} 
Output the list $\calL$ of all codewords $(f(a_1),f(a_2),\ldots, f(a_n))$ corresponding to
univariate polynomials $f(X) \in \calL'$ of degree smaller than $k$ which satisfy that $f(a_i) = w_i$ for at least $\sqrt{2 nk}+1$ indices $i \in [n]$.

\end{enumerate}

  \medskip

  \end{boxedminipage}

\caption{List decoding of Reed-Solomon Codes}
\label{fig:rs_list}
\end{figure}

\paragraph{Correctness.}  Suppose that $f(X) \in \F_q[X]$ is a univariate polynomial of degree smaller than $k$ so that $f(a_i) = w_i$ for at least $\sqrt{2nk}+1$ indices $i \in [n]$, we shall show that $f(X) \in \calL$. By Lemma \ref{lem:rs_list_interpolation}, there exists a non-zero bivariate polynomial $Q(X,Y)$ satisfying the requirements on Step \ref{step:rs_list_interpolation}. By Lemma \ref{lem:rs_list_P_is_root}, we have that $Q(X,f(X))=0$, and by Lemma \ref{lem:rs_list_solution_size} this implies in turn that $Y-f(X)$ divides $Q(X,Y)$, and so  $f(X)$ will be included in $\calL'$ on Step \ref{step:rs_unique_root_finding}. Clearly, $f(X)$ will also pass the checks on Step \ref{step:rs_list_check}, and thus $f(X)$ will also be included in $\calL$.

\paragraph{Efficiency.} Step \ref{step:rs_list_interpolation} can be performed in time $\poly(n, \log (q))$ by solving a system of $n$ linear equations in $O(n)$ variables using Gaussian elimination. Step \ref{step:rs_list_root_finding} can be performed in time $\poly(n,\log (q))$ by the algorithm in \autoref{thm:bivariate-factorization}. Step \ref{step:rs_list_check} can also clearly be performed in time $\poly(n, \log (q))$.

\subsection{List decoding Reed-Solomon Codes up to the Johnson Bound}\label{subsec:rs_johnson}

In this section, we present an efficient algorithm for list decoding of  Reed-Solomon Codes \emph{up to the Johnson Bound} (cf., Theorem \ref{thm:johnson}). 

Once more, fix  a prime power $q$, distinct evaluation points $a_1, \ldots, a_n \in \F_q$, and a positive integer $k<n$, and let $\RS_q(a_1, \ldots, a_n;k)$ be the corresponding Reed-Solomon code. 
Suppose that $w\in \F_q^n$ is a received word. 
We shall show how to compute the list of all polynomials $f(X) \in \F_q[X]$ of degree smaller than $k$ so that $f(a_i) \neq w_i$ for at most $e: = n-\sqrt{k n}-1$ indices $i \in [n]$. That is, $f(a_i) = w_i$ for at least $t: =n-e= \sqrt{ k n}+1$ indices $i \in [n]$. 

In particular, if we let $R:=\frac{k} {n}$ denote the rate of the $\RS_{q}(a_1, \ldots, a_n;k)$ code, then this algorithm will be able to list-decodable from up to roughly $(1- \sqrt{R})$-fraction of errors. This roughly shaves off a factor of $\sqrt{2}$ from the agreement needed by the algorithm discussed in the previous section and matches the Johnson Bound of Theorem \ref{thm:johnson} (recalling that the relative distance $\delta$ of the $\RS_q(a_1, \ldots, a_n;k)$ code satisfies that $R \leq 1-\delta$). 

\paragraph{Overview. }
The algorithm 
builds upon the ideas in the previous algorithms that we saw, and adds one new technical ingredient to the framework, namely, the method of \emph{multiplicities}. To describe this idea, we first briefly recall a high level sketch of the algorithm in the last section. 

The algorithm in Figure \ref{fig:rs_list} essentially proceeds by finding a non-zero, low degree (in an appropriate sense) bivariate polynomial $Q(X,Y)$ with the following property: for every polynomial $f(X)$ of degree smaller than $k$, if the $i^{th}$ entry $w_i$ of the received word satisfies $w_i = f(a_i)$, then the univariate polynomial $P_{f}(X) := Q(X, f(X))$ has a zero at $a_i$. Therefore, for an $f$ whose encoding is close to the received word $w$, the univariate $P_f$ has a \emph{lot} of distinct zeroes and if the number of such agreements exceeds the degree of $P_f$ (which is at most the $(1,k)$-weighted degree of $Q$), then $P_f$ must be identically zero, or equivalently, $f(X)$ must be a root of $Q$ (when $Q$ is viewed as a univariate polynomial in the indeterminate $Y$, with coefficients coming from the univariate polynomial ring $\F_q[X]$). 

The polynomial $Q$ is obtained by solving an appropriate linear system in its coefficients which guarantees that  $Q(a_i, w_i)=0$ for any $i$, and we need to set the degree of $Q$ to be not too low in order to guarantee a non-zero solution. 
On the other hand, it is also crucial that the degree of $Q$ is not to high, 
since the minimum number of agreements needed for this algorithm to work is essentially one more than the degree of $P_f$, and hence the error tolerance is higher if $Q$ has low degree. These two requirements, that are in tension with each other, determine the quantitative bounds needed for the algorithm to succeed. 

At a high level, the main new idea that makes it possible to tolerate more errors is to obtain a $Q(X,Y)$ in the interpolation step with the stronger property that if a polynomial $f(X)$ and the received word $w$ agree on some entry $i$, i.e., $f(a_i) = w_i$, then, the univariate polynomial $P_f(X)$ vanishes with \emph{high} multiplicity at $a_i$, i.e.,  $P_f^{(j)}(a_i)=0$ for $j=0,1,\ldots, r-1$ for a sufficiently large parameter $r$, where $P_f^{(j)}(X)$ denotes the $j$-th \emph{Hasse Derivative} (cf., Section \ref{subsec:prelim_poly}). This intuitively seems to indicate that we now need fewer agreements between $f$ and $w$ to ensure that $P_f(X)$ is identically zero. To obtain  a non-zero bivariate $Q(X,Y)$ that satisfies this property, the idea is to strengthen the interpolation step and look for a $Q(X,Y)$ that (essentially) vanishes on every point $(a_i, w_i)$ with sufficiently high multiplicity. Since the number of constraints has increased in the process, we now need to work with bivariates $Q(X,Y)$ of somewhat higher $(1,k)$-weighted degree for this system to have a non-zero solution. 

Once again, these two steps are in tension with each other - the interpolation step is likely to succeed if the $(1,k)$-weighted degree of $Q$ is not too low, while for $P_f(X)$ to be identically zero for a polynomial $f(X)$ of degree smaller than $k$, whose encoding is close to $w$, it helps to have the degree of $P_f$, which is at most the $(1,k)$-weighted degree of $Q$, to be low. It is apriori unclear that the quantitative bounds obtained using this approach will be an improvement over the bounds obtained in the algorithm of Figure \ref{fig:rs_list} that was presented in the previous section. However, as we shall see, this is indeed the case, and the new algorithm list decodes all the way up to the Johnson Bound. 

We now describe the main steps of the algorithm formally. 

\paragraph{Step 1: Interpolation. } The following lemma is a high multiplicity variant of 
Lemma \ref{lem:rs_list_interpolation} from the previous section.

\begin{lemma}\label{lem:gs-interpolation}
There exists a non-zero bivariate polynomial $Q(X,Y)$ over $\F_q$ with $(1,k)$-weighted degree at most $\sqrt{kn(r+1)r}$, such that for every $i \in [n]$, 
$Q$ vanishes with multiplicity at least $r$ on $(a_i,w_i)$. 
Furthermore, such a polynomial $Q(X,Y)$ can be found by solving a system of linear equations which is obtained by viewing the  coefficients of the monomials in $Q(X,Y)$ as unknowns, and viewing each vanishing constraint for a Hasse Derivative of $Q$ as a homogeneous linear constraint on these unknowns.      
\end{lemma}
\begin{proof}
The proof is essentially the same as that of Lemma \ref{lem:rs_list_interpolation}, which corresponds to the $r=1$ case. We view the coefficients of $Q$ as variables, and notice that the property that any Hasse Derivative of $Q$ vanishes at $(a_i, w_i)$ is just a homogeneous linear constraint on the coefficients of $Q$. Thus, if the parameters are set in a way that the number of unknowns is more than the number of constraints, then this system must have a non-zero solution. Thus, to complete the proof, we just count the number of variables and constraints in the linear system. 

The number of Hasse Derivatives of order less than $r$ of a bivariate polynomial is equal to the number of monomials in two variables of total degree at most $(r-1)$, which is at most $\binom{r+1}{2}$. Thus, for every $i \in [n]$, we have $\binom{r+1}{2}$ homogeneous linear constraints, which gives $n\binom{r+1}{2}$ constraints in total. Let $D$ be a parameter (to be set soon), which denotes the $(1,k)$-weighted degree of $Q(X,Y)$. If $D$ is large enough to ensure that the number of bivariate monomials of $(1,k)$-weighted degree at most $D$ exceeds $n\binom{r+1}{2}$, then the above homogeneous linear system has more variables than constraints, and hence has a non-zero solution. Thus, by the same computation as in Lemma \ref{lem:rs_list_interpolation} (setting $\ell = \frac{D} {k}$), it suffices to choose $D$ such that 
\[
\frac{(D+2)(D+k)}{2k} > n \binom{r+1}{2} \, .
\]
Setting $D$ to be equal to $\sqrt{kn (r+1) r}$ (which we assume is an integer for simplicity) satisfies this constraint.  
\end{proof}

 \paragraph{Step 2: Root finding.}   
 We now argue that any polynomial $f(X)$ of degree smaller than $k$ whose encoding has large agreement with the received word $w$ satisfies that $Q(X, f(X))$ is identically zero. 
To this end, we first observe that if $Q(X,Y)$ vanishes on a point $(a_i, f(a_i))$ with multiplicity at least $r$, then $P_f(X):=Q(X,f(X))$ vanishes on 
 $a_i$ with multiplicity at least $r$.
 
 \begin{claim}\label{clm:gs-zero-at-agreement}
 Let $Q(X,Y) \in \F_q[X,Y]$ and $f(X) \in \F_q[X]$ be polynomials, and let  $P_f(X):=Q(X,f(X))$. Suppose that $Q(X,Y)$ vanishes on a point $(a, f(a))$ for some $a \in \F_q$ with multiplicity at least $r$. Then $P_f(X)$ vanishes on 
 $a$ with multiplicity at least $r$.
 \end{claim}
 \begin{proof}

 The statement follows by a standard application of the chain rule for Hasse Derivatives. 
 In more detail, recall from the definition of Hasse Derivatives that 
 $$f(a + Z) = \sum_{i = 0}^{k-1} f^{(i)}(a) Z^i = f(a) + Z \cdot \tilde{f}(Z),$$ for some univariate polynomial $\tilde f(Z)$. 
 From the definition of Hasse Derivatives, we also get 
\[ 
 P_f(a + Z) = Q(a + Z, f(a) + Z \cdot \tilde{f}(Z)) = \sum_{u,v} Q^{(u,v)}(a, f(a)) Z^{u + v} \tilde{f}(Z)^v \, .
 \]
Now, by assumption that all Hasse Derivatives of order less than $r$ of $Q$ are zero at $(a, f(a))$, we have that 
\[ 
 P_f(a + Z) = \sum_{u,v \in \Z_{\geq 0}: u + v \geq r} Q^{(u,v)}(a, f(a)) Z^{u + v} \tilde{f}( Z)^v \, .
\]
Thus, the lowest degree of a monomial in $Z$ that can possibly have a non-zero coefficient in $P_f(a + Z)$ is at least $r$, and hence, by the definition of multiplicity, we get that the multiplicity of $P_f$ at $a$ is at least $r$. 
\end{proof}

The following lemma is analogous to Lemma \ref{lem:rs_list_P_is_root} from the previous section.

  \begin{lemma}\label{lem:gs-close-poly-are-roots}
 Suppose that $Q(X,Y)$ is a non-zero bivariate polynomial over $\F_q$ which satisfies that 
 $\deg_{(1,k)}(Q) \leq \sqrt{kn(r+1)r}$ for $r=2kn$. Suppose furthermore that $Q$ vanishes on each point $(a_i,w_i)$ for $i \in [n]$ with multiplicity at least $r$. 
 Let $f(X) \in \F_q[X]$ be a univariate polynomial of degree smaller than $k$ so that $f(a_i) = w_i$ for at least $\sqrt{nk} + 1$ indices $i \in [n]$. Then $Q(X, f(X))=0$.   
  \end{lemma}
  \begin{proof}
  From the bound on the $(1,k)$-weighted degree of $Q$, we have that $P_f(X) := Q(X, f(X))$ is a univariate polynomial of degree at most $D := \sqrt{kn(r+1)r}$. From Claim \ref{clm:gs-zero-at-agreement}, we know that for any $i \in [n]$, if $f(a_i) = w_i$, then $P_f$ has a zero of multiplicity at least $r$ at $a_i$. Thus, if $f$ and $w$ agree on greater than $\frac D r = \frac{\sqrt{kn(r+1)r}} {r} = \sqrt{nk(1 + \frac 1 r)}$ many distinct $a_i$'s, we have from Fact \ref{fact:prelim_hasse_sz_univ} that $P_f$ must be identically zero. For sufficiently large $r$, e.g. $r \geq  kn$, we have that $\sqrt{nk(1 + \frac 1 r)} < \sqrt{nk} + 1$. Thus, by picking large enough $r$ (for example, $r=2kn$), we get that more than $\sqrt{nk}$ agreements suffice for $P_f$ to be identically zero. 
   \end{proof}
   
Thus, as in the previous section, we get that any univariate polynomial of degree smaller than $k$ with large agreement with $w$ is contained as a root of $Q(X,Y)$ and can be recovered  using Theorem \ref{thm:bivariate-factorization}. Once again, the time complexity of the list decoding algorithm is polynomially bounded in the input size, since the main steps just involve solving a linear system of polynomially bounded size and a special case of bivariate polynomial factorization. The full description of the algorithm is given in Figure \ref{fig:rm_list_johnson} below.

\begin{figure}[h]
  \begin{boxedminipage}{\textwidth} \small \medskip \noindent
    $\;$

    \underline{\textbf{List decoding up to Johnson Bound of $\RS_{q}(a_1, \ldots, a_n;k)$:}}
    
    \medskip

     $\triangleright$ Let $r := 2nk$.  
\begin{itemize}
\item \textbf{INPUT:} $ w\in \F_q^n$.
\item \textbf{OUTPUT:} The list $\calL$ of all polynomials $f(X) \in \F_q[X]$ of degree smaller than $k$ so that $f(a_i)= w_i$ for at least $\sqrt{nk} +1$ indices $i \in [n]$.

\end{itemize}

    \begin{enumerate}
\item \label{step:gs_list_interpolation} Find a non-zero bivariate polynomial $Q(X,Y)$ over $\F_q$ with $(1,k)$-weighted degree at most $\sqrt{nk(r+1)r}$, and which satisfies that $Q^{(u,v)}(a_i, w_i)=0$ for any $i \in [n]$, and for any $u, v \in \N$ with $u + v < r$. 

Such a polynomial $Q(X,Y)$ can be found by solving a system of linear equations which is obtained by viewing the  coefficients of the monomials in $Q(X,Y)$ as unknowns, and viewing each requirement of the form $Q^{(u,v)}(a_i,w_i)=0$ as a homogeneous linear constraint on these unknowns.       

\item  \label{step:rs_list_root_finding}
Find the list $\calL'$ of all univariate polynomials $f(X)$ so that $Q(X, f(X))=0$ by factoring $Q(X,Y)$ using the algorithm in \autoref{thm:bivariate-factorization}, and including in $\calL'$ all univariate polynomials $f(X)$ so that $Y-f(X)$ is a factor of $Q(X,Y)$.
\item \label{step:rs_list_check} 
Output the list $\calL$ of all codewords $(f(a_1),f(a_2),\ldots, f(a_n))$ corresponding to 
univariate polynomials $f(X) \in \calL'$ of degree smaller than $k$ which satisfy that $f(a_i) = w_i$ for at least $\sqrt{nk}+1$ indices $i \in [n]$.

\end{enumerate}

  \medskip

  \end{boxedminipage}

\caption{List decoding of Reed-Solomon Codes up to Johnson Bound}
\label{fig:rs_list_johnson}
\end{figure}

\subsection{Limitations on list decoding of Reed-Solomon Codes beyond the Johnson bound}\label{subsec:rs_limit}

In the previous section, we showed that Reed-Solomon codes are efficiently list-decodable up to the Johnson Bound. That is, we presented an algorithm that efficiently outputs all polynomials of degree smaller than $k$ whose encoding agrees with the received word on at least $\sqrt{R}n $ points. In this section, we shall present a combinatorial lower bound on the list size of Reed-Solomon Codes, which shows
that Reed-Solomon Codes are generally \emph{not} list-decodable from a slightly smaller agreement of $R^{\frac 1 2 +\epsilon} n$ with a list-size that is polynomial of $n$.
This implies in turn  that there is no polynomial-time algorithm that can list-decode Reed-Solomon Codes up to this amount of agreement, and in particular, from an agreement arbitrarily close to $R$. This bound follows as a consequence of the following theorem.

\begin{theorem}\label{thm:rs_limitation}
For any $\gamma \in (0,1)$, $\beta \in (\gamma, \sqrt{\gamma})$, a sufficiently large integer $m$, and $n=2^m$, there is a word $w \in (\F_n)^n$, so that there are at least $n^{(\gamma - \beta^2) \log (n)}$ 
codewords in $ \RS_n(n^\gamma)$  that agree with $w$ on at least $n^\beta$ entries.
\end{theorem}

Setting $\gamma := 1-\epsilon$ and $\beta := (\frac 1 2 - \epsilon)+ \gamma(\frac 1 2 +\epsilon)$ in the above lemma (noting that $\beta \in (\gamma, \sqrt{\gamma})$), shows, for infinitely many  $n$, the existence of a word $w \in (\F_n)^n$, for which there exist a super-polynomial number of codewords in $ \RS_n(k)$  that agree with $w$ on at least  $n^{\frac 1 2 -\epsilon}k^{\frac 1 2 + \epsilon} = R^{\frac 1 2 +\epsilon} n$ entries, where $k=n^{1-\epsilon}$. 

The proof of the above Theorem \ref{thm:rs_limitation}
is based on \emph{subspace polynomials}, defined as follows.

\begin{definition}[Subspace polynomials]
Let $V \subseteq \F_{2^m}$ be an $\F_2$-linear subspace of $\F_{2^m}$. The \textsf{subspace polynomial} for $V$ is defined by 
$P_V(X) := \prod_{v \in V} (X-v).$
\end{definition}

The two main properties of subspace polynomials we shall use are that they have relatively many roots (the whole subspace $V$), and on the other hand, these polynomials are rather sparse, as shown by the following claim.

\begin{claim}\label{clm:subspace_poly}
Let $V \subseteq \F_{2^m}$ be an $\F_2$-linear subspace of $\F_{2^m}$ of dimension $d$. Then 
there exist $b_0,b_1, \ldots, b_{d-1} \in \F_{2^m}$ so that 
$P_V (X)=X^{2^d} + \sum_{i=0}^{d-1} b_i X^{2^i}.$
\end{claim}

\begin{proof}
For a non-negative integer $i$, let $f_i(X): V \to \F_{2^m}$ denote the $\F_2$-linear function given by 
$f_i(v):=v^{2^i}$ for $v \in V$. Next observe that the collection of $\F_2$-linear functions from $V$ to $\F_{2^m}$ forms an $\F_{2^m}$-linear subspace of dimension $d$, and so $f_0,f_1, \ldots, f_d$ must have a non-trivial linear dependency. That is, there exist scalars $b_0,b_1, \ldots,b_d \in \F_{2^m}$, not all zero, so that $P(X):=\sum_{i=0}^d b_i f_i(X) = \sum_{i=0}^d b_i X^{2^i}$ evaluates to zero on any $v \in V$. 
Furthermore, since $V$ has $2^d$ elements, we must have that $\deg(P(X)) =2^d$, and we may further assume that $P(X)$ is monic. So $P(X)$ and $P_V(X)$ are two monic polynomials of degree $2^d$ that vanish on all $2^d$ elements of $V$, 
and therefore we must have that $P_V(X) = P(X)=X^{2^d} + \sum_{i=0}^{d-1} b_i X^{2^i}$ (since $P(X)-P_V(X)$ is a polynomial of degree smaller than $2^d$ with at least $2^d$ roots, and so must be the zero polynomial).
\end{proof}

We now turn to the proof of the above Theorem \ref{thm:rs_limitation}.

\begin{proof}[Proof of Theorem \ref{thm:rs_limitation}.]
Let $\gamma \in (0,1)$, $\beta \in (\gamma, \sqrt{\gamma})$, let $m$ be a sufficiently large integer. Let $\mathcal{V}$ be the collection of $\F_2$-linear subspaces $V\subseteq \F_{2^m}$ of dimension $\beta m$. Then the size of $\mathcal{V}$ is  
$\prod_{i=0}^{\beta m -1} \frac{2^m -2^i} {2^{\beta m}-2^i} \geq 2^{\beta(1-\beta)m^2}$ (since to choose a subspace $V$ of dimension $d$ in $\F_{2^m}$, we need for each $i=0,1, \ldots, d-1$ to choose a vector which is not in the linear span of all previous $i$ vectors, and we also need to divide by the number of options to choose $d$ linearly independent vectors in $V$). 
Furthermore, by Claim \ref{clm:subspace_poly}, for each $V \in \mathcal{V}$ there exist $b_0,b_1, \ldots, b_{\beta m -1} \in \F_{2^m}$ so that 
$P_V(X) = X^{2^{\beta m}} + \sum_{i=0}^{\beta m -1} b_i X^{2^i}$.

By the pigeonhole principle, there exist $b_{\gamma m+1}, b_{\gamma m+2}, \ldots, b_{\beta m-1 } \in \F_{2^m}$ such that the number of subspace polynomials $P_V(X)$ for $V \in \mathcal{V}$ with their coefficient of $X^{2^i}$ equal to $b_i$ for $i=\gamma m+1, \ldots, \beta m-1$ is at least 
$2^{\beta(1-\beta)m^2} / 2^{ (\beta - \gamma)m^2}\geq  2^{ (\gamma - \beta^2)m^2}$.  Let $\mathcal{V}' \subseteq \mathcal{V}$ denote the collection of subspaces $V \in \mathcal{V}$ for which $P_V(X)$ has this property.

Let $P(X):= \sum_{i=\gamma m+1}^{\beta m -1} b_i X^{2^i}$, and let $w \in (\F_n)^n$ be the evaluation table of $P(X)$ over $\F_{2^m}$, where $n=2^m$. For $V \in \mathcal{V}'$, let $P'_V(X) = P_V(X) + P(X)$. Then the polynomials $P'_V(X)$ for $V \in \mathcal{V}'$ are a collection of $2^{ (\gamma - \beta^2)m^2 }= n^{(\gamma -\beta^2)\log (n)}$ distinct polynomials of degree at most $2^{\gamma m}=n^\gamma$. Moreover, for each $V \in \mathcal{V'}$,  $P'_V(X)$ agrees with $P(X)$ on all $2^{\beta m} = n^{\beta}$ elements of $V$.  
We conclude that  there are at least $n^{(\gamma - \beta^2) \log (n)}$ 
codewords in $ \RS_n(n^\gamma)$ (the evaluations of $P'_V(X)$ over $\F_{2^m}$ for $V \in \mathcal{V}'$)
that agree with $w$ on at least $n^\beta$ entries.
\end{proof}

\subsection{Bibliographic notes}

The first efficient algorithms for \textsf{unique decoding of Reed-Solomon codes} were given by Peterson \cite{Peterson60}, Gorenstein and Zierler \cite{GZ61}, Berlekamp \cite{Berlekamp_book}, and Massey \cite{Massey69}. 
The unique decoding algorithm for Reed-Solomon Codes described in Section \ref{subsec:rs_unique} is due to  Berlekamp and Welch \cite{BW87}, with the exposition based on that of Gemmell
and Sudan \cite{GS92} (see also \cite[Section 12.1]{GRS_survey}). 
The \textsf{list-decoding algorithm for Reed-Solomon codes beyond unique decoding radius} described in Section \ref{subsec:rs_list} is due to Sudan \cite{Sudan97}, and the \textsf{list-decoding algorithm for Reed-Solomon Codes up to the Johnson Bound} is due to Guruswami and Sudan \cite{GS-list-dec} (see also \cite[Section 4]{Gur-survey} and \cite[Section 12.2]{GRS_survey}). 

The root finding algorithm (\autoref{thm:bivariate-factorization}) used in the  algorithms for list decoding Reed-Solomon Codes follows from the more general algorithms known for factorization of bivariate (and multivariate) polynomials, for instance from the works of Kaltofen \cite{Kaltofen1985} and that of Lenstra \cite{Lenstra1985}. 

The limitation on list-decoding of Reed-Solomon Codes beyond the Johnson Bound, presented in Section \ref{subsec:rs_limit},  was discovered by Ben-Sasson, Kopparty, and Radhakrishnan \cite{BKR10}.

\section{Algorithmic list decoding up to capacity}\label{sec:capacity}

In the previous section, we presented efficient (polynomial-time) algorithms for list decoding of Reed-Solomon Codes \emph{up to the Johnson Bound},  and we also showed that Reed-Solomon Codes are generally \emph{not} list-decodable much beyond the Johnson Bound, and in particular, up to capacity. 

In this section, we present algorithms for list-decoding variants of Reed-Solomon codes \emph{up to capacity}. 
Specifically, in Sections \ref{subsubsec:mult_no_mult} and \ref{subsubsec:mult_mult} we present efficient  algorithms for list decoding \emph{multiplicity codes} up to capacity, and in Section \ref{subsec:rs_subfield} we also show that a similar algorithm can be used to list decode Reed-Solomon codes over subfield evaluation points  in slightly non-trivial sub-exponential time. These algorithms in particular show that the list size of multiplicity codes at capacity is at most polynomial in the block length, while the list-size of Reed-Solomon codes over subfield evaluation points is at most sub-exponential in the block length. 
In Section \ref{sec:list}, we shall present combinatorial upper bounds on the list size of these codes, which effectively reduce the list size to a constant, independent of the block length.

\subsection{List decoding multiplicity codes beyond unique decoding radius}\label{subsubsec:mult_no_mult}

Next we present efficient (polynomial-time) algorithms for 
list decoding the family of \emph{multiplicity codes}, which extend
 Reed-Solomon Codes by also evaluating \emph{derivatives} of polynomials, \emph{up to capacity}. Towards this, we first present in this section, as a warmup, a simple list-decoding algorithm, similar to the \emph{unique decoding} algorithm for  Reed-Solomon Codes presented in Section \ref{subsec:rs_unique}, that list decodes multiplicity codes \emph{beyond the unique decoding radius}. Then in Section \ref{subsubsec:mult_mult}, we shall show that augmenting this algorithm with multiplicities, similarly to  the list-decoding algorithm for Reed-Solomon Codes \emph{up to the Johnson Bound} presented in Section \ref{subsec:rs_johnson}, leads to an algorithm that list decodes multiplicity codes \emph{up to capacity}. 

In what follows, fix a prime power $q$,  distinct evaluation points $a_1, \ldots, a_n \in \F_q$, and positive integers $s,k$ so that $k< sn$. Recall that 
  the \emph{multiplicity code} $\MULT^{(s)}_q(a_1, \ldots, a_n;k)$ is the code which associates with each polynomial $f(X) \in \F_q[X]$ of degree smaller than $k$ a codeword $(f^{(<s)}(a_1), \ldots, f^{(<s)}(a_n)) \in (\F_q^s)^n$,
 where  for $a \in \F_q$,  $f^{(<s)}(a) = (f(a), f^{(1)}(a), \ldots, f^{(s-1)}(a))$.
For simplicity, in what follows we shall assume that  $\max\{k,s\} \leq \char(\F_q)$.\
Suppose that $w\in (\F_q^s)^n$ is a received word, where $w_i = (w_{i,0}, \ldots, w_{i,s-1}) \in \F_q^s$ for any $ i \in[n]$.
We shall show an algorithm, running in time roughly $q^s$, which computes the list of all polynomials $f(X) \in \F_q[X]$ of degree smaller than $k$ so that $f^{(<s)}(a_i) \neq w_i$ for at most $e :=\frac{s \cdot  (n-k+1)-1} {s+1}   $ 
 indices $i \in [n]$. That is, $f^{(<s)}(a_i) = w_i$ for at least $t: =n -e = \frac{n+s (k-1) +1} {s+1} $ indices $i \in [n]$. 
 
 Note that the special case of $s=1$ corresponds to the unique decoding algorithm for Reed-Solomon codes presented in Section \ref{subsec:rs_unique}. 
Moreover, if we let $R:=\frac{k} {s \cdot n}$ denote the rate of the $\mult_q^{(s)}(a_1, \ldots, a_n; k)$ code, then this algorithm will be able to list decode from up to  roughly $\frac {s} {s+1} (1- s \cdot R)$-fraction of errors in time roughly $q^s$. It can be verified that this is strictly larger than the unique decoding radius of $\frac{1 -R} {2}$ for any $s>1$ and $R \leq \frac{1} {3s}$. Furthermore, for any constant $s$ and $q=\poly(n)$ the running time (and list size) is polynomial in the block length.

\paragraph{Overview.}
Recall that in the unique decoding algorithm for Reed-Solomon Codes (Algorithm \ref{fig:rs_unique}), the decoding algorithm was divided into two main steps: The interpolation step in which we searched  for a non-zero low-degree bivariate polynomial $Q(X,Y) \in \F_q[X,Y]$
of the form $Q(X,Y) = A (X) +  B (X) \cdot Y$  
so that $Q(a_i, w_i)=0$ for any $i \in [n]$, and the root finding step in which we searched for the univariate polynomial $f(X) \in \F_q[X]$ satisfying that $Q(X, f(X)) = 0$. 

Furthermore, the main properties we needed out of $Q$ 
were that the number of unknown coefficients in $Q$ is greater than the number $n$ of constraints of the form $Q(a_i,w_i)=0$, and that the $(1,k)$-weighted degree $\deg_{(1,k)}(Q)$ of $Q$ is smaller than the agreement parameter $t$. The first property guarantees the existence of a non-zero $Q$, while the second property  guarantees that for any univariate polynomial $f(X) \in \F_q[X]$ of degree smaller than $k$ so that $f(a_i)=w_i$ for at least $t$ of the indices $i \in [n]$, it holds that
 $Q(X,f(X))$ is a univariate polynomial of degree smaller than $t$ that has at least $t$ roots, and consequently $Q(X,f(X))=0$.

In the list-decoding setting, we would like the agreement parameter $t$ to be smaller, which means that the $(1,k)$-weighted degree of $Q$ is required to be smaller. Consequently, $Q$ may potentially have less than $n$ monomials, and so we may not be able to argue the existence of a non-zero polynomial $Q$. In the list-decoding algorithm for Reed-Solomon Codes beyond the unique decoding radius (Algorithm \ref{fig:rs_list}), we coped with this by no longer requiring that $Q$ has the form  $Q(X,Y) = A (X) +  B (X) \cdot Y$ that is linear in $Y$, which can potentially increase the number of monomials in $Q$.
For list decoding of multiplicity codes beyond unique decoding radius we deal with this differently by setting $Q$ to be a low-degree \emph{$(s+1)$-variate polynomial} $Q(X,Y_0, \ldots, Y_{s-1})$ that is \emph{linear in $Y_0, \ldots, Y_{s-1}$}, which can also potentially increase the number of monomials in $Q$.

More specifically, in the interpolation step, we  search for a non-zero  $(s+1)$-variate polynomial $Q(X,Y_0, \ldots, Y_{s-1}) \in \F_q[X,Y_0, \ldots, Y_{s-1}]$ 
of the form 
$$
Q(X,Y_0, \ldots, Y_{s-1})=A(X)+B_0(X) \cdot Y_0+ \cdots + B_{s-1}(X) \cdot Y_{s-1},
$$
of $(1,k, \ldots, k)$-weighted degree smaller than $t$,
so that $Q(a_i,w_i)=0$ for any $i \in [n]$ (where we view $w_i = (w_{i,0}, \ldots, w_{i,s-1}) \in \F_q^s$ as an assignment to $s$ variables in $\F_q$). Then, in the root finding step, we search for all univariate polynomials $f(X) \in \F_q[X]$ satisfying that $Q(X, f(X), f^{(1)}( X), \ldots, f^{(s-1)}(X)) = 0$. 

The analysis is similar to the algorithm for unique decoding of Reed-Solomon Codes, but a crucial point exploited in the proof is that there are not too many polynomials $f(X)$ solving the \emph{differential equation} $Q(X, f(X), f^{(1)}(X), \ldots, f^{(s-1)}(X)) = 0$. 
In what follows, we describe separately the interpolation and root finding steps, and analyze their correctness and efficiency.

\paragraph{Step 1: Interpolation.}

In this step, we search for a non-zero low-degree $(s+1)$-variate polynomial $Q(X,Y_0, \ldots, Y_{s-1})$ so that $Q(a_i, w_i)=0$  for any $i \in [n]$.

\begin{lemma}\label{lem:mult_no_mult_list_interpolation}
Let $t= \frac{n+s (k-1) +1} {s+1}$.
Then there exists a non-zero $(s+1)$-variate polynomial $Q(X,Y_0, \ldots, Y_{s-1})$ over $\F_q$ of the form
$$Q(X,Y_0, \ldots, Y_{s-1})=A(X)+B_0(X) \cdot Y_0+ \cdots + B_{s-1}(X) \cdot Y_{s-1},$$ 
where $A(X)$ has degree smaller than $t $ and each $B_\ell(X)$ has degree smaller than $t-k+1$, and which satisfies that $Q(a_i,w_i)=0$ 
for any $i \in [n]$ (where we view $w_i = (w_{i,0}, \ldots, w_{i,s-1}) \in \F_q^s$ as an assignment to $s$ variables in $\F_q$).  

Furthermore, such a polynomial $Q(X,Y_0, \ldots, Y_{s-1})$ can be found by solving a system of linear equations which is obtained by viewing the  coefficients of the monomials in $A(X),B_0(X), \ldots, B_{s-1}(X)$ as unknowns, and viewing each requirement of the form $Q(a_i,w_i)$ as a homogeneous linear constraint on these unknowns.       
\end{lemma}

\begin{proof}
If we view the coefficients of the monomials in $A(X),B_0(X), \ldots, B_{s-1}(X)$ as unknowns, then each requirement of the form  $Q(a_i,w_i)=0$ imposes a homogeneous linear constraint on these unknowns. Further note that 
the number of unknowns is 
$$t + s \cdot (t-k+1) = (s+1) \cdot t - s \cdot (k-1) = n+1,$$
while the number of linear constraints is $n$, and consequently this linear system has a non-zero solution.
\end{proof}

\paragraph{Step 2: Root finding.} 

 In this step, we search for all univariate polynomials $f(X) \in \F_q[X]$ of degree smaller than $k$ that solve the differential equation 
$$ 
 Q\left(X, f(X), f^{(1)}(X), \ldots, f^{(s-1)}(X)\right) = 0,
$$
 and we would like to show that any 
univariate polynomial $f(X) \in \F_q[X]$ of degree smaller than $k$ so that  $f^{(<s)}(a_i) = w_i$ for at least $t$ indices $i \in [n]$ satisfies the above differential equation. We start with the latter.

\begin{lemma}\label{lem:mult_no_mult_list_P_is_root}
Suppose that $Q(X,Y_0, \ldots, Y_{s-1})$ is a non-zero $(s+1)$-variate polynomial over $\F_q$ of the form
$$Q(X,Y_0, \ldots, Y_{s-1})=A(X)+B_0(X) \cdot Y_0+ \cdots + B_{s-1}(X) \cdot Y_{s-1},$$ 
where $A(X)$ has degree smaller than $t $ and each $B_\ell(X)$ has degree smaller than $t-k+1$, and which satisfies that $Q(a_i,w_i)=0$ 
for any $i \in [n]$ (where we view $w_i = (w_{i,0}, \ldots, w_{i,s-1}) \in \F_q^s$ as an assignment to $s$ variables in $\F_q$).  
Let $f(X) \in \F_q[X]$ be a univariate polynomial  of degree smaller than $k$ so that $f^{(<s)}(a_i) = w_i$ for at least $t$ indices $i \in [n]$. Then 
$$ Q(X, f(X), f^{(1)}(X), \ldots, f^{(s-1)}(X)) = 0.$$
\end{lemma}

\begin{proof}
The polynomial $P_f(X):= Q(X, f(X), f^{(1)}(X), \ldots, f^{(s-1)}(X))$ is a univariate polynomial  of degree smaller than $t$, which satisfies that 
$$P_f(a_i) = Q(a_i, f(a_i), f^{(1)}(a_i), \ldots, f^{(s-1)}(a_i))= Q(a_i, w_i)=0$$ 
for any $i \in [n]$ for which $f^{(<s)}(a_i) = w_i$.  Thus $P_f(X)$ is a univariate polynomial of degree smaller than $t$ with at least $ t$ roots, and so it must be the zero polynomial. 
\end{proof}

The main new ingredient in the list decoding algorithm for multiplicity codes is the next lemma which bounds the number of univariate polynomials $f(X) \in \F_q[X]$ satisfying the differential equation.

\begin{lemma}\label{lem:mult_no_mult_list_solution_size} 
Suppose that $Q(X,Y_0, \ldots, Y_{s-1})$ is a non-zero $(s+1)$-variate polynomial over $\F_q$ of the form
$$Q(X,Y_0, \ldots, Y_{s-1}) = A (X) +  B_0 (X) \cdot Y_0+ \cdots +B_{s-1}(X)\cdot Y_{s-1},$$
and let $\calL$ be the list containing all polynomials $f(X) \in \F_q[X]$ of degree smaller than $k$ so that  
$$Q(X, f(X), f^{(1)}(X), \ldots, f^{(s-1)}(X))=0.$$ Then  if $\max\{k,s\}\leq \char(\F_q)$, then  $|\calL| \leq q^{s-1}$. 

Furthermore, $\calL$ forms an affine subspace over  $\F_q$ of dimension at most $s-1$. A basis for this subspace can be found by solving a system of linear equations which is obtained by viewing the  coefficients of $f(X)$ as unknowns, and viewing the requirement that each of the coefficients of $Q(X, f(X), f^{(1)}(X), \ldots, f^{(s-1)}(X))$ is zero as a (non-homogeneous) linear constraint on these unknowns.     
\end{lemma}

\begin{proof}
First note that if there is no polynomial $f(X) \in \F_q[X]$ satisfying that $Q(X, f(X), f^{(1)}(X), \ldots, f^{(s-1)}(X))=0$, then 
we are done, so we may assume that there exists a polynomial $f(X) \in \F_q[X]$ satisfying that
$
Q(X, f(X), f^{(1)}(X), \ldots, f^{(s-1)}(X)) = 0.
$
Next we observe that in this case  $B_0(X), \ldots, B_{s-1}(X)$ are not all zero. 
To see this, suppose on the contrary that $B_0(X), \ldots, B_{s-1}(X)$ are all zero. Then by assumption that $Q(X, f(X), f^{(1)}(X), \ldots, f^{(s-1)}(X))= A(X) + \sum_{
\ell=0}^{s-1} B_\ell(X)  \cdot f^{(\ell)}(X)=0$, we also have that $A(X)=0$, which implies in turn that 
$Q(X,Y_0, \ldots, Y_{s-1})= A(X)+ \sum_{\ell=0}^{s-1} B_\ell(X) \cdot Y_\ell =0$, contradicting the assumption that $Q(X,Y_0, \ldots, Y_{s-1})$ is non-zero. Without loss of generality, we may further assume that $B_{s-1}(X) \neq 0$, otherwise we prove the statement for the maximal value $s' <s$ for which $B_{s'}(X) \neq 0$, which implies the statement also for the value $s$. 
Finally, since $B_{s-1}(X)$ is a non-zero polynomial over $\F_q$, there exists an element $a \in \F_q$  
so that $B_{s-1}(a) \neq 0$. The main observation is the following.

\begin{claim}\label{clm:mult_no_mult_list_solution_size}
For any $u \in \F_q^{s-1}$, there is a \emph{unique} polynomial $f(X) \in \F_q[X]$ of degree smaller than $k$ so that $f^{(<s-1)}(a) = u$ and $Q(X, f(X), f^{(1)}(X), \ldots, f^{(s-1)}(X))=0$.
\end{claim}

\begin{proof}
  To show this, fix $u \in \F_q^{s-1}$, and let $f(X) \in \F_q[X]$ be a polynomial of degree smaller than $k$ so that $f^{(<s-1)}(a) = u$ and
$$
P_f(X):=Q(X, f(X), f^{(1)}(X), \ldots, f^{(s-1)}(X))=A(X) + \sum_{\ell=0}^{s-1} B_\ell(X)  \cdot f^{(\ell)}(X) = 0.
$$
By the definition of the Hasse Derivative, we can write
$f(X) = \sum_{i=0}^{k-1} f^{(i)}(a) \cdot (X-a)^i$. It therefore suffices to show that the values of $f^{(<s-1)}(a)$ determine the rest of the derivatives $f^{(s-1)}(a), \ldots, f^{(k-1)}(a)$.

To show the above, we first compute the derivatives of $P_f(X)$. 
By the properties of the Hasse Derivative (Lemma \ref{lem:prelim_hasse_properties_univ}), for any $j \in \N$ we have that:
\begin{eqnarray}\label{eq:mult_no_mult_uni_der}
P_f^{(j)} (X)& = & A^{(j)}(X) + \sum_{\ell=0}^{s-1} (B_\ell(X) \cdot f^{(\ell)} (X) )^{(j)} \nonumber \\
& = & A^{(j)}(X) + \sum_{\ell=0}^{s-1} \sum_{h=0}^{j} B_\ell^{(j - h)}(X)  \cdot (f^{(\ell)})^{(h)}(X)  \nonumber \\ 
& = & A^{(j)}(X) + \sum_{\ell=0}^{s-1} \sum_{h=0}^{j} {h + \ell \choose \ell} B_\ell^{(j - h)}(X)  \cdot f^{(h+\ell)}(X).
\end{eqnarray}

Fix $j \in \{0,1,\ldots, k-s\}$. 
Since $P_f(X)$ is the zero polynomial, we have that $P_f^{(j)}(a)=0$, 
and so by the above,
$$ A^{(j)}(a) + \sum_{\ell=0}^{s-1} \sum_{h=0}^{j} {h + \ell \choose \ell} B_\ell^{(j - h)}(a)  \cdot f^{(h+\ell)}(a)=0.$$
Inspecting the above expression, we see that the highest order derivative of $f$ that appears there is $f^{(j+s-1)}(a)$,  which is only obtained for the choice of $\ell=s-1$ and $h=j$. Furthermore, the coefficient multiplying 
$f^{(j+s-1)}(a)$ is ${j +s-1 \choose s-1} B_{s-1}(a)$, which is non-zero by assumptions that $\max\{k,s\} \leq \char(\F_q)$ and that $B_{s-1}(a) \neq 0$. Consequently, the derivative $f^{(j+s-1)}(a)$ is a linear combination of lower derivatives $f^{(<j+s-1)}(a)$. In particular, $f^{(s-1)}(a), \ldots, f^{(k-1)}(a)$, and so also the polynomial $f(X)$, are uniquely determined by $f^{(<s-1)}(a)=u$.  
\end{proof}

By the above claim, we clearly have that $|\calL| \leq q^{s-1}$. The moreover part follows by the linear structure of $Q$. 
\end{proof}

\medskip

The full description of the algorithm appears  in Figure \ref{fig:mult_no_mult_list} below, followed by correctness and efficiency analysis.

\begin{figure}[h]
  \begin{boxedminipage}{\textwidth} \small \medskip \noindent
    $\;$

    \underline{\textbf{List decoding of $\mult_q^{(s)}(a_1, \ldots, a_n; k)$:}}
    
    \medskip
\begin{itemize}
\item \textbf{INPUT:} $w \in (\F_q^s)^n$, where $w_i = (w_{i,0}, \ldots, w_{i,s-1}) \in \F_q^s$ for each $i \in [n]$.
\item \textbf{OUTPUT:} 
The list $\calL$ of all polynomials $f(X) \in \F_q[X]$ of degree smaller than $k$ so that $f^{(<s)}(a
_i)= w_i$ for at least $t=\frac{n+s(k-1) +1} {s+1}$ indices $i \in [n]$.
\end{itemize}

    \begin{enumerate}
\item \label{step:mult_no_mult_list_interpolation} Find a non-zero $(s+1)$-variate polynomial $Q(X,Y_0, \ldots, Y_{s-1})$ over $\F_q$ of the form
$$Q(X,Y_0, \ldots, Y_{s-1}) = A (X) +  B_0 (X) \cdot Y_0 + \cdots+ B_{s-1}(X)\cdot Y_{s-1} ,$$ 
where $A(X)$ has degree smaller than $t$ and each $B_\ell(X)$ has degree smaller than $t-k+1$, and which satisfies that $Q(a_i,w_i)=0$ 
for any $i \in [n]$ (where we view $w_i = (w_{i,0}, \ldots, w_{i,s-1}) \in \F_q^s$ as an assignment to $s$ variables in $\F_q$).

Such a polynomial $Q(X,Y_0, \ldots, Y_{s-1})$ can be found by solving a system of linear equations which is obtained by viewing the  coefficients of the monomials in $Q$ as unknowns, and viewing each requirement of the form $Q(a_i,w_i)=0$ as a homogeneous linear constraint on these unknowns.       

\item  \label{step:mult_no_mult_list_root_finding}
Find a basis for the list $\calL'$ of all univariate polynomials $f(X)$ of degree smaller than $k$ so that 
$$Q(X, f(X), f^{(1)}(X), \ldots, f^{(s-1)}(X))=0.$$
This basis can be found by solving a system of linear equations which is obtained by viewing the  $k$ coefficients of $f(X)$ as unknowns, and viewing the requirement that each of the $t$ coefficients of $Q(X, f(X), f^{(1)}(X), \ldots, f^{(s-1)}(X))$ is zero as $t$ (non-homogeneous) linear constraints on these unknowns.     
 
\item \label{step:mult_no_mult_list_check} Output the list $\calL$ of all codewords $(f(a_1),f(a_2),\ldots, f(a_n))$ corresponding to 
univariate polynomials $f(X) \in \calL'$ of degree smaller than $k$ which satisfy that $f^{(<s)}(a_i) = w_i$ for at least $t$ indices $i \in [n]$.
\end{enumerate}

  \medskip

  \end{boxedminipage}

\caption{List decoding of multiplicity codes}
\label{fig:mult_no_mult_list}
\end{figure}

\paragraph{Correctness.}  Suppose that $f(X) \in \F_q[X]$ is a univariate polynomial of degree smaller than $k$ so that $f^{(<s)}(a_i) = w_i$ for at least $t$ indices $i \in [n]$, we shall show that $f(X) \in \calL$. By Lemma \ref{lem:mult_no_mult_list_interpolation}, there exists a non-zero $(s+1)$-variate polynomial $Q(X,Y_0, \ldots, Y_{s-1})$ satisfying the requirements on Step \ref{step:mult_no_mult_list_interpolation}. By Lemma \ref{lem:mult_no_mult_list_P_is_root}, we have that 
$Q(X,f(X), f^{(1)}(X), \cdots, f^{(s-1)}(X))=0,$ and so  $f(X)$ will be included in $\calL'$ on Step \ref{step:mult_no_mult_list_root_finding}. Clearly, $f(X)$ will also pass the checks on Step \ref{step:mult_no_mult_list_check}, and thus $f(X)$ will also be included in $\calL$.

\paragraph{Efficiency.} Step \ref{step:mult_no_mult_list_interpolation} can be performed in time $\poly(n, s, \log (q))$ by solving a system of $n$ linear equations in $n+1$ variables using Gaussian elimination. Step \ref{step:mult_no_mult_list_root_finding} can be performed in time $\poly(n,s, \log (q))$ by solving a system of $t$ linear equations in $k$ variables. 
By Lemma \ref{lem:mult_no_mult_list_solution_size},
Step \ref{step:mult_no_mult_list_check} can be performed in time $|\calL| \cdot \poly(n,s, \log (q)) \leq q^s \cdot \poly(n,s, \log (q)) $, where the latter bound is polynomial in the block length for a constant $s$ and $q=\poly(n)$. 

\subsection{List decoding multiplicity codes up to capacity}\label{subsubsec:mult_mult}

We now show how to extend the list decoding algorithm of the previous section to an algorithm that list decodes multiplicity codes up to capacity.

As before, let $\mult_q^{(s)}(a_1, \ldots, a_n; k)$ be the multiplicity code, 
where $q$ is a prime power, $a_1,a_2,\ldots,a_n$ are distinct points in $\F_q$, and $s,k$ are 
positive integers so that   $k< sn$ and $\max\{k,s\} \leq \char(\F_q)$. Suppose also that $w\in (\F_q^s)^n$ is a received word, where $w_i = (w_{i,0}, \ldots w_{i,s-1}) \in \F_q^s$ for any $ i \in[n]$.
Let $r \in \{1,2,\ldots, s\}$ be a parameter to be determined later on.
We shall show how to compute the list of all polynomials $f(X) \in \F_q[X]$ of degree smaller than $k$ so that $f^{(<s)}(a_i) \neq w_i$ for at most $e := \frac{(s-r+1)  r  n  - r   (k-1) - 1} {(s-r+1)  (r+1)} $  indices $i \in [n]$. That is, $f^{(<s)}(a_i) = w_i$ for at least $t: =n -e =  \frac{(s-r+1)  n  + r  (k-1)+1} {(s-r+1)  (r+1)} $ indices $i \in [n]$. 

Note that the algorithm from the previous section corresponds to the special case that $r=s$. 
Moreover, if we let $R:=\frac{k} {s \cdot n}$ denote the rate of the  $\mult_q^{(s)}(a_1, \ldots, a_n; k)$ code, then this algorithm will be able to list decode from up to  roughly a $\frac{r} {r+1} \left(1 -  \frac {s} {s-r+1} \cdot R\right)$-fraction of errors in time roughly $q^r$. Thus, for any $\epsilon >0$, we can choose $r = \Theta( 1/\epsilon)$ and $s = \Theta( 1/\epsilon^2)$ so that the fraction of errors is at least $1-R-\epsilon$ and the 
running time (and list size) is roughly $q^{r}  =q^{1/\epsilon}$, which is polynomial in the block length for $q=\poly(n)$. Later, in Section \ref{subsec:const_list_mult} we shall show that the list size is in fact a constant independent of the block length (and even matches the generalized Singlton Bound).  

\paragraph{Overview.}
The improved algorithm is based on the \emph{method of multiplicities}, similarly to the way the algorithm presented in Section \ref{subsec:rs_johnson} for list decoding of Reed-Solomon Codes up to the Johnson Bound improves on the algorithm presented in Section \ref{subsec:rs_list}. 

In more detail, recall that in the previous Algorithm \ref{fig:mult_no_mult_list} for list-decoding of the multiplicity code, in the interpolation step, we searched  for a non-zero low-degree $(s+1)$-variate polynomial $Q(X,Y_0, \ldots, Y_{s-1}) = A(X)+\sum_{\ell=0}^{s-1} B_\ell(X) Y_\ell$
so that $Q(a_i, w_i)=0$ for any $i \in [n]$, and in the root finding step, we searched for all univariate polynomials $f(X) \in \F_q[X]$ of degree smaller than $k$ satisfying the differential equation $Q(X,f(X), f^{(1)}(X), \cdots, f^{(s-1)}(X)) = 0$. 

Furthermore, the main properties we needed out of $Q$ 
were that the number of unknown coefficients in $Q$ is greater than the number $n$ of constraints of the form $Q(a_i,w_i)=0$, and that the $(1,k, \ldots, k)$-weighted degree of $Q$ is smaller than the agreement parameter $t$. The first property guarantees the existence of a non-zero $Q$, while the second property  guarantees that for any univariate polynomial $f(X) \in \F_q[X]$ of degree smaller than $k$ so that $f^{(<s)}(a_i)=w_i$ for at least $t$ of the indices $i \in [n]$, it holds that
 $Q(X,f(X), f^{(1)}(X), \cdots, f^{(s-1)}(X)) $ is a univariate polynomial of degree smaller than $t$ that has at least $t$ roots, and consequently it is the zero polynomial.

In order to list-decode the multiplicity code up to capacity,  we would like to set the agreement parameter $t$ to be even smaller than in Algorithm \ref{fig:mult_no_mult_list}. 
As was previously the case, this means that $Q(X,f(X), f^{(1)}(X), \cdots, f^{(s-1)}(X))$ has a smaller number of roots, which forces the degree of $Q$ to be smaller in order to argue that it is the zero polynomial. Consequently, $Q$ may potentially have less than $n$ monomials, and so we may not be able to argue the existence of a non-zero polynomial $Q$.

To overcome this, 
we define $Q$ to be an $(r+1)$-variate polynomial  $Q(X,Y_0, \ldots, Y_{r-1})$ for some $r<s$, 
and we require that $Q(X, f(X), f^{(1)}(X) ,\ldots, f^{(r-1)}(X))$ vanishes on $a_i$ with \emph{multiplicity} $s-r+1$
 for any point $i \in [n]$ for which $f^{(<s)}(a_i)=w_i$. 
Thus $Q(X,f(X), f^{(1)}(X), \cdots, f^{(s-1)}(X))$ has more roots, counted with multiplicities, and consequently it can have higher degree and more coefficients, which may be helpful in arguing the existence of a non-zero polynomial $Q$. However, this also comes with a price, as we also increased the number of constraints on $Q$ by a multiplicative factor of $s-r+1$, and consequently we now require that $Q$ has more coefficients than before. It turns out however that this trade-off is favorable, and allows to decrease the agreement parameter $t$ to the minimum possible for an appropriate setting of $r$ and $s$.

In what follows, we describe separately the changes made in the interpolation and root finding steps, and analyze their correctness and efficiency.

\paragraph{Step 1: Interpolation.}

In the interpolation step, we would now like to find a non-zero low degree $(r+1)$-variate polynomial $Q(X,Y_0, \ldots, Y_{r-1})$ which satisfies that $P_f(X):=Q(X, f(X), f^{(1)}(X) ,\ldots, f^{(r-1)}(X)) $ vanishes on $a_i$ with multiplicity $s-r+1$  for any point $i \in [n]$ for which $f^{(<s)}(a_i)=w_i$.  In the algorithm for list decoding of Reed-Solomon Codes up to the Johnson Bound (cf., Figure \ref{fig:rs_list_johnson}) this was guaranteed by requiring that $Q$ vanishes with high \emph{multivariate} multiplicity on any point $(a_i,w_i)$ (cf., Claim \ref{clm:gs-zero-at-agreement}). However, 
requiring that $Q$ vanishes with high multivariate multiplicity on all these points turns out to be too costly in the current setting, and instead we directly enforce the property that $P_f(X)$ vanishes on each $a_i$ for which $f^{(<s)}(a_i)=w_i$
with high \emph{univariate} multiplicity.\footnote{Specifically, as we show below, enforcing that  $P_f$ vanishes with univariate multiplicity at least $s-r+1$ on $a_i$ imposes $s-r+1$ linear constraints on the coefficients of $Q$, while enforcing that $Q$ vanishes with mulivariate multiplicity at least $s-r+1$ on $a_i$ imposes $ {(r+1)+(s-r+1)-1 \choose r+1} = {s+1 \choose r+1} \approx s^r$ linear constraints on the coefficients of $Q$.}

To understand what requirement this imposes on $Q$, recall that by 
(\ref{eq:mult_no_mult_uni_der}), for any $j \in \N$, 
we have that 
$$
P_f^{(j)} (X)  =  A^{(j)}(X) + \sum_{\ell=0}^{r-1} \sum_{h=0}^{j} {h+ \ell \choose \ell} B_\ell^{(j - h)}(X)  \cdot f^{(h+\ell)}(X). 
$$
Further note that in the case that $f^{(<s)}(a_i)=w_i$, for any $j \in \{0,1,\ldots, s-r\}$,
we have all derivatives of the form $f^{(h+\ell)}(a_i)$ for $\ell \in \{0,1, \ldots, r-1\}$  and $h \in \{0,1, \ldots,j\} $
available for us.

\begin{lemma}\label{lem:mult_cap_interpolation}
Let $t  =  \frac{(s-r+1) n  + r (k-1)+1} {(s-r+1)  (r+1)}.$
Then there exists a non-zero $(r+1)$-variate polynomial $Q(X,Y_0, \ldots, Y_{r-1})$ over $\F_q$ of the form
$$Q(X,Y_0, \ldots, Y_{r-1})=A(X)+B_0(X) \cdot Y_0+ \cdots + B_{r-1}(X) \cdot Y_{r-1},$$ 
where $A(X)$ has degree smaller than $(s-r+1) \cdot t$ and each $B_\ell(X)$ has degree smaller than
$(s-r+1) \cdot t-k+1$, and which satisfies that
\begin{equation}\label{eq:mult_cap_interpolation}
 A^{(j)}(a_i) + \sum_{\ell=0}^{r-1} \sum_{h=0}^{j} {h + \ell \choose \ell} B_\ell^{(j - h)}(a_i)  \cdot w_{i,h+\ell}=0
\end{equation}
for any $i \in [n]$ and $j=0,1,\ldots, s-r$.

Furthermore, such a polynomial $Q(X,Y_0, \ldots, Y_{r-1})$ can be found by solving a system of linear equations which is obtained by viewing the  coefficients of the monomials in $A(X),B_0(X), \ldots, B_{r-1}(X)$ as unknowns, and viewing each requirement of the form  (\ref{eq:mult_cap_interpolation})  as a homogeneous linear constraint on these unknowns.       
\end{lemma}

\begin{proof}
If we view the coefficients of the monomials in $A(X),B_0(X), \ldots, B_{r-1}(X)$ as unknowns, then each requirement of the form  (\ref{eq:mult_cap_interpolation}) imposes a homogeneous linear constraint on these unknowns. 

Further note that the number of unknowns is 
$$(s-r+1) \cdot t+ r \cdot ( (s-r+1)\cdot t-k+1) = (r+1) \cdot  (s-r+1) \cdot t- r \cdot (k-1) = (s-r+1)\cdot n+1,$$
while the number of linear constraints is $(s-r+1) \cdot n $, and consequently this linear system has a non-zero solution.
\end{proof}

\paragraph{Step 2: Root finding.} 

 In this step, we search for all univariate polynomials $f(X) \in \F_q[X]$ of degree smaller than $k$ satisfying  that $Q(X, f(X), f^{(1)}(X), \ldots, f^{(r-1)}(X) )=0$, and we would like to show that any 
univariate polynomial $f(X) \in \F_q[X]$ of degree smaller than $k$ so that  $f^{(<s)}(a_i) = w_i$ for at least $t$ indices $i \in [n]$ satisfies that $Q(X, f(X), f^{(1)}(X), \ldots, f^{(r-1)}(X) )=0$. We have already shown how to perform the former in Lemma \ref{lem:mult_no_mult_list_solution_size}, the following lemma shows the latter.

\begin{lemma}\label{lem:mult_cap_P_is_root}
Suppose that $Q(X,Y_0, \ldots, Y_{r-1})$ is a non-zero $(r+1)$-variate polynomial over $\F_q$ of the form
$$Q(X,Y_0, \ldots, Y_{r-1})=A(X)+B_0(X) \cdot Y_0+ \cdots + B_{r-1}(X) \cdot Y_{r-1},$$ 
where $A(X)$ has degree smaller than $(s-r+1) \cdot t$ and each $B_\ell(X)$ has degree smaller than $(s-r+1) \cdot t-k+1$, and which satisfies (\ref{eq:mult_cap_interpolation}) for any $i \in [n]$ and $j=0,1, \ldots, s-r$.
Let $f(X) \in \F_q[X]$ be a univariate polynomial  of degree smaller than $k$ so that $f^{(<s)}(a_i) = w_i$ for at least $t$ indices $i \in [n]$. Then 
$$Q(X, f(X), f^{(1)}(X), \ldots, f^{(r-1)}(X) )=0.$$
\end{lemma}

\begin{proof}
The polynomial $P_f(X):=Q(X, f(X), f^{(1)}(X), \ldots, f^{(r-1)}(X) )$ is a univariate polynomial  of degree smaller than $(s-r+1) \cdot t$. 
Furthermore, by  (\ref{eq:mult_no_mult_uni_der}), for any $i \in [n]$ for which $f^{(<s)}(a_i) = w_i$, it holds that:
\begin{eqnarray*}
P_f^{(j)}(a_i) & = & A^{(j)}(a_i) + \sum_{\ell=0}^{r-1} \sum_{h=0}^{j} {h + \ell \choose \ell} B_\ell^{(j - h)}(a_i)  \cdot f^{(h+\ell)}(a_i) \\
& = & A^{(j)}(a_i) + \sum_{\ell=0}^{r-1} \sum_{h=0}^{j} {h + \ell \choose \ell} B_\ell^{(j - h)}(a_i)  \cdot w_{i,h+\ell} 
\end{eqnarray*}
for any $j=0,1, \ldots, s-r$. So $P_f(X)$ vanishes on $a_i$ with multiplicity $s-r+1$ for any $i \in [n]$ for which $f^{(<s)}(a_i) = w_i$.
Thus $P_f(X)$ is a univariate polynomial of degree smaller than $(s-r+1) \cdot t$ with at least $ t$ roots of multiplicity $s-r+1$, and so by Fact \ref{fact:prelim_hasse_sz_univ} it must be the zero polynomial. 
\end{proof}

The full description of the algorithm appears  in Figure \ref{fig:mult_cap} below, followed by correctness and efficiency analysis.

\begin{figure}[h]
  \begin{boxedminipage}{\textwidth} \small \medskip \noindent
    $\;$

    \underline{\textbf{List decoding up to capacity of $\MULT_q^{(s)}(a_1, \ldots,a_n;k)$:}}
    
    \medskip
\begin{itemize}
\item \textbf{INPUT:} $w\in (\F_q^s)^n$, where $w_i = (w_{i,0}, \ldots, w_{i,s-1}) \in \F_q^s$ for each $i \in [n]$.
\item \textbf{OUTPUT:} 
The list $\calL$ of all polynomials $f(X) \in \F_q[X]$ of degree smaller than $k$ so that $f^{(<s)}(a_i)= w_i$ for at least $t  =  \frac{(s-r+1) n  + r (k-1)+1} {(s-r+1)  (r+1)}$ indices $i \in [n]$.
\end{itemize}

    \begin{enumerate}
\item \label{step:mult_cap_interpolation} Find a non-zero $(r+1)$-variate polynomial $Q(X,Y_0, \ldots, Y_{r-1})$ over $\F_q$ of the form
$$Q(X,Y_0, \ldots, Y_{r-1}) = A (X) +  B_0 (X) \cdot Y_0 + \cdots+ B_{r-1}(X)\cdot Y_{r-1} ,$$ 
where $A(X)$ has degree smaller than $(s-r+1) \cdot t$ and each $B_\ell(X)$ has degree smaller than $(s-r+1)\cdot (t-k+1)$, and which satisfies that 
 $$A^{(j)}(a_i) + \sum_{\ell=0}^{r-1} \sum_{h=0}^{j} {h + \ell \choose \ell} B_\ell^{(j - h)}(a_i)  \cdot w_{i,h+\ell}=0$$
for any $i \in [n]$ and $j=0,1,\ldots, s-r$.

Such a polynomial $Q(X,Y_0, \ldots, Y_{r-1})$ can be found by solving a system of linear equations which is obtained by viewing the  coefficients of the monomials in $Q$ as unknowns, and viewing each requirement as above  as a homogeneous linear constraint on these unknowns.       

\item  \label{step:mult_cap_root_finding}
Find a basis for the list $\calL'$ of all univariate polynomials $f(X)$ of degree smaller than $k$ so that 
$$Q(X, f(X), \ldots, f^{(r-1)}(X))=0.$$
This basis can be found by solving a system of linear equations which is obtained by viewing the  $k$ coefficients of $f(X)$ as unknowns, and viewing the requirement that each of the $t$ coefficients of $Q(X, f(X), f^{(1)}(X), \ldots, f^{(r-1)}(X) )$ is zero as $t$ (non-homogeneous) linear constraints on these unknowns.     
 
\item \label{step:mult_cap_check} Output the list $\calL$ of all codewords $(f(a_1),f(a_2), \ldots, f(a_n))$ corresponding to 
univariate polynomials $f(X) \in \calL'$ of degree smaller than $k$ which satisfy that $f^{(<s)}(a_i) = w_i$ for at least $t$ indices $i \in [n]$.
\end{enumerate}

  \medskip

  \end{boxedminipage}

\caption{List decoding of multiplicity codes up to capacity}
\label{fig:mult_cap}
\end{figure}

\paragraph{Correctness.}  Suppose that $f(X) \in \F_q[X]$ is a univariate polynomial of degree smaller than $k$ so that $f^{(<s)}(a_i) = w_i$ for at least $t$ indices $i \in [n]$, we shall show that $f(X) \in \calL$. By Lemma \ref{lem:mult_cap_interpolation}, there exists a non-zero $(r+1)$-variate polynomial $Q(X,Y_0, \ldots, Y_{r-1})$ satisfying the requirements on Step \ref{step:mult_cap_interpolation}. By Lemma \ref{lem:mult_cap_P_is_root}, we have that 
$Q(X, f(X), f^{(1)}(X), \ldots, f^{(r-1)}(X) )=0,$ and so  $f(X)$ will be included in $\calL'$ on Step \ref{step:mult_cap_root_finding}. Clearly, $f(X)$ will also pass the checks on Step \ref{step:mult_cap_check}, and thus $f(X)$ will also be included in $\calL$.

\paragraph{Efficiency.} Step \ref{step:mult_cap_interpolation} can be performed in time $\poly(n, s, \log (q))$ by solving a system of $(s-r+1) \cdot n$ linear equations in $(s-r+1) \cdot n+1$ variables using Gaussian elimination. Step \ref{step:mult_cap_root_finding} can be performed in time $\poly(n,s, \log (q))$ by solving a system of $t$ linear equations in $k$ variables. By Lemma
\ref{lem:mult_no_mult_list_solution_size}, 
Step \ref{step:mult_cap_check} can be performed in time $|\calL| \cdot \poly(n,s, \log(q)) \leq q^{r} \cdot \poly(n,s, \log (q)) $, where the latter bound is polynomial in the block length for a constant $r$ and $q=\poly(n)$.

\subsection{List decoding Reed-Solomon Codes over subfield evaluation points up to capacity}\label{subsec:rs_subfield}

We now turn back to the basic family of Reed-Solomon Codes, and consider a special subclass of these codes which are defined over an extension field, and \emph{evaluated over a subfield}. We show that an algorithm similar to the one presented in Section \ref{subsubsec:mult_no_mult} can be used to list decode these codes  up to capacity, albeit with only a slightly non-trivial sub-exponential running time. This in particular implies a sub-exponential upper bound on the list size of these codes, and later on, in Section \ref{subsec:const_list_rs_subfield} we shall show that the list-size can be reduced to a constant independent of the block length (and the running time to polynomial in the block length) by bypassing to a (carefully chosen) subcode. 

In more detail, in what follows fix a prime power $q$, distinct points $a_1, \ldots,a_n \in \F_q$, 
a positive integer $k$ such that $k<n$, and a positive integer $s$, and consider the Reed-Solomon Code $\RS_{q^s}(a_1, \ldots, a_n; k)$ that is defined over the \emph{extension field} $\F_{q^s}$.
We shall show that  $\RS_{q^s}(a_1, \ldots, a_n; k)$ can be list-decoded with similar guarantees to that of the list-decoding algorithm for multiplicity codes beyond the unique decoding radius (Algorithm \ref{fig:mult_no_mult_list}).
Suppose that $w\in (\F_{q^s})^n$ is a received word.
Let $r \in \{1,\ldots, s\}$ be a parameter to be determined later on.
We shall show how to compute the list of all polynomials $f(X) \in \F_{q^s}[X]$ of degree smaller than $k$ so that $f(a_i) \neq w_i$ for at most $e = \frac{r (n-k+1)-1} {r+1} $  indices $i \in [n]$. That is, $f(a_i) = w_i$ for at least $t: =n -e = \frac{n+r (k-1) +1} {r+1} $ indices $i \in [n]$. 

The main advantage however of the Reed-Solomon Code $\RS_{q^s}(a_1, \ldots,a_n;k)$  over the corresponding multiplicity code $\MULT_q^{(s)}(a_1, \ldots,a_n;k)$ is that the rate of  $\RS_{q^s}(a_1, \ldots,a_n;k)$  is higher by a multiplicative factor of $s$. 
More specifically, if we let $R:=\frac{k} { n}$ denote the rate of the $\RS_{q^s}(a_1, \ldots,a_n;k)$ code, then this algorithm will be able to list decode from up to roughly a $\frac {r} {r+1} (1- R)$-fraction of errors in time roughly $q^{r k}  $.  
In particular, for any constant $\epsilon >0$, we can choose $r = \Theta( 1/\epsilon)$  and 
$s = \Theta( 1/\epsilon^2)$ so that the fraction of errors is at least $1-R-\epsilon$, and the running time (and list size) is roughly $q^{r k}  =q^{\epsilon sk}$, which is
sub-polynomial in the number $q^{sk}$ of distinct codewords.  Later, in Section \ref{subsec:const_list_rs_subfield}, we shall show that the list size can be reduced to a constant independent of the block length (and the running time to polynomial in the block length) by passing to an appropriate subcode. 

\paragraph{Overview.}
The list-decoding algorithm is very similar to the list-decoding algorithm for multiplicity codes beyond the unique decoding radius (Algorithm \ref{fig:mult_no_mult_list}). Recall that in Algorithm \ref{fig:mult_no_mult_list}, in 
the interpolation step, we searched for a non-zero low-degree $(s+1)$-variate polynomial $Q(X,Y_0, \ldots, Y_{s-1})$ of the form 
$$
Q(X,Y_0, \ldots, Y_{s-1})=A(X)+B_0(X) \cdot Y_0+ \cdots + B_{s-1}(X) \cdot Y_{s-1},
$$
so that 
$
Q(a_i,w_i)=0
$
for any $i \in [n]$, and in the root finding step, we searched for all univariate polynomials $f(X) \in \F_q[X]$ satisfying that 
$$
Q(X, f(X), f^{(1)}(X), \ldots, f^{(s-1)}(X)) =  0.
$$

The main observation we shall use for list decoding of Reed-Solomon Codes over subfield evaluation points is that for any  polynomial $f(X)= \sum_{i=0}^{k-1} f_i X^i $ over the extension field $\F_{q^s}$ and for any point $a$ in the base field $\F_q$, 
$$ (f(a))^q=\left(\sum_{i=0}^{k-1} f_i a^i\right)^q=\sum_{i=0}^{k-1} (f_i)^q (a^q)^i =  \sum_{i=0}^{k-1} (f_i)^q a^i,$$
where the second equality uses the fact 
that raising to the power of $q$ is an $\F_q$-linear transformation, and the third equality uses the fact that $a^q = a$ for any $a \in \F_q$.
Motivated by this, for a polynomial $f(X)= \sum_{i=0}^{k-1} f_i X^i \in \F_{q^s}[X]$, we let $\sigma(f) = \sum_{i=0}^{k-1} (f_i)^q X^i \in \F_{q^s}[X]$. Under this notation, we have that 
$  (f(a))^q = \sigma(f)(a)$ for any polynomial $f(X) \in \F_{q^s}[X]$ and  point $a \in \F_q$.

Using the above, in the list-decoding algorithm for $\RS_{q^s}(a_1, \ldots,a_n;k)$, in the interpolation step, we once more search for a non-zero low degree $(r+1)$-variate polynomial of the form $Q(X)=A(X)+\sum_{\ell=0}^{r-1} B_\ell(X) \cdot Y^\ell$
(with coefficients in the extension field $\F_{q^s}$). But now in the interpolation step, we require that
$$
Q\left(a_i, w_i, (w_i)^q, \ldots, (w_i)^{q^{r-1}}\right)=0
$$
for any $i \in [n]$. 
Then in the root finding step, we we search for all univariate polynomials $f(X) \in \F_{q^s}[X]$
satisfying that 
$$
Q\left(X, f(X), \sigma(f)(X), \ldots, \sigma^{r-1}(f)(X) \right)=0.
$$

In what follows, we describe separately the interpolation and root finding steps, and analyze their correctness and efficiency.

\paragraph{Step 1: Interpolation.}

In this step, we would like to find a non-zero low-degree $(r+1)$-variate polynomial $Q(X,Y_0, \ldots, Y_{r-1})$ satisfying 
that 
$Q\left(a_i, w_i, (w_i)^q, \ldots, (w_i)^{q^{r-1}}\right)=0
$.
The proof of the following lemma is identical to the proof of Lemma \ref{lem:mult_no_mult_list_interpolation}.

\begin{lemma}\label{lem:rs_subfield_cap_interpolation}
Let $t= \frac{n+r (k-1) +1} {r+1}$.
Then there exists a non-zero $(r+1)$-variate polynomial $Q(X,Y_0, \ldots, Y_{r-1})$ over $\F_{q^s}$ of the form
$$Q(X,Y_0, \ldots, Y_{r-1})=A(X)+B_0(X) \cdot Y_0+ \cdots + B_{r-1}(X) \cdot Y_{r-1},$$ 
where $A(X)$ has degree smaller than $t $ and each $B_\ell(X)$ has degree smaller than $t-k+1$, and which satisfies that 
$$
Q\left(a_i, w_i, (w_i)^q, \ldots, (w_i)^{q^{r-1}}\right)=0
$$
for any $i \in [n]$.  

Furthermore, such a polynomial $Q(X,Y_0, \ldots, Y_{r-1})$ can be found by solving a system of linear equations which is obtained by viewing the  coefficients of the monomials in $A(X),B_0(X), \ldots, B_{r-1}(X)$ as unknowns, and viewing each requirement of the form 
$Q\left(a_i, w_i, (w_i)^q, \ldots, (w_i)^{q^{r-1}}\right)=0
$ as a homogeneous linear constraint on these unknowns.       
\end{lemma}

\paragraph{Step 2: Root finding.} 

 In this step, we search for all univariate polynomials $f(X) \in \F_{q^s}[X]$ satisfying 
 $
Q\left(X, f(X), \sigma(f)(X), \ldots, \sigma^{r-1}(f)(X) \right)=0,
$
 and we would like to show that any 
univariate polynomial $f(X) \in \F_{q^s}[X]$ of degree smaller than $k$ so that  $f(a_i) = w_i$ for at least $t$ indices $i \in [n]$ satisfies that $
Q\left(X, f(X), \sigma(f)(X), \ldots, \sigma^{r-1}(f)(X) \right)=0
$. We start by showing the latter.

\begin{lemma}\label{lem:rs_subfield_cap_P_is_root}
Suppose that $Q(X,Y_0, \ldots, Y_{r-1})$ is a non-zero $(r+1)$-variate polynomial over $\F_{q^s}$ of the form
$$Q(X,Y_0, \ldots, Y_{r-1})=A(X)+B_0(X) \cdot Y_0+ \cdots + B_{r-1}(X) \cdot Y_{r-1},$$ 
where $A(X)$ has degree smaller than $t $ and each $B_\ell(X)$ has degree smaller than $t-k+1$, and which satisfies that
$$
Q\left(X, f(X), \sigma(f)(X), \ldots, \sigma^{r-1}(f)(X) \right)=0
$$
for any $i \in [n]$.  
Let $f(X) \in \F_{q^s}[X]$ be a univariate polynomial  of degree smaller than $k$ so that $f(a_i) = w_i$ for at least $t$ indices $i \in [n]$. Then 
$$Q\left(X, f(X), \sigma(f)(X), \ldots, \sigma^{r-1}(f)(X) \right) = 0.$$
\end{lemma}

\begin{proof} 
The polynomial $P_f(X):=Q\left(X, f(X), \sigma(f)(X), \ldots, \sigma^{r-1}(f)(X) \right) $ is a univariate polynomial over $\F_{q^s}$ of degree smaller than $t$, which satisfies that 
\begin{eqnarray*}
P_f(a_i) & = & Q \left(a_i, f(a_i), \sigma(f)(a_i), \ldots, \sigma^{r-1}(f)(a_i) \right) \\
& = & Q \left(a_i, f(a_i), (f(a_i))^q, \ldots, (f(a_i))^{q^{r-1}} \right) \\
& = &  Q \left(a_i, w_i, (w_i)^q, \ldots, (w_i)^{q^{r-1}} \right) =0
\end{eqnarray*} 
for any $i \in [n]$ for which $f(a_i) = w_i$.  Thus $P_f(X)$ is a univariate polynomial of degree smaller than $t$ with at least $ t$ roots, and so it must be the zero polynomial. 
\end{proof}

The next lemma bounds the number of univariate polynomials $f(X) \in \F_{q^s}[X]$ satisfying the equation 
$
Q\left(X, f(X), \sigma(f)(X), \ldots, \sigma^{r-1}(f)(X) \right)=0.
$

\begin{lemma}\label{lem:rs_subfield_cap_solution_size}
Suppose that $Q(X,Y_0, \ldots, Y_{r-1})$ is a non-zero $(r+1)$-variate polynomial over $\F_{q^s}$ of the form
$$Q(X,Y_0, \ldots, Y_{r-1}) = A (X) +  B_0 (X) \cdot Y_0+ \cdots +B_{r-1}(X)\cdot Y_{r-1},$$
and let $\calL$ be the list containing all polynomials $f(X) \in \F_{q^s}[X]$  of degree smaller than $k$ so that  
$$Q\left(X, f(X), \sigma(f)(X), \ldots, \sigma^{r-1}(f)(X)\right) = 0.$$ Then $|\calL| \leq q^{(r-1) \cdot k}$. 

Furthermore, $\calL$ forms an  affine subspace over $\F_q$ of dimension at most $(r-1)k$. A basis for this subspace can be found by solving a system of linear equations over $\F_q$ which is obtained by viewing the  coefficients of $f(X)$ (represented as strings in $\F_q^s$) as unknowns, and viewing the requirement that each of the coefficients of $Q\left(X, f(X), \sigma(f)(X), \ldots, \sigma^{r-1}(f)(X)\right)$ is zero as a  (non-homogeneous) $\F_q$-linear constraint on these unknowns.     
\end{lemma}

\begin{proof}
As in the proof of Lemma \ref{lem:mult_no_mult_list_solution_size}, we may assume that there exists an element $a \in \F_q^s$ so that $B_{r-1} (a) \neq 0$. Without loss of generality, we may further assume that $a=0$, since otherwise we can translate all polynomials $A(X), B_0(X), \ldots, B_{r-1}(X)$ by $a$, which corresponds to translating all polynomials in 
$\calL$ by $a$, which does not change the size of $\calL$.

Let $f(X)=\sum_{j=0}^{k-1} f_j X^j $ be a polynomial of degree smaller than $k$ over $\F_{q^s}$ satisfying that 
\begin{eqnarray*}
P_f(X) & := & Q\left(X, f(X), \sigma(f)(X), \ldots, \sigma^{r-1}(f)(X)\right) \\
& = &A(X) + \sum_{\ell=0}^{r-1} B_\ell(X)  \cdot \sigma^{\ell}(f)(X) \\ 
& = &A(X) + \sum_{\ell=0}^{r-1}  B_\ell(X)  \sum_{j=0}^{k-1}  (f_j)^{q^\ell} X^j  \\
& = & 0.
\end{eqnarray*}
We shall show that for any $j \in \{0,1,\ldots, k-1\}$, given the first $j-1$ coefficients $f_0, f_1,\ldots, f_{j-1}$, there are at most $q^{r-1}$ options for $f_j$, and consequently $|\calL| \leq q^{(r-1)k}$.

To show the above, fix $j \in \{0,1,\ldots, k-1\}$. Since $P_f(X)$ is the zero polynomial, the coefficient of $X^j$ in 
$P_f(X)$ must be zero. On the other hand, inspecting the above expression for $P_f(X)$, we see that the coefficient of $X^j$ in $P_f(X)$ can be written as $y_j+\sum_{\ell=0}^{r-1} B_\ell(0)  \cdot (f_j)^{q^{\ell}}$, where $y_j \in \F_{q^s}$ is determined by $f_0, f_1, \ldots, f_{j-1}$. 

Let
$$
B(X):= B_0(0) \cdot X  + B_1(0) \cdot X^q+ \cdots + B_{r-1}(0) \cdot X^{q^{r-1}}.
$$
Then by the above,  $f_j \in \F_{q^s}$ must satisfy that $B(f_j) = -y_j$. 
On the other hand, by assumption that $B_{r-1}(0)\neq 0$, $B(X)$ is a non-zero polynomial of degree $q^{r-1}$, and so it attains any value for at most $q^{r-1}$ assignments in $\F_{q^s}$. 
Consequently, for any $j \in \{0,1,\ldots, k-1\}$, given prior coefficients $f_0, \ldots, f_{j-1}$, there are at most $q^{r-1}$ options for the coefficient $f_j$, which implies in turn that $|\calL| \leq q^{(r-1)k}$.
 
The moreover part follows by the linear structure of $Q$, and by the fact that raising to a power of $q$ is an $\F_q$-linear transformation.    
\end{proof}

\medskip

The full description of the algorithm appears  in Figure \ref{fig:rs_subfield_cap} below, followed by correctness and efficiency analysis.

\begin{figure}[h]
  \begin{boxedminipage}{\textwidth} \small \medskip \noindent
    $\;$

    \underline{\textbf{List decoding of $\RS_{q^s}(a_1, \ldots,a_n;k)$ for $a_1, \ldots, a_n \in \F_q$:}}
    
    \medskip
\begin{itemize}
\item \textbf{INPUT:} $w\in (\F_{q^s})^n$. 
\item \textbf{OUTPUT:} 
The list $\calL$ of all polynomials $f(X) \in \F_{q^s}[X]$ of degree smaller than $k$ so that $f(a_i)= w_i$ for at least $t=\frac{n+r(k-1) +1} {r+1}$ indices $i \in [n]$.
\end{itemize}

    \begin{enumerate}
\item \label{step:rs_subfield_cap_interpolation} Find a non-zero $(r+1)$-variate polynomial $Q(X,Y_0, \ldots, Y_{r-1})$ over $\F_q$ of the form
$$Q(X,Y_0, \ldots, Y_{r-1}) = A (X) +  B_0 (X) \cdot Y_0 + \cdots+ B_{r-1}(X)\cdot Y_{r-1} ,$$ 
where $A(X)$ has degree smaller than $t$ and each $B_\ell(X)$ has degree smaller than $t-k+1$, and which satisfies that $Q(a_i,(w_i)^q, \ldots, (w_i)^{q^{r-1}})=0$ 
for any $i \in [n]$.

Such a polynomial $Q(X,Y_0, \ldots, Y_{r-1})$ can be found by solving a system of linear equations which is obtained by viewing the  coefficients of the monomials in $Q$ as unknowns, and viewing each requirement of the form $Q(a_i,(w_i)^q, \ldots, (w_i)^{q^{s-1}})=0$  as a homogeneous linear constraint on these unknowns.       

\item  \label{step:rs_subfield_cap_root_finding}
Find a basis for the list $\calL'$ of all univariate polynomials $f(X)$ over $\F_{q^s}$ of degree smaller than $k$ so that 
$$Q\left(X, f(X), \sigma(f)(X), \ldots, \sigma^{r-1}(f)(X)\right)=0.$$
This basis can be found by solving a system of linear equations over $\F_q$ which is obtained by viewing the  $k$ coefficients of $f(X)$ (represented as strings over $\F_q^s$) as unknowns, and viewing the requirement that each of the $t$ coefficients of $Q\left(X, f(X), \sigma(f)(X), \ldots, \sigma^{r-1}(f)(X)\right)$ is zero as $t$ (non-homogeneous) $\F_q$-linear constraints on these unknowns.     
 
\item \label{step:rs_subfield_cap_check} Output the list $\calL$ of all codewords $(f(a_1), f(a_2), \ldots, f(a_n))$
corresponding to univariate polynomials $f(X) \in \calL'$ of degree smaller than $k$ which satisfy that $f(a_i) = w_i$ for at least $t$ indices $i \in [n]$.
\end{enumerate}

  \medskip

  \end{boxedminipage}

\caption{List decoding of Reed-Solomon Codes over subfield evaluation points}
\label{fig:rs_subfield_cap}
\end{figure}

\paragraph{Correctness.}  Suppose that $f(X) \in \F_{q^s}[X]$ is a univariate polynomial of degree smaller than $k$ so that $f(a_i) = w_i$ for at least $t$ indices $i \in [n]$, we shall show that $f(X) \in \calL$. By Lemma \ref{lem:rs_subfield_cap_interpolation}, there exists a non-zero $(r+1)$-variate polynomial $Q(X,Y_0, \ldots, Y_{r-1})$ satisfying the requirements on Step \ref{step:rs_subfield_cap_interpolation}. By Lemma \ref{lem:rs_subfield_cap_P_is_root}, we have that 
$Q\left(X, f(X), \sigma(f)(X), \ldots, \sigma^{r-1}(f)(X)\right)=0,$ and so  $f(X)$ will be included in $\calL'$ on Step \ref{step:rs_subfield_cap_root_finding}. Clearly, $f(X)$ will also pass the checks on Step \ref{step:rs_subfield_cap_check}, and thus $f(X)$ will also be included in $\calL$.

\paragraph{Efficiency.} Step \ref{step:rs_subfield_cap_interpolation} can be performed in time $\poly(n, s, \log (q))$ by solving a system of $n$ linear equations in $n+1$ variables using Gaussian elimination. 
Step \ref{step:rs_subfield_cap_root_finding} can be performed in time $\poly(n,s, \log (q))$ by solving a system of $t$ linear equations in $ks$ variables. 
By Lemma \ref{lem:rs_subfield_cap_solution_size},
Step \ref{step:rs_subfield_cap_check} can be performed in time $|\calL| \cdot \poly(n,s, \log (q)) \leq q^{(r-1)k} \cdot \poly(n,s, \log (q)) $, where the latter bound is sub-polynomial in the number $q^{sk}$ of distinct codewords for $r \leq \epsilon s$.

\subsection{Bibliographic notes}

In \cite{GR08_folded_RS}, Guruswami and Rudra introduced the family of \textsf{Folded Reed-Solomon (FRS) Codes}, 
and gave an efficient algorithm for list decoding these codes up to capacity, which gave the first explicit family of codes that can be efficiently list-decoded up to capacity. 
Folded Reed-Solomon Codes are defined similarly to Reed-Solomon Codes,  except that each codeword entry $f(a_i)$ is replaced with a \emph{tuple} of evaluations $(f(a_i), f(\gamma \cdot a_i), \ldots, f(\gamma^{s-1} \cdot a_i)) \in \F_q^s$, where $\gamma$ is a generator of the multiplicative group $\F_q^*$, $s \geq 1$ is a \emph{folding parameter}, and 
all sets $\{a_i, \gamma a_i, \ldots, \gamma^{s-1}a_i \}$ for $i \in [n]$ are pairwise disjoint. Guruswami and Rudra showed that for a sufficiently large $s$ (depending on the gap to capacity $\epsilon$), these codes are list-decodable up to capacity, with an algorithm similar to the one presented in Section \ref{subsec:rs_johnson} for list decoding of Reed-Solomon Codes up to the Johnson Bound. Later,  Kopparty \cite{Kop15} showed that a similar algorithm can also be used to list decode multiplicity codes up to capacity. 

The algorithm for list decoding multiplicity (and FRS) codes up to capacity with an interpolating polynolmial $Q$ with a \emph{linear structure}, presented in Section \ref{subsubsec:mult_mult},  was discovered by Vadhan \cite[Section 5.2.4]{Vadhan_survey} and Guruswami and Wang \cite{GW13}. 
The simpler algorithm  for list decoding multiplicity codes beyond the unique decoding radius, presented in Section \ref{subsubsec:mult_no_mult} as a warmup, follows the presentation of \cite[Section 17.2]{GRS_survey} for list-decoding of Folded Reed-Solomon Codes. 

It was shown by Bhandari, Harsha, Kumar, and Sudan \cite{BHKS24_ideal} that list-decoding algorithms for Folded Reed-Solomon Codes and multiplicity codes can be extended to the more general family of 
\textsf{polynomial-ideal codes}  that includes both Folded Reed-Solomon Codes and multiplicity codes as a special case.  The algorithm for list-decoding of Reed-Solomon Codes over subfield evaluation points up to capacity, presented in Section \ref{subsec:rs_subfield}, was discovered by Guruswami and Xing \cite{GX22}.

All the aforementioned list-decoding algorithms can also be extended to the setting of \textsf{list recovery}, where one is given for each codeword entry a small list of $\ell$ possible values, and the goal is to return the list of all codewords that are consistent with most of these lists (list-decoding corresponds to the special case of $\ell=1$).
We refer the reader to the survey \cite{RV_survey} for motivation and more information on the list recovery model and the extension of the above algorithms to the setting of list recovery.

\section{Combinatorial upper bounds on list size}\label{sec:list}

In the previous section, we presented efficient (polynomial-time)  algorithms for list decoding multiplicity codes up to capacity, and we also showed that a similar algorithm can be used to list decode Reed-Solomon codes over subfield evaluation points  in slightly non-trivial sub-exponential time. 
These algorithms in particular imply combinatorial upper bounds on the list sizes of these codes, however, 
the obtained bounds were quite large.
Specifically, 
the list size was a large polynomial $n^{\Theta(1/\epsilon)}$ for multiplicity codes, and only slightly sub-exponential $N^{\epsilon}$ for Reed-Solomon codes over subfield evaluation points, where $n,N$ denote the block length and code size, respectively, and $\epsilon$ denotes the gap to capacity.

 In  this section, we present tighter combinatorial upper bounds on the list size of these codes.  
 Specifically, we first show in 
Section \ref{subsec:random_rs} below that Reed-Solomon Codes over \emph{random} evaluation points are (combinatorially) list decodable \emph{up to the generalized Singleton Bound}, and so 
are in particular list-decodable up to capacity with list-size $O(1/\epsilon)$. 
Then, in Section \ref{subsec:const_list_mult},  we show that multiplicity codes attain the generalized Singleton Bound over \emph{any} set of evaluation points.  
Finally, in Section \ref{subsec:const_list_rs_subfield}, we show that a certain \emph{subcode} (i.e., a code formed by only picking a subset of the codewords) of Reed-Solomon Codes over subfield evaluation points is (efficiently) list-decodable up to capacity with a \emph{constant} list size (depending on $\epsilon$).

\subsection{Reed-Solomon Codes over random evaluation points}\label{subsec:random_rs}

In this section, we show that Reed-Solomon Codes, evaluated over \emph{random} points, chosen from a sufficiently large (exponential-size) field, are list-decodable all the way up to the \emph{generalized Singleton Bound} (cf., Theorem \ref{thm:gen_singleton}), with high probability. 
The proof relies on two main ingredients: higher-order MDS codes and hypergraph connectivity. We first review each of these ingredients in Sections \ref{subsubsec:higher-order-mds} and \ref{subsubsec:hypergraph} below, and then use these to establish the result about random Reed-Solomon Codes in Section \ref{subsec:rs_singleton}.

\subsubsection{Higher-order MDS codes}\label{subsubsec:higher-order-mds}

The \emph{Singleton Bound} states that for any code $C \subseteq \Sigma^n$ of distance $\Delta$ and size $|C|=\Sigma^k$ it must hold that $\Delta \leq n-k+1$ (which in particular implies that $\delta \leq 1-R$). The code $C$ is said to be \textsf{maximum distance separable } (MDS) if it attains the Singleton Bound, that is, $\Delta = n-k+1$. It is well-known that a \emph{linear} code $C$ is an MDS code if and only if its \emph{dual code}  $C^\perp$ is an MDS code. 
In Section \ref{subsec:prelim_list}, we presented the \emph{generalized Singleton Bound} (cf., Theorem \ref{thm:gen_singleton}) which extends the notion of MDS codes to the setting of list decoding. In this section, we shall present another extension of the notion of \emph{linear} MDS codes which we shall refer to as \emph{higher-order MDS codes}. Later, in Section  \ref{subsec:rs_singleton}, we shall connect the two notions by showing that if $C$ is a higher-order MDS code then its \emph{dual code} $C^\perp$ attains the generalized Singleton Bound. 

Let
 $C \subseteq \F^n$ be a linear code of dimension $k$, and 
let
 \(G\in\F^{n\times k}\) be a generator matrix for $C$, where $G_1, \ldots, G_k \in \F^n$ denote the columns of $G$. For a string $w \in \F^n$, let $\mathcal{Z}(w):=\{i \in [n] \mid w_i=0\}$ denote the set of zero entries of $w$. Note that any column $G_j$ of $G$ is a non-zero codeword of $C$, and so, if $C$ is an MDS code, then  $G_j$ has weight at least $n-k+1$, or equivalently, $|\mathcal{Z}(G_j)| \leq k-1$. The following fact extends this observation by stating that for a linear MDS code, any $t$ columns of $G$ have at most $k-t$ common zeros. 
	
	\begin{fact}\label{fact:mds}
	Suppose that $C \subseteq \F^n$ is a linear MDS code with generating matrix $G \in \F^{n\times k}$ with columns $G_1, \ldots, G_k$.  Then $\bigg|\bigcap_{j \in J}\mathcal{Z}(G_j)\bigg|\le k-|J| $ for any non-empty subset $J \subseteq [k]$.
	\end{fact}

\begin{proof}
 Assume towards a contradiction that there exists a non-empty subset $J \subseteq [k]$ of size $t$ so that
 $\bigg|\bigcap_{j \in J}\mathcal{Z}(G_j)\bigg| \geq k-t+1$. Without loss of generality we may assume that $J=\{1,2,\ldots, t\}$ and $\bigcap_{j \in J}\mathcal{Z}(G_j) \supseteq \{1, 2, \ldots, k-t+1 \}$. Let $A$ denote the top $k \times k$ minor of $G$. Then $A$ has a $(k-t+1) \times t$  block of zeros on its top left corner, and so its first $t$ columns have rank at most 
 $t-1$, and so are linearly dependent.  Consequently, $A$ is not full rank, and so there exists $0 \neq m \in \F^k$ so that $A \cdot m = 0$. But this implies in turn that $G\cdot m$ is a non-zero codeword of $C$ of weight at most $n-k$, which contradicts the assumption that $C$ is an MDS code. 
\end{proof}

Higher-order MDS codes are codes which satisfy the converse of the above fact, i.e., for any zero patterns satisfying the intersection condition given by the above fact, there exists a generator matrix for the code with this zero pattern.
More formally, we say that a series of (not necessarily distinct) subsets $Z_1, Z_2, \ldots, Z_k \subseteq [n]$ is a \textsf{generic zero pattern} (GZP) if $\bigg|\bigcap_{j \in J} Z_j \bigg| \le k-|J| $ for any $\emptyset \neq J \subseteq [k]$. Further, we say that a matrix
 $G\in\F^{n\times k}$ with columns $G_1, \ldots, G_k \in \F^n$ 
 \textsf{attains} $Z_1,\ldots,Z_k$ if $Z_j\subseteq\mathcal{Z}(G_j)$ for all $j\in[k]$. 
	
	\begin{definition}[Higher-order MDS codes]
    Let $C \subseteq \F^n$ be  a linear $\MDS$ code of dimension $k$. We say that $C$ is a \textsf{higher-order MDS code} if for any $\GZP$ $Z_1, \ldots, Z_k \subseteq [n]$, there exists a generator matrix $G$  for $C$ which attains $Z_1, \ldots, Z_k$.
\end{definition}
	
\subsubsection{Hypergraph connectivity}\label{subsubsec:hypergraph}

To relate the notion of higher-order MDS codes to list decoding we shall use a certain notion of \emph{hypergraph connectivity}. To describe this notion, we shall first need to introduce some notation and definitions.
A \textsf{hypergraph} $H = (V, E)$ consists of a ground set of vertices $V=\{v_1, \ldots, v_n\}$ and a collection (possibly a multiset) of (hyper-)edges $E=\{e_1, \ldots, e_m\}$ which are (possibly empty)  subsets of the vertices $V$. 
We define the \textsf{weight} $\wt(e)$ of an edge $e \in E$ as $\wt(e):=\max\{|e|-1,0\}$, and the \textsf{weight} $\wt(H)$ of the hypergraph $H$ as the total edge weight  $\wt(H):=\sum_{i=1}^m  \wt(e_i)$. For $U \subseteq V$, let $H|_U$ denote the \emph{sub-hypergrpah} $H|_U=(U, E|_U)$, where  $E|_U:=\{ e_1 \cap U, \ldots, e_m \cap U\}$. 

A hypergraph $H=(V,E)$ is said to be $k$-\textsf{edge connected} if for any non-empty proper subset $U \subseteq V$, the number of \emph{crossing edges} between $U$ and $V \setminus U$ (i.e., the number of edges $e \in E$ with at least one vertex in $U$ and at least one vertex in $V \setminus U$) is at least $k$. The $k=1$ case corresponds to the standard notion of hypergraph connectivity which requires that there is a path between any pair of vertices in $V$.\footnote{A  \textsf{path} in a hypergraph $H=(V,E)$  is a sequence $v_1, e'_1, v_2, e'_2, \ldots, v_{\ell-1}, e'_{\ell-1}, v_\ell$, where $v_1, \ldots, v_\ell \in V$, $e'_1, \ldots, e'_{\ell-1} \in E$, and 
 $v_i, v_{i+1} \in e'_i$  for all $i =1, \ldots, \ell-1$.} By the \emph{Hypergraph Menger's Theorem} \cite[Theorem 1.11]{Kiraly-thesis}, the $k\geq 2$ case is equivalent to the property that every pair of vertices in $V$ have $k$ edge-disjoint paths between them. 
 
 The notion of $k$-\textsf{partition connectivity} is a strengthening of the above notion of $k$-edge connectivity which requires that for any integer $t \geq 2$, and for any partition $\mathcal{P}$ of $V$ into $t$ parts, the number of crossing edges (i.e., the number of edges which are not contained in a single part of the partition) is at least $k(t-1)$. Note that $k$-edge connectivity corresponds to the special case of $t=2$. We shall use the following weakening of partition connectivity.

\begin{definition}[Weak partition connectivity]
A hypergraph $H=(V,E)$ is \textsf{$k$-weakly partition connected} if for every partition $\mathcal{P}$ of $V$ into $t$ parts, it holds that:
$$
 \sum_{e \in E} \max\{|\mathcal{P}(e)|-1,0\} \geq k (t-1),
$$ 
 where $|\mathcal{P}(e)|$ denotes the number of parts in the partition intersecting $e$. 
\end{definition}
Note that for any partition $\mathcal{P}$, any crossing edge $e \in E$ contributes at least $1$ to the left-hand side of the above inequality, and so $k$-partition connectivity implies $k$-weak partition connectivity. Moreover, for a partition $\mathcal{P}$ into two parts, 
the left-hand side of the above inequality equals the number of crossing edges, and so $k$-weak partition connectivity also implies $k$-edge connectivity.

Next we show that any hypergraph of large weight has a sub-hypergraph that is weakly partition connected.

\begin{lemma}\label{lem:weight_partition}
Let $H=(V,E)$ be a hypergraph with at least two vertices and of weight at least $k (|V|-1).$ Then there exists a subset $U \subseteq V$ 
of at least two vertices so that $H|_U$ is  $k$-weakly partition connected.
\end{lemma}

\begin{proof}
Let $U \subseteq V$ be an inclusion-minimal subset with $|U|\geq 2$ so that 
\begin{equation}\label{eq:partition_1}
\wt(H|_U)\geq k(|U|-1).
\end{equation} 
Note that for any subset $U' \subseteq V$ of size one, $\wt(H|_{U'})$ and $k(|U'|-1)$ are both zero, and so by minimality of $U$, for any non-empty $U' \subsetneq U$, we have that 
\begin{equation}\label{eq:partition_2}
\wt(H|_{U'})\leq k(|U'|-1).
\end{equation}

Let $\mathcal{P}$ be a partition of $U$ into $t$ parts $P_1, P_2, \ldots, P_t$, we shall show that
$$ \sum_{e \in E|_U} \max\{|\mathcal{P}(e)|-1, 0 \} \geq k (t-1),$$
and so $H|_U$ is $k$-weakly partition connected. If $t=1$, then we clearly have that both sides in the above inequality are zero, and so the inequality holds. Hence we may assume that $t\geq 2$.

Then we have that:
\begin{eqnarray*}
  \sum_{e \in E|_U} \max\{|\mathcal{P}(e)|-1, 0 \}  & = & \sum_{e \in E|_U} \bigg( \wt(e) -  \sum_{i=1}^t\wt(e \cap P_i) \bigg) \\
    & = & \sum_{e \in E|_U}  \wt(e) - \sum_{i=1}^t \sum_{e \in E|_U}  \wt(e \cap P_i) \ \\
    & = & \wt(H|_U) - \sum_{i=1}^t\wt(H|_{P_i}) \\
    & \geq & k (|U|-1) -  \sum_{i=1}^t k (|P_i| -1) \\
    & = & k \bigg( (|U|-1) - \sum_{i=1}^t (|P_i| -1) \bigg) \\
   & = & k (t-1),
  \end{eqnarray*}
    where the
    inequality follows by (\ref{eq:partition_1}) and (\ref{eq:partition_2}) (noting that $ P_i$  is a non-empty  proper subset of $U$ for any $i \in [t]$ by assumption that $t \geq 2$).
 \end{proof}

 The main property of weakly-partition-connected hypergraphs we shall use is that they can be \emph{oriented}.  A \textsf{directed hypergraph} $H=(V,E)$ is a hypergraph where in each edge $e \in E$, one vertex is assigned as the \textsf{head}, and the rest of the vertices are assigned as \textsf{tails}. The \textsf{in-degree} $\deg_{\mathrm{in}}(v)$ of a vertex $v \in V$ is the number of edges for which $v$ is the head. 
 A \textsf{path} in $H$ is a sequence $v_1, e'_1, v_2, e'_2, \ldots, v_{\ell-1}, e'_{\ell-1}, v_\ell$, where $v_1, \ldots, v_\ell \in V$, $e'_1, \ldots, e'_{\ell-1} \in E$, 
 and for all $i =1, \ldots, \ell-1$, vertex $v_i$ is a tail of the edge $e'_i$, and vertex $v_{i+1}$ is the head of $e'_i$. 
 An \textsf{orientation} of an (undirected) hypergraph is obtained by assigning head to each hyperedge. 
 
 \begin{theorem}\label{thm:partition_orient}
 A hypergraph $H=(V,E)$ is $k$-weakly-partition-connected if and only if it has an orientation such that for some vertex $v$ (the 'root'), every other vertex $u$ has $k$ edge-disjoint paths to $v$. 
 \end{theorem}
 
 The proof of the above theorem is beyond the scope of the current survey. The interested reader is referred to \cite[Theorem 9.4.13 and 15.4.4]{Frank11}.
 
 Finally, the following lemma relates the notion of weakly-partition-connected-hypergraphs to the notion of GZPs, discussed in the previous section.
 
 \begin{lemma}\label{lem:partition_gzp}
 Let $H=(V,E)$ be a $k$-weakly-partition-connected hypergraph with $V=\{v_1, \ldots, v_n\}$ 
and $E=\{e_1, \ldots, e_m\}$, where $n \geq 2$. For $i \in [n]$, let $Z_i = \{j \in [m] \mid v_i \notin e_j \} \subseteq [m]$. Then 
 there exist non-negative integers $\delta_1, \ldots, \delta_n$ summing to $m-k$ so that taking $\delta_i$ copies of each $Z_i$ gives a GZP. 
 \end{lemma}
 
 \begin{proof}
 Consider the orientation of $H$ given by Theorem \ref{thm:partition_orient}. Without loss of generality, we may assume that $v_1$ is the root in this orientation. Let $\delta_1:= \deg_{\mathrm{in}}(v_1) - k$ and $\delta_i :=  \deg_{\mathrm{in}}(v_i)$ for any $i >1$. Then we clearly have that $\delta_i \geq 0$ for $i >1$, and we also have that $\delta_1 \geq 0$ since there are $k$ edge-disjoint paths from $v_2$ to $v_1$, and so $\deg_{\mathrm{in}}(v_1) \geq k$. Moreover, since there are $m$ edges, the $\delta_i$'s must sum to $m-k$.

It remains to show the GZP property. Consider an arbitrary non-empty multiset $U \subseteq V$ such that each vertex $v_i$ appears at most $\delta_i$ times in $U$. We shall show that 
$
  \bigg|\bigcap_{i : v_i \in U} Z_i \bigg| \leq \sum_{i: v_i \notin U} \delta_i,
$
and so
 $$
   \bigg|\bigcap_{i : v_i \in U} Z_i \bigg| \leq \sum_{i: v_i \notin U} \delta_i = m-k- \sum_{i: v_i \in U} \delta_i \leq m-k-|U|,
   $$
which gives the GZP property.

  To show that $\bigg|\bigcap_{i : v_i \in U} Z_i \bigg| \leq \sum_{i: v_i \notin U} \delta_i$, first observe that the left-hand side in this inequality 
  equals the number of edges in $E$ which do not contain any vertex from $U$.  This number is clearly at most the sum of the indegrees in vertices not in $U$, which equals the right-hand side in the case that $v_1 \in U$. Next assume that  $v_1 \notin U$, and fix an arbitrary vertex $v_i \in U$. Then there are $k$ edge-disjoint paths  from $v_o$ to $v_1$. Each of these paths must contain an edge whose head is not in $U$, but contains some vertex from $U$. Thus, the left-hand side is at most $\left(\sum_{i: v_i \notin U} \deg_{\mathrm{in}}(v_i) \right)- k$, which also equals the right-hand side in the case that $v_1 \notin U$. 
 \end{proof}

\subsubsection{Reed-Solomon Codes over random evaluation points are list-decodable up to the generalized Singleton Bound}\label{subsec:rs_singleton}

Next we use the notions of higher-order MDS codes and hypergraph connectivity, introduced in the previous sections, to show that Reed-Solomon Codes over random evaluation points, chosen from an exponentially-large field, are list-decodable up to the generalized Singleton Bound.  To do so, we shall use the following notion of agreement hypergraph.

\begin{definition}[Agreement hypergraph]\label{def:agreement_hypergraph}
For strings $c_0,c_1, \ldots, c_\ell, w \in \Sigma^n$, we define the \textsf{agreement hypergraph of $c_0, c_1,\ldots,c_\ell$ and $w$}  as the hypergraph on vertex set $\{0,1,\ldots, \ell\}$ whose edges are 
$e_i:=\{j \in \{0,1,\ldots, \ell\} \mid (c_j)_i =w_i  \} \subseteq \{0,1,\ldots, \ell\}$ 
for $i \in [n]$ (i.e., each edge $e_i$ corresponds to all strings $c_j$ which agree with $w$ on the $i$-th entry). 
\end{definition}

\begin{claim}\label{clm:agreement_hypergraph}
Suppose that $c_0,c_1, \ldots, c_\ell, w \in \Sigma^n$ are strings so that $\Delta(c_j, w) \leq \frac {\ell} {\ell+1} \cdot (n- k)$ for any $j \in \{0,1,\ldots,\ell\}$. 
Then the  agreement hypergraph $H$ of $c_0,c_1,\ldots,c_\ell$ and $w$ has weight at least $\ell \cdot k$. 
\end{claim}

\begin{proof}
We have that:
\begin{eqnarray*}
\wt(H) & =&  \sum_{i=1}^n \wt(e_i) \\
& \geq&  \sum_{i=1}^n (|e_i|-1)\\
 &= &\sum_{i=1}^n \left( \left( \sum_{j=0}^\ell \mathbf{1}_{ (c_j)_i =w_i}\right) - 1\right)  \\
 &=&  -n + \sum_{j=0}^\ell  (n - \Delta(w,c_j)) \\
 &\geq&  -n + (\ell+1) n - \ell (n-k) 
 =  \ell \cdot k.
\end{eqnarray*}
\end{proof}

The following theorem says that the \emph{dual code} of a higher-order MDS code attains the generalized Singleton Bound. 

\begin{theorem}\label{thm:high_mds_singleton}
Let $C \subseteq \F^n$ be a linear $\MDS$ code of rate $R$, and suppose that $C^{\perp}$ is a higher-order MDS code. 
Then $C$ is $(\frac{L} {L+1} (1- R ),L)$-list decodable for any positive integer $L$.
\end{theorem}

\begin{proof}
The theorem clearly holds for $L=1$ by assumption that $C$ is MDS. 
Next assume that $C$ is not $(\frac{L} {L+1} (1- R),L)$-list decodable for some integer $L\geq 2$. Then there exist a string $w \in \F^n$ and distinct codewords $c_0,c_1, \ldots c_L \in C$ so that $\Delta(w,c_j) \leq \frac{L} {L+1} (1- R)  n =\frac{L} {L+1} (n-k)$ for any $j \in \{0,1,\ldots,L\}$, where $k=\dim(C)=Rn$. 
Let $H=(\{0,1,\ldots,L\}, E)$ be the agreement hypergraph of $c_0, c_1, \ldots, c_L$ and $w$. 
Then by Claim \ref{clm:agreement_hypergraph} we have that $\wt(H) \geq L \cdot k$.  
By Lemma \ref{lem:weight_partition}, there exists a subset $J \subseteq \{0,1,\ldots,L\}$ 
of at least two vertices so that $H|_J$ is  $k$-weakly partition connected. Without loss of generality, we may assume that $J=\{0,1,\ldots,\ell\}$ for some $\ell \in [L]$.  For $j \in \{0,1,\ldots,\ell\}$, let 
$Z_j =  \{i \in [n] \mid (c_j)_i \neq w_i \} \subseteq [n]$.  
By Lemma \ref{lem:partition_gzp}, there exist  non-negative integers $\delta_0, \delta_1, \ldots, \delta_\ell$ summing to $n-k$ so that taking $\delta_j$ copies of each $Z_j$ gives a GZP. 

By assumption that $C^{\perp}$ is a higher-order MDS code, there exists a generating matrix $G \in \F^{n \times (n-k)}$
 for $C^{\perp}$ so that $G$ attains the above GZP. Then $H:=G^T \in \F^{(n-k) \times n}$ is a parity-check matrix for $C$. Since $H  c =0$ for any $c \in C$, we have that
 $H  w = H(w-c_j)$ for any $j \in \{0,1,\ldots,\ell\}$. 
Furthermore, by the definition of the $Z_j$'s, for any $j \in \{0,1,\ldots, \ell\}$, the entries in $Hw=H (w-c_j)$ which correspond to $Z_j$ are zero. We conclude that $H  w=0$, and so $w \in C$. 
Finally, since $\ell \geq 1$, there must exist some $j \in \{0,1,\ldots,\ell\}$ so that $w \neq c_j$. So $w$ and $c_j$ are two distinct codewords so that $\Delta(w,c_j) \leq \frac{L}{L+1}(n-k) <n-k+1$, which contradicts the assumption that $C$ is MDS. 
\end{proof}

Next we would like to apply the above theorem to Reed-Solomon Codes. To do so, we need to show that the \emph{dual} of a Reed-Solomon Code is higher-order MDS. It is well-known that the dual of a Reed-Solomon Code $\RS_q(k)$, evaluated over the whole field, is also a Reed-Solomon Code $\RS_q(n-k)$. For a Reed-Solomon Code $\RS_q(a_1,\ldots,a_n;k)$, evaluated over a subset of field elements, it is known that there exist non-zero scalers $b_1, \ldots, b_n \in \F_q$ so that 
the dual code contains all codewords of the form 
$(b_1\cdot c_1, \ldots, b_n \cdot c_n)$ for $(c_1, \ldots, c_n) \in \RS_q(a_1,\ldots,a_n;n-k)$. Note that if $G$ is a generator matrix for $\RS_q(a_1,\ldots,a_n;n-k)$, then $B \cdot G$ is a generator matrix for  
$\RS_q(a_1,\ldots,a_n;k)^\perp$, where $B$ is a diagonal matrix with diagonal entries $b_1, \ldots, b_n$. 
Since $G$ and $B \cdot G$ have the same zero entries, to show that $\RS_q(a_1,\ldots,a_n;k)^\perp$ is higher-order MDS, it suffices to show that $\RS_q(a_1,\ldots,a_n;n-k)$ is higher-order MDS.

We shall show that with high probability, a Reed-Solomon code with random evaluation points is higher-order MDS. The proof relies on the following  theorem, known as the 'Generator-Matrix-MDS (GM-MDS)' Theorem (the proof of this theorem is beyond the scope of the survey, and we refer the reader to \cite{lovett21, YH19} for more details). In what follows, for $Z \subseteq [n]$, let $f_Z$ denote the $(n+1)$-variate polynomial given by $f_Z(X,Y_1, \ldots, Y_n)= \prod_{i \in Z} (X-Y_i)$. Also, for a field $\F$, let $\F(Y_1, \ldots, Y_n)$ denote the \emph{field of rational functions} over $\F$ in the indeterminates $Y_1, \ldots, Y_n$.

	\begin{theorem}[GM-MDS Theorem,  \cite{lovett21, YH19}]\label{thm:gm_mds}
	Let $\F$ be a field, and let  $Z_1, \ldots, Z_k \subseteq [n]$ be a GZP. Then the polynomials  $f_{Z_1}, f_{Z_2}, \ldots, f_{Z_k}$, when viewed as univariate polynomials in  $X$ with coefficients in   $\F(Y_1, \ldots,Y_n)$, 
    are linearly independent over 
    $\F(Y_1, \ldots,Y_n)$.
	\end{theorem}

	\begin{corollary}\label{cor:gm_mds}
    Let $q$ be a prime power, and let $k,n$ be positive integers so that $k <n$ and 
	$2^{kn} \cdot (k n) < q$. 
			Then with probability at least  $ 1-  \frac{2^{kn} \cdot (kn)} {q}$ over the choice of independent and uniform $a_1, \ldots, a_n \in \F_q$, it holds that
			 $\RS_q(a_1,\ldots,a_n;k)$ is a higher-order MDS code.
	\end{corollary}
	
	\begin{proof}
    Fix a GZP $Z_1, \ldots, Z_k \subseteq [n]$, and note that by the GZP property, $|Z_j| \leq k-1$ for any $j \in [k]$.
   For $j \in [k]$, let $f'_{Z_j}(X) =f_{Z_j}(X,a_1, \ldots, a_n) =  \prod_{i \in Z_j} (X-a_i),$ and note that $f'_{Z_j}$ is a univariate polynomial over $\F_q$ of degree  smaller than $k$. 
   We shall show that with probability at least $1 -\frac{k\cdot n} {q}$ 
over the choice of $a_1, \ldots, a_n \in \F_q$, $f'_{Z_1}(X), \ldots, f'_{Z_k}(X)$ are linearly independent over $\F_q$. 

To see the above, consider the $k \times k$ matrix $M$ whose $j$-th column 
is  the coefficient vector of the polynomial $f_{Z_j}$, when viewed as a univariate polynomial in $X$ with coefficients in $\F_q(Y_1, \ldots, Y_n)$. By Theorem \ref{thm:gm_mds}, 
 $f_{Z_1}, \ldots, f_{Z_k}$, when viewed as univariate polynomials in  $X$ with coefficients in   $\F_q(Y_1, \ldots,Y_n)$, 
    are linearly independent over 
    $\F_q(Y_1, \ldots,Y_n)$, and so $M$ is a full-rank matrix over $\F_q(Y_1, \ldots, Y_n)$. 
 Further note that each entry in $M$ is a multilinear polynomial in $\F_q[Y_1, \ldots, Y_n]$, and  consequently, $\det(M)$ over $\F_q(Y_1, \ldots, Y_n)$ is a non-zero polynomial in $\F_q[Y_1, \ldots, Y_n]$ of  total degree at most $ k \cdot n$. 
By the Schwartz-Zippel Lemma (Lemma \ref{lem:prelim_sz}), with probability at least  
$ 1 -\frac{k\cdot n} {q}$ over the choice of $a_1,\ldots,a_n \in \F_q$, the matrix $M' \in \F_q^{k \times k}$, obtained from $M$ by assigning $a_1, \ldots, a_n$ to the indeterminates $Y_1, \ldots, Y_n$, has a non-zero determinant over $\F_q$, and so is full rank. But note that the columns of $M'$ are precisely the coefficient vectors of  $f'_{Z_1}(X), \ldots, f'_{Z_k}(X)$, and so if $M'$ is full rank, then these polynomials are linearly independent over $\F_q$. 

Next assume that $f'_{Z_1}(X), \ldots, f'_{Z_k}(X)$ are linearly independent, and let $G \in \F_q^{n\times k}$ be the matrix whose columns are the codewords of $\RS_q(a_1,\ldots,a_n;k)$ corresponding to these polynomials. Then the columns of $G$ are $k$ linearly independent codewords in $\RS_q(a_1,\ldots,a_n;k)$, and consequently we clearly have that $\img(G)=\RS_q(a_1,\ldots,a_n;k)$, and so $G$ is a generating matrix for this code. Furthermore, by the definition of $f'_{Z_1}(X), \ldots, f'_{Z_k}(X)$, we also clearly have that $G$ attains $Z_1, \ldots, Z_k$. 

So we conclude that for any GZP $Z_1, \ldots, Z_k \subseteq [n]$, with probability at least $1 -\frac{k\cdot n} {q}$ 
over the choice of $a_1, \ldots, a_n \in \F_q$, there exists a generator matrix for $\RS_q(a_1,\ldots,a_n;k)$ which attains $Z_1, \ldots, Z_k$. 
Since the number of GZPs $Z_1, \ldots, Z_k \subseteq [n]$ is at most $2^{kn}$, by a union bound, we conclude that with probability at least $1 - 2^{kn} \cdot \frac{k \cdot n} {q}$ over the choice of $a_1, \ldots, a_n \in \F_q$, for any GZP $Z_1, \ldots, Z_k \subseteq [n]$, 
 there exists a generator matrix for \(\RS_q(a_1,\ldots,a_n;k)\)  which attains $Z_1, \ldots, Z_k$.
\end{proof}

Combining the above corollary with Theorem \ref{thm:high_mds_singleton}, we conclude that Reed-Solomon Codes over random evaluation points, chosen from an exponentially large field, attain the generalized Singleton Bound with high probability (and in particular, are list-decodable up to capacity with high probability). With some more effort, it can be shown that the field size can be reduced to polynomial, and even linear, in $n$, but at the cost of only \emph{approximately} attaining the generalized Singleton Bound, that is, achieving a list-decoding radius of $\frac{L}{L+1}(1- \frac k n - \epsilon)$ for any positive itneger $L$ and constant $\epsilon >0$ (which can be shown to be necessary), we refer the reader to \cite{AGGLZ25} for more details. Interesting open problems are to find \emph{explicit} evaluation points so that Reed-Solomon Codes evaluated on these points achieve list-decoding capacity, as well as \emph{efficient} list-decoding algorithms for Reed-Solomon Codes up to capacity.

\subsection{Multiplicity codes}\label{subsec:const_list_mult}

In this section, we show that multiplicity codes attain the \emph{generalized Singleton bound} over \emph{any} set of evaluation points. 
More specifically,  we first present in Section \ref{subsubsec:const_list_mult} below a 
simple probabilistic argument 
which shows that multiplicity codes are list-decodable up to capacity with a constant list size that depends \emph{exponentially} on $1/\epsilon$.  This argument is quite general and applies to \emph{any} linear code that is list-decodable up to capacity with a list that is contained in a low-dimensional subspace.  
 Later, in Section  \ref{subsubsec:singleton_list_mult}, we shall show that multiplicity codes even attain the \emph{generalized Singleton Bound} (and in particular the list size only depends  \emph{linearly} on $1/\epsilon$), using more specific properties of multiplicity codes.

 \subsubsection{Constant list size}\label{subsubsec:const_list_mult}
 
In this section, we show that multiplicity codes are list decodable up to capacity with a  \emph{constant} list size (depending on the gap to capacity $\epsilon$), based on a simple probabilistic argument which shows that for a linear code of large distance, the list cannot contain many codewords coming from a low dimensional subspace. 

More specifically, we shall prove the following fairly general lemma which applies to \emph{any} linear code, and
which says that if $C$ is a linear code of relative distance $\delta$ that is list-decodable from a $(\delta-\epsilon)$-fraction of errors with a list that is contained in a linear subspace of dimension at most $r$, then the list size is in fact upper bounded by a quantity that only depends on $\delta, \epsilon$, and $r$. In particular, if $\delta, \epsilon, r$ are all constants, independent of the block length, than so is the list size.

\begin{lemma}\label{lem:const_list_mult}
Let $\F$ be a finite field, and let $C \subseteq \Sigma^n$ be an $\F$-linear code of relative distance $\delta$. Suppose furthermore that $C$ is
list decodable from a $(\delta -\epsilon)$-fraction of errors with a list $\calL$ that is contained in an $\F$-linear  subspace $V\subseteq C$ of dimension $r$. Then 
$$|\calL| \leq  \left( \frac {r} {\epsilon(1-\delta)} \right)^{O\left(   \frac {r} {\epsilon  (1-\delta)} \cdot \log\left( \frac {1} {1-\delta} \right)   \right)}.$$
\end{lemma}

\bigskip

Recall that multiplicity codes of rate $R$ have relative distance at least $\delta = 1-R$, and that by Lemma  \ref{lem:mult_no_mult_list_solution_size}, these codes are list-decodable from a $(\delta -\epsilon)$-fraction of errors with a list that is contained in an $\F$-linear subspace of  dimension $r=O(1/\epsilon)$. 
The above lemma then implies that  these codes are list-decodable from an $(1-R-\epsilon)$-fraction of errors with a constant list size on the order of $(1/\epsilon)^{O(1/\epsilon^2)}$. 
 Later, in Section  \ref{subsubsec:singleton_list_mult}, we shall show that the list size can be further reduced to  $O(1/\epsilon)$ (and even all the way up to the generalized Singleton Bound!), using some additional properties that are more specific to multiplicity codes. We now turn to the proof of the above Lemma \ref{lem:const_list_mult}. 

\begin{proof}[Proof of Lemma \ref{lem:const_list_mult}]
The proof is algorithmic: we will give a simple randomized algorithm $\PRUNE$,
which when given the received word $w \in \Sigma^n$,
either outputs a vector $v \in V$ or outputs $\bot$. The guarantee is 
that for any $v \in \calL$, $v$ is output by the algorithm $\PRUNE$
with probability at least $p_0=p_0(\delta, \epsilon, r)$, which implies in turn that $|\calL| \leq \frac{1} {p_0}$.

The algorithm $\PRUNE$ works as follows. For some parameter $t$, to be determined later on, 
it picks entries $i_1, i_2, \ldots, i_t \in [n]$ uniformly and independently at random, and lets $I:=\{i_1,\ldots,i_t\}$. Then the algorithm 
checks if there is a unique $v \in V$ that agrees with $w$ on $I$
(that is, $v_i =w_i$ for all $i \in I$).
If so, it outputs that unique element $v$; otherwise (i.e., either there are zero or greater than one 
such $v$'s) it outputs $\bot$.

We would like to show that for any $v \in \calL$, the algorithm outputs $v$ with constant probability $p_0$ (depending only on $\delta, \epsilon$, and $r$). 
Fix such a $v \in \calL$. Let $E_1$ be the event that $v$ agrees with $w$ on $I$, and let $E_2$ be the event that two different codewords in $V$ agree on $I$.
Note that the algorithm will output $v$ if and only if the event $E_1$ holds and the event $E_2$ does not hold, so the probability that $v$ is output is at least $\Pr[E_1] - \Pr[E_2]$. The following two claims give lower and upper bounds on the probabilities of the events $E_1$ and $E_2$, respectively.

\begin{claim}\label{clm:const_list_mult_e1}
$\Pr[E_1] \geq (1-\delta+\epsilon)^t.$
\end{claim}

\begin{proof}
Follows since $v \in \calL$, and so $v_i = w_i$ for at least a ($1-\delta+\epsilon$)-fraction of the entries $i \in [n]$, and since $i_1, \ldots, i_t \in [n]$ are chosen uniformly and independently at random.
\end{proof}  

\begin{claim}\label{clm:const_list_mult_e2}
$ \Pr[E_2] \leq (1-\delta)^t\cdot \left(\frac{t} {1-\delta}\right)^{r}.$
\end{claim}

\begin{proof}
For $i \in [n]$, let $A_i:=\{v \in C \mid v_i=0 \},$
let $V_0:=V$, 
and for $ j=1,\ldots, t$, let $$ V_j :=  V \cap A_{i_1}\cap A_{i_2} \cap \cdots \cap A_{i_j},$$ and $r_j: =  \dim_{\F}( V_j ).$
Observe  that $r = r_0 \geq r_1 \geq \ldots \geq r_t$, and that the event $E_2$ holds if and only if $r_t >0$. Note also  that $|v| \geq \delta n$ for any non-zero $v \in V$, since $V \subseteq C$ and $C$ has relative distance $\delta$. 

Next we claim that for any $j=0,1,\ldots,t-1$, it holds that $r_{j+1} \leq  \max\{0,r_j-1\}$ with probability at least $\delta$ over the choice of $i_{j+1}$. To see this, note first that if $r_j=0$, then $r_{j+1} \leq r_j \leq 0$ and so we are done. Otherwise, if $r_{j} \neq 0$ then there exists a non-zero vector $v \in V_j$. Recalling that $V_j \subseteq V $, we have that $|v| \geq \delta n$, and so  $v_{i_{j+1}} \neq 0$ with probability at least $\delta$ over the choice of $i_{j+1}$. But recalling that $V_{j+1} \subseteq A_{i_{j+1}}$, this implies in turn that $v \notin V_{j+1}$. So in this case there exists a non-zero $v \in V_j$ so that $v \notin V_{j+1}$, and so the dimension of $V_{j+1}$ is strictly smaller than that of $V_j$.

Finally, note that if $r_t > 0$, then it must hold that $r_{j+1} > \max \{0, r_j-1\}$ for at least $t-r+1$ of the $j$'s in $0,1,\ldots,t-1$. By the above and the union bound, the probability of this event is at most 
$$ {t \choose t-r+1} \cdot (1-\delta)^{t-r+1} \leq  (1-\delta)^t \cdot \left(\frac t {1-\delta} \right)^{r} .$$
\end{proof}

By the above Claims \ref{clm:const_list_mult_e1} and \ref{clm:const_list_mult_e2}, we conclude that any $v \in \calL$ is output by the algorithm $\PRUNE$ with probability at least
$$
p_0 \geq \Pr[E_1] - \Pr[E_2] \geq (1-\delta+\epsilon)^t - (1-\delta)^t \cdot \left(  \frac{t} {1-\delta}  \right)^r.
$$

Finally, by setting $t := 3 \cdot \frac {r} {\epsilon  (1-\delta)} \cdot \log\left( \frac {r} {\epsilon  (1-\delta)}\right)$ in the above expression we get that
\begin{eqnarray*}
p_0:  & = & (1-\delta+\epsilon)^t - (1-\delta)^t\cdot \left(\frac{t} {1-\delta}\right)^{r} \\
& \geq & (1-\delta)^t \cdot \left[ (1+\epsilon)^t - \left(\frac{t} {1-\delta}\right)^{r}\right] \\
& \geq & (1-\delta)^t \cdot \left[ \left( \frac {r} {\epsilon  (1-\delta)} \right) ^{3 \cdot r}- \left(3 \cdot \frac {r} {\epsilon  (1-\delta)^2} \cdot \log\left( \frac {r} {\epsilon  (1-\delta)}\right)\right)^r \right] \\
& \geq & (1-\delta)^t \\
& \geq & (1-\delta)^{O\left(   \frac {r} {\epsilon  (1-\delta)} \cdot \log\left( \frac {r} {\epsilon  (1-\delta)} \right)   \right)},
\end{eqnarray*}
where the penultimate inequality holds when the ratio $\frac {r} {\epsilon  (1-\delta)}$ is sufficiently large.
This implies in turn that 
$$|\calL| \leq \frac 1 {p_0} \leq  \left( \frac {1} {1-\delta} \right)^{O\left(   \frac {r} {\epsilon  (1-\delta)} \cdot \log\left( \frac {r} {\epsilon  (1-\delta)} \right)   \right)}= \left( \frac {r} {\epsilon(1-\delta)} \right)^{O\left(   \frac {r} {\epsilon  (1-\delta)} \cdot \log\left( \frac {1} {1-\delta} \right)\right)},$$
which concludes the proof of the lemma.
\end{proof}

\subsubsection{Generalized Singleton Bound}\label{subsubsec:singleton_list_mult}

In this section, we show that multiplicity codes (approximately) attain the  \emph{generalized Singleton Bound} of Theorem \ref{thm:gen_singleton}, which in particular implies that the list size of multiplicity codes only depends  \emph{linearly} on $1/\epsilon$.
 
More specifically, in what follows, fix a prime power $q$,  distinct evaluation points $a_1, \ldots, a_n \in \F_q$, and positive integers $s,k$ so that $k< sn$ and $ \max\{k,s\} \leq \char(\F_q)$, and let $\mult_q^{(s)}(a_1, \ldots, a_n; k)$ be the corresponding multiplicity code. We shall show that for any positive integer $L <s$, $\mult_q^{(s)}(a_1, \ldots, a_n; k)$ is $(\alpha,L)$-list decodable for 
$\alpha < \frac{L}{L+1} (1- \frac{k}{n(s-L+1)})$. In particular, if we let $R:=\frac{k} {s \cdot n}$ denote the rate of this code, then it is list decodable from a $\frac {L} {L+1} (1- \frac{s} {s-L+1}\cdot R)$-fraction of errors with list size $L$. So for any $\epsilon >0$, we can choose $s = \Theta(L/\epsilon)$, so that the fraction of errors is at least $\frac{L}{L+1}(1- R - \epsilon)$. In particular, we can choose $s=\Theta(1/\epsilon^2)$ so that the multiplicity code $\mult_q^{(s)}(a_1, \ldots, a_n; k)$ is list decodable from an $(1-R-\epsilon)$-fraction of errors with list size $O(1/\epsilon)$. 

By Claim \ref{clm:agreement_hypergraph}, this would be a consequence of the following lemma.

\begin{lemma}\label{lem:mult_singleton}
Let $q$ be a prime power,  let $a_1, \ldots, a_n$ be distinct points in $\F_q$, and let $s,k,\ell$ be positive integers so that $\ell <s$, $k< sn$,  and $\max\{k,s\} \leq \char(\F_q)$. Let $w \in (\F_q^s)^n$ be a string, let $c_0, c_1, \ldots, c_\ell$ be distinct codewords in 
$\mult_q^{(s)}(a_1, \ldots, a_n; k)$, and let $H$ be the agreement hypergraph of $c_0, c_1, \ldots, c_\ell$ and $w$ (cf., Definition \ref{def:agreement_hypergraph}). Then $\wt(H) < \frac{\ell}{s-\ell+1} \cdot k$. 
\end{lemma}

Before we prove the above lemma, we show that this lemma implies that $\mult_q^{(s)}(a_1, \ldots, a_n; k)$ is $(\alpha,L)$-list decodable for 
$\alpha < \frac{L}{L+1} (1- \frac{k}{n(s-L+1)})$. To see this, suppose on the contrary that $\mult_q^{(s)}(a_1, \ldots, a_n; k)$ is not $(\alpha,L)$-list decodable. Then there exist a string $w \in (\F_q^s)^n$ and codewords $c_0,c_1 \ldots, c_L \in \mult_q^{(s)}(a_1, \ldots, a_n; k)$ so that $\Delta(w,c_j) \leq \frac{L}{L+1} (n- \frac{k}{s-L+1})$ for any $j \in \{0,1,\ldots, L\}$.  
Let $H$ be the agreement hypergraph of  $c_0,c_1, \ldots, c_L$ and $w$. Then by Claim \ref{clm:agreement_hypergraph}, $H$ has weight at least $L \cdot \frac{k}{s-L+1}$ which contradicts the above Lemma \ref{lem:mult_singleton}. 

For the proof of Lemma \ref{lem:mult_singleton}, we first recall that
in Section \ref{subsubsec:const_list_mult} we showed that for \emph{any} $\F$-linear code $C \subseteq \Sigma^n$, and for any $\F$-linear subspace $V \subseteq C$, the dimension of the intersection $V \cap A_i$ for a random $i \in [n]$ is typically small, where  $A_i:=\{v \in C \mid v_i=0 \}$.
For the proof of Lemma  \ref{lem:mult_singleton}, we first 
show a tighter bound on the expected dimension of $V \cap A_i$ for a random $i \in [n]$ for the special case of multiplicity codes. 

\begin{lemma}\label{lem:mult_singleton_wronskian}
Let $q$ be a prime power,  let $a_1, \ldots, a_n$ be distinct points in $\F_q$, and let $s,k$ be positive integers so that $k< sn$ and $ \max\{k,s\}\leq \char(\F_q)$.  Let $V \subseteq \mult_q^{(s)}(a_1, \ldots, a_n; k)$ be an $\F_q$-linear subspace of dimension $r\leq s$, and for $i \in [n]$, let $A_i:=\{v \in \mult_q^{(s)}(a_1, \ldots, a_n; k) \mid v_i=0 \}$.
Then we have that:
$$\sum_{i=1}^n \dim(A_i \cap V) \leq \frac{r } { s-r+1} \cdot k.$$
\end{lemma}

\begin{proof}
Let $c_1, \ldots, c_r$ be a basis for $V$, and let $f_1(X), \ldots, f_r(X) \in (\F_q)_{<k}[X]$ be the polynomials corresponding to these codewords. Let $W(X)$ be the $r \times r$ \textsf{Wronskian matrix} over the polynomial ring $\F_q[X]$ given by
$$W(X) = \begin{pmatrix}
f_1(X) & f_2(X) & \cdots & f_r(X) \\
f_1^{(1)}(X) & f_2^{(1)}( X) & \cdots & f_r^{(1)}( X) \\
\vdots & \vdots & \vdots & \vdots \\
f_1^{(r-1)}(X) & f_2^{(r-1)}(X) & \cdots & f_r^{(r-1)}( X)
\end{pmatrix} ,$$
and let $D(X):= \det(W(X)) \in \F_q[X]$.

Then  $D(X)$ is a polynomial of degree at most $(k-1) \cdot r$. Next we claim that $D(X) \neq 0$. 
Interestingly, one way to see this is using Lemma \ref{lem:mult_no_mult_list_solution_size} which shows that the list of multiplicity codes is contained in a low-dimensional subspace. In more detail, suppose on the contrary that $D(X) \neq 0$. Then $W(X)$ is singular over the field of rational functions $\F_q(X)$, and so there exists a non-trivial linear combination of rows that sums to zero. That is, there exist $B_0(X), B_1(X), \ldots, B_{r-1}(X) \in \F_q[X]$, not all zeros, so that for any $t =1,\ldots, r$,
$$ B_0(X) \cdot f_t(X) + B_1(X) \cdot f_t^{(1)}(X) + \cdots +B_{r-1}(X) \cdot f_t^{(r-1)}(X) =0 .$$
Lemma  \ref{lem:mult_no_mult_list_solution_size}  then implies that $f_1(X), \ldots, f_r(X)$ lie in an $(r-1)$-dimensional subspace, which contradicts our assumption that they span an $r$-dimensional subspace.

Next we show that for any $i \in [n]$, 
$D(X)$ vanishes on each element $a_i$ with multiplicity at least $(s-r+1)\cdot \dim( A_i \cap V)$. To this end, fix $i \in [n]$, and let $r_i:= \dim(A_i \cap V)$. 
We would like to show  that $D(X)$ vanishes on $a_i$ with multiplicity at least
 $(s-r+1) \cdot r_i$, which by  the definition of the Hasse Derivative,  is equivalent to the property that $(X-a_i)^{(s-r+1) \cdot r_i}$ divides $D(X)$. 

To show the above, first note that performing elementary operations on the columns of $W(X)$ only changes $D(X)$ by a constant multiplicative factor, and hence we may assume without loss of generality that $f_1(X), \ldots, f_{r_i}(X)$ are a basis for $A_i \cap V$, and so 
$f_t^{(j)}(a_i) =0$ 
for any $t =1,2,\ldots,r_i$ and $j=0,1,\ldots, s-1$.
So $f_1(X), \ldots, f_{r_i}(X)$ all vanish on $a_i$ with multiplicity at least $s$, and so $(X-a_i)^s$ divides all these polynomials. Furthermore, $(X-a_i)^{s-r+1}$ divides $f_t^{(j)}(X)$ for any $t=1,2, \ldots, r_i$ and $j=0,1,\ldots, r-1$ 
since taking derivative reduces the degree of each monomial by at most $1$. So we conclude that $(X-a_i)^{s-r+1}$ divides each entry in the first $r_i$ columns, which implies in turn that $(X-a_i)^{(s-r+1) \cdot r_i}$ divides  the determinant $D(X)$ of $W(X)$.

So we conclude that $D(X)$ is a non-zero polynomial of degree at most $(k-1) \cdot r$, which vanishes on each $a_i$ with multiplicity at least $(s-r+1)\cdot \dim( A_i \cap V)$. Consequently, we have that $$\sum_{i=1}^n (s-r+1) \cdot \dim( A_i \cap V) \leq (k-1) \cdot r.$$  Dividing both sides in the above inequality by $ s-r+1$ gives the desired conclusion.
\end{proof}

We now turn to the proof of Lemma \ref{lem:mult_singleton}, based on the above Lemma \ref{lem:mult_singleton_wronskian}. 

\begin{proof}[Proof of Lemma \ref{lem:mult_singleton}]
The proof is by induction on $\ell$. For the base case $\ell=1$, consider a string $w \in (\F_q^s)^n$ and a pair of distinct codewords $c_0, c_1\in \mult_q^{(s)}(a_1, \ldots, a_n; k)$, and let $H=(\{0,1\}, E)$ be the corresponding agreement hypergraph. Then we have that
$$\wt(H) = \sum_{e \in E} \max\{|e|-1,0\} = | \{i \in [n] \mid (c_0)_i=(c_1)_i=w_i\}| \leq n - \Delta(c_0,c_1) <\frac k s,$$
which proves the $\ell=1$ case.

For the induction step, fix $1<\ell<s$. We shall assume that Lemma \ref{lem:mult_singleton} holds for any $1\leq \ell'<\ell$, and we shall show that it also holds for $\ell$. Let $w \in (\F_q^s)^n$ be a string, and let $c_0,c_1,  \ldots, c_\ell$ be distinct codewords in 
$\mult_q^{(s)}(a_1, \ldots, a_n; k)$. Without loss of generality, we may assume that $c_0=0$, since otherwise we can translate $w$ and $c_0, c_1, \ldots, c_\ell$ by $c_0$. 
Let $H=(\{0,1,\ldots,\ell\}, E=\{e_1, \ldots, e_n\})$ be the corresponding agreement hypergraph, and 
assume towards a contradiction that $\wt(H) \geq \frac{\ell}{s-\ell+1}\cdot k$.

Let $V:=\spn\{c_0,c_1, \ldots, c_\ell\}$, let  $r := \dim(V)$, and note that $1 \leq r\leq \ell$ by assumption that $c_0=0$. Without loss of generality, we may further assume that $c_1, \ldots, c_r$ is a basis for $V$. 
Let $\mathcal{P}$ be a partition of $\{0,1,\ldots, \ell\}$ into $r+1$ parts $P_0,P_1,\ldots, P_r$, where  
for $t \in \{0,1, \ldots, r\}$, 
$$P_t:= \left\{j \in [\ell] \mid c_j \in 
\spn\{c_0,c_1, \ldots,c_t\} \setminus  \spn\{c_0, c_1, \ldots,c_{t-1}\} \right\} .$$

For $i \in [n]$, let $A_i:=\{v \in \mult_q^{(s)}(a_1, \ldots, a_n; k) \mid v_i=0 \}$. 
Then the main observation is the following.
\begin{claim}
For any $i \in [n]$, it holds that 
$\dim(A_i \cap V)\geq \max\{|\mathcal{P}(e_i)| -1,0\}$, 
 where $|\mathcal{P}(e_i)|$ denotes the number of parts in the partition intersecting $e_i$. 
 \end{claim}

\begin{proof}
Fix $i \in [n]$. 
The claim clearly holds if $e_i = \emptyset$, so assume next that $e_i \neq \emptyset$, and let $V_i=\spn\{c_j | j \in e_i\}$ denote the span of all codewords in $e_i$. 

Assume first that $w_i=0$. In this case the $i$-th entry of all codewords in $V_i$ is zero, and so $V_i \subseteq A_i\cap V$. Furthermore, we have that $c_0=0 \in e_i$, and so $e_i$ intersects $P_0$. Suppose that $e_i$ intersects additional $t$ parts in the partition $\mathcal{P}$. Then $e_i$ contains $t$ linearly independent vectors, and so $t \leq \dim(V_i) \leq \dim(A_i\cap V)$. So in this case we have that $\dim(A_i\cap V) \geq t = |\mathcal{P}(e_i)| -1$. 

Next assume that $w_i \neq 0$, and let $c_j$ be an arbitrary codeword in $e_i$.
Then since all codewords in $e_i$ agree on the $i$-th entry, we have that all codewords in $e_i$ are contained in $c_j + A_i \cap V$, 
and so $\dim(V_i) \leq \dim(A_i\cap V) + 1$. Furthermore, $0 \notin e_i$ and so $e_i$ does not intersect $P_0$. Suppose that $e_i$ intersects $t$ parts in the partition $\mathcal{P}$. Then $e_i$ contains $t$ linearly independent vectors, and so $t \leq \dim(V_i) \leq \dim(A_i\cap V) + 1$. So in this case we also have that $\dim(A_i\cap V) \geq t-1 = |\mathcal{P}(e_i)| -1$.
\end{proof}
 
By the above claim, and using the same calculation as in the proof of Lemma \ref{lem:weight_partition}, 
we have that 
\begin{eqnarray*}
  \sum_{i \in [n]} \dim(A_i \cap V) 
 & \geq & \sum_{e \in E} \max\{|\mathcal{P}(e)| -1,0\}  \\
& = & \wt(H) - \sum_{t=0}^r \wt(H|_{P_t}) \\
 &> & \frac{\ell}{s-\ell+1} \cdot k - \sum_{t=0}^r \frac{|P_t|-1}{s-|P_t|+2}  \cdot k\\
  &> & \frac{\ell}{s-\ell+1} \cdot k -  \frac{  \sum_{t=0}^r(|P_t|-1)}{s-\ell+2}  \cdot k\\
  &= & \frac{\ell}{s-\ell+1} \cdot k - \frac{\ell-r}{s-\ell+1} \cdot k \\
   & =&  \frac{r} {s-\ell+1} \cdot k, \\
\end{eqnarray*}
where in the second inequality we used our assumption that  $\wt(H) \geq \frac{\ell}{s-\ell+1} \cdot k$, and that by the induction hypothesis $\wt(H|_{P_t}) <\frac{|P_t|-1}{s-|P_t|+2} \cdot k$ for any $t\in \{0,1,\ldots,r\}$ (since  $\mathcal{P}$ is a proper partition). 
Finally,  recalling that $r \leq \ell$, the above inequality contradicts Lemma \ref{lem:mult_singleton_wronskian}, which concludes the proof of this lemma.
\end{proof}

\subsection{Reed-Solomon Codes over subfield evaluation points}\label{subsec:const_list_rs_subfield}

In this section, we show that a certain \emph{subcode} of Reed-Solomon Codes over subfield evaluation points is (efficiently) list decodable up to capacity with a \emph{constant} list size (depending on the gap to capacity $\epsilon$). 
For Reed-Solomon Codes over subfield evaluation points, Lemma \ref{lem:rs_subfield_cap_solution_size} only gave a \emph{linear} 
upper bound on the dimension of the list, and so the methods for reducing the list size from Section \ref{subsubsec:const_list_mult} do not apply in this setting. 

The main observation that lets us reduce the list size in this setting is that the proof of Lemma  \ref{lem:rs_subfield_cap_solution_size}  in fact shows that the list has a more refined low-dimensional \emph{periodic} structure. 
Namely, there exists an $\F_q$-linear subspace $\hat V \subseteq \F_{q^s}$ of constant dimension $r=O(1/\epsilon)$ so that the following holds: Any polynomial $f(X) =\sum_{i=0}^{k-1} f_i \in \F_{q^s}[X]$ whose associated codeword is in the list satisfies that for any $i \in \{0,1,\ldots, k-1\}$, given the first $i$ coefficients $f_0,f_1,\ldots,f_{i-1} \in \F_{q^s}$, the next coefficient $f_i$ belongs to an affine shift of $\hat V$ (where the affine shift only depends on $f_0 ,f_1 ,\ldots, f_{i-1}$).\footnote{This follows by recalling from the proof of Lemma  \ref{lem:rs_subfield_cap_solution_size} 
that for any $i \in \{0,1,\ldots,k-1\}$, 
$f_i$ must satisfy that  $B(f_i) = -y_i$, where $y_i$ only depends on $f_0, f_1,  \ldots, f_{i-1}$, and $B(X)= \sum_{\ell=0}^{r-1} B_\ell(0) X^{q^\ell}$. Consequently, we have that  $f_i \in z_i + \hat V$, where $\hat V \subseteq \F_{q^s}$ denotes the set of roots of $B(X)$, and $z_i \in \F_{q^s}$ is such that 
$B(z_i) =-y_i$. Finally, note that by the linear structure of $B(X)$, $\hat V$ is an $\F_q$-linear subspace of dimension at most $r-1$. 
}
This motivates the following definition. In what follows, for a vector $v \in \F^{sk}$, we let $v=(v_{s,1}, \ldots, v_{s,k})$, where $v_{s,i}$ denotes the $i$-th block of $v$ of length $s$.

\begin{definition}[Periodic subspace]\label{defn:periodic_subspace}
Let $\F$ be a finite field, and let $k,s,r$ be positive integers.
A linear subspace $V \subseteq \F^{sk}$ is a \textsf{$(k,s,r)$-periodic subspace} if there exists a linear subspace $\hat V \subseteq \F^s$ of dimension $r$, so that the following holds: For any $v \in V$ and $i \in \{1, 2, \ldots, k\}$, there exists $z_i \in \F^s$, that only depends on $v_{s,1}, \ldots, v_{s,i-1}$, so that $v_{s,i}  \in z_i +\hat V$.
\end{definition}

By the above, the Reed-Solomon Code over subfield evaluation points $\RS_{q^s} (a_1, \ldots, a_n ;k)$ for $a_1, \ldots, a_n \in \F_q$  
is list decodable from an $(1-R-\epsilon)$-fraction of errors, where the polynomials whose associated codewords are in the list are contained in a $(k,s,r)$-periodic subspace for $r=O(1/\epsilon)$ (when viewing each polynomial $f(X) \in \F_{q^s}[X]$ of degree smaller than $k$ as a concatenation of $k$ coefficients in 
$\F_{q^s}$, identifying $\F_{q^s}$ with $\F_q^s$ via some $\F_q$-linear bijection, and identifying $(\F_q^s)^k$ with $\F_q^{sk}$ in the natural way).
This periodic structure can lead to an $\F_q$-linear subcode with lists of \emph{constant dimension} (and so also of \emph{constant size} by Lemma \ref{lem:const_list_mult})  when only keeping in the code codewords whose associated polynomials fall in a \emph{periodic evasive subspace}, which is a large subspace whose intersection with any periodic subspace has a low dimension.

\begin{definition}[Periodic evasive subspace]
Let $\F$ be a finite field, and let $k,s,r,t$ be positive integers.
A linear subspace $W \subseteq \F^{sk}$ is a \textsf{$(k,s,r,t)$-periodic evasive subspace} if  $\dim(W \cap V) \leq t$ for any $(k,s,r)$-periodic subspace $V \subseteq \F^{sk}$. 
\end{definition}

One way to obtain a periodic evasive subspace is via the notion of a \emph{subspace design}, 
which is a collection of large subspaces of $\F^s$, whose intersection with any low-dimensional subspace has low dimension \emph{on average}.

\begin{definition}[Subspace design]
Let $\F$ be a finite field, and let $k,s,r,t$ be positive integers.
A collection of $k$ subspaces $H_1, \ldots, H_k \subseteq \F^s$ is an \textsf{$(r,t)$-subspace design} if $\sum_{i=1}^{k} \dim(\hat V \cap H_i)\leq t$ for any linear subspace  $\hat V \subseteq \F^s$ of dimension $r$. 
\end{definition}

The next claim shows that if $H_1, \ldots, H_k$ is a subspace design, then $H_1 \times H_2 \times \cdots \times H_k$ is a periodic evasive subspace with the same parameters.

\begin{claim}\label{clm:subspace_design}
Let $\F$ be a finite field, and let $k,s,r,t$ be positive integers.
 Suppose that a collection of $k$ subspaces  $H_1, \ldots, H_k \subseteq \F^s$  is an $(r,t)$-subspace design. Then $W:= H_1 \times H_2 \times \cdots \times H_k \subseteq \F^{sk}$ is a $(k,s,r,t)$-periodic evasive subspace.
\end{claim}

\begin{proof}
Let $V \subseteq \F^{sk}$ be a $(k,s,r)$-periodic subspace, and let $V':=V \cap W$. Since both $V$ and $W$ are linear subspaces over $\F$, then so is $V'$. Thus, to show that $V'$ has dimension at most $t$, it suffices to show that $V'$ has size at most $q^t$. 

To see the above, recall that since $V$ is a $(k,s,r)$-periodic subspace, there exists a subspace $\hat V \subseteq \F^s$ of dimension at most $r$, so that for any $v \in V$ and $i \in \{1, \ldots,k\}$, there exists $z_i \in \F^s$ that only depends on $v_{s,1}, \ldots, v_{s,i-1}$, so that $v_{s,i} \in z_i +\hat V$. By the definition of $W$, this implies in turn that for any $v \in V' = V \cap W$ and  $i \in \{1, \ldots,k\}$, there exists $z_i \in \F^s$ that only depends on $v_{s,1}, \ldots, v_{s,i-1}$, so that $v_{s,i} \in (z_i +\hat V) \cap H_i$.
Thus, given the values of  $v_{s,1}, \ldots, v_{s,i-1}$, there are at most $q^{\dim(\hat V \cap H_i)}$ possible values for $v_{s,i}$. Consequently, $V'$ has total size at most $$\prod_{i=1}^k q^{ \dim(\hat V \cap H_i)} = q^{\sum_{i=1}^k \dim(\hat V \cap H_i)} \leq q^t,$$ where the upper bound follows since $H_1, \ldots, H_k$ is a $(k,s,r,t)$-subspace design.
\end{proof}

The next theorem gives an explicit construction of a subspace design with large subspaces $H_1, \ldots, H_k$ whose total intersection with any low-dimensional subspace is small. Interestingly, the construction relies on the list decoding properties of multiplicity codes!

\begin{theorem}\label{thm:subspace_design_explicit}
For any prime power $q$, a prime $d$, 
and positive integers $k,s,r$ satisfying that  $s<2rk$, $\max\{2r, s\} < \char(\F_q)$, and $k \leq  \frac{q^d -q} {d }$, there exists an explicit construction of an
$(r, \frac s d)$-subspace design $H_1, \ldots, H_k \subseteq \F_q^s$, where each subspace $H_i$ has co-dimension $2r d$.
\end{theorem}

\begin{proof}
We first note that  Lemma \ref{lem:mult_singleton_wronskian} implies a version of this theorem for $d=1$ and $k\leq q$. To see this, let $a_1,a_2,\ldots,a_k$ be distinct elements in $\F_q$, and 
for $i \in [k]$, let 
$$H_i := \left\{ f(X) \in (\F_q)_{<s}[X] \; \mid \;  f (a_i) = f^{(1)}(a_i) = \cdots =  f^{(2r-1)}(a_i) =  0 \right\},$$
where we view a polynomial $f(X) \in (\F_q)_{<s}[X]$ as a vector of coefficient in $\F_q^s$. 
Then each $H_i$ has co-dimension at most $2r$ since any requirement of the form $f^{(j)}(a_i)=0$ imposes a single linear constraint on the coefficients of polynomials in $H_i$. Let $\hat V \subseteq (\F_q)_{<s}[X]$ be an $r$-dimensional subspace. For $i \in [n]$, let $A_i$ be the set of codewords in $\mult_q^{(2r)}(a_1, \ldots, a_k; s)$ corresponding to the polynomials in $H_i$, and let $V$ be the set of codewords in this code corresponding to the polynomials in $\hat V$. Then applying Lemma \ref{lem:mult_singleton_wronskian} with $k$ distinct evaluation points, multiplicity parameter $2r$, and degree parameter $s$, we have that 
$$\sum_{i=1}^n \dim(H_i \cap \hat V) = \sum_{i=1}^n \dim(A_i \cap V)\leq \frac{r} {r+1}\cdot s\leq s.$$

To extend the proof to $d>1$, and obtain a larger collection of subspaces $H_i$, one can pick the $a_i$'s from the extension field $\F_{q^d}$. In more detail, let   $a_1,a_2,\ldots,a_k$ be elements in $\F_{q^d}$ so that all elements of the form $a_i^{q^\ell}$ for $i \in [k]$ and $\ell \in \{0,1,\ldots, d-1\}$ are distinct. Note that if $d$ is a prime, then any element $a \in \F_{q^d} \setminus \F_q$ has distinct powers $a, a^q, \ldots, a^{q^{d-1}}$. So in this case, one can find at least $k=\frac {q^d -q} {d}$ elements $a_1, \ldots, a_k$ inside $\F_{q^d}$ which satisfy the above requirement.

Similarly to the $d=1$ case, for  $i \in [k]$, let
$$H_i := \left\{ f(X) \in (\F_q)_{<s}[X] \; \mid \;  f (a_i^{q^\ell}) = f^{(1)}(a_i^{q^\ell}) = \cdots =  f^{(2r-1)}\left(a_i^{q^\ell}\right) =  0 \; \text{for any}\; \ell=0,1,\ldots, d-1 \right\},$$
where we view a polynomial $f(X) \in (\F_q)_{<s}[X]$ as a vector of coefficient in $\F_q^s$.
Then it can be verified that each $H_i$ has co-dimension at most $2 r d $. Moreover, the proof of  Lemma \ref{lem:mult_singleton_wronskian} shows that 
$$\sum_{i \in [k]} \dim(H_i \cap \hat V) \leq \frac{r} {d \cdot (r+1)} \cdot s \leq \frac{ s} {d}$$
for any $r$-dimensional subspace $\hat V \subseteq (\F_q){< s}[X]$.
\end{proof}

By setting $d = \frac {\epsilon s} {2r} $ (assuming that $r  \leq \frac {\epsilon s} {2}$) in the above Theorem \ref{thm:subspace_design_explicit}, we get the following corollary.

\begin{corollary}[Explicit subspace design]\label{cor:subspace_design_explicit}
For any prime power $q$, $\epsilon >0$, and positive integers $k,s,r$ satisfying that 
$s<2rk$, $s<\char(\F_q)$, 
$q^{s} \geq  (k \cdot \frac{\epsilon s} {2r})^{2r /\epsilon}$, 
and $r <\frac {\epsilon s} {2}$, there exists an explicit construction of an 
$(r, \frac {2r} \epsilon)$-subspace design $H_1, \ldots, H_k \subseteq \F_q^s$, where each subspace $H_i$ has co-dimension $\epsilon s$.
\end{corollary}

Recall that by Lemma \ref{lem:rs_subfield_cap_solution_size},  the Reed-Solomon Code over subfield evaluation points $\RS_{q^s}(a_1, \ldots, a_n;k)$ for $a_1, \ldots, a_n \in \F_q$ and $s = \Theta(1/\epsilon^2)$ 
is list decodable from an $(1-R-\epsilon)$-fraction of errors, where the polynomials whose associated codewords are in the list are contained in a $(k,s,r)$-periodic subspace for $r=O(1/\epsilon)$. Let $H_1, \ldots, H_k$ be the subspace design given by the above corollary for these parameters, and recall that by Claim \ref{clm:subspace_design},  $W:= H_1 \times H_2 \times \cdots \times H_k$ is a $(k,s,r,\frac{2r} {\epsilon})$-periodic evasive subspace. Let $C \subseteq \RS_{q^s}(a_1, \ldots, a_n;k)$ be the subcode which consists of all codewords associated with polynomials $f(X) =\sum_{i=0}^{k-1} f_iX^i \in \F_{q^s}[X]$ so that $f_i \in H_{i+1}$ for any $i \in \{0,1,\ldots,k-1\}$. Then $C$ has rate $R-\epsilon$ and is list-decodable from an $(1-R-\epsilon)$-fraction of errors with list of \emph{dimension} $O(1/\epsilon^2)$ (which by Lemma \ref{lem:const_list_mult}, also implies that this subcode has a list of \emph{size} at most $\exp(\poly(1/\epsilon))$). 

Finally, note that $C$ can be list decoded up to capacity in \emph{polynomial time} as follows. First note that by Lemma \ref{lem:rs_subfield_cap_solution_size}, one can find a basis for the $(k,s,r)$-periodic subspace $V$ containing the list of $\RS_{q^s}(a_1, \ldots, a_n;k)$ in polynomial time by solving a system of linear equations. Then one can find a basis for the $O(1/\epsilon^2)$-dimensional subspace $V'$ containing the list of $C$ in polynomial time by solving another system of linear equations to find the intersection of $V$ with the subspace containing the encodings of all degree $k$ polynomials with coefficients in $W= H_1 \times \cdots \times H_k$.
Finally, since the size of $V'$ is $q^{O(1/\epsilon^2)}$, one can find all elements in the list in time $q^{O(1/\epsilon^2)}$, which is polynomial in the block length for a constant $\epsilon$ and $q=\poly(n)$. 

\subsection{Bibliographic notes}\label{subsec:list_bio}

\paragraph{Higher-order MDS codes.} 
The notion of \textsf{higher-order MDS codes} was introduced independently by Brakensiek, Gopi, and Makam \cite{BGM22} and  Roth \cite{Roth22}. In \cite{Roth22}, higher-order MDS codes were defined as codes achieving the \emph{generalized Singleton Bound}, while in \cite{BGM22}, they were defined
based on the \emph{dimension of the intersection of the subspaces} spanned by subsets of columns of their generator matrix. The definition of higher-order MDS codes presented in Section \ref{subsubsec:higher-order-mds}, based on \emph{size of the intersection of the zero-patterns} of columns of the generator matrix,  was given later in another paper of Brakensiek, Gopi, and Makam  \cite{BGM23}. In the same paper  it was also shown that all the above three definitions are equivalent (up to duality). Fact \ref{fact:mds} was called the \textsf{MDS condition} in \cite{DSY14}, and its proof seems to be folklore.

\paragraph{Hypergraph connectivity.}
Kir{\'{a}}ly \cite{Kiraly-thesis} introduced the notion of  \textsf{edge connectivity} for hypergraphs, and proved the \textsf{Menger Theorem} for hypergraphs which says that edge connectivity is equivalent to having multiple edge-disjoint paths between any pair of vertices (this theorem generalizes the well-known Menger theorem for graphs \cite{Menger27}). The notion of \textsf{(weak) partition connectivity} for hypergraphs is well-studied in combinatorics \cite{FKK03a, FKK03b, Kiraly-thesis} and optimization \cite{JMS03, FK08, Frank11, CX18}, and it generalizes the well-known notion of graph partition connectivity (In particular, the well-known \textsf{Nash-Williams-Tutte Tree-Packing theorem} shows that in graphs, partition connectivity is equivalent to having edge-disjoint \emph{packing trees}). 
Our proof of Lemma \ref{lem:weight_partition} follows the proof of \cite[Lemma 2.4]{AGGLZ25}. Theorem \ref{thm:partition_orient} about \textsf{hypergraph orientaitons} is stated most explicitly in \cite{Frank11}, but is also implicit in \cite{Kiraly-thesis, FKK03b}.

\paragraph{Reed-Solomon Codes over random evaluation points.} 
Shangguan and Tamo \cite{ST20} conjectured that 
\textsf{random Reed-Solomon Codes} are higher-order MDS codes and proved several special cases of this conjecture, and this conjecture was proven in full in \cite{BGM23}. We present an alternative view of their proof, based on \emph{hypergraph connectivity}, that was presented by 
Alrabiah, Guo, Guruswami, Li, and Zihan \cite{AGGLZ25}, and was based in turn on a  hypergraph perspective on list-decoding that was first suggested by Guo, Li, Shangguan, Tamo, and Wootters
\cite{GLSTW24}. 
Lemma \ref{lem:partition_gzp} which connects weak partition connectivity with \emph{generalized zero-patterns} follows \cite[Corollary A.4]{AGGLZ25}.  Theorem \ref{thm:high_mds_singleton} which says that duals of higher-order MDS codes (according to the GZP-based definition) satisfy the generalized Singleton Bound follows the proof of \cite[Theorem 2.11]{AGGLZ25}. 
The GM-MDS Theorem (Theorem \ref{thm:gm_mds}) was proven independently by Lovett \cite{lovett21} and Yildiz and Hassibi \cite{YH19}, following a conjecture made by  Dau,  Song, and Yuen \cite{DSY14}. Corollary \ref{cor:gm_mds} follows the proof of \cite[Proposition 4.5]{BGM23}. 

An exponential lower bound on the field size of higher-order MDS codes was shown in  \cite[Corollary 4.2]{BGM22}, while 
in \cite{AGGLZ25} it was shown that random Reed-Solomon Codes over a \emph{linear-size alphabet} \emph{approximately} attain the generalized Singleton Bound \cite{GZ23, AGGLZ25}. An alternative approach for showing that random Reed-Solomon Codes attain the generalized Singleton Bound was given by Levi, Mosheiff, and Shagrithaya
\cite{LMS25}. This approach was based on showing that random Reed-Solomon Codes behave similarly to random linear codes with respect to list-decodabability properties (and more generally, any \emph{local} property, see below). 

\paragraph{Multiplicity codes and Reed-Solomon Codes with subfield evaluation points.} The probabilistic argument, presented in Section \ref{subsubsec:const_list_mult}, which shows that any linear code that is list decodable up to capacity with a constant dimensional list is also list-decodable up to capacity with a \emph{constant list size} was discovered by Kopparty, Ron-Zewi, Saraf, and Wootters \cite[Lemma 3.1]{KRSW23} (some improvements and extensions were later given by Tamo \cite{Tamo24}).  The proof that multiplicity codes (as well and FRS codes) attain the \emph{generalized Singleton Bound}, presented in Section \ref{subsubsec:singleton_list_mult}, 
was discovered by Chen and Zhang \cite{CZ25}, and an algorithmic version
was given in \cite{AHS26}. Lemma \ref{lem:mult_singleton_wronskian} that is used in this proof was proven by Guruswami and Kopparty \cite{GK16}. 
Brakensiek, Chen, Dhar, and Zhang \cite{BCDZ25} extended the proof of \cite{CZ25} by showing that multiplicity codes (as well as Folded Reed-Solomon Codes) 
satisfy any \textsf{local property} that is satisfied by random linear codes (Informally speaking, local properties
\cite{MRRSW24, LMS25} are properties for which non-membership can be certified using a small number of codewords, and they include list decoding and list recovery as special cases).

The  approach for reducing the list size for a subcode of Reed-Solomon Codes over subfield evaluation points using \textsf{subspace designs}, presented in Section \ref{subsec:const_list_rs_subfield}, was suggested by Guruswami and Xing \cite{GX22}. 

\paragraph{AG codes.} A disadvantage of all the explicit capacity-achieving list decodable codes we discussed so far is that their alphabet is a large polynomial in the block length (at least $n^{1/\epsilon^2}$, where $n$ is the block length, and $\epsilon$ is the gap to capacity). 
 To reduce the alphabet size to a constant, one can resort to \textsf{algebraic-geometric  (AG) codes}. 
 Loosely speaking, 
AG codes are an extension of Reed-Solomon Codes  to the setting of {\em function fields}, in which codewords correspond to the evaluation of certain functions on rational points of some 
{\em algebraic curve} over some finite field $\F$. 
Reed-Solomon Codes then correspond to the special case in which the algebraic curve is the affine line over $\F$ and the functions are low degree polynomials on the line.
As the affine line has only $|\F|$ rational points, this forces the alphabet of the Reed-Solomon Code to be at least as large as its block length. The advantage of AG codes is that algebraic curves can generally have many more rational points, and so the alphabet can potentially be much smaller than the block length. In particular, by a certain choice of parameters, the alphabet size can be made constant for increasing block lengths.  

Algebraic-geometric codes were first introduced by Goppa \cite{Gop82}. 
Guruswami and Sudan  \cite{GS-list-dec} showed that their algorithm for list-decoding of Reed-Solomon Codes up to the \emph{Johnson Bound }, presented in Section \ref{subsec:rs_johnson}, 
can be extended to the setting of AG codes. Later, Guruswami and Xing \cite{GX22} introduced versions of \emph{folded} AG codes, as well as AG codes with \emph{subfield evaluation points}, and used these to extend the results that Folded Reed-Solomon Codes and Reed-Solomon Codes over subfield evaluation points 
achieve list-decoding capacity, presented in Section \ref{subsubsec:mult_mult} (for the related family of multiplicity codes) and Section \ref{subsec:rs_subfield}, respectively, to the setting of AG codes. 
More recently, Brakensiek, Dhar, Gopi,  and Zhang
\cite{BDGZ25} extended the result that random Reed-Solomon Codes achieve the \emph{generalized Singleton Bound}, presented in Section \ref{subsec:random_rs}, to the AG code setting. The  approach for reducing the list size for a subcode of Reed-Solomon Codes using \textsf{subspace designs}, presented in Section \ref{subsec:const_list_rs_subfield}, was also used by Guruswami and Xing \cite{GX22} to reduce the list-size of a subcode of AG codes
to nearly-constant (on the order of $\log^*(n)$), and
later to a constant by Guo and Ron-Zewi \cite{GR22}.

\section{Near-linear time list decoding}\label{sec:linear}
In this section, we discuss faster versions of the list decoding algorithms for Reed-Solomon Codes and multiplicity codes. Recall that we have already seen versions of such algorithms in Sections \ref{sec:johnson} and \ref{sec:capacity} 
that run in polynomial time in the input size, and decode Reed-Solomon Codes up to the Johnson Bound and multiplicity codes up to list decoding capacity. The goal now is to see alternative algorithms for these problems whose time complexity is much faster - they run in nearly linear time in the input size, in certain regime of parameters. We start with such an algorithm for Reed-Solomon Codes in Section \ref{sec:fast-RS-decoding} below, followed by a fast list decoding algorithm for multiplicity codes in Section \ref{subsec:fast-multiplicity-decoding}.

\subsection{Fast list decoding of Reed-Solomon Codes up to Johnson Bound} \label{sec:fast-RS-decoding}

In this section, we present a near-linear time list decoding algorithm for Reed-Solomon Codes up to a radius approaching the Johnson Bound. Specifically, in what follows, fix  a prime power $q$, distinct evaluation points $a_1, \ldots, a_n \in \F_q$, and a positive integer $k<n$, and let $\RS_q(a_1, \ldots, a_n;k)$ be the corresponding Reed-Solomon Code. We shall show an algorithm which for any $\epsilon >0$, list decodes this code from $(1-\sqrt{\frac k n} - \epsilon)$-fraction of errors in time $n\cdot \poly(\frac n k, \log (q), \log (n), \frac 1 \epsilon)$. 
Thus, when both $\epsilon$ and the rate $R:= \frac k n$ of the code are constant, the running time of this algorithm is $n \cdot \poly(\log (n), \log (q))$. Note that as opposed to the polynomial-time list-decoding algorithm for Reed-Solomon Codes presented in Section \ref{subsec:rs_johnson}, 
the algorithm we present in this section does not achieve the Johnson Bound, but it only approaches it. 

\paragraph{Overview.}
Recall that the polynomial-time algorithm for list decoding of Reed-Solomon Codes up to the Johnson Bound, presented in Section \ref{subsec:rs_johnson} (cf., Figure \ref{fig:rs_list_johnson}), consists of two main steps, 
namely, that of constructing a non-zero bivariate polynomial that \emph{encodes} all the close enough codewords in its belly, and that of finding roots of a bivariate polynomial. Moreover, the first step is done by setting up and solving a linear system, while the second step relies on some off the shelf root computation algorithm for bivariate polynomials. While the time complexity of both these steps turns out to be polynomially bounded, it is unclear  whether the time complexity can be improved to something that is nearly linear in the input size. In fact, even for the seemingly simpler linear system that appears in the interpolation step of the unique decoding algorithm for Reed-Solomon Codes, presented in Section \ref{subsec:rs_unique}, it is unclear if the system can be solved in nearly linear time. Indeed, even fast unique decoding algorithms for Reed-Solomon Codes, that have been known since the 60's rely on some non-trivial insights. 

The fast list decoding algorithm we present in this section has two main technical components. The first component shows that the interpolation step in the algorithm of Figure \ref{fig:rs_list_johnson} can indeed be performed in nearly linear time 
using appropriate lattices over the univariate polynomial ring. The second component, essentially uses (off the shelf) an algorithm of Roth and Ruckenstein \cite{RothR2000} for computing low degree roots of bivariate polynomials. We now discuss these ideas in some more detail. 

\subsubsection{Fast implementation of the interpolation step}
Let $w \in \F_q^n$ be the received word. Our goal in this step is to find, in nearly linear time, a non-zero bivariate polynomial $Q(X,Y)$ of small $(1,k)$-weighted degree, such that for every $i$, $Q$ has a high multiplicity zero at $(a_i, w_i)$.  To this end, we start with some definitions that, apriori, might appear to be slightly mysterious. 

Let $R(X)$ be the unique polynomial of degree at most $n-1$ such that for every $i$, $R(a_i) = w_i$. $R$ can be obtained in nearly linear time using the standard fast interpolation algorithms for univariate polynomials based on the Fast Fourier Transform (e.g Algorithm 10.22 in \cite{GathenG-MCA}). The utility of $R$ stems from the following simple observation: if $f(X)$ is such that $f(a_i) = w_i$, then $X-a_i$ divides $f(X)-R(X)$, or equivalently, 
$f(X) \equiv R(X) \mod (X-a_i)$. 
Let $\ell, m \in \N$ be parameters with $\ell \leq  m$ to be set later, and for every $j \in \{0,1,\ldots, m-1\}$, let $P_j(X,Y)$ be the bivariate polynomial defined as follows. 
\begin{align*}
    \forall j \in \{0,1, \ldots, \ell-1\}, \quad \quad \quad P_j(X,Y) :=& (Y - R(X))^j \cdot \prod_{i = 1}^n (X - a_i)^{\ell - j} \\
    \forall j \in \{\ell, \ell+1, \ldots, m-1\}, \quad P_j(X,Y) :=& (Y - R(X))^j
\end{align*}

Let us now consider the set of bivariate polynomials in the $\F_q[X]$-linear span of $P_0, P_1, \ldots, P_{m-1}$, that is, 
\[
\mathcal{L}^{(\ell,m)} := \left\{ \sum_{j = 0}^{m-1} g_j(X) \cdot P_j(X,Y) : g_j \in \F_q[X] \right\} \, .
\]
The set $\mathcal{L}^{(\ell,m)}$ 
could be alternatively thought of as follows. 
We view each $P_j$ as a univariate in $Y$ with coefficients from the ring $\F_q[X]$; and thus, $P_j$ can be thought of as a vector of dimension $m$ (namely the coefficient vector, when viewed as a univariate in $Y$), where every entry is an element of $\F_q[X]$. The set $\mathcal{L}^{(\ell,m)}$ is then the $\F_q[X]$ linear span of these vectors. Furthermore, we note that since the $Y$-degree of $P_j$ equals $j$ (and hence is distinct), the coefficient vectors of $P_j$ (when viewed as univariates in $Y$) are linearly independent over the field $\F_q(X)$.

Recall that over the ring of integers, the notion of a \textsf{lattice} is defined as follows: for any postive integer $n$ and linearly independent $n$ dimensional vectors $v_0, \dots, v_m$ over real numbers, the lattice generated by $v_0, v_1, \ldots, v_m$ is defined as 
$\{\sum_i a_i v_i : a_i \in \Z\}$. 
There are some obvious syntactic similarities in the above definition of integer lattices and the definition of the set $\mathcal{L}^{(\ell,m)}$: in both cases, we have \emph{integral domains}  (integers and $\F_q[X]$ respectively) and we take a set of vectors that are linearly independent over an underlying field containing these domains (the field of real numbers and the field of fractions $\F_q(X)$ respectively) and consider their weighted linear combinations, where the weights come from the respective integral domains. Furthermore, both of the integral domains are in fact \emph{Euclidean}. Given these syntactic similarities, we follow the convention of referring to  $\mathcal{L}^{(\ell,m)}$ as a \textsf{lattice generated by $P_0, P_1, \ldots, P_{m-1}$} throughout this survey.

The following lemma shows the relevance of the lattice $\mathcal{L}^{(\ell,m)}$ 
to the interpolation step of the list decoding algorithm for Reed-Solomon Codes. 

\begin{lemma}\label{lem:RS-lattices-interpolants-connections}
Let $Q(X,Y)$ be any non-zero polynomial in the lattice $\mathcal{L}^{(\ell,m)}$. Then for every $i \in [n]$, $Q(X,Y)$ vanishes with multiplicity at least $\ell$ at $(a_i,w_i)$. 
\end{lemma}
\begin{proof}
    Let $u, v$ be non-negative integers such that $u + v < \ell$. Then, for the proof of the lemma, it suffices to show that  the Hasse Derivative $Q^{(u,v)}$ of $Q$ vanishes at $(a_i,w_i)$. Recall that $Q$ is of the form $\sum_{j = 0}^{m-1} g_j(X) P_j(X,Y)$. Thus, by linearity of Hasse Derivatives, it suffices to show that the $(u,v)$-Hasse Derivative of every summand $g_j(X) P_j(X,Y)$ vanishes at $(a_i, w_i)$. From the product rule of Hasse Derivatives (cf., Lemma \ref{lem:prelim_hasse_properties}), 
    we have that
    \begin{align*}
    (g_j(X) P_j(X,Y))^{(u,v)} =& \sum_{0 \leq u' \leq u, 0 \leq v' \leq v} P_j(X,Y)^{(u',v')} \cdot g_j(X)^{(u-u', v-v')},   
    \end{align*}
   and it therefore suffices to show that $P_j^{(u',v')}$ vanishes at $(a_i,w_i)$  for any 
   $u' + v' < \ell$. 

To see the above, we consider two cases depending on $j$, since the structure of $P_j$ changes when $j \geq \ell$. 
\begin{itemize}
        \item $j \geq \ell: $  From its definition, we have $P_j = (Y - R(X))^j$. We now argue that for any 
       non-negative $u', v'$ such that $u' + v' < \ell$, $P_j^{(u',v')}$ is divisible by $(Y - R(X))$. This would complete the argument since $(Y - R(X))$ (and hence $P_j^{(u',v')}$) evaluates to zero at $(a_i,w_i)$. 

       To see the divisibility claim, we consider  $P_j(X + Z_1, Y+ Z_2) = (Y + Z_2 - R(X+Z_1))^j$. Viewing $R(X+Z_1)$ as a polynomial in $Z_1$, we get that $R(X + Z_1)$ equals $R(X) + Z_1\cdot \Gamma(X,Z_1)$ for some bivariate polynomial $\Gamma$. Thus, $P_j(X + Z_1, Y+ Z_2) = ((Y - R(X)) - ( Z_1\Gamma(X,Z_1)-Z_2))^j$. From the binomial theorem, we further get that \[
       P_j(X + Z_1, Y+ Z_2) = \sum_{t = 0}^j \binom{j}{t} (Y-R)^{j-t} ( Z_1\Gamma- Z_2)^t . 
       \]     
    By definition of Hasse Derivatives, we have that $P_j^{(u',v')}$ equals the coefficient of $Z_1^{u'}Z_2^{v'}$ in the above expansion. Note that every monomial in the polynomial $(Z_1\Gamma - Z_2)^t$ has total degree at least $t$ in the $Z$ variables. Thus, the coefficient of $Z_1^{u'}Z_2^{v'}$ in $P_j(X + Z_1, Y+ Z_2)$ equals the coefficient of $Z_1^{u'}Z_2^{v'}$ in the truncated sum 
    $\sum_{t = 0}^{u' + v'} \binom{j}{t} (Y-R)^{j-t} ( Z_1\Gamma - Z_2)^t$. Furthermore, from $t \leq u' + v' < \ell$ and $\ell \leq j$, we get that each of the terms in the sum above is divisible by $(Y-R)$ (which has $Z$-degree zero). Thus,  $P_j^{(u',v')}$ must be divisible by $(Y-R)$. 
    
    \item $j < \ell: $ This case also proceeds along the lines of the above case, albeit with a little more care. For $j < \ell$, we have from the definition of $P_j$ that $P_j(X,Y) = (Y-R(X))^j \prod_{t=1}^n (X-a_t)^{\ell - j}$. For notational convenience we can re-write $P_j$ as $P_j (X,Y)= (Y-R(X))^j \cdot (X - a_i)^{\ell - j} \cdot \prod_{t \neq i} (X-a_t)^{\ell - j}$. The advantage of this rearrangement is that both $(Y-R)$ and $(X-a_i)$ are zero at $(a_i, w_i)$. Thus, counted with multiplicities, $P_j$ has at least $\ell$ factors that vanish at $(a_i, w_i)$. Intuitively, this should indicate that the multiplicity of $P_j$ at $(a_i, w_i)$ must be at least $\ell$ and hence its $(u', v')$-Hasse Derivative should vanish at $(a_i, w_i)$ for all $u' + v' < \ell$. This intuition can be made formal by using the definition of Hasse Derivatives to conclude that $P_j^{(u',v')}$ can be written as a sum of terms, each of which is divisible by either $(Y-R(X))$ or $(X-a_i)$ (or both). We skip the formal proof of this claim here.

\end{itemize}
This completes the proof of the lemma. 
\end{proof}
\autoref{lem:RS-lattices-interpolants-connections} shows that every non-zero polynomial in the lattice $\mathcal{L}^{(\ell, m)}$ vanishes with high multiplicity at $(a_i,w_i)$ and hence satisfies the interpolation conditions in the algorithm of Figure \ref{fig:rs_list_johnson} (for a careful choice of $\ell$). Thus, if we can somehow find a polynomial in this lattice that is non-zero and has low $(1,k)$-weighted degree in nearly linear time, we would have made excellent progress towards a fast list decoder for Reed-Solomon Code. 

The main insight 
is to associate a notion of \emph{norm} or \emph{length} to the vectors in the lattice $\mathcal{L}^{(\ell, m)}$ such that under this norm, the question of finding a non-zero low $(1,k)$-weighted degree polynomial $Q(X,Y)$ in this lattice is precisely the question of finding a \emph{short} (under this norm) non-zero vector in this lattice. Thus, proving a bound on the norm of this short vector and finding one in nearly-linear time will together complete a fast implementation of the interpolation step. We now discuss the details, starting with the definition of the notion of lattices over the univariate polynomial ring and the length of the vectors in the lattice.

\paragraph{Lattices over the univariate polynomial ring.} We start with a couple of definitions.
Let $\F$ be a finite field, and let $\mathcal{P} = \{P_0, P_1, \ldots, P_{m-1}\} \subseteq \F[X]^r$ be a set of $m$ vectors in $\F[X]^r$. A \textsf{lattice $\mathcal{L}$ generated by $\mathcal{P}$ over the ring $\F[X]$} is defined as 
\[
\mathcal{L} := \left\{ \sum_{i = 0}^{m-1} g_i(X)\cdot P_i: g_i(X) \in \F[X] \right\} \, .
\]
A set $\mathcal P \subseteq \F[X]^r$ is said to be a \textsf{basis} for a lattice $\mathcal{L}$ in $\F[X]^r$ if $\mathcal{P}$ generates the lattice $\mathcal{L}$ in the above sense, and the vectors in $\mathcal{P}$ are linearly independent over the field $\F(X)$ of rational functions in $X$ over $\F$. The size of a basis of a lattice is a fundamental property of the lattice and in particular, all bases of a lattice have the same size. 
Note that the vectors in $\mathcal P$ can alternatively be thought of as coefficient vectors of bivariate polynomials in $\F[X,Y]$ when they are viewed as univariates in $Y$. We switch between these views in the course of discussion in this survey as per convenience. 

For our applications in this survey, we confine ourselves to lattices that are \textsf{full rank} in the sense that they are generated by a basis $\mathcal{P} \subseteq \F[X]^m$ such that $|\mathcal{P}| = m$, i.e. the size of the basis equals the dimension of the ambient space. 
For instance, we note that the lattice $\mathcal{L}^{(\ell, m)}$ defined earlier in this section is full rank. To see this, observe that the generators of this lattices, i.e the  bivariate polynomials $P_0, P_1, \ldots, P_{m-1}$ have distinct $Y$-degrees, and so they are linearly independent over $\F_q(X)$. 

For a basis $\mathcal P$ of a full rank lattice $\mathcal L$ over $\F[X]$, we define the determinant of $\mathcal P$, denoted by $\det(\mathcal P)$ to be the determinant of the square matrix whose columns are the vectors in $\mathcal P$. Recall that $\det(\mathcal P)$ is a polynomial in the variable $X$. As is the case with full rank lattices over the ring of integers, the determinants of all bases  for a lattice $\mathcal L$ are all equal to each other, and this is a fundamental property of the lattice $\mathcal L$, that we denote by $\det(\mathcal L)$. 

Next we define a notion of norm of vectors in these lattices. 
\begin{definition}[Degree norm]\label{defn:length-of-vectors} 
For any vector $P = (p_0, p_1, \ldots, p_{r-1}) \in \F[X]^r$, we define the \textsf{degree norm} of $P$, denoted as $\mu(P)$ as 
\[
\mu(P) := \max_{j \in \{0, 1, \ldots, r-1\}} \deg(p_j) \, , 
\]
where $\deg(p_j)$ is the maximum degree (in $X$) of any entry $p_j \in \F[X]$ of $P$. 

In other words, $\mu(P)$ is the degree in the variable $X$ of the bivariate polynomial $\sum_{j=0}^{r-1} p_j(X) Y^j$.  
\end{definition}
The following theorem is a combination of the key technical result of Alekhnovich in \cite{Alekhnovich} and a classical theorem of Minkowski for lattices. The theorem is implicit in \cite{Alekhnovich} and we refer to Section 3.3 in \cite{GoyalHKS2024} for a more detailed discussion on this. 

\begin{theorem}[\cite{Alekhnovich}] \label{thm:algorithmic-minkowski}

    Let $\mathcal{L} \subseteq \F[X]^m$ be a full rank lattice over $\F[X]$ generated by a basis $\mathcal{P}$. Then, there is a non-zero vector $v$ in $\mathcal{L}$ such that 
    \[
    \mu(v) \leq \frac{1}{m} \deg (\det (\mathcal L)) \, . 
    \]

    Moreover, there is an algorithm that when given the basis $\mathcal{P}$ as input finds a non-zero vector $v \in \mathcal{L}$ which satisfies the above inequality, and runs in time $\Tilde{O}(\mu(\mathcal{P})) \cdot \poly(m) $, where $\mu(\mathcal{P})$ equals $\max_{P \in \mathcal{P}} \mu(P)$. 
\end{theorem}
Given this brief digression to the properties of lattices, we now get back to the task of obtaining a fast implementation of the interpolation step of the algorithm of Figure \ref{fig:rs_list_johnson}, and prove the following lemma. 

\begin{lemma}\label{lem:fast-GS-interpolation-setup}
Let $P_0(X,Y), \ldots, P_{m-1}(X,Y)$ be the polynomials and let  $\mathcal{L}^{(\ell,m)}$ be the lattice as defined earlier in this section with respect to a received word $w \in \F_q^n$ and evaluation points $a_1, \ldots, a_n \in \F_q$, and let $k<n$ be a positive integer.
Then, there is a non-zero polynomial $Q(X,Y)$ in $\mathcal{L}^{(\ell,m)}$ of $(1,k)$-weighted degree at most
\[
\left(\frac{n(\ell+1) \ell }{2m} + \frac{(m-1)k}{2}\right).
\]
Moreover, there is an algorithm that when given the coefficient vectors of $P_0, P_1, \ldots, P_{m-1}$ as input, outputs the coefficient vector of $Q$ in time $\Tilde{O}(n) \cdot \poly(m)$. 
\end{lemma}

\begin{proof}
In order to argue about the $(1,k)$-weighted degree of polynomials in $\mathcal{L}^{(\ell,m)}$, it would be helpful to instead work with the following related lattice $(\mathcal{L'})^{(\ell, m)}$ generated by the $\F_q[X]$ span of the bivariates $P_0', P_1', \ldots, P_{m-1}'$ defined as 
\begin{align*}
& \forall j \in \{0,1, \ldots,\ell-1\},  &P'_j(X,Y) :=& (X^kY - R(X))^j \cdot \prod_{i = 1}^n (X - a_i)^{\ell - j} \\
& \forall j \in \{\ell,\ell+1, \ldots, m-1\},  &P'_j(X,Y) :=& (X^kY - R(X))^j\, . 
\end{align*}
We again view these polynomials as univariates in $Y$ with coefficients in $\F_q[X]$. Thus, $P_j'$ is a polynomial with $Y$ degree equal to $j$ and for $j < \ell$, its leading coefficient equals $X^{kj} \cdot \prod_{i = 1}^n (X - a_i)^{\ell - j}$, whereas for $j \geq \ell$, its leading coefficient equals  $X^{kj}$. Thus, if we look at the $m \times m$ matrix consisting of the coefficient vectors of $P_j'$ (we think of each of them as a polynomial in $Y$ of degree at most $m-1$), we get a triangular matrix, with the diagonal entries corresponding to precisely the leading coefficients of $P_0', P_1', \ldots, P_{m-1}'$. Hence, the lattice generated by them is full rank and its determinant is the product of the diagonal entries of this matrix, which equals
\[
\left(\prod_{j = 0}^{\ell - 1}X^{kj} \cdot \prod_{i = 1}^n (X - a_i)^{\ell - j} \right) \cdot \left( \prod_{j = \ell}^{m-1} X^{kj}\right) \, .
\]
Thus, the degree of the determinant of $\mathcal{L}^{(\ell, m)}$ equals $\left(\frac{n (\ell+1) \ell}{2} + \frac{m(m-1)k}{2}\right)$. We now invoke \autoref{thm:algorithmic-minkowski} to $(\mathcal{L}')^{(\ell, m)}$ to get that there is a non-zero polynomial $\Tilde{Q}(X,Y)$ in $(\mathcal{L}')^{(\ell, m)}$ such that $\mu(\Tilde{Q})$, which equals its $X$-degree, is at most 
\[
\frac{1}{m} \cdot \left(\frac{n(\ell+1)\ell}{2} + \frac{m(m-1)k}{2}\right) = \left(\frac{n(\ell+1)\ell}{2m} + \frac{(m-1)k}{2}\right) \, .
\]
Moreover, since the $X$-degree of any $P_j'$ is at most $(mk + n\ell)$, there is an algorithm that when given $P_0', \ldots, P_{m-1}'$ as inputs, outputs $\Tilde{Q}$ in $\tilde{O}(mk + n\ell)\cdot \poly(m) \leq \Tilde{O}(n) \cdot  \poly(m)$ time. Since $\Tilde{Q}$ is in $(\mathcal{L}')^{(\ell,m)}$, and $P_0', P_1', \ldots, P_{m-1}'$ are linearly independent over the field $\F_q(X)$, we have that there are unique univariate polynomials $g_0, g_1, \ldots, g_{m-1} \in \F_q[X]$ such that $\tilde{Q} = \sum_{j=0}^{m-1} g_j(X) P_j'(X,Y)$. Let $Q(X,Y)$ be the polynomial defined as 
\[
Q(X,Y) := \sum_{j=0}^{m-1} g_j(X) P_j(X,Y) \, .
\]
From the definition of $Q, P_0, P_1, \ldots, P_{m-1},P_0',P_1', \ldots, P_{m-1}'$, it follows immediately that the $(1,k)$-weighted degree of $Q$ is at most the $X$-degree of $\Tilde{Q}$, which is at most $\left(\frac{n(\ell+1) \ell }{2m} + \frac{(m-1)k}{2}\right)$. Moreover, there is an algorithm that when given the coefficient vector of $\tilde{Q}$, outputs the coefficient vector of $Q$ in nearly linear time in the input size. 

To see the moreover part, we just observe that in each $P_j'$, we have the property that the coefficient of $Y^i$ is divisible by $X^{i k}$. Since $\Tilde{Q}$ is obtained by taking an $\F_q[X]$-linear combination of such polynomials, this property continues to be true for $\Tilde{Q}$. Finally, to obtain $Q$ from $\Tilde{Q}$, we just need to divide the coefficient of $Y^i$ in $\Tilde{Q}$ by $X^{ik}$, for every $i$, and this operation can be done in nearly linear time in the description of $\Tilde{Q}$, and hence has time complexity at most $\Tilde{O}(n) \cdot \poly(m)$. 

Thus, the overall time complexity of obtaining $Q$, given $P_0, P_1, \ldots, P_{m-1}$ includes the time complexity of obtaining $P_0', P_1', \ldots, P_{m-1}'$ from $P_0, P_1, \ldots, P_{m-1}$, invoking \autoref{thm:algorithmic-minkowski} to obtain $\Tilde{Q}$, and then doing the aforementioned transformation on $\Tilde{Q}$ to obtain $Q$. We have already bounded the time complexity of each of these steps, apart from the complexity of constructing $P_0', \ldots, P_{m-1}'$, by $\Tilde{O}(n) \cdot \poly(m)$. Finally, since $P_j'(X, Y)$ equals $P_j(X, X^k Y)$, we can obtain the coefficient vector of $P_j'$ by multiplying the coefficient of $Y^i$ in $P_j(X,Y)$ by $X^{ki}$ for every $i < m$. This takes an additional $\Tilde{O}(n) \cdot \poly(m)$ time, and gives us the desired bound of $\Tilde{O}(n) \cdot \poly(m)$ on the total complexity. 
\end{proof}

To complete the interpolation step of the algorithm of Figure \ref{fig:rs_list_johnson}, we now invoke the above \autoref{lem:fast-GS-interpolation-setup} with the right setting of parameters, and account for the time complexity of obtaining $P_j$. 

\begin{lemma}\label{lem:fast-GS-interpolation-final}
Let $q$ be a prime power, let $a_1, \ldots, a_n$ be distinct elements of $\F_q$, and let $k<n$ and $\ell \in \N$ be parameters.
Then there is an algorithm that runs in time $n\cdot \poly(\log (n), \log (q), \ell, \frac n k)$, and on any input $w \in \F_q^n$, outputs a non-zero bivariate polynomial $Q(X,Y)$ with $(1,k)$-weighted degree at most $\sqrt{n k (\ell+1)\ell}$, such that for every $i \in [n]$, $Q$ vanishes with multiplicity at least $\ell$ on $(a_i,w_i)$.  
\end{lemma}
\begin{proof}
    The algorithm proceeds by first setting the parameter $m$ such that  $m$ satisfies $n (\ell+1)\ell = m(m-1)k$. For simplicity, we assume that $m$ can be chosen to be a positive integer while satisfying this equation. So, $m = O(\ell \sqrt{\frac n k})$. 

    Next, we construct, via fast univariate polynomial interpolation, the polynomial $R(X)$ of degree at most $n-1$ such that for every $i \in [n]$, $R(a_i) = w_i$. This can be done in $ n \cdot \poly(\log (q), \log (n))$ time. 
    We now construct the basis $P_0, P_1, \ldots, P_{m-1}$ of the lattice $\mathcal{L}^{(\ell,m)}$ as follows. We focus on $j < \ell$ since the argument for larger $j$ is similar (and only easier). For $j < \ell$, $P_j$ is defined as $(Y-R(X))^j \cdot \prod_{i = 1}^n (X-a_i)^{\ell - j}$. By expanding $(Y-R(X))^j$ and rearranging the terms, we get 
    \[
    P_j = \sum_{t = 0}^j Y^{t}\left(\binom{j}{t}R^{j-t} \prod_{i = 1}^n (X-a_i)^{\ell - j} \right) \, .
    \]
    Thus, for every $t$, the coefficient of $Y^t$ in $P_j$ is a polynomial in $X$ of degree at most $O(nm)$. Furthermore, given the coefficient vector of $R$ (that we have already computed above) and $a_1, a_2, \dots, a_n$, we can compute the polynomial $R^{j-t} \prod_{i = 1}^n (X-a_i)^{\ell - j}$ using $\Tilde{O}(nm)$ arithmetic operations over the underlying field using the standard FFT based algorithms for fast polynomial multiplcation. Thus, each $P_j$ (and hence, the entire basis) can be computed in time $n \cdot \poly(m, \log (n), \log (q))$, which is at most $n\cdot \poly(\log (n), \log (q), \ell, \frac n k)$ for our choice of $m = O(\ell \sqrt{\frac n k})$.
    
    Now, from \autoref{lem:fast-GS-interpolation-setup}, we have that there is an algorithm that takes the basis of the lattice $\mathcal{L}^{(\ell, m)}$ as input and outputs the coefficient vector of a non-zero $Q(X,Y)$ in this lattice in time $n \cdot \poly(m) \leq n \cdot \poly(\ell,\frac n k)$, such that the $(1,k)$-weighted degree of $Q$ is at most $\left(\frac{n(\ell+1)\ell}{2m} + \frac{(m-1)k}{2}\right)$. For our choice of $m$, the two terms in this weighted degree bound are equal to each other, and hence the $(1,k)$-degree of $Q$ is at most $\sqrt{\frac{nk (\ell+1) \ell(m-1)}{m}} \leq \sqrt{n k(\ell+1) \ell}$. 
    
    Finally, from \autoref{lem:RS-lattices-interpolants-connections}, we get that this polynomial $Q$ vanishes with multiplicity at least $\ell$ on $(a_i, w_i)$ for every $i \in [n]$. This observation together with the fact that the overall running time of the algorithm is at most $n\cdot \poly(\log (n), \log (q), \ell, \frac n k)$ completes the proof of the lemma. 
    \end{proof}

\subsubsection{Fast implementation of the root finding step}

We first note that from the properties of the polynomial $Q$ constructed in \autoref{lem:fast-GS-interpolation-final}, the following observation follows. 
\begin{lemma}\label{lem:fast-GS-close-enough-codewords-are-roots}
Let  $Q(X,Y)$ be the polynomial constructed in \autoref{lem:fast-GS-interpolation-final},  and let $f(X) \in \F_q[X]$ be a univariate polynomial of degree less than $k$ so that $f(a_i) = w_i$ for at least $\frac{\sqrt{n k (\ell+1)\ell}}{\ell}$ many indices $i \in [n]$. Then, $Q(X,f(X))=0$.
\end{lemma}
\begin{proof}
The lemma follows immediately from \autoref{clm:gs-zero-at-agreement}, analogously to the proof of \autoref{lem:gs-close-poly-are-roots}. We just sketch the details, while keeping track of the changed parameters. From the same argument as in the proof of \autoref{lem:gs-close-poly-are-roots}, we have that $Q(X,f(X))$, which is a univariate polynomial of degree at most  $\sqrt{nk(\ell+1)\ell}$, vanishes with multiplicity at least $\ell$ on every $a_i$ where $f(a_i) = w_i$. Thus, if the number of agreements of $f$ and $w$ is at least $\frac{\sqrt{n k (\ell+1)\ell}}{\ell} $, it must be the case that $Q(X,f(X))$ is identically zero. 
\end{proof}

One final technical tool needed for the fast variant of a list decoding algorithm for Reed-Solomon Codes is the following theorem of Roth \& Ruckenstein \cite{RothR2000}. The theorem gives a faster algorithm for finding roots of bivariate polynomials than what is given by \autoref{thm:bivariate-factorization} in certain settings of parameters.  We refer to Theorem 1.2 in \cite{Alekhnovich} for a proof of the theorem. 
\begin{theorem}[\cite{RothR2000}, \cite{Alekhnovich}]\label{thm:fast-root-finder}
Let $\F$ be any finite field. Then, there is a randomized algorithm that takes as input the coefficient vector of a bivariate polynomial $Q(X,Y) \in \F[X,Y]$ and outputs all 
 its factors of the form $Y-f(X)$ where $f \in \F[X]$ in time
$\Tilde{O}(\deg_X(Q)) \cdot \poly(\deg_Y(Q),\log( |\F|) )$.
\end{theorem}

\subsubsection{Fast list decoding of Reed-Solomon Codes up to the Johnson Bound}

By setting the parameters appropriately, we can now prove the following theorem which gives a fast list-decoding algorithm for Reed-Solomon Codes approaching the Johson bound.

\begin{theorem}\label{thm:fast-RS-alekhnovich}
Let $q$ be a prime power, let  $a_1, a_2, \ldots, a_n$ be distinct points in $\F_q$, let $k<n$ be a positive integer, and let $\epsilon>0$ be a parameter. Then the Reed-Solomon Code  $\RS_{q}(a_1, \ldots, a_n;k)$ can be list decoded from $(1-\sqrt{\frac k n} - \epsilon)$-fraction of  errors in time $n\cdot \poly( \frac n k, \log (q), \log (n), \frac 1 \epsilon)$.  
\end{theorem}

\begin{proof}
    Given the parameter $\epsilon > 0$ in the theorem, we set the parameter $\ell$ to be the smallest natural number greater than $\frac 1 \epsilon$. Note that for this choice of $\ell$, we have $\frac{\sqrt{n k (\ell+1)\ell}}{\ell} < (1+\epsilon)\sqrt{n  k}$.
    
    We now invoke \autoref{lem:fast-GS-interpolation-final} with the given received word as input. This gives us a non-zero bivariate polynomial $Q(X,Y)$ that, as \autoref{lem:fast-GS-close-enough-codewords-are-roots} shows, has the property that any polynomial $f(X) \in \F_q[X]$ of degree less than $k$ with agreement at least $(1+\epsilon)\sqrt{n k}$ with the received word satisfies $Q(X,f(X)) = 0$. 
    We now invoke \autoref{thm:fast-root-finder} to compute all such roots $f(X)$ of $Q(X,Y)$ and output the ones that have degree less than $k$ and have sufficiently large agreement with the received word. 

    The time complexity of the algorithm depends on the time complexity of the interpolation step (\autoref{lem:fast-GS-interpolation-final}), which is 
\[
n\cdot \poly\left(\log (n), \log (q), \ell, \frac n k\right) \leq n \cdot \poly\left(\log (n), \log (q), \frac n k, \frac 1 \epsilon\right)
\]
and the time complexity of the root finding step (\autoref{thm:fast-root-finder}), which is  
\[
\Tilde{O}(\deg_X(Q)) \cdot \poly(\deg_Y(Q),\log( q) ) \leq \sqrt{nk} \cdot \poly\left(\frac n k, \log (nk), \log (q), \frac 1 \epsilon\right)    \,  ,
\]
 where we used the fact that $\deg_X(Q)$ is at most the $(1,k)$-weighted degree of $Q$ and hence is at most $O( \sqrt{nk}/\epsilon)$ and $\deg_Y(Q)$ is at most $\frac{1}{k}$ of the   $(1,k)$-weighted degree of $Q$ and hence is at most $O( \sqrt{\frac n k} \cdot \frac 1 \epsilon)$. Finally, recalling that $k < n$ gives us the upper bound on the time complexity of $n\cdot \poly( \frac n k, \log (q), \log (n), \frac 1 \epsilon)$ of the algorithm, and completes the proof of the theorem.

\end{proof}

\subsection{Fast list decoding of multiplcity codes up to capacity}\label{subsec:fast-multiplicity-decoding}

In this section, we build upon the ideas discussed in the preceding section to get an algorithm that list decodes multiplicity codes up to capacity in nearly linear time (in an appropriate parameter regime). In what follows, fix a prime power $q$, distinct evaluation points $a_1, \ldots, a_n \in \F_q$, and positive integers $k,s$ so that $k<sn$ and $\max\{k,s\} \leq \char(\F_q)$, and let $\MULT^{(s)}_q(a_1, \ldots, a_n;k)$ denote the corresponding multiplicity code. 
We shall show that for any constant  $\delta :=1- \frac{k} {sn}$ and constant $\epsilon>0$, and sufficiently large $s$ (depending on $\epsilon$),  the multiplicity code $\MULT^{(s)}_q(a_1, \ldots, a_n;k)$ can be list decoded from $(1-R - \epsilon)$-fraction of errors in time $\Tilde{O}(n) \cdot \polylog(q)$, where $R:= \frac{k} {sn}$ is the rate of the code. 

\paragraph{Overview.}
The fast algorithm for list decoding multiplicity codes up to capacity is essentially a fast implementation of the polynomial-time algorithm 
for list decoding multiplicity codes up to capacity that was presented in Section \ref{subsubsec:mult_mult}.
We will proceed by first describing a faster algorithm for the interpolation step that lets us construct an ordinary differential equation of high order satisfied by all close enough codewords. This is based on the machinary of lattices over the univariate polynomial ring that we have already seen. The second part of the proof is a fast algorithm for solving these differential equations whose solutions are low dimensional subspaces, and recovering a basis for the subspace of solutions. 
Eventually, we will observe that a constant size (depending on $\epsilon$) list of close enough codewords can be recovered in nearly linear time using the results of Section \ref{subsubsec:const_list_mult}.

\subsubsection{Fast implementation of the interpolation step}
Let $w \in (\F_q^s)^n$ be the received word, where $w_i=(w_{i,0}, \ldots, w_{i,s-1})$ for any $ i \in [n]$,
and let $r \leq s$ be a parameter. Our goal in this step is to find, in nearly-linear time, a non-zero low-degree $(r+1)$-variate polynomial $Q(X,Y_0, \ldots, Y_{r-1})$ that is linear in $Y_0, \ldots,Y_{r-1}$ which satisfies that $P_f(X):=Q(X, f(X), f^{(1)}(X) ,\ldots, f^{(r-1)}(X)) $ vanishes on $a_i$ with high multiplicity for any point $i \in [n]$ for which $f^{(<s)}(a_i)=w_i$.  
The following definition 
will be helpful in stating some of the technical statements succinctly. 
\begin{definition}[Derivative operator]\label{lem:tau-derivative}

   The \textsf{derivative operator}
    $\tau$ is an $\F_q$-linear map that maps polynomials in $\F_q[X,Y_0, \ldots, Y_{r-1}]$ that are linear in $Y_0, \ldots, Y_{r-1}$ to polynomials in $\F_q[X,Y_0, \ldots, Y_r]$ that are linear in $Y_0,\ldots, Y_r$ as follows:
    \[
    \tau\left(A(X) + \sum_{\ell = 0}^{r-1} B_\ell(X)Y_\ell \right) := A^{(1)}(X) + \sum_{\ell = 0}^{r-1} \left(B_\ell^{(1)}(X) \cdot Y_\ell + (\ell+1)B_\ell(X)\cdot Y_{\ell+1}  \right) \, .
    \]

    More generally, for a non-negative integer $j $, let $\tau^{(j)}$  denote the following linear operator:
    \[
      \tau^{(j)}\left(A(X) + \sum_{\ell = 0}^{r-1} B_\ell(X)Y_\ell \right) :=  A^{(j)}(X) + \sum_{\ell=0}^{r-1} \sum_{h=0}^{j} {h + \ell \choose \ell} B_\ell^{(j - h)}(X) \cdot Y_{\ell+h} .
\]
Thus, the image of $\tau^{(j)}$ is contained in $\F_q[X,Y_0, \ldots, Y_{r-1+j}]$. 

\end{definition}
Note that by (\ref{eq:mult_no_mult_uni_der}), for any polynomials 
$Q(X)=A(X) + \sum_{\ell = 0}^{r-1} B_\ell(X)Y_\ell$ and
$f(X) \in \F_q[X]$, if we let $P_f(X): = Q(X,f(X),f^{(1)}(X), \ldots,f^{(r-1)}(X))$, then
we have that $$\tau^{(j)}(Q)(X,f(X), f^{(1)}(X), \ldots,f^{(r-1+j)}(X)) = P_f^{(j)}(X) ,$$ for any non-negative integer $j$. 

The following lemma is a faster version of \autoref{lem:mult_cap_interpolation}.

\begin{lemma}
\label{lem:fast_mult_cap_interpolation}
There is an algorithm that takes as input a received word $w \in (\F_q^s)^n$, runs in time $\Tilde{O}(n) \cdot \poly(s)$, and  outputs a non-zero $(r+1)$-variate polynomial $Q(X,Y_0, \ldots, Y_{r-1})$ over $\F_q$ of the form 
$$Q (X,Y_0, \ldots, Y_{r-1})= A(X) + B_0(X) \cdot Y_0 + \cdots +B_{r-1}(X) \cdot Y_{r-1}  $$
such that:
\begin{enumerate}
    \item The $X$-degree of $Q$ is at most $\frac {n(s-r+1)} {r+1}$. 
    \item For any $i \in [n]$ and $j \in \{0, 1, \ldots, s-r\}$, we have that 
    \[
    \tau^{(j)}(Q)(a_i, w_{i,0}, w_{i,1}, \ldots, w_{i,s-1}) = 0 \, .
    \]
\end{enumerate}
\end{lemma}

\begin{proof}
    We will start by defining an appropriate lattice over the ring $\F_q[X]$ such that all polynomials in the lattice have degree at most one in $Y_0, Y_1, \ldots, Y_{r-1}$ and satisfy the constraints in the second item in the lemma above. We then show that there is a polynomial in this lattice with $X$-degree at most $\frac{n(s-r+1)} {r+1}$, and hence the shortest vector with respect to $X$-degree being the measure satisfies this property. Based on these observations, the algorithm is quite natural - it will construct a basis for this lattice and then run the algorithm given in \autoref{thm:algorithmic-minkowski} to find a shortest vector in this lattice.

    In more detail, let $R_0, R_1, \ldots, R_{r-1} \in \F_q[X]$ be polynomials with degree at most $sn-1, (s-1)n-1, \ldots, (s-r+1)n-1$ respectively such that for every $\ell \in \{0,1\ldots, r-1\}$, $j \in \{0,1,\ldots, s-\ell-1\}$ and $i \in [n]$, we have that \[
    R_\ell^{(j)}(a_i) = \binom{\ell + j}{j}w_{i, \ell+j}.
    \]

    Clearly, such polynomials exist due to degree considerations, and in fact are unique and can be found in time $\Tilde{O}(ns)$ using FFT based algorithms (algorithm 10.22 in \cite{GathenG-MCA}). 
    Given these $R_\ell$'s, let us consider the lattice $\mathcal{L}$ defined as follows. 
    $$
        \mathcal{L}:= \left\{ \tilde{g}(X) \cdot \prod_{i=1}^n (X-a_i)^{s-r+1} + \sum_{\ell = 0}^{r-1} g_\ell (X) \cdot (Y_\ell -R_\ell(X)) : \tilde{g}, g_0, \ldots, g_{r-1} \in \F_q[X]\right\}
   $$ 

We first show that for every  polynomial $Q \in \mathcal{L}$, 
$i \in \{1, 2, \ldots, n\}$, and $j \in \{0,1,\ldots, s-r\}$, it holds that
    $
     \tau^{(j)}(Q)(a_i, w_{i,0}, w_{i,1}, \ldots, w_{i,s-1}) = 0 \, .
    $
    To prove the claim, it suffices to show that the property sought is true for every polynomial of the form $\tilde{g}(X) \cdot \prod_{t=1}^n (X-a_t)^{s-r+1}$ and of the form $g_{\ell}(X) \cdot (Y_{\ell}-R_{\ell}(X))$ for $\ell \in \{0,1,\ldots, r-1\}$. Since every polynomial in the lattice is an $\F_q$-linear combination of polynomials of these form, the claim will follow via the $\F_q$-linearity of $\tau$. Moreover, the property is immediate for polynomials of the form $\tilde{g}(X) \cdot \prod_{t=1}^n (X-a_t)^{s-r+1}$ since for polynomials with no $Y$ variables, the operator $\tau$ is just the derivative operator with respect to $X$. So we just focus on polynomials of the form  $g_\ell(X) \cdot (Y_\ell-R_\ell(X))$. The property is clear for $j=0$ by the definition of the $R_\ell$'s, next we discuss the argument for $j = 1$, since the extension for larger $j$ is very similar. 

    By definition of $\tau$, $\tau^{(1)}(g_\ell(Y_\ell - R_\ell))$ equals 
    \[
   \tau^{(1)}(g_\ell(Y_\ell - R_\ell)) = g_\ell^{(1)}(X)(Y_\ell - R_\ell(X)) + g_\ell(X)((\ell + 1)Y_{\ell+1} - R_\ell^{(1)}(X)) \, .
    \]
    Setting $X$ to $a_i$, $Y_\ell$ to $w_{i,\ell}$ and $Y_{\ell+1}$ to $w_{i,\ell+1}$, we notice that  both $(w_{i,\ell} - R_\ell(a_i))$ and
    $((\ell + 1)w_{i,\ell+1} - R_\ell^{(1)}(a_i))$ are zero from the definition of $R_\ell$. Thus, $\tau^{(1)}(g_\ell (Y_\ell - R_\ell))$ 
    is zero on the input $(a_i, w_{i,0}, w_{i,1}, \ldots, w_{i,s-1})$. This argument naturally extends to $\tau^{(j)}$ for larger values of $j$ by carefully following the definition of the operator $\tau^{(j)}$ and the interpolation constraints on $R_{\ell}$.

Next we show that there exists $Q \in \mathcal{L}$ of $X$-degree at most $\frac{n(s-r+1)}{r+1}$.
        The generators of the lattice $\mathcal{L}$, when viewed as an $(r+1)$-dimensional vector space with entries from $\F_q[X]$ (where entries $0,1,2\ldots, r-1$ correspond to the coefficient of $Y_0, Y_1, \ldots, Y_{r-1}$ respectively, and the $r$-st entry is the $Y$-free term) form a triangular system. Thus, we have a full rank lattice at our hand. Thus, the determinant of the $(r+1)\times (r+1)$ dimensional matrix whose columns correspond to these generators has a non-zero determinant. Moreover, this determinant is a polynomial in $\F_q[X]$ of degree at most the degree of $\prod_{i = 1}^n (X-a_i)^{s-r+1}$, and hence is at most $n(s-r+1)$. Therefore, from \autoref{thm:algorithmic-minkowski}, we get that there is a polynomial in $\mathcal{L}$ that is non-zero and has $X$-degree at most $\frac{n(s-r+1)}{r+1}$. Moreover, there is an algorithm that takes the generators of this lattice as input and outputs this non-zero vector in time $\Tilde{O}(n) \cdot \poly(s)$.  
\end{proof}

\subsubsection{Fast implementation of the root finding step}
Next observe that from the proof of \autoref{lem:mult_cap_P_is_root}, we have that for any polynomial $Q$ constructed in \autoref{lem:fast_mult_cap_interpolation} and any univariate $f$ of degree less than $k$ such that $f$ agrees with $w$ on at least 
$$t : =  \frac{(s-r+1) n  + (r+1)(k-1)} {(s-r+1)  (r+1)} = \frac{ n} {r+1} + \frac{k-1} {s-r+1}$$ evaluation points, we have that 
\[
Q(X,f(X), f^{(1)}(X), \ldots, f^{(r-1)}(X)) = 0 \, .
\]
Moreover, similarly to the discussion in the beginning of Section \ref{subsubsec:mult_mult}, the parameters $s,r$ can be set based on $\epsilon$ such that this error tolerance gets us $\epsilon$-close to list decoding capacity as desired. 

Therefore, to complete the list decoding task, it suffices to solve the above equation. To be able to do this in nearly linear time, we first recall that by Lemma \ref{lem:mult_no_mult_list_solution_size}, the solution space of polynomials of degree less than $k$ satisfying $Q(X,f(X), f^{(1)}(X), \ldots, f^{(r-1)}(X)) = 0$ is an affine space over the field $\F_q$ of dimension at most $r$. Hence, this subspace can be described by specifying a basis (whose total description size is at most $O(kr\cdot \log(q))$). 

Moreover, by Lemma \ref{lem:const_list_mult}
the number of codewords contained in this affine space that are close enough to $w$ is a constant (depending only on $\epsilon$, the relative distance $\delta$ of the code, and the dimension $r$). To complete the algorithm, we shall describe a randomized nearly-linear time algorithm that outputs these codewords.  This is essentially immediate from the proof of Lemma \ref{lem:const_list_mult}: Recall that in this lemma we described a randomized algorithm $\PRUNE$,
which when given $w \in \Sigma^n$, outputs any codeword in the list with constant probability 
(depending on $\epsilon$, $\delta$, and $r$), and used this to conclude that  the list size is constant. Thus, we can simply run $\PRUNE$ for a constant number of times and return the union of the output lists.
By a union bound, all close-by codewords will appear in the union of the output lists with high (constant) probability.
Finally, note that the algorithm $\PRUNE$ proceeds by sampling constantly many entries in $[n]$ at random, and solving a linear system of equations on these entries, 
and so this step can be implemented in near linear time in $n$ by using the standard linear system solver. 

 Thus, it suffices to show how to find the subspace of solutions to the differential equation
$$Q(X,f(X), \ldots, f^{(r-1)}(X)) = 0$$ in nearly-linear time.

Such an algorithm is given in the following lemma.
\begin{lemma}\label{lem:fast-diff-eqn-solver}
Suppose that $Q(X,Y_0, \ldots, Y_{r-1})$ is a non-zero $(r+1)$-variate polynomial over $\F_q$ of the form
$$Q(X,Y_0, \ldots, Y_{r-1}) = A (X) +  B_0 (X) \cdot Y_0+ \cdots +B_{r-1}(X)\cdot Y_{r-1}$$
and of degree at most $D$.
Let $\calL$ be the list containing all polynomials $f(X) \in \F_q[X]$ of degree smaller than $k$ so that  
$$Q(X, f(X), f^{(1)}(X), \ldots, f^{(r-1)}(X))=0.$$ 
Then if $\max\{k,s\} \leq \char(\F_q)$, then $\calL$ forms an affine subspace over  $\F_q$ of dimension at most $r-1$, and there is an algorithm that takes as input the coefficient vector of $Q$ and $k$ and outputs a basis of this affine space of solutions in time $\Tilde{O}(D+k) \cdot \poly(r, \log (q))$.     
\end{lemma}

Note that for our choice of $D = \frac {n(s-r+1)} {r+1}$ and $r<s$, the algorithm runs in time 
$\Tilde{O}(D+k) \cdot \poly(r,\log (q)) = n \cdot \poly(s,\log (q), \log (n)).$
In the rest of this section, we sketch some of the main ideas in the proof of the above lemma, and we refer the interested reader to \cite[Theorem 5.1]{GoyalHKS2024} for a full proof. In particular, for notational simplicity, we only focus on the cases of small $r$, namely $r=0$ and $r=1$, which already contain most of the technical insights.  

\paragraph{Solving the $r=0$ case. }
For $r=0$, we are just looking for $f$ of degree less than $k$ such that $A(X) + f(X) \cdot B(X) = 0$. Thus, in this case, the solution $f(X)$ to this equation, if it exists is unique and is equal to the quotient when $A(X)$ is divided by $B(X)$. Thus, a natural way of solving this problem will be to compute the quotient and remainder when $A(X)$ is divided by $B(X)$ and output the quotient if the remainder is zero and the quotient has degree less than $k$. This task can be done in nearly linear time in the input size using known classical algorithms in computational algebra that are based on the Fast Fourier Transform. We describe the sketch of one such algorithm below based on divide and conquer. As we shall see later, some of the ideas behind this algorithm will naturally lead us towards the case of larger $r$.

For the discussion below, let us assume that $k$ is a power of $2$. First observe that any  polynomial $f$ of degree smaller than $k$ can be uniquely written as $f = f_0(X) + X^{k/2} f_1(X)$ where $f_0, f_1$ are (possibly zero) polynomials of degree less than $\frac k 2$. The idea is to try and recover $f_0$ and $f_1$ by solving equations of smaller size, and then put them together to compute $f$. The first point is to note that if $f$ satisfies $A(X) + f(X)\cdot B(X) = 0$, then it satisfies  
\begin{equation}\label{eqn:f}
A + f\cdot B \equiv 0 \mod X^{k} \, .    
\end{equation}
Moreover, since $f$ has degree less than $k$, it suffices to recover $f$ modulo $X^k$. So, to solve the original equation, we will instead solve the above modular equation and then check if the solution satisfies the original equation. To solve the modular equation above, we try to understand the properties that $f_0, f_1$ must satisfy. 

By substituting $f = f_0 + X^{k/2}f_1$ in $(A + f B \equiv 0 \mod X^k)$, we get  that 
\[
A + f_0 \cdot B + X^{k/2} f_1 \cdot B \equiv 0  \mod X^{k}.
\]
Now, every monomial in the term $X^{k/2} f_1 \cdot B $ has degree at least $\frac k 2$. Thus, the above equation holds if and only if 
\begin{equation}\label{eqn:f0}
A + f_0 B \equiv 0 \mod X^{k/2}    
\end{equation}
and 
\[
(A + f_0 B) + X^{k/2} f_1 B \equiv 0 \mod  X^{k}
\]
Moreover, for any $f_0$ satisfying $(A + f_0 B \equiv 0 \mod  X^{k/2})$, we have that $A+ f_0 B$ equals  $\tilde{A} X^{k/2}$ for some polynomial $\tilde{A}(X)$. Thus, $f_1$ must satisfy 
\[
X^{k/2}\tilde{A} + X^{k/2} f_1 B \equiv 0 \mod  X^{k} , 
\]
which is equivalent to 
\begin{equation}\label{eqn:f1}
\tilde{A} + f_1 B \equiv 0 \mod  X^{k/2} .     
\end{equation}
Thus, to solve for $f$ in Equation (\ref{eqn:f}) it suffices to solve Equation (\ref{eqn:f0}) and then construct $\tilde{A}$ and solve Equation (\ref{eqn:f1}). Moreover, in Equation(\ref{eqn:f0}) and Equation (\ref{eqn:f1}), the polynomials $A, B, \tilde{A}$ can be replaced by their residues modulo $X^{k/2}$. Thus, we have reduced the original problem of solving Equation (\ref{eqn:f}) to solving two problems of exactly half the size. If at any stage of the recursion, we end up with an equation with no solution, we get that the original equation did not have a solution. Moreover, if any of these equations has a solution, it is unique (modulo $X^k$). 

Finally, to argue about the running time, we note that to solve instances where we work modulo  $X^k$, we make two recursive calls sequentially, each involving working modulo $X^{k/2}$. Since we are only interested in solutions modulo $X^k$ (or lower powers of $X$) in these modular equations, we can assume without loss of generality (by just ignoring all monomials of degree greater than $k$, since they do not participate in any computation modulo $X^k$) that $A$ and $B$ are also of degree at most $k$. Constructing the problem instance for the first recursive call (\ref{eqn:f0}) takes $O(k)$ time, and given a solution $f_0$ to this instance, we need to construct the polynomial $\Tilde{A}$ to make the second recursive call (\ref{eqn:f1}). Given $f_0$, we can obtain $\Tilde{A}$ by computing the polynomial $(A+f_0B) \mod X^{k}$ in time $\Tilde{O}(k)$ (using the Fast Fourier Transform for polynomial multiplication). Given solutions $f_0$ and $f_1$ of the two subproblems, the solution $f := f_0 + X^{k/2}f_1$ can be again constructed in time $O(k)$. Thus, the overall time complexity $T(k)$ can be upper bounded by the recurrence 
\[
T(k) \leq 2T(\frac k 2) + k\log^c (k) \, ,
\]
for some fixed constant $c$, where $c$ just depends on the complexity of Fast Fourier Transform over the underlying field $\F$. 

This gives the desired near linear upper bound on the running time. 

\paragraph{Solving the $r=1$ case. }
We now consider the case of $r = 1$. This case essentially captures all of the main ideas of the proof of the above \autoref{lem:fast-diff-eqn-solver} for large $r$, but keeps the exposition simpler in terms of notation. 

Recall that for $r=1$,  we are looking for solutions to the following differential equation:
 $$ A + fB_0 + f^{(1)} B_1 = 0.$$
Without loss of generality, we assume that $B_1(X)$ is a non-zero polynomial (else we are in the $r=0$ case), and moreover, $B_1(0)$ is non-zero, else, we shift $X$ to $X+a$ for an $a$ which is not a root of $B_1$. Such a shift can be found in time $\Tilde{O}(D)$ deterministically (using fast algorithms for univariate multipoint evaluation, e.g. Chapter 10 in \cite{GathenG-MCA}) or by simply taking a random $a$ from the underlying field. 

Once again, it will be helpful to work with the modular equation instead. 
\begin{equation}\label{eqn:differential_original_modular}
   A + fB_0 + f^{(1)} B_1 \equiv 0  \mod X^k . 
\end{equation}
A natural first attempt for this problem is again to try and recover $f$ by recovering $f_0, f_1$ recursively as in the $r=0$ case. However, there is an issue here - unlike in the $r=0$ case, equations of the form of (\ref{eqn:differential_original_modular}) can have multiple solutions. Thus, even if we set up modular equations for $f_0$ and $f_1$ and solve them, the number of candidate solutions $f$ of the original problem that are obtained by a combination of these can be as large as the product size of the solution sets for $f_0$ and $f_1$. This explosion of potential solutions appears to be an issue that needs to be handled for this strategy to work. 

To get around this issue, 
we rely on the fact that the space of solutions of (\ref{eqn:differential_original_modular}) is an affine space over $\F_q$ of dimension $1$. Furthermore, the proof of \autoref{lem:mult_no_mult_list_solution_size} (cf., Claim \ref{clm:mult_no_mult_list_solution_size}) tells us more about the structure of this affine space: if $B_1(0) \neq 0$ (which we have assumed here without loss of generality), then any solution $f$ of (\ref{eqn:differential_original_modular})  is uniquely determined by its constant term, i.e. $f(0)$.\footnote{Note that here, we are still looking at solutions modulo $X^k$.} This simple observation, combined with elementary linear algebra,  leads to the following claim that turns out to be very useful.

\begin{claim}
\label{clm:basis-for-soln-diff-eqn}
    Let $g(X), h(X)$ be the unique polynomials of degree less than $k$ with constant term zero and one respectively that are solutions of (\ref{eqn:differential_original_modular}) (assuming that $B_1(0) \neq 0$). Then, $g, h$ form a basis of the affine space of solutions of degree less than $k$ (over the field $\F$) of (\ref{eqn:differential_original_modular}). In other words,  the set of such solutions equals the set $\{\lambda g + (1-\lambda)h : \lambda \in \F\}$. 
\end{claim}
\begin{proof}
Let $g, h$ be any solutions to (\ref{eqn:differential_original_modular}). Thus, we have that $A + gB_0 + g^{(1)}B_1 \equiv 0 \mod X^k$ and $A + hB_0 + h^{(1)}B_1 \equiv 0 \mod X^k$.  Now, for any $\lambda \in \F$, we claim that $\lambda g + (1-\lambda)h$ is also a solution of (\ref{eqn:differential_original_modular}). To see this, we observe using the linearity of the derivative operator that 
\[
    A + (\lambda g + (1-\lambda)h)B_0 + (\lambda g + (1-\lambda)h)^{(1)}B_1 = A + \lambda \left( gB_0 + g^{(1)}B_1\right) + (1-\lambda) \left( hB_0 + h^{(1)}B_1\right) .
\]
Since $g, h$ are solutions to (\ref{eqn:differential_original_modular}), we observe that $\left(gB_0 + g^{(1)}B_1 \equiv -A \mod X^k\right)$ and\\
$\left(hB_0 + h^{(1)}B_1 \equiv -A \mod X^k\right)$. Thus,  
\[
A + \lambda \left( gB_0 + g^{(1)}B_1\right) + (1-\lambda) \left( hB_0 + h^{(1)}B_1\right) \equiv A + \lambda(-A) + (1-\lambda)(-A) \mod X^k ,
\]
which is $0$ modulo $X^k$. So, we get that every element of the set $\{\lambda g + (1-\lambda)h : \lambda \in \F\}$ is a solution of (\ref{eqn:differential_original_modular}), if $g$ and $h$ are solutions. 

Furthermore, we note that $$\{\lambda g + (1-\lambda)h : \lambda \in \F\} = \{\lambda(g-h) + h : \lambda \in \F\}$$ and thus this set is indeed an affine space of dimension one whenever $(g-h)$ is non-zero (which is the case here, since the constant terms of $g$ and $h$ are distinct). 
Since we had already argued that the solutions to (\ref{eqn:differential_original_modular}) form an affine space of dimension at most one, we have that $\{\lambda g + (1-\lambda)h : \lambda \in \F\}$ is indeed the set of all such solutions. 
\end{proof}

Given the above claim, a natural strategy towards solving (\ref{eqn:differential_original_modular}) is to try and recover the unique polynomials $f_0, f_1$ of degree less than $k$, with constant term zero and one respectively that are solutions of the equation. It turns out that the uniqueness of solutions to (\ref{eqn:differential_original_modular}), once the constant terms are fixed, makes this subproblem amenable to a divide and conquer like strategy. 

Let us consider a solution $f$ of degree less than $k$ of (\ref{eqn:differential_original_modular}) with constant term $\alpha$. Thus, $f(X)$ can be written as $f(X) = \alpha + Xg(X)$ for a polynomial $g$ of degree less than $k-1$. Since $f$ is a solution, we have the following sequence of identities. 

\begin{align*}
A + f B_0 + f^{(1)} B_1 & \equiv 0  \mod X^k \\
     \implies A + (\alpha + Xg)B_0 + (\alpha + Xg)^{(1)} B_1 & \equiv 0  \mod X^k \\ 
     \implies (A + \alpha B_0) + (XB_0 + B_1)g + X B_1 g^{(1)} & \equiv 0 \mod X^k  
\end{align*}
Thus, for every $\alpha$ such that $f = \alpha + Xg$ is a solution of (\ref{eqn:differential_original_modular}), $g$ must satisfy that
\begin{equation}\label{eqn:differential_new}
  (A + \alpha B_0) + (XB_0 + B_1)g + XB_1 g^{(1)}  \equiv 0 \mod X^k  \, .
\end{equation}

The above new differential equation, while being similar in structure to (\ref{eqn:differential_original_modular}), has the useful property that there is a unique polynomial $g(X)$ of degree less than $(k-1)$ which satisfies this equation. 
To see this, note that by the proof of \autoref{lem:mult_no_mult_list_solution_size} (cf., Claim \ref{clm:mult_no_mult_list_solution_size}), it suffices to show that the constant term $g(0)$ of $g$ is uniquely determined. To see this, we note that evaluating the left hand side of   (\ref{eqn:differential_new}) at zero gives the following equation: 
\[
  (A(0) + \alpha B_0(0)) +  B_1(0)g(0) = 0  \, .
\]
Since $B_1(0)$ is a non-zero field element (by our assumption at the beginning of this section), the above degree one equation in $g(0)$ has a unique solution given by $g(0) = -\frac{A(0) + \alpha B_0(0)}{B_1(0)}$. 

The uniqueness of solutions now enables us to use a natural divide and conquer style algorithm, almost identical to the one we discussed for the $r=0$ case. We write $g$ as $g = g_0 + X^{k/2} g_1$ for polynomials $g_0, g_1$ of degree less than $\frac k 2$ and try to understand the conditions that $g_0, g_1$ must satisfy. If $g$ is a solution of Equation (\ref{eqn:differential_new}), then, working modulo $X^{k/2}$ and discarding terms that are divisible by $X^{k/2}$, we get that $g_0$ must satisfy 
$$
  (A + \alpha B_0) + (XB_0 + B_1)g_0 + XB_1 g_0^{(1)}  \equiv 0 \mod X^{k/2}  \, .
$$
Moreover, for any $g_0$ satisfying the above equation, we have that $(A + \alpha B_0) + (XB_0 + B_1)g_0 + XB_1 g_0^{(1)} $ is divisible by $X^{k/2}$, and hence equals $X^{k/2}\Tilde{A}$ for some polynomial $\Tilde{A}$, and $g_1$ must satisfy 
$$
  X^{k/2}\tilde{A} + X^{k/2}\left(XB_0 + \left(1 + \frac k 2\right)B_1\right)g_1 + X^{k/2}\cdot XB_1 g_1^{(1)}  \equiv 0 \mod X^{k}  \, .
$$
This is equivalent to 
$$
\tilde{A} + \left(XB_0 + \left(1+\frac k 2\right)B_1\right)g_1 +  XB_1 g_1^{(1)}  \equiv 0 \mod X^{k/2}  \, .
$$
Thus, we have reduced the question of finding $g$ to the question of finding $g_0, g_1$ that satisfy similar differential equations, but now, we are working modulo $X^{k/2}$. Thus, the instance size has effectively halved. Moreover, these new instances can be generated in nearly linear time in the input size. Thus, a divide and conquer based algorithm will find $g$ in nearly linear time. 

\medskip

As mentioned above, the setting of larger $r$ can be handled using very similar ideas, but with some care. We skip the details here, and refer the interested reader to \cite{GoyalHKS2024}.

\subsection{Bibliographic notes}
The question of designing nearly linear time encoding and decoding algorithms has been an important theme of study in algebraic coding theory. The Berlekamp-Massey algorithm (Chapter 7 in \cite{GathenG-MCA}) discovered in the 1960s decodes Reed-Solomon Codes (and BCH Codes) in the unique decoding regime in nearly linear time. The algorithm is based on an algorithm of Berlekamp \cite{berlekamp-bch} for decoding BCH Codes. Massey \cite{Massey69} gave a simplified description of this algorithm and showed that it can also be used to solve other problems in computational algebra, notably the problem of computing minimal polynomials of linear recurrent sequences. These algorithms essentially gave a reduction from the problem of decoding Reed-Solomon Codes (or BCH Codes) in the unique decoding regime to the well studied problem of \emph{rational function reconstruction}. Subsequently, rational function reconstruction was shown to be solvable in nearly linear time using a near linear time algorithm for computing the GCD of univariate polynomials, thereby giving a near linear time algorithm for unique decoding of Reed-Solomon Codes. The half GCD algorithm is also referred to as the Knuth-Sch{\"o}nhage algorithm, and we refer the interested reader to Chapter 11 of \cite{GathenG-MCA} for a description of the algorithm and detailed references. 

In the list decoding setting, Roth \& Ruckenstein \cite{RothR2000} gave a faster (though not near linear time) implementation of the algorithms of Sudan and Guruswami-Sudan for list decoding Reed-Solomon Codes. On the way to their algorithm, they designed a faster algorithm for root computation for bivariates stated in \autoref{thm:fast-root-finder}. In a subsequent work, Alekhnovich \cite{Alekhnovich} introduced lattice based techniques and relied on some of the ideas in \cite{RothR2000} to give nearly linear time algorithms for list decoding Reed-Solomon Codes of constant rate in parameter regimes approaching the Johnson Bound. 
For multiplicity codes (as well as folded Reed-Solomon Codes), near linear time algorithms for list decoding up to capacity (again in the constant rate regime) were described in a work of Goyal, Harsha, Kumar and Shankar \cite{GoyalHKS2024}. The presentation  in this chapter follows the presentation in \cite{GoyalHKS2024}.

\section{Local list decoding}\label{sec:local}

In this section we show how to list decode polynomial codes in \emph{sublinear time}. Sublinear-time algorithms (also known as \emph{local algorithms}) are highly-efficient algorithms which run in time that is much smaller than the input length. In particular,  such algorithms cannot even read the whole input, and many times they also cannot afford to write down the whole output. For this to be at all possible, we shall need to relax the computational model, and provide the algorithm with a \emph{query access to the input}, and also in the case the  output is too long, only require that the algorithm provides \emph{query access to the output}. Sublinear-time algorithms are also typically \emph{randomized}, and they have a small chance of error. 

In the context of (unique) decoding algorithms, this means that the local algorithm runs in time which is much smaller than the block length $n$ of the code. To enable this, we provide the algorithm with a query access to the received word $w \in \Sigma^n$, and additionally only require that for any given codeword entry $i \in [n]$, the algorithm correctly recovers the $i$-th entry of the closest codeword $c$ with high probability. Formally, we have the following definition.

\begin{definition}[Locally correctable code $(\LCC)$]\label{def:lcc}
Let $C\subseteq\Sigma^{n}$ be a code of relative distance $\delta$, and let $\alpha \in (0, \frac \delta 2)$. We say that $C$
is $(Q,\alpha)$-\textsf{locally correctable $(\LCC)$}
if there exists a randomized algorithm $A$ which satisfies the following requirements: 
\begin{itemize}
\item \textbf{Input:} $A$ takes as input a codeword entry $i\in\left[n\right]$
and also gets query access to a string $w\in\Sigma^{n}$ so that
$\delta(w,c)\leq \alpha$ for some codeword $c\in C$.  
\item \textbf{Query complexity:} $A$ makes at most $Q$ queries to the
oracle $w$. 
\item \textbf{Output:} $A$ outputs $c_{i}$ with probability at least $\frac 3 4$.
\end{itemize}
\end{definition}

\begin{remark}\label{rem:lcc}
A few remarks are in order:
\begin{enumerate}
\item We could have chosen the success probability to be any constant greater than $\frac 1 2$, since the success probability  can be amplified by repeating the local correction procedure independently for a constant number of times and outputting the majority value (at the cost of increasing the query complexity by a constant multiplicative factor). For simplicity, we fixed the success probability to $\frac 3 4$ in the above definition.
\item Often one is also interested in the \emph{running time} of the local correction algorithms. As is typically the case, in all the algorithms presented in this section, the running time will be polynomial in the query complexity $Q$.
\item A useful (and arguably more natural) variant of local correction is \emph{local decoding}. 
Locally decodable codes are defined similarly to locally correctable codes, except that
we view the codeword~$c$ as the encoding of some message~$m$ (by fixing some injective encoding map $\Phi: \Sigma^k \to C$, where $k$ is such that $|C|= |\Sigma|^k$), 
and the local decoder is required to decode individual entries in~$m$. 
If we restrict ourselves to \emph{linear codes}, then locally decodable codes are a weaker
object than locally correctable codes, since any linear locally correctable code can be converted into a locally decodable code by choosing a systematic encoding map.\footnote{ An encoding map $\Phi: \Sigma^k \to C$ is \emph{systematic} if any $m \in \Sigma^k$ is a prefix of $\Phi(m)$. For linear codes, such a map can be found by transforming the generating matrix $G$ into a reduced column echelon form, and letting $\Phi(m)= G \cdot m$.} For  simplicity, we focus in this section on the model of local correction, but we note that all of the algorithms we present can be easily adapted for the local decoding model as well.  
\end{enumerate}
\end{remark}

We will be interested in an extension of the above definition of local correction to the setting of list decoding. However, the actual definition requires some subtlety: we want to return a list of answers corresponding to all near-by codewords in the list, and we want the algorithm to be local, but returning a list of possible symbols in $\Sigma$ for individual entries of near-by codewords is pretty useless, since typically there would be many entries $i \in [n]$ so that the values of the $i$-th entry of all near-by codewords cover all of $\Sigma$. Furthermore, for applications one typically needs some form of \emph{consistency}, i.e., that for any fixed codeword $c$ in the list it is possible to locally correct  individual entries of $c$ with high probability in a consistent way. 
To ensure this, we require that the local list-decoding algorithm returns a list $A_1,A_2, \ldots, A_L$ of randomized local algorithms, so that for each near-by codeword $c$ in the list there would be some algorithm $A_j$ in the list that locally corrects $c$.
Formally, we have the following definition.

\begin{definition}[Locally list decodable code ($\LLDC$)]\label{def:lldc}
We say that a code $C\subseteq\Sigma^{n}$ is $(Q,\alpha, L)$-\textsf{locally list decodable ($\LLDC$)}
if there exists a randomized algorithm $A$ which outputs $L$ randomized algorithms $A_1,A_2, \ldots,A_L$ that
satisfy the following requirements for any $w \in \Sigma^n$: 
\begin{itemize}
\item \textbf{Input:}  Each $A_j$ takes as input a codeword entry $i\in\left[n\right]$
and also gets query access to $w$.
\item \textbf{Query complexity:} Each $A_j$ makes at most $Q$ queries to $w$.
\item \textbf{Output:}
 For every codeword $c\in C$ so that $\dist(c,w)\leq \alpha$, with probability at least $\frac 1 2$ over the randomness of $A$ the following event happens:
there exists some $j\in [L]$ such that for all $i \in [n]$,  $$ \Pr[A_j(i) = c_i] \geq \frac 3 4,$$ where the probability is over the internal randomness of $A_j$.
\end{itemize}
\end{definition}

\begin{remark}\label{rem:LLDC}
The same remarks as in Remark \ref{rem:lcc} also apply  to the above definition. We also note the following:
\begin{enumerate}
\item The success probability of~$\frac 1 2$ for the local list decoding algorithm $A$ in the above definition can be amplified by invoking the local list decoding algorithm $A$ independently for multiple times, and outputting the union of all randomized algorithms $A_j$ that are output in any of the invocations, at the cost of increasing the list size   
(Specifically, amplifying the success probability  to $1 - 2^{-t}$ requires increasing the list size by a multiplicative factor of $O(t)$).

\item In the above definition we required that for any \emph{fixed} close-by codeword $c$, the list of local algorithm $A_j$ contains  a local corrector for $c$ with high probability. We can also guarantee the stronger requirement that with high-probability,  the output list contains a local corrector for \emph{any} close-by codeword $c$, by first amplifying the success probability of $A$ to $1 - \frac{1} {4L}$, and then applying a union bound over all near-by codewords (noting that by the definition of local list decoding, there could be at most $2L$ such codewords).
\end{enumerate}
\end{remark}

Though the above definition of local list decoding may seem a bit strange and convoluted at first glance, it turns out to be just the right definition for various applications in coding theory and theoretical computer science. 

For a code to be locally (list) decodable, its codeword entries need to satisfy some non-trivial local relations. In the setting of polynomial-based codes, this property can be achieved by considering codes based on \emph{multivariate} polynomials, where the main property being used is that the restriction of a low-degree multivariate polynomial to any line (or more generally, a curve or a low-dimension subspace) is a low-degree \emph{univariate} polynomial. Thus, given a codeword entry, one can recover the value of the entry by passing a random line through the point, and recovering the closest low-degree univariate polynomial on the line (or more generally, a curve or a low-dimension subspace).

\medskip

In what follows, we first focus on  the  family of \emph{Reed-Muller $(\RM)$  Codes}, which extends Reed-Solomon Codes to the setting of \emph{multivariate} polynomials. We first show in Section \ref{subsec:rm_unique} below, as a warmup, and since this will also be used as a sub-procedure in the local list-decoding algorithm, how to \emph{locally correct} Reed-Muller Codes from a small fraction of errors (below half the minimum distance of the code). Then in Section \ref{subsec:rm_list} we show how to \emph{locally list decode} Reed-Muller Codes  beyond half the minimum distance, and in Section \ref{subsec:rm_list_johnson} we present an improved algorithm that is able to locally list decode Reed-Muller Codes  \emph{up to Johnoson bound}. Finally, in Section  \ref{subsec:mult_mult_list} we consider the family of \emph{multivariate multiplicity codes} which extends Reed-Muller Codes  by evaluating also \emph{multivariate derivatives} of multivariate polynomials, and outline a local list decoding algorithm for these codes \emph{up to their minimum distance}, and we further explain how this can be used to obtain capacity-achieving locally list decodable codes.

\subsection{Local correction of Reed-Muller Codes}\label{subsec:rm_unique}

In this section, we first show, as a warmup, and since this will also be used as a sub-procedure in the local list-decoding algorithm presented in the next section, how to \emph{locally correct} Reed-Muller Codes  from a small fraction of errors (below half the minimum distance of the code).  
More specifically, fix  a prime power $q$, and positive integers $k,m$ so that $k<q$. 
Recall that the  \emph{Reed-Muller $(\RM)$  Code} $\RM_{q,m}(k)$ is the code which associates with each $m$-variate polynomial $f(X_1, \ldots, X_m) \in \F_q[X_1, \ldots, X_m]$
of (total) degree smaller than $k$ a codeword $(f(\ba))_{\ba \in \F_q^m} \in \F_q^{q^m}$. Let $\delta: = 1 - \frac{k} {q}$ denote the relative distance of this code, below we shall show that this code is a $(q,\alpha)$-LCC for $\alpha:= \frac{\delta }{8} - \frac 1 q$. Note that for $m>1$, the query complexity of $q$ is much smaller than the block length $q^m$ of the code. 

Suppose that $w \in \F_q^{q^m}$ is a received word, in what follows it will be convenient to view $w$ as a function $w: \F_q^m \to \F_q$. 
Suppose that there exists a (unique) polynomial $f(X_1, \ldots, X_m) \in \F_q[X_1, \ldots, X_m]$ of degree smaller than $k$ so that $f(\ba) \neq w(\ba)$ for at most an $\alpha$-fraction of the entries $\ba \in \F_q^m$. Below we shall describe a randomized local correction algorithm which given as input an entry $\ba \in \F_q^m$, makes a few queries to $w$, and outputs $f(\ba)$ with probability at least $\frac 3 4$.

 \paragraph{Overview.} To decode $f(\ba)$, we pick a random line through $\ba$, and query the restriction of $w$ to the line. We then
 find the \emph{univariate} polynomial $g(X) \in \F_q[X]$ of degree smaller than $k$ that is closest to $w$ on the line (this can be done efficiently using the algorithm of Figure \ref{fig:rs_unique}), and return the value of $g$ on $\ba$. The main properties we use to show correctness are that the restriction of $f$ to any line is a univariate polynomial of degree smaller than $k$, and also that random lines \emph{sample-well} in the sense that the fraction of errors on a random line is typically close to the total fraction of errors. The formal description of the algorithm is presented in Figure \ref{fig:rm_unique_low_error} below, followed by correctness analysis.

\begin{figure}[h]
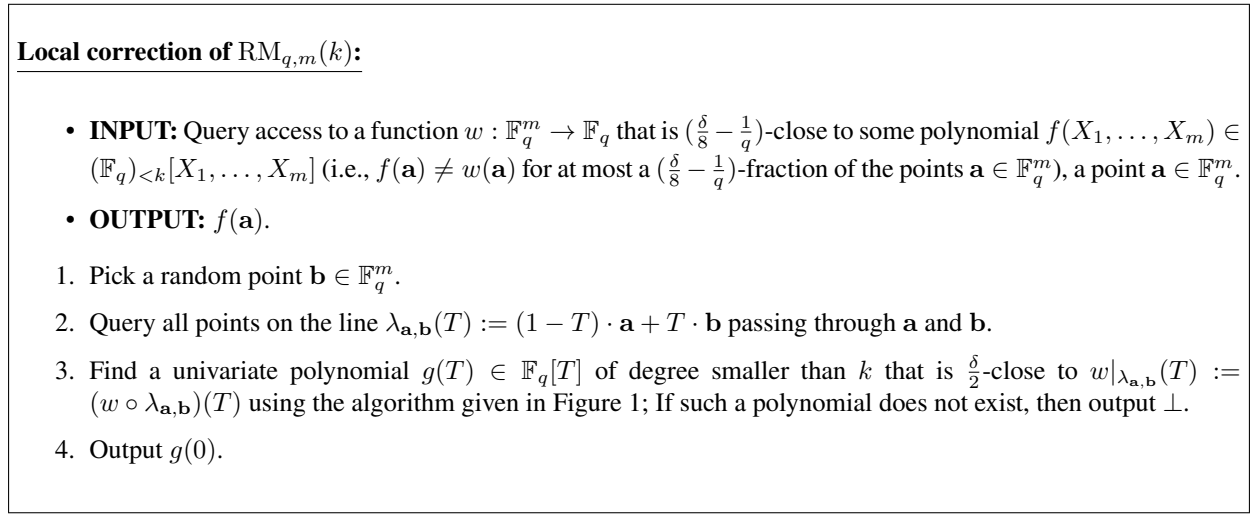

  \begin{boxedminipage}{\textwidth} \small \medskip \noindent
    $\;$

    \underline{\textbf{Local correction of $\RM_{q,m}(k)$:}}

    \medskip

\begin{itemize}
\item \textbf{INPUT:} Query access to a function $w: \F_q^m \to \F_q$ that is $(\frac \delta 8 - \frac 1 q)$-close to some polynomial $f(X_1, \ldots, X_m) \in (\F_q)_{<k}[X_1, \ldots, X_m]$ (i.e., $f(\ba) \neq w(\ba)$ for at most a $(\frac \delta 8 - \frac 1 q)$-fraction of the points $\ba \in \F_q^m$), 
a point $\ba \in \F_q^m$. 
\item \textbf{OUTPUT:} $f(\ba)$.

\end{itemize}

    \begin{enumerate}
\item Pick a random point $\bb \in \F_q^m$.
\item Query all points on the line $\lambda_{\ba, \bb}(T):=(1-T)
\cdot\ba+T\cdot\bb$ passing through $\ba$ and $\bb$.
\item Find a univariate polynomial $g(T) \in \F_q[T]$ of degree smaller than $k$ that is $\frac \delta 2$-close to $w|_{\lambda_{\ba, \bb}}(T):=(w \circ {\lambda_{\ba, \bb}})(T)$ using the algorithm given in Figure \ref{fig:rs_unique}; If such a polynomial does not exist, then output $\bot$.
\item Output $g(0)$.
\end{enumerate}

  \medskip

  \end{boxedminipage}

\caption{Local correction of Reed-Muller Codes}
\label{fig:rm_unique_low_error}
\end{figure}

\paragraph{Correctness.}
Note that for a random point $\bb \in \F_q^m$, any point on the line $\lambda_{\ba,\bb}(T) $ other than the zero point is uniformly random. Consequently, the expected fraction of errors on the line is at most $(\frac \delta 8 - \frac 1 q) + \frac 1 q =  \frac \delta 8$. By Markov's inequality, this implies in turn that with probability at least $\frac 3 4$ over the choice of $\bb$, the fraction of errors on the line 
$\lambda_{\ba,\bb}$ is less than $\frac \delta 2$. Lastly, assuming that the latter event happens,  by the correctness of the algorithm of Figure \ref{fig:rs_unique} (noting that $f|_{\lambda_{\ba,\bb}}$ is a univariate polynomial of degree smaller than $k$), we have that $g(T) = f|_{\lambda_{\ba, \bb}}(T)$, and in particular, $g(0)=f(\ba)$. So we conclude that the algorithm outputs $f(\ba)$ with probability at least $\frac 3 4$.

\paragraph{Query complexity.} The query complexity is clearly the number $q$ of points on a line.  

\begin{remark}\label{rem:rm_lcc}
While we have not made an attempt to optimize the decoding radius of the above local correction algorithm, we note that to get a non-trivial success probability greater than $\frac 1 2$, one must set the decoding radius to be smaller than $\frac \delta 4$. 
This is because the use of  \emph{Markov's inequality}, which guarantees that with probability at least $\frac 1 2$ over the choice of the random line through $\ba$, the number of errors is smaller than $\frac \delta 2$ only below this radius.

To locally correct Reed-Muller Codes  from a larger fraction of errors, one can replace in the algorithm of Figure \ref{fig:rm_unique_low_error} the random line $\lambda_{\ba,\bb}(T)$ passing through $\ba$ with a random \emph{curve} passing through $\ba$ of the form $\chi_{\ba,\bb,\bc}(T) = \ba + \bb T + \bc T^2$, where $\bb, \bc \in \F_q^m$ are uniformly random. The advantage of decoding on curves is that any two points on a random curve through $\ba$ are \emph{pairwise independent}, and consequently one can use  \emph{Chebyshev inequality} to get a better bound on the number of errors on the curve. However,  resorting to a curve also increases the degree of the restriction of $f$ to the curve, and so one can tolerate less errors on the curve. But it turns out that this tradeoff is favorable, and allows one to locally correct up to a radius of roughly $ \delta - \frac 1 2$, which is close to half the minimum relative distance for $\delta$ which is close to $1$ (see \cite[Section 2.2.3]{Yekhanin_survey} for a full description and analysis of this algorithm). 
\end{remark}

\subsection{Local list decoding of Reed-Muller Codes beyond unique decoding radius}\label{subsec:rm_list}

Next we show how to \emph{locally list-decode} Reed-Muller Codes  beyond the unique decoding radius.  Recall that in local list decoding, we require that the local list-decoding algorithm returns a list 
 $A_1,A_2, \ldots, A_L$ of randomized local algorithms, so that for each near-by codeword $c$ in the list there would be some algorithm $A_j$ in the list that locally corrects $c$.
In fact, since we have already shown that Reed-Muller Codes  are \emph{locally correctable}, it suffices to 
show that they are \emph{approximately} locally list decodable, in the sense that the local algorithms $A_j$  are only required to correctly recover \emph{most} of the entries of $c$. Such an approximate local list decoding algorithm can then be turned into a local list decoding algorithm by composing 
each local algorithm $A_j$ with the local correction algorithm presented in the previous section. 
Furthermore, using an averaging argument, it can be shown that in approximate local list decoding each of the local algorithms $A_j$ can be assumed to be \emph{deterministic}. 
Formally, we have the following definition.

\begin{definition}[Approximately locally list-decodable code $(\ALLDC)$]\label{def:alldc}
We say that a code $C\subseteq\Sigma^{n}$ is $(Q,\alpha,\epsilon, L)$-\textsf{approximately locally list decodable $(\ALLDC)$}
if there exists a randomized algorithm $A$ which outputs $L$ \emph{deterministic} algorithms $A_1,A_2, \ldots,A_L$ that
satisfy the following requirements for any $w \in \Sigma^n$: 
\begin{itemize}
\item \textbf{Input:}  Each $A_j$ takes as input a codeword entry $i\in\left[n\right]$
and also gets query access to $w$.
\item \textbf{Query complexity:} Each $A_j$ makes at most $Q$ queries to $w$.
\item \textbf{Output:}  
For every codeword $c\in C$ so that $\dist(c,w)\leq \alpha$, with probability at least $\frac 1 2$ over the randomness of $A$ the following event happens: there exists some $j\in [L]$ such that $$ \Pr_{i \in [n]}[A_j(i) = c_i] \geq 1-\epsilon,$$ where the probability is over the choice of a \emph{random entry} $i \in [n]$.
\end{itemize}
\end{definition}

\begin{remark}\label{rem:ALLDC}
The same remarks as in  Remark \ref{rem:LLDC} also apply to the above definition. We note however that in the case of ALLDC, the success probability of each local algorithm $A_j$ cannot be amplified by repetition, and therefore we have not fixed its value to $\frac 3 4$ as in the previous $\LLDC$ definition (Definition \ref{def:lldc}).
\end{remark}

The following lemma shows that one can turn an approximate local list decoding algorithm into a local list decoding algorithm by composing each local algorithm $A_j$ with a local correction algorithm.

\begin{lemma}[$\ALLDC + \LCC \Rightarrow \LLDC$]\label{lem:approx-local-to-local}
Suppose that $C \subseteq \Sigma^n$ is a code that is $(Q,\alpha,\epsilon,L)$-$\ALLDC$ and $(Q',\alpha')$-LCC for $\alpha' \geq \epsilon$. Then $C$ is also a $(Q\cdot Q', \alpha,L)$-$\LLDC$.
\end{lemma}

\begin{proof}
The local list decoding algorithm  replaces each of the local algorithms $A_j$ generated by the approximate local list decoding algorithm $A$ with a local algorithm $A'_j$, operating as follows: on input $i \in [n]$, the local algorithm $A'_j$ invokes the local correction algorithm $A'$ on input $i$, and forwards each of the queries of $A'$ in $[n]$ to the local algorithm $A_j$.

Correctness follows since by our assumption that $\alpha' \geq \epsilon$, $A_j$ outputs the correct value on at least a $(1-\alpha')$-fraction of the entries in $[n]$, and $A'$ outputs the correct value with probability at least $\frac 3 4$ provided that at most an $\alpha'$-fraction of the codeword entries are corrupted.
The query complexity is $Q \cdot Q'$, because for each of the $Q'$ queries that $A'$ would make, $A_j'$ runs a copy of $A_j$ which makes $Q$ queries.
\end{proof}

Given the above lemma, our task now is to design an \emph{approximate} local list decoding algorithm for Reed-Muller Codes beyond half their minimum distance. 
As before, fix  a prime power $q$ and positive integers $k,m$ so that $k<q$, and let $\RM_{q,m}(k)$ be the corresponding Reed-Muller Code. Let $\delta: = 1 - \frac{k} {q}$ denote the relative distance of this code, let $\sigma>1$ and $\xi >0$ be parameters, and let
$\alpha:=1-\sqrt{ \sigma \cdot(1-\delta)} - \xi$. 
We shall show that this code is a $(q, \alpha, \epsilon, L)$-ALLDC for 
$\epsilon =  \frac{2\delta} {\sigma -1} +  O(\frac{1} {\xi^2 q})$ and $L=q$. Thus, for sufficiently large $\sigma$ and $q$, the failure probability $\epsilon$ is smaller than $\frac \delta 8$, 
and so by the results of Section \ref{subsec:rm_unique} and by
Lemma \ref{lem:approx-local-to-local}, we also have that this code is a $(q^2, \alpha,L)$-LLDC. Note that for any fixed $\sigma>1$, the decoding radius of $1-\sqrt{ \sigma \cdot(1-\delta)}$ is larger than half the minimum relative distance for a sufficiently large relative distance $\delta$. Note that for any $m>2$, the query complexity of $q^2$ is much smaller than the block length $q^m$ of the code. 

Below we shall describe a \emph{randomized} approximate local list decoding algorithm $A$ which outputs a list of \emph{deterministic} $q$-query local algorithms $A_1,\ldots, A_L$ with the following guarantee. 
Suppose that $w \in \F_q^{q^m}$ is a received word (as before it will be convenient to view $w$ as a function $w: \F_q^m \to \F_q$), and suppose that 
$f(X_1, \ldots, X_m) \in \F_q[X_1, \ldots, X_m]$ is a polynomial of degree smaller than $k$ so that $f(\ba) \neq w(\ba)$ for at most an $\alpha$-fraction of the entries $\ba \in \F_q^m$. We shall show that with probability at least $\frac 1 2$ over the randomness of $A$, there exists a local algorithm $A_j$ which correctly recovers all but at most an $
\epsilon$-fraction of the values $f(\ba)$ for $\ba \in \F_q^m$.

\paragraph{Overview.}

To generate the list of local algorithms, the algorithm $A$ picks a uniformly random point $\bb \in \F_q^m$ and for any value $v \in \F_q$, it outputs a local algorithm $A_{\bb, v}$ (note that the number of local algorithms is indeed $L=q$). We think of the value $v$ as a ``guess" or ``advice" for the value of $f(\bb)$.
Once we have this advice, we use it to locally correct $f$. Specifically, to recover $f(\ba)$, the local algorithm $A_{\bb,v}$ first considers the line $\lambda_{\ba, \bb}$ passing through $\ba$ and $\bb$,
and globally list decodes the restriction of $w$ to this line to obtain a list $\calL \subseteq \F_q[T]$ of univariate polynomials (this can be done efficiently using the algorithm of Figure \ref{fig:rs_list_johnson}).
These univariate polynomials are candidates for $f|_{\lambda_{\ba,\bb}}$. Which of these univariate polynomials
is $f|_{\lambda_{\ba,\bb}}$? We use our guess $v$ for $f(\bb)$: if there is a unique
univariate polynomial in the list with value $v$ at $\bb$, then we deem that to be our candidate for
$f|_{\lambda_{\ba,\bb}}$, and output its value at the point $\ba$ as our guess for $f(\ba)$. 

The above algorithm 
will correctly recover $f(\ba)$ if (1) there are not too many errors on the line through $\ba$ and $\bb$,
and (2) no other polynomial in $\calL$ takes the same value at $\bb$ as $f$ does.  
As we have already shown in the previous section, the first event is high probability by standard sampling bounds. As to the second event, using  
that any pair of polynomials in $\calL$ differ by at least a $\delta$-fraction of the points, we get that any other polynomial $h(T) \in \calL$ will agree with $f|_{\lambda_{\ba,\bb}}$ on $\bb$ with probability at most $1-\delta$ over the choice of $\bb$. Assuming that the list $\calL$ is sufficiently small (which happens if the decoding radius is sufficiently below the Johnson Bound), one can then apply a union bound over all elements of 
$\calL$  to show that with high probability, no other polynomial $h(T) \in \calL$ agrees with $f|_{\lambda_{\ba,\bb}}$ on $\bb$.
 
The formal description of the approximate local list decoding algorithm is presented Figure \ref{fig:rm_list} below, followed by correctness analysis.

\begin{figure}[h]
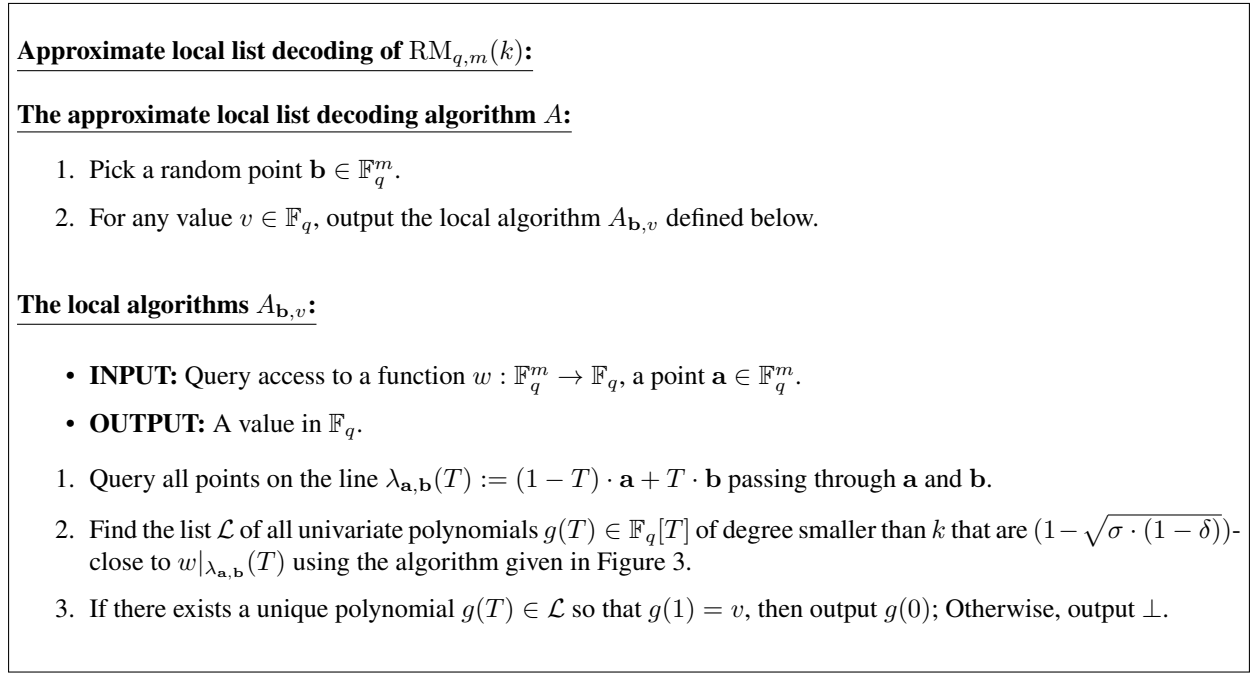

  \begin{boxedminipage}{\textwidth} \small \medskip \noindent
    $\;$

    \underline{\textbf{Approximate local list decoding of $\RM_{q,m}(k)$:}}

    \bigskip

   \underline{\textbf{The approximate local list decoding algorithm $A$:}}

    \begin{enumerate}
\item Pick a random point $\bb \in \F_q^m$.
\item For any value $v\in \F_q$, output the local algorithm $A_{\bb,v}$ defined below.
\end{enumerate}

\bigskip

 \underline{\textbf{The local algorithms $A_{\bb,v}$:}}

 \medskip

 \begin{itemize}
\item \textbf{INPUT:} Query access to a function $w: \F_q^m \to \F_q$, 
a point $\ba \in \F_q^m$. 
\item \textbf{OUTPUT:} A value in $ \F_q$.

\end{itemize}

    \begin{enumerate}
\item  
Query all points on the line $\lambda_{\ba, \bb}(T):=(1-T)
\cdot\ba+T\cdot\bb$ passing through $\ba$ and $\bb$.
\item Find the list $\mathcal{L}$ of all univariate polynomials $g(T) \in \F_q[T]$ of degree smaller than $k$ that are $(1-\sqrt{ \sigma\cdot(1-\delta)})$-close to $w|_{\lambda_{\ba,\bb}}(T)$ using the algorithm given in Figure \ref{fig:rs_list_johnson}.
\item \label{step:rm_list_candidate} If there exists a unique polynomial $g(T)  \in \mathcal{L}$ so that $g(1)=v$, then output $g(0)$; Otherwise, output $\bot$.
\end{enumerate}

  \medskip

  \end{boxedminipage}

\caption{Local list decoding of Reed-Muller Codes}
\label{fig:rm_list}
\end{figure}

\paragraph{Correctness.}
Let $f(X_1, \ldots, X_m) \in \F_q[X_1, \ldots, X_m]$ be a polynomial of degree smaller than $k$ so that $f(\ba) \neq w(\ba)$ for at most an $\alpha$-fraction of the entries $\ba \in \F_q^m$. We would like to show that with probability at least $\frac 1 2$ over the choice of $\bb$, the local algorithm $A_{\bb, f(\bb)}$ correctly recovers $f(\ba)$ for all but at most an $\epsilon$-fraction of  $\ba \in \F_q^m$. To show this, we shall first prove that with high probability  over the choice of \emph{both} random $\ba \in \F_q^m$ and random $\bb \in \F_q^m$, it holds that $A_{\bb, f(\bb)}$ correctly computes $f(\ba)$ on input $\ba$. 

We start by showing that with high probability over the choice of $\ba, \bb$, there exists a polynomial $g(T) \in \mathcal{L}$ that agrees with $f$ on the line $\lambda_{\ba,\bb}$.
 
 \begin{claim}\label{clm:rm_list_dec_1}
 With probability at least $1 - \frac{1} {\xi^2 q}$ over the choice of random $\ba, \bb \in \F_q^m$,
it holds that
  $f|_{\lambda_{\ba,\bb}} \in \mathcal{L}$.
 \end{claim}

 \begin{proof}
Note that for random points $\ba, \bb \in \F_q^m$, the line $\lambda_{\ba,\bb}$ is a uniformly random line. Consequently, any single point on the line is uniformly random, and any pair of points on the line are \emph{pairwaise independent}. Let $Z_1,\dots,Z_q$ be indicator random variables where $Z_t=1$ if $w$ differs from $f$ on the $t$'th point of the line $\lambda_{\ba,\bb}$, let $Z:=\frac{\sum_{t =1}^q Z_t} {q}$, and note that  by the correctness of the algorithm of Figure \ref{fig:rs_list_johnson}, we have that $f|_{\lambda_{\ba,\bb}} \in \mathcal{L}$ if $Z \leq 1-\sqrt{ \sigma \cdot(1-\delta)}$. Moreover, by linearity of expectation and since any point on the line is uniformly random, we have that
$$\mathbb{E}[Z] = \frac {\sum_{t=1}^q \mathbb{E}[Z_t]} {q} \leq  1-\sqrt{ \sigma \cdot(1-\delta)} - \xi,$$ while by pairwise independence, we have that
$$\Var(Z) = \frac{\sum_{t=1}^q \Var[Z_t]} {q^2} \leq \frac{\sum_{t=1}^q \mathbb{E}[Z_t] } {q^2}\leq \frac 1 q.$$
Consequently, by Chebyshev's inequality, 
$$
\Pr\left[Z > 1-\sqrt{ \sigma\cdot(1-\delta)} \right]  \leq 
\Pr \left[ Z - \mathbb{E}[Z] > \xi \right]
 \leq \frac{ \Var(Z)} { \xi^2 } \leq \frac 1 {\xi^2 q}.
$$
So we conclude that $Z \leq 1-\sqrt{ \sigma \cdot(1-\delta)}$ with probability at least $ 1 -  \frac 1 {\xi^2 q}$, in which case we have that $f|_{\lambda_{\ba,\bb}} \in \mathcal{L}$.
\end{proof}

Next we show that with high probability over the choice of $\ba, \bb$, there is no polynomial other than $f|_{\lambda_{\ba,\bb}}$  in $\mathcal{L}$ that agrees with
$f$ on $\bb$.

 \begin{claim}\label{clm:rm_list_dec_2}
 With probability at least $1 - \frac{\delta} {\sigma -1}$ over the choice of random $\ba, \bb \in \F_q^m$, there is no polynomial $h(T) \in \mathcal{L}$ other than
  $f|_{\lambda_{\ba,\bb}}$ so that $h(1) = f(b)$.
 \end{claim}

 \begin{proof}
We can  view the process of picking random $\ba, \bb \in \F_q^m$ as first picking a random $\ba \in \F_q^m$, then picking a random line $\lambda$ through $\ba$, and finally letting $\bb \in \F_q^m$ be a random point on the line $\lambda$. So in what follows, fix a point $\ba$ and a line $\lambda$ through it, and let $\mathcal{L}$ be the list of univariate polynomials of degree smaller than $k$ that are $(1 - \sqrt{\sigma (1-\delta)})$-close to $w$ on the line.
We shall show that with probability at least $1 - \frac{\delta} {\sigma -1}$ over the choice of $\bb$, it holds that there is no polynomial $h(T) \in \mathcal{L}$ other than
  $f|_\lambda$ so that $h$ and $f$ agree on $\bb$.

Fix a polynomial $h(T) \in \mathcal{L}$ so that $h(T) \neq f|_\lambda(T)$. Then since both $h$ and $f|_\lambda$ are polynomials of degree smaller than $k$, they must differ by at least a 
$\delta$-fraction of the points on the line. Consequently, we have that 
$h$ and $f|_\lambda$  agree on the point $\bb$ with probability at most $1- \delta$ over the choice of $\bb$.
Furthermore, note that by the Johnson Bound (cf., Theorem \ref{thm:johnson}), we have that 
\begin{equation}\label{eq:rm_list_dec_johnson_1}
|\mathcal{L}| 
\leq 
\frac{\delta }{\sigma\cdot(1-\delta)-(1-\delta)} \leq 
\frac{\delta } {(\sigma-1)(1-\delta)}.
\end{equation}
Consequently, by a union bound over all elements in $\mathcal{L}$, the probability that there exists a polynomial $h(T) \in \mathcal{L}$ other than $f|_\lambda$ so that $h$ and $f$ agree on $\bb$ is at most $\frac{\delta} {\sigma -1}$.
 \end{proof}

By the above Claims \ref{clm:rm_list_dec_1} and \ref{clm:rm_list_dec_2}, we conclude that with probability at least $1 - \frac{\delta} {\sigma -1} - \frac 1 {\xi^2 q}$ over the choice of \emph{both} $\ba$ and $\bb$, 
it holds that $f|_{\lambda_{\ba,\bb}} \in \mathcal{L}$ and
there is no polynomial $h(T) \in \mathcal{L}$ other than
  $f|_{\lambda_{\ba,\bb}}$ so that $h(1) = f(b)$, in which case the local algorithm $A_{\bb, f(\bb)}$ will set $g =f|_{\lambda_{\ba,\bb}}$ in Step \ref{step:rm_list_candidate}, and will output  
 $f|_{\lambda_{\ba,\bb}}(0) = f(\bb)$. By Markov's inequality, this implies in turn that with probability at least $\frac 1 2$ over the choice of $\bb$, the local algorithm $A_{\bb, f(\bb)}$ outputs $f(\ba)$ for at least a $(1-\epsilon)$-fraction of the points $\ba$ for
 $ \epsilon = \frac{2\delta} {\sigma -1} + \frac 2 {\xi^2 q}.
$

\paragraph{Query complexity and list size.} The list size is clearly the number $q$ of possible values in $\F_q$, and the
query complexity is clearly the number $q$ of points on a line.

\subsection{Local list decoding of Reed-Muller Codes up to the Johnson Bound}\label{subsec:rm_list_johnson}

The local list decoding algorithm for Reed-Muller Codes  that we presented in the previous section only gives a decoding radius approaching $1-\sqrt{ \sigma \cdot(1-\delta)}$ for a sufficiently large constant $\sigma>1$,
which is smaller than the Johnson Bound of  $1 - \sqrt{1-\delta}$. While we have not made an attempt to optimize the value of $\sigma$, we note that it cannot be smaller than $2$.
The reason is that in order to apply the union bound in Claim \ref{clm:rm_list_dec_2} we need the list size $L':=|\mathcal{L}|$ to be smaller than $\frac{1} {1-\delta}$ (since the failure probability for each individual polynomial in the list can be as large as $1-\delta$), and by (\ref{eq:rm_list_dec_johnson_1}), such a list size is only guaranteed for a decoding radius smaller than $1-\sqrt{2(1-\delta)}$. 

In this section we show an improved algorithm which can locally list decode Reed-Muller Codes  all the way up to the Johnson Bound. In more detail, as before, fix  a prime power $q$ and positive integers $k,m$ so that $k<q$, and let $\RM_{q,m}(k)$ be the corresponding Reed-Muller Code. Let $\delta: = 1 - \frac{k} {q}$ denote the relative distance of this code, let $\gamma, \xi >0$ be parameters, and let
$\alpha:=1-\sqrt{ (1+\gamma) \cdot(1-\delta)} - \xi$. 
We shall show that for any $s\geq 1$, this code is a $(Q, \alpha,   \epsilon, L)$-ALLDC for 
$Q = r^m \cdot q$, 
$\epsilon =  \frac{2\delta} {\gamma s} +  O\left(\frac{1} {\xi^2 q}\right)$, and 
$L = (\frac r s)^s$, where $r:=\frac{8 \delta s} {\gamma(1-\delta)}$.
Thus, for  $s > \frac{16 \delta } {\gamma}$ and sufficiently large $q$, the failure probability $\epsilon$ is smaller than $\frac \delta 8$, 
and so by the results of Section \ref{subsec:rm_unique} and by
Lemma \ref{lem:approx-local-to-local}, we also have that this code is a $(r^m \cdot q^2, \alpha,L)$-LLDC. Note that for constant $m,s, \gamma$, and $\delta$, the query complexity is  $r^m\cdot q = O(q^2)$ which is much smaller than the block length $q^m$ of the code for $m>2$ and a growing $q$. 

\paragraph{Overview.} 
The main idea behind the improved algorithm 
 is to also use {\em derivatives} to disambiguate the list.
Specifically, as before the advice is generated by choosing a uniformly random point $\bb \in \F_q^m$, but now the guess $v$
is supposed to equal to all multivariate derivatives of order smaller than $s$
of $f(X_1, \ldots, X_m)$ at $\bb$ for some integer $s\geq 1$ (cf., Section \ref{subsec:prelim_poly} for the definition of multivariate derivatives). 
To take advantage of derivatives, we recall that by Fact \ref{fact:prelim_hasse_sz_univ}, any pair of  univariate degree $k$ polynomials agree  on all derivatives of order smaller than $s$ on at most a fraction of 
$\frac{k} {sq} = \frac{1-\delta} {s}$ of the points. Consequently, the advice fails to disambiguate any pair of polynomials in $\calL$ with probability at most $\frac{1-\delta} {s}$, and taking $s \gg L'$ now allows us to apply a union bound over  $\calL$.

However, a disadvantage of this approach is that it significantly increases the list size of the local list decoding algorithm, which is the number of different guesses for all $m$-variate derivatives of order smaller than $s$ at the point $\bb$. Specifically, the number of different guesses is $q^{{m +s-1 \choose m}}$ (which is much larger than the block length $q^m$ of the code!), since the number of $m$-variate derivatives of order smaller than $s$ is ${m +s-1 \choose m}$, and each such derivative can take any value in $\F_q$.

To reduce the list size, we slightly extend the model of local list decoding by providing the local list decoding algorithm $A$ with a query access to the received word $w$ (so the algorithm $A$ can make a small number of queries to $w$ before outputting the list of local algorithms $A_1, \ldots, A_L$). 
To take advantage of this model, the local list decoding algorithm $A$ first chooses $t$ random lines through $\bb$, then (globally) list decodes the restriction of $f$ to these lines, and finally only outputs  guesses $v$ for $f^{(<s)}(\bb)$ that are consistent with  the directional derivatives of these lines at $\bb$ for a significant fraction of the lines. 
For a careful choice of the joint distribution of the $t$ random lines, bounding the number of guesses can be cast as an instance of a {\em list recovery} problem for Reed-Muller Codes (a certain generalization of list decoding), and in particular, we can use an extension of the Johnson Bound to the setting of  list recovery to bound the number of such guesses.

\medskip

The formal description of the approximate local list decoding algorithm is given in Figure \ref{fig:rm_list_johnson} below, followed by correctness analysis. 
  Before we formally state the algorithm, we first fix some notation. 
This notation is motivated by the following \emph{chain rule} for Hasse Derivatives, which relates the directional derivatives of a multivariate polynomial on a line to its global derivatives (the proof follows directly from the definition of the Hasse derivative and is left as an exercise). 
 
 \begin{fact}[Chain rule for Hasse Derivatives]\label{clm:chain}
Let $f(X_1, \ldots, X_m) \in \F[X_1, \ldots, X_m]$, let $\ba, \bu \in \F^m$, and let $\lambda(T) = \ba + T\bu$ be the line passing through $\ba$ in \emph{direction} $\bu$. 
Then for any non-negative integer $i$ and $T \in \F$,
\begin{equation}\label{eq:chain_1}
(f |_\lambda)^{(i)}(T) = \sum_{\bi : \wt(\bi) = i} f^{(\bi)}(\ba+T\bu)\cdot \bu^{\bi}.
\end{equation}
\end{fact}

For integers $m,s$, let $I_{m,s} := \{ \bi \in \N^m \mid \wt(\bi) < s \}$, and 
 for an element $v \in \F_q^{I_{m,s}}$, and a direction
$\bu \in \F_q^m$, we define the restriction $v|_\bu \in \F_q^s$ of $v$ to direction
$\bu$ to equal
$
v|_\bu:= \left(\sum_{\bi : \wt(\bi) = i} v_\bi \cdot \bu^{\bi} \right)_{i=0,1,\ldots, s-1}. 
$
Note that in this notation, (\ref{eq:chain_1}) can be rephrased as
\begin{equation}\label{eq:chain_2}
(f|_\lambda)^{(<s)}(T) =  \left(f^{(<s)}(\ba+T\bu)\right)|_{\bu}.
\end{equation}

We now describe the algorithm. 
\begin{figure}[H]
  \begin{boxedminipage}{\textwidth} \small \medskip \noindent

    \underline{\textbf{Approximate local list decoding of $\RM_{q,m}(k)$:}}

    \bigskip

   $ \triangleright$ Let $s\geq 1$ be a parameter, let $r:= \frac {8\delta s} {\gamma (1-\delta)} $, and let $U$ be an arbitrary subset of $\F_q$ of size $r$.

     \bigskip

   \underline{\textbf{The approximate local list decoding algorithm $A$:}}

\begin{itemize}
\item \textbf{INPUT:} Query access to a function $w: \F_q^m \to \F_q$. 
\item \textbf{OUTPUT:} A list of local algorithms.
\end{itemize}

    \begin{enumerate}
\item Pick random points $\bb \in \F_q^m$ and $\bu_0 \in \F_q^m$.
\item For each $\bu \in U^m:$
\begin{enumerate}
\item  
Query all points on the line $\lambda_\bu(T):= \bb + T (\bu_0+\bu)$ through $\bb$ in direction $\bu_0+\bu$.
\item Find the list $\mathcal{L}_\bu$ of all univariate polynomials $g(T) \in \F_q[T]$ of degree smaller than $k$ that are $(1-\sqrt{ (1+\gamma)\cdot(1-\delta)})$-close to $w|_{\lambda_\bu}$ using the algorithm given in Figure \ref{fig:rs_list_johnson}.

\item Let $S_\bu = \left\{ g^{(<s)}(0) \mid g(T) \in \mathcal{L}_\bu \right\}$.
\end{enumerate}
\item Let $V$ be the set containing all $v \in \F_q^{I_{m,s}}$ which satisfy  that $v|_{\bu_0+\bu} \in S_{\bu}$ for at least $\frac 1 2$ of the points $\bu \in  U^m$. 
\item For any $v\in V$, output the local algorithm $A_{\bb,v}$ defined below.
\end{enumerate}

\bigskip

 \underline{\textbf{The local algorithms $A_{\bb,v}$:}}

 \medskip

 \begin{itemize}
\item \textbf{INPUT:} Query access to a function $w: \F_q^m \to \F_q$, 
a point $\ba \in \F_q^m$. 
\item \textbf{OUTPUT:} A value in $ \F_q$.

\end{itemize}

    \begin{enumerate}
\item  
Query all points on the line $\lambda_{\ba, \bb}(T):=(1-T)
\cdot\ba+T\cdot\bb = \ba + T(\bb - \ba)$ passing through $\ba$ and $\bb$.
\item Find the list $\mathcal{L}$ of all univariate polynomials $g(T) \in \F_q[T]$ of degree smaller than $k$ that are $(1-\sqrt{ (1+\gamma)\cdot(1-\delta)})$-close to $w|_{\lambda_{\ba, \bb}}(T)$ using the algorithm given in Figure \ref{fig:rs_list_johnson}.
\item \label{step:rm_list_johnson_candidate} If there exists a unique polynomial $g(T)  \in \mathcal{L}$ so that $g^{(<s)}(1)=  v|_{\bb - \ba}$, then output $g(0)$; Otherwise, output $\bot$.
\end{enumerate}

  \medskip

  \end{boxedminipage}

  \caption{Local list decoding of Reed-Muller Codes up to Johnson Bound}
\label{fig:rm_list_johnson}
\end{figure}

\paragraph{Correctness.} Let $f(X_1, \ldots, X_m) \in \F_q[X_1, \ldots, X_m]$ be a polynomial of degree smaller than $k$ so that $f(\ba) \neq w(\ba)$ for at most an $\alpha$-fraction of the entries $\ba \in \F_q^m$. 
We first show that with high probability over the choice of $\bb, \bu_0 \in \F_q^m$, it holds that $f^{(<s)}(\bb) \in V$.

\begin{claim}\label{clm:rm_list_dec_johnson_1}
With probability at least $1 - \frac{2} {\xi^2 q}$ over the choice of $\bb, \bu_0 \in \F_q^m$, it holds that $f^{(<s)}(\bb) \in V$.
\end{claim}

\begin{proof}
First note that for any fixed $\bu \in U^m$, since $\bb$ and $\bu_0$ are chosen uniformly at random, then the line $\lambda_\bu(T) = \bb + T(\bu_0+\bu)$ is uniformly random. Consequently, for any fixed $\bu \in U^m$, the same analysis as in Claim \ref{clm:rm_list_dec_1} shows that with probability at least $1 - \frac{1} {\xi^2 q}$ over the choice of $\bb$ and $\bu_0$, it holds that $f|_{\lambda_{\bu}} \in \mathcal{L}_{\bu}$. By Markov's inequality, this implies in turn that with probability at least $1 - \frac{2} {\xi^2 q} $ over the choice of $\bb$ and $\bu_0$, it holds that $f|_{\lambda_{\bu}} \in \mathcal{L}_{\bu}$ for at least $\frac 12 $ of the points $\bu \in U^m$.

Next observe that for any $\bu \in  U^m$ so that $f|_{\lambda_{\bu}} \in \mathcal{L}_{\bu}$ we have that $(f|_{\lambda_{\bu}})^{(<s)} (0) \in S_{\bu}$. Moreover, by (\ref{eq:chain_2}), it holds that
$
(f|_{\lambda_{\bu}})^{(<s)}(0) =  \left(f^{(<s)}(\bb)\right)|_{\bu_0+\bu}
$ for each such $\bu$.
We conclude that with probability at least $1 - \frac{2} {\xi^2 q} $ over the choice of $\bb$ and $\bu_0$, we have that $\left(f^{(<s)}(\bb)\right)|_{\bu_0+\bu} \in S_{\bu}$ for at least $\frac 1 2$ of $\bu \in U^m$, in which case $f^{(<s)}(\bb)$ will be included in $V$.
\end{proof}

Next we would like to show that with probability at least $\frac 1 2$ over the choice of $\bb$,
the local algorithm $A_{\bb, f^{(<s)}(\bb)}$ correctly recovers $f(\ba)$ for most points  $\ba \in \F_q^m$. This part of the anlysis is very similar to the analysis of the algorithm of Figure \ref{fig:rm_list}, where the main difference is that we replace the standard bound on the maximum number of zeros of a low-degree polynomial with the stronger version that also accounts for multiplicities given by Fact \ref{fact:prelim_hasse_sz_univ}, which allows us in turn to handle larger lists $\mathcal{L}$.

In more detail, as before, to show the above, we shall first show that with high probability  over the choice of \emph{both} random $\ba \in \F_q^m$ and random $\bb \in \F_q^m$, it holds that $A_{\bb,  f^{(<s)}(\bb)}$ correctly computes $f(\ba)$ on input $\ba$. 
The following claim says that with high probability over the choice of $\ba, \bb$, there exists a polynomial $g(T) \in \mathcal{L}$ that agrees with $f$ on the line $\lambda_{\ba,\bb}$. The proof is very similar to that of Claim \ref{clm:rm_list_dec_1}.

 \begin{claim}\label{clm:rm_list_dec_johnson_2}
 With probability at least $1 - \frac{1} {\xi^2 q}$ over the choice of random $\ba, \bb \in \F_q^m$,
it holds that
  $f|_{\lambda_{\ba,\bb}} \in \mathcal{L}$.
 \end{claim}

The next claim says that with high probability over the choice of $\ba, \bb$, there is no polynomial $h(T) \in \mathcal{L}$ 
other than $f|_{\lambda_{\ba,\bb}}$   so that 
$h^{(<s)}(1) = f^{(<s)}(\bb)|_{\bb - \ba}$.

 \begin{claim}\label{clm:rm_list_dec_johnson_3}
 With probability at least $1 - \frac \delta {\gamma s}$ over the choice of random $\ba, \bb \in \F_q^m$, there is no polynomial $h(T) \in \mathcal{L}$ other than
  $f|_{\lambda_{\ba,\bb}}$ so that $h^{(<s)}(1) = f^{(<s)}(\bb)|_{\bb - \ba}$.
 \end{claim}

 \begin{proof}
 The proof is similar to that of Claim \ref{clm:rm_list_dec_2},
 except that we use Fact \ref{fact:prelim_hasse_sz_univ} to bound the probability that  
 $h$ and $f|_\lambda$ agree on $\bb$ on all derivatives of order smaller than $s$.
 In more detail, as in the proof of Claim \ref{clm:rm_list_dec_2}, we can  view the process of picking random $\ba, \bb \in \F_q^m$ as first picking a random $\ba \in \F_q^m$, then picking a random line $\lambda$ through $\ba$, and finally letting $\bb \in \F_q^m$ be a random point on the line $\lambda$. So in what follows, fix a point $\ba$ and a line $\lambda$ through it, and let $\mathcal{L}$ be the list of univariate polynomials of degree smaller than $k$ that are $(1 - \sqrt{(1+\gamma)(1-\delta)})$-close to $w$ on the line.
We shall show that with probability at least $1 - \frac \delta {\gamma s}$ over the choice of $\bb$, it holds that there is no polynomial $h(T) \in \mathcal{L}$ other than
  $f|_\lambda$ so that $h$ and $f|_\lambda$ agree on $\bb$ on all derivatives of order smaller than $s$.

Fix a polynomial $h(T) \in \mathcal{L}$ so that $h(T) \neq f|_\lambda(T)$. Then, since both $h$ and $f|_\lambda$ are polynomials of degree smaller than $k$, by Fact \ref{fact:prelim_hasse_sz_univ}, they agree on all of their derivatives of order smaller than  $s$ on at most a $ \frac{1-\delta} s$-fraction of the points on the line. 
Moreover, note that by the Johnson Bound (cf., Theorem \ref{thm:johnson}), we have that 
\begin{equation}\label{eq:rm_list_dec_johnson_2}
|\mathcal{L}| 
\leq 
\frac{\delta  }{(1+\gamma) \cdot(1-\delta)-(1-\delta)} \leq 
\frac{\delta } { \gamma(1-\delta)}.
\end{equation}
Consequently, by a union bound over all elements in $\mathcal{L}$, the probability that there exists a polynomial 
$h(T) \in \mathcal{L}$ other than $f|_\lambda$ that agrees with $f|_\lambda$ on  $\bb$ on all derivatives of order smaller than $s$ is at most $\frac{\delta} {\gamma s}$.
By (\ref{eq:chain_2}), this implies in turn that with probability at least $1 - \frac \delta {\gamma s}$ over the choice of $\ba, \bb \in \F_q^m$, it holds that 
$h^{(<s)}(1) \neq  f^{(<s)}(\bb)|_{\bb - \ba}$.
 \end{proof}

By the above Claims \ref{clm:rm_list_dec_johnson_1}, \ref{clm:rm_list_dec_johnson_2}, and \ref{clm:rm_list_dec_johnson_3}, we conclude  that with probability at least $1 - \frac{\delta} {\gamma s} - \frac 3 {\xi^2 q}$ over the choice of  $\ba$, $\bb$, and $\bu_0$, 
it holds that $f^{(<s)}(\bb) \in V$, and that during the execution of $A_{\bb, f^{(<s)}(\bb)}$, it holds that  $f|_{\lambda_{\ba,\bb}} \in \mathcal{L}$ and
there is no polynomial $h(T) \in \mathcal{L}$ other than
  $f|_{\lambda_{\ba,\bb}}$ so that $h^{(<s)}(1) = f^{(<s)}(\bb)|_{\bb - \ba}$, in which case the local algorithm $A_{\bb,  f^{(<s)}(\bb)}$ will set $g =f|_{\lambda_{\ba,\bb}}$ in Step \ref{step:rm_list_johnson_candidate}, and will output  
 $f|_{\lambda_{\ba,\bb}}(0) = f(\ba)$. By Markov's inequality, this implies in turn that with probability at least $\frac 1 2$ over the choice of $\bb$ and $\bu_0$, the local algorithm $A_{\bb, f^{(<s)}(\bb)}$ will be output by $A$, and will output $f(\ba)$ for at least a $(1-\epsilon)$-fraction of the points $\ba$ for $\epsilon=\frac {2\delta} {\gamma s}+ \frac{6} {\xi^2 q}$.

  \paragraph{Query complexity.} The query complexity of the approximate local list decoding algorithm $A$ is $r^m \cdot q$ since it queries $r^m$ lines, where each line has $q$ points, while the query complexity of each of the local algorithms $A_{\bb, v}$ is the number $q$ of points on the line.

\paragraph{List size.}
  It remains to bound the number of local algorithms output by $A$, which amounts to bounding the size of the set $V$. To this end, fix $\bb, \bu_0 \in \F_q^m$, and recall that for any $v\in V$, we have that
 $ v|_{\bu_0+\bu}= \left(\sum_{\bi : \wt(\bi) = i} v_\bi \cdot  (\bu_0+\bu)^{\bi} \right)_{i=0,1,\ldots, s-1} \in S_{\bu}$ for at least $\frac 1 2$ of the points $\bu \in  U^m$, where $S_\bu = \left\{ g^{(<s)}(0) \mid g(T) \in \mathcal{L}_\bu \right\} \subseteq \F_q^s$. 
For $\bu \in U^m$ and $i \in \{0,1,\ldots, s-1\}$,  let $S_{\bu,i} := \left\{ g^{(i)}(0) \mid g(T) \in \mathcal{L}_\bu \right\} \subseteq \F_q$, and let $V_i$ be the set containing 
all vectors $v \in \F_q^{\{\bi \mid \wt(\bi)=i\}}$ so that $\sum_{\bi : \wt(\bi) = i} v_\bi \cdot (\bu_0+\bu)^{\bi} \in S_{\bu, i}$ for at least $\frac 1 2$ of the points $\bu \in U^m$. Note that by (\ref{eq:rm_list_dec_johnson_2}) and our choice of $r=\frac{8 \delta s} {\gamma(1-\delta)}$, we have that 
$$|S_{\bu,i}| \leq |S_{\bu}| \leq |\mathcal{L}_\bu| \leq \frac{\delta} {(1-\delta) \gamma} = \frac{r} {8s}.$$ 
We shall show that for any $i \in \{0,1,\ldots, s-1\}$, $V_i$ has size at most $\frac r s$, 
and so the size of $V$ is at most $(\frac rs)^s $.

To see the above, fix $i \in \{0,1,\ldots, s-1\}$, and let 
$$\mathcal{L}_i:=\left\{ h \in (\F_q)_{<s}[X_1, \ldots , X_m] \mid h(\bu_0+\bu) \in S_{\bu,i} \;\; \text{for at least half of the points $\bu \in U^m$}\right\}.$$
Next we observe that $|V_i| \leq |\mathcal{L}_i|$. To see this, map each $v \in V_i$ to a polynomial $h_v(\bX):=\sum_{\bi : \wt(\bi)=i} v_\bi \bX^{\bi}$, and note that $h_v \in \mathcal{L}_i$ and $h_v \neq h_{v'}$ for any distinct $v,v'\in V_i$. 
Thus, it suffices to bound the size of 
$\mathcal{L}_i$.

The problem of bounding the size of $\mathcal{L}_i$ can be cast as follows. 
For a prime power $q$, positive integers $k,m$, and a subset $U \subseteq \F_q$ of size $r$, let 
$\RM_{q,m}(U;k)$ be the code over $\F_q$ of block length $r^m$ which associates with 
 any polynomial $f(X_1, \ldots, X_m) \in (\F_q)_{<k}[X_1, \ldots, X_m]$ a codeword $(f(\bu))_{\bu \in U^m} \in \F_q^{r^m}$. Note that the Reed-Muller Code $\RM_{q,m}(k)$ corresponds to the special case in which $U = \F_q$. It is well-known that $\RM_{q,m}(U;k)$ has relative distance at least $1 - \frac{k} {r}$ \cite{Schwartz80, Zippel79,DL78}. 
 
 In the above terminology, we 
 are  given for each codeword entry $\bu \in U^m$ in the code $\RM_{q,m}(U;s)$ (of relative distance $\delta':=1- \frac{s} {r}$) a small \emph{input list} $S_{\bu,i}$ of possible values in $\F_q$ of size at most $\ell:=\frac{r} {8s}$, and the goal is to bound the number codewords $c\in \RM_{q,m}(U;k)$ so that  $c_\bu \notin S_{\bu,i}$ for at most $\alpha':=\frac 1 2$ of the codeword entries  $\bu$. This is a well-known problem called \emph{list recovery}, which extends the problem of list decoding that corresponds to the case where all input lists $S_{\bu}$ have size one. In particular, an extension of the Johnson Bound to the setting of list-recovery (see e.g., \cite[Lemma 5.2]{GKORS17}) says that the number of codewords is at most 
 $$\frac{\delta' \ell} {(1-\alpha')^2 - \ell(1-\delta')} \leq \frac{\ell} {\frac 14 - \frac {\ell s} r} = \frac r s.$$

\subsection{Local list decoding of multivariate multiplicity codes up to minimum distance}\label{subsec:mult_mult_list}

In the previous section we showed how to locally list decode Reed-Muller Codes  up to the \emph{Johnson Bound}. The main barrier that prevented us from bypassing the Johnson Bound was that the local list decoding algorithm for Reed-Muller Codes   used as a subroutine the algorithm for list decoding of Reed-Solomon Codes from roughly the same decoding radius, and we currently do not know how to list decode Reed-Solomon Codes beyond the Johnson Bound with small list size.

However, for (univariate) multiplicity codes of sufficiently large multiplicity $s$, 
the results of Sections \ref{subsubsec:mult_mult}
and \ref{subsec:const_list_mult}
imply that they are list decodable up to their \emph{minimum distance} with constant list size, and  thus it is reasonable to expect that we could also locally list decode \emph{multivariate} multiplicity codes of sufficiently large multiplicity $s$ up to their minimum distance. Indeed, an algorithm very similar to that of Figure \ref{fig:rm_list_johnson} can be used to locally list decode multivariate multiplicity codes up to their minimum distance. Since the local list decoding algorithm for multivariate multiplicity codes is very similar to that of Reed-Muller Codes, but significantly more notationally heavy, we do not provide here a formal description of the algorithm for multivariate multiplicity codes, but instead just point out the main changes that need to be done in the algorithm of Figure \ref{fig:rm_list_johnson} to adapt it to the setting of multivariate multiplicity codes. We refer the reader to \cite[Section 4]{KRSW23} for a formal description of the algorithm. 

First, we note that in the multivariate multiplicity code setting, in order to list decode the restriction $w|_\lambda$ of $w$ to some line $\lambda$, we need to also receive oracle access to the \emph{directional derivatives} on the line, while $w$ only provides oracle access to the \emph{global derivatives}. However, 
we can use (\ref{eq:chain_2}) to compute the directional derivative of $w$ at any point on the line from the global derivatives of $w$ at this point.

Second, in the setting of multivariate multiplicity codes, the local algorithm $A_{\bb, v}$ needs to output the \emph{global derivatives} $f^{(<s)}(\ba)$ of $f$ at $\ba$, instead of just the value $f(\ba)$. Using the polynomial $g(T)$ determined in Step \ref{step:rm_list_johnson_candidate} of the algorithm $A_{\bb,v}$ in Figure \ref{fig:rm_list_johnson}, we can recover with high probability the \emph{directional derivatives} 
of $f$ at $\ba$ in the direction of the line $\lambda_{\ba,\bb}$, however these do not generally determine the global derivatives. To cope with this, we can first use exactly the same procedure that was used by the approximate local list decoding algorithm $A$ to obtain a small number of possible guesses $V'$ for the value of $f^{(<s)}(\ba)$, and then instead of outputting $g(0)$ in Step \ref{step:rm_list_johnson_candidate}, output the unique $v' \in V'$ such that $v'|_{\bb - \ba} = g^{(<s)}(0)$ as a guess for $f^{(<s)}(\ba)$, if such a $v'$ exists. 

To show that the above solution works, we first recall that by Claim \ref{clm:rm_list_dec_johnson_1}, we have that $v_0:=f^{(<s)}(\ba) \in V'$ with high probability. So it remains to show that with high probability, for any $v' \in V'$ other than $v_0$ it holds that $v'|_{\bb - \ba} \neq v_0|_{\bb - \ba}$. To see this, fix $v' \in V'$ other than $v_0$, and recall that $v'|_{\bb - \ba} = \left(\sum_{\bi : \wt(\bi) = i} v'_\bi \cdot  (\bb-\ba)^{\bi} \right)_{i=0,1,\ldots, s-1}$ and $v_0|_{\bb - \ba} = \left(\sum_{\bi : \wt(\bi) = i} (v_0)_\bi \cdot  (\bb-\ba)^{\bi} \right)_{i=0,1,\ldots, s-1}$. Since $v' \neq v_0$, there exists some $i \in \{0,1,\ldots, s-1\}$ so that the polynomials $h_{v'}(\bX):=\sum_{\bi : \wt(\bi) = i} v'_\bi \cdot  \bX^{\bi}$ and  $h_{v_0}(\bX):=\sum_{\bi : \wt(\bi) = i} (v_0)_\bi \cdot  \bX^{\bi}$ are distinct. But since both polynomials are of degree at most $s$, this implies in turn that with probability at least $1 - \frac s q$ over the choice of $\ba$ and $\bb$ it holds that $h_{v'}(\bb-\ba) \neq h_{v_0}(\bb-\ba)$, in which case we shall also have that $v'|_{\bb - \ba} \neq v_0|_{\bb - \ba}$. By applying a union bound over all elements of $V'$, we conclude that with probability at least
$1 - \frac {s|V'|} {q}$ there is no $v' \in V'$ other than $v_0=f^{(<s)}(\ba)$ so that 
$v'|_{\bb - \ba} = v_0|_{\bb - \ba}$.

\paragraph{Local list decoding up to capacity.} 
Above we  outlined an algorithm which can local list decode multivariate multiplicity codes of sufficiently large multiplicity $s$ up to their minimum relative distance of $\delta:= 1 - \frac{k} {s q}$. Recall however that the rate of $m$-variate multiplicity codes of multiplicity $s$ is 
$$\frac{ {m+k-1 \choose m} } { {m+s-1 \choose m} \cdot q^m} 
\leq \frac{\frac 1{m!} \cdot (m+k)^m } {\frac 1{m!}\cdot  s^m \cdot q^m}
= \left(\frac{k+m} {k}\right)^m \cdot \left( \frac{k} {sq} \right)^m = \left(1+ \frac m k\right)^m \cdot (1-\delta)^m,$$ and consequently the decoding radius achieved by the above algorithm is well below list-decoding capacity for fixed $m >1$ and growing $k$. Nevertheless, this algorithm can be used to obtain another family of codes that is locally list decodable up to capacity.

In more detail, recall that
in Section \ref{subsec:rm_list_johnson} we encountered the notion of \emph{list recovery}, which extends the notion of list decoding. Recall that in list recovery, one is given for each codeword entry a small input list of possible values, and the goal is to return the list of all codewords that are consistent with most of these input lists. More formally, for $\alpha \in (0,1)$ and positive integers $\ell, L$ we say that a code $C \subseteq \Sigma^n$ is \textsf{$(\alpha,\ell, L)$-list recoverable} if for any subsets $S_1, \ldots, S_n \subseteq \Sigma$ of size at most $\ell$ each there are at most $L$ different codewords $c \in C$ which satisfy that $w_i \notin S_i$ for at most an $\alpha$-fraction of the indices $i \in [n]$. Note that list decoding corresponds to the special case of $\ell=1$. The \textsf{problem of list-recovering from an $\alpha$-fraction
of errors} is the task of finding the list of these codewords.

We have already seen in Section \ref{subsec:rm_list_johnson} that the notion of list recovery could be useful for local list decoding. Another motivation to consider the notion of list recovery is that it is known that  \emph{high rate} codes (of rate close to $1$) that are list-recoverable from a constant fraction of errors and constant-size input lists $S_i$ 
can be transformed into capacity-achieving list-decodable codes, and furthermore, this transformation also preserves locality (the transformation is based on expander graphs, and is beyond the scope of the current survey, see \cite[Appendix C]{KRSW23} for more details). 

As is typically the case, all the local list decoding algorithms mentioned above could be extended to the setting of list recovery with similar performance. Furthermore, another advantage of multivariate multiplicity codes over Reed-Muller Codes  is that they can achieve a higher rate. Specifically, while the rate  of $m$-variate multiplicity codes of multiplicity $s$  of relative distance $\delta:=1-\frac{k} {sq}$ is at least  
$$\frac{ {m+k-1 \choose m} } { {m+s-1 \choose m} \cdot q^m} 
\geq \frac{\frac 1{m!} \cdot k^m } {\frac 1{m!}\cdot  (s+m)^m \cdot q^m}
= \left(\frac{s} {s+m}\right)^m \cdot \left( \frac{k} {sq} \right)^m = \left(1 - \frac m {s+m}\right)^m \cdot (1-\delta)^m,$$
the rate of $m$-variate Reed-Muller Codes  of relative distance $\delta':=1 - \frac kq$ is only
$$\frac{ {m+k-1 \choose m} } { q^m} \leq \frac{\frac 1{m!} \cdot (m+k)^m } { q^m}
= \frac 1{m!} \cdot \left(\frac{k+m} {k}\right)^m \cdot \left( \frac{k} {q} \right)^m = \frac 1{m!} \cdot\left(1+ \frac m k\right)^m \cdot (1-\delta')^m,$$ 
In particular, this means that for any $m>1$, Reed-Muller Codes  have rate at most $\frac 1 2$, while multivariate multiplicity codes can have rate close to $1$ for $\delta \approx \frac 1 m$,  $s \approx m^2$, and growing $m$. Consequently, one can use the extension of the local list decoding algorithm for multivariate multiplicity codes to the setting of list recovery to get high-rate locally list recoverable codes, which can be transformed in turn into capacity-achieving locally list decodable codes. Once more, we refer the reader to \cite[Appendix C]{KRSW23} for more details. 

\subsection{Bibliographic notes}

\textsf{Locally decodable and locally correctable codes} were introduced in \cite{BFLS91, KT00, STV01} and \cite{Lip90, BK95}, respectively, see \cite{Yekhanin_survey} for a comprehensive survey on locally decodable and locally correctable codes. \textsf{Locally list decodable codes} were introduced in \cite{GL89, STV01}, and since their introduction they found various applications in theoretical computer science, in cryptography~\cite{GL89}, learning theory~\cite{KM93}, average-to-worst-case reductions \cite{Lip90}, and hardness amplification~\cite{BFNW93, STV01}.

The \emph{local correction} algorithm for Reed-Muller Codes , presented in Section \ref{subsec:rm_unique}, was suggested in \cite{Lip90, BF90, GLRSW91} (see also \cite[Section 2.2.2]{Yekhanin_survey}), while the improvement using curves, described in Remark \ref{rem:rm_lcc}, was suggested by Gemmell and Sudan \cite{GS92} (see also \cite[Section 2.2.3]{Yekhanin_survey}). The \emph{local list decoding} algorithm for Reed-Muller Codes  beyond the unique decoding radius, presented in Section \ref{subsec:rm_list}, was suggested by Arora and Sudan \cite{AS03} and Sudan, Trevisan, and Vadhan \cite{STV01}. The local list decoding algrotihm for Reed-Muller Codes  \emph{up to the Johnson Bound}, presented in Section \ref{subsec:rm_list_johnson}, follows the ideas suggested by Kopparty, Ron-Zewi, Saraf, and Wootters \cite{KRSW23} for local list decoding of multivariate multiplicity codes up to their minimum distance. Here we present a simpler version of their algorithm for local list decoding of Reed-Muller Codes  up to the Johnson Bound. Alternative approaches for local list decoding of Reed-Muller Codes  up to the Johnson Bound, based on restriction to planes instead of lines, were suggested by Brander and Kopparty \cite{BK09} (see also \cite{GK16_lifted, Kop15}). 

A \emph{local correction algorithm} for \emph{multivariate multiplicity codes} was suggested by Kopparty, Saraf, and Yekhanin \cite{KSY14}, who also introduced these codes and observed that they can achieve higher rate than the corresponding Reed-Muller Codes. The \emph{local list decoding} algorithm for multivariate multiplicity codes up to their minimum distance, outlined in Section \ref{subsec:mult_mult_list}, was suggested in \cite{KRSW23}. 

The \emph{expander-based transformation} that turns a high-rate locally list recoverable code into a capacity-achieving locally list decodable code, mentioned in Section \ref{subsec:mult_mult_list}, was suggested by Gopi, Kopparty, Oliveira, Ron-Zewi, and Saraf \cite{GKORS17} and Hemenway, Ron-Zewi, and Wootters \cite{HRW17}. It is based on the expander-based transformation that was suggested by Alon, Edmonds, and Luby in the unique decoding setting \cite{AEL95, AL96_codes}, and its extensions to list decoding by Guruswami and Indyk \cite{GI02}, and local decoding by Kopparty, Meir, Ron-Zewi, and Saraf \cite{KMRS}.  More recently, this transformation was used to obtain expander-based constructions of capacity-achieving (non-local) list decodable codes (not relying on polynomial-based codes) \cite{st25}, and even ones achieving the generalized Singleton Bound \cite{JMST25}.

In this section we discussed \emph{local} list decoding algorithms for multivariate codes. Such algorithms are inherently randomized (since an adversary can easily corrupt the answers to all of the queries of a deterministic algorithm), and only work when the code is evaluated over all of $\F_q^m$ (because of the need to pass a random line -- or a related algebraic low-dimensional structure -- through the domain). Over $\F_q^m$, Pellikaan and Wu \cite{PWu04} gave an efficient \emph{deterministic} (global) list decoding algorithm for Reed-Muller Codes up to the Johnson Bound, and this algorithm was extended by Kopparty \cite{Kop15} to list decoding of multivariate multiplicity codes of arbitrary multiplicity up to the Johnson Bound.  
More recently, Bhandari, Harsha, Kumar, and Sudan \cite{BHKS2024_grid} gave an efficient deterministic algorithm for list decoding  multivariate multiplicity codes, over a constant number of variables and of sufficiently large multiplicity, \emph{up to their minimum distance}. 

Interestingly, the bound on the number of roots with multiplicities given by Theorem \ref{thm:prelim_hasse_sz} also applies when the domain is a \emph{product set} of the form $A_1 \times \cdots \times A_m$, where $A_1, \ldots, A_m$ are arbitrary subsets of $\F_q$. 
 However, over arbitrary product sets, Reed-Muller Codes were only shown relatively recently to be  efficiently uniquely decodable up to half their minimum distance  by Kim and Kopparty \cite{KimK2017}, and this algorithm was extended by Bhandari, Harsha, Kumar, and Shanka \cite{BhandariHKS2023} to unique decoding of multivariate multiplicity codes of arbitrary multiplicity over arbitrary product sets.  We are not aware of any efficient algorithm that list decodes Reed-Muller Codes, or multivariate multiplicity codes of arbitrary multiplicity, over arbitrary product sets beyond half the minimum distance. 
 Curiously, the aforementioned algorithm of  \cite{BHKS2024_grid} for list decoding constant-variate multiplicity codes of sufficiently large multiplicity up to their minimum distance does work over arbitrary product sets.

\section{Conclusion and open problems}\label{sec:conclusion}

In this book, we surveyed recent advances on list decoding of polynomial codes. Specifically, we investigated the list-decoding properties of Reed-Solomon Codes, showed that related families of polynomial codes such as multiplicity codes are efficiently list decodable up to capacity, and discussed how to obtain improved combinatorial upper bounds on the list sizes using a more refined analysis.
We also presented fast (near-linear time) implementations of these list-decoding algorithms, and local (sublinear-time) list-decoding algorithms for multivariate polynomial codes such as Reed-muller Codes and multivariate multiplicity codes. 

\paragraph{Open questions.} We end this book with a couple of intriguing questions for further research that are still open by the time of writing of this survey. 

\begin{enumerate}

\item \textbf{Explicit and efficiently list decodable Reed-Solomon Codes achieving list-deocding capacity.} In Section \ref{subsec:rs_limit}, we showed that Reed-Solomon Codes are generally \emph{not} list decodable much beyond the Johnson bound, while in Section \ref{subsec:random_rs} we showed that Reed-Solomon Codes over \emph{random evaluation points} are (combinatorially) list decodable up to capacity with optimal list size, with high probability. A very interesting question is to find 
explicit evaluation points for which Reed-Solomon Codes are list decodable up to capacity, as well as efficient list decoding algorithms.

\item  \textbf{Explicit and efficiently list decodable codes achieving list-decoding capacity over a fixed-size alphabet. }
In Section \ref{subsec:list_bio}, we mentioned that one can obtain capacity-achieving list-decodable codes over a constant-size alphabet, whose size only depends on the gap $\epsilon$ to capacity, by resorting to Algebraic-Geometric (AG) codes. A major open problem is to obtain explicit constructions of capacity-achieving list decodable codes, as well as efficient list-decoding algorithms, over a \emph{fixed-size} alphabet (independent of $\epsilon$), for example over the \emph{binary} alphabet. Over a $q$-ary alphabet, the list-decoding capacity is known to be $H_q^{-1}(1-R)$, where $$
H_q(\alpha)=\alpha\log_q(q-1)+\alpha\log_q\left(\frac 1 \alpha\right) +(1-\alpha)\log_q\left(\frac 1 {1-\alpha}\right)
$$
is the \textsf{$q$-ary entropy function}. 

\item \textbf{Linear-time encodable and list decodable codes achieving list-decoding capacity.}
In Section \ref{sec:linear} we presented fast near-linear time algorithms for list deocding multiplicity codes up to capacity. Obtaining \emph{truly linear-time} list-decoding algorithms (as well as truly linear-time encoding algorithms for capacity-acheiving list-decodable codes) is still open. Since in the unique decoding setting,  expander codes can be used to obtain linear-time encodable and decodable codes, it could be that the recent line of work on list-decoding of expander codes  \cite{st25, JMST25} could be used to obtain such codes.

\item \textbf{Applications in theoretical computer science.}
We mentioned in Section \ref{sec:intro} that list decodable codes have many applications in theoretical computer science (some references are given in Section \ref{subsec:intro_scope}). It would be very interesting to investigate if any of the developments we surveyed in this book could be useful for such applicaitons.

\end{enumerate}

\paragraph{Acknowledgement.} 

%Noga Ron-Zewi
We would like to thank Zeyu Guo, Prahladh Harsha, Swastik Koppparty, Ramprasad Saptharishi, Shubhangi Saraf,  Madhu Sudan, Amnon Ta-Shma, S. Venkitesh, and Mary Wootters for many enlightening discussions over the years about list decoding of polynomial codes  which inspired and influenced the writing of this survey.

Noga Ron-Zewi was partially suppported by the European Union (ERC, ECCC, 101076663) while writing this survey. Views and opinions expressed are however those of the author(s) only and do not necessarily reflect those of the European Union or the European Research Council. Neither the European Union nor the granting authority can be held responsible for them. 

Mrinal Kumar's work is supported by the Department of Atomic Energy, Government  of India, under project no. RTI4014, and partially by grants from ANRF, Google Research and Premji Invest. 

\bibliographystyle{alpha} 
\bibliography{ref}

@misc{Sudan-survey-2025,
      title={Algebra in Algorithmic Coding Theory}, 
      author={Madhu Sudan},
      year={2025},
      eprint={2512.06478},
      archivePrefix={arXiv},
      primaryClass={cs.IT},
      url={https://arxiv.org/abs/2512.06478}, 
}

@inproceedings{AHS26,
  author       = {Vikrant Ashvinkumar and Mursalin Habib and Shashank Srivastava},
  title        = {Algorithmic Improvements to List Decoding of Folded Reed-Solomon Codes},
  booktitle    = {Proceedings of the 2026 {ACM-SIAM} Symposium on Discrete Algorithms,
                  {SODA}},
  publisher    = {{SIAM}},
  year         = {2026},
}

@article{Trevisan_survey,
  author       = {Luca Trevisan},
  title        = {Some Applications of Coding Theory in Computational Complexity},
  journal      = {Electron. Colloquium Comput. Complex.},
  volume       = {{TR04-043}},
  year         = {2004},
  url          = {https://eccc.weizmann.ac.il/eccc-reports/2004/TR04-043/index.html},
  eprinttype    = {ECCC},
  eprint       = {TR04-043},
  bibsource    = {dblp computer science bibliography, https://dblp.org}
}

@article{Sudan_list_dec_survey,
author = {Sudan, Madhu},
title = {List decoding: algorithms and applications},
year = {2000},
issue_date = {March 2000},
publisher = {Association for Computing Machinery},
address = {New York, NY, USA},
volume = {31},
number = {1},
issn = {0163-5700},
url = {https://doi.org/10.1145/346048.346049},
doi = {10.1145/346048.346049},
journal = {SIGACT News},
month = mar,
pages = {16–27},
numpages = {12}
}

@INPROCEEDINGS{Guruswami_list_dec_app_survey,
  author={Guruswami, V.},
  booktitle={2006 IEEE Information Theory Workshop - ITW '06 Punta del Este}, 
  title={List Decoding in Average-Case Complexity and Pseudorandomness}, 
  year={2006},
  volume={},
  number={},
  pages={32-36},
  doi={10.1109/ITW.2006.1633776}}

@article{Gop82,
    author = {V.~D.~Goppa},
    title = {Algebraico-geometric codes},
    journal = {Math. USSR-Izv.},
    year = {1983},
volume = {21},
issue = {1},
pages = {75--91},
}

@article{hamming50,
	Author = {Hamming, Richard},
	Journal = {The Bell System Technical Journal},
	Number = {2},
	Pages = {147--160},
	Title = {Error detecting and error correcting codes},
	Volume = {29},
	Year = {1950}}

@article{shannon48,
	Author = {Shannon, Claude},
	Journal = {The Bell System Technical Journal},
	Pages = {379–-423, 623–-656},
	Title = {A mathematical theory of communication},
	Volume = {27},
year={1948}}

@ARTICLE{berlekamp-bch,
  author={Berlekamp, E.},
  journal={IEEE Transactions on Information Theory}, 
  title={Nonbinary BCH decoding (Abstr.)}, 
  year={1968},
  volume={14},
  number={2},
  pages={242-242},
  doi={10.1109/TIT.1968.1054109}}

@article{Kaltofen1985,
  author  = {Erich L. Kaltofen},
  journal = {SIAM J. Comput.},
  number  = {2},
  pages   = {469--489},
  title   = {Polynomial-Time Reductions from Multivariate to Bi-
             and Univariate Integral Polynomial Factorization},
  volume  = {14},
  year    = {1985},
  url     = {https://doi.org/10.1137/0214035}
}

@article{Lenstra1985,
  author  = {Arjen K. Lenstra},
  journal = {J. Comput.\ Syst.\ Sci.},
  note    = {(Preliminary version in \emph{15th STOC}, 1983)},
  number  = {2},
  pages   = {235--248},
  title   = {Factoring Multivariate Polynomials over Finite
             Fields},
  volume  = {30},
  year    = {1985},
  url     = {https://doi.org/10.1016/0022-0000(85)90016-9}
}

@article{BGM22,
  author       = {Joshua Brakensiek and
                  Sivakanth Gopi and
                  Visu Makam},
  title        = {Lower Bounds for Maximally Recoverable Tensor Codes and Higher Order
                  {MDS} Codes},
  journal      = {{IEEE} Trans. Inf. Theory},
  volume       = {68},
  number       = {11},
  pages        = {7125--7140},
  year         = {2022},
  url          = {https://doi.org/10.1109/TIT.2022.3187366},
  doi          = {10.1109/TIT.2022.3187366},
  bibsource    = {dblp computer science bibliography, https://dblp.org}
}

@article{BGM23,
author = {Brakensiek, Joshua and Gopi, Sivakanth and Makam, Visu},
title = {Generic Reed–Solomon Codes Achieve List-Decoding Capacity},
journal = {SIAM Journal on Computing},
year = {2025},
pages = {STOC23-118-STOC23-154},
doi = {10.1137/23M1598064},
Note = {To appear, \url{https://doi.org/10.1137/23M1598064}},
eprint = { https://doi.org/10.1137/23M1598064}
}

@INPROCEEDINGS{DSY14,
	author={Dau, Son Hoang and Song, Wentu and Yuen, Chau},
	booktitle={2014 IEEE International Symposium on Information Theory}, 
	title={{On the existence of MDS codes over small fields with constrained generator matrices}}, 
	year={2014},
	volume={},
	number={},
	pages={1787-1791},
	doi={10.1109/ISIT.2014.6875141},
	note = {\url{https://doi.org/10.1109/ISIT.2014.6875141}}
}

@article{AGGLZ25,
   title={Random Reed-Solomon codes achieve list-decoding capacity with linear-sized alphabets},
   ISSN={2517-5599},
   note={\url{http://dx.doi.org/10.19086/aic.2025.8}},
   DOI={10.19086/aic.2025.8},
   journal={Advances in Combinatorics},
   publisher={Alliance of Diamond Open Access Journals},
   author={Alrabiah, Omar and Guo, Zeyu and Guruswami, Venkatesan and Li, Ray and
           Zhang, Zihan},
   year={2025}}

@ARTICLE{YH19,
	author={Yildiz, Hikmet and Hassibi, Babak},
	journal={IEEE Transactions on Information Theory}, 
	title={{Optimum Linear Codes With Support-Constrained Generator Matrices Over Small Fields}}, 
	year={2019},
	volume={65},
	number={12},
	pages={7868-7875},
	doi={10.1109/TIT.2019.2932663},
	note = {\url{https://doi.org/10.1109/TIT.2019.2932663}}
}

@article{lovett21,
	author = {Lovett, Shachar},
	title = {{Sparse MDS Matrices over Small Fields: A Proof of the GM-MDS Conjecture}},
	journal = {SIAM Journal on Computing},
	volume = {50},
	number = {4},
	pages = {1248-1262},
	year = {2021},
	doi = {10.1137/20M1323345},
	URL = { 
	https://doi.org/10.1137/20M1323345
	},
	eprint = { 
	https://doi.org/10.1137/20M1323345	}
	,
	note = {\url{https://doi.org/10.1137/20M1323345}}
}

@article{RV_survey,
  title={List Recoverable Codes: The Good, the Bad, and the Unknown (hopefully not Ugly)},
    author={Nicolas Resch and S. Venkitesh},
  journal={arXiv preprint 2510.07597},
  year={2025},
    note={\url{https://arxiv.org/abs/2510.07597}}, 
}

@article{BKR10,
  author       = {Eli Ben{-}Sasson and
                  Swastik Kopparty and
                  Jaikumar Radhakrishnan},
  title        = {Subspace polynomials and limits to list decoding of Reed-Solomon codes},
  journal      = {{IEEE} Trans. Inf. Theory},
  volume       = {56},
  number       = {1},
  pages        = {113--120},
  year         = {2010},
  url          = {https://doi.org/10.1109/TIT.2009.2034780},
  doi          = {10.1109/TIT.2009.2034780},
  bibsource    = {dblp computer science bibliography, https://dblp.org}
}

@inproceedings{ST20,
  author       = {Chong Shangguan and
                  Itzhak Tamo},
  editor       = {Konstantin Makarychev and
                  Yury Makarychev and
                  Madhur Tulsiani and
                  Gautam Kamath and
                  Julia Chuzhoy},
  title        = {Combinatorial list-decoding of Reed-Solomon codes beyond the Johnson
                  radius},
  booktitle    = {Proceedings of the 52nd Annual {ACM} {SIGACT} Symposium on Theory
                  of Computing, {STOC} 2020, Chicago, IL, USA, June 22-26, 2020},
  pages        = {538--551},
  publisher    = {{ACM}},
  year         = {2020},
  url          = {https://doi.org/10.1145/3357713.3384295},
  doi          = {10.1145/3357713.3384295},
  bibsource    = {dblp computer science bibliography, https://dblp.org}
}

@article{Roth22,
  author       = {Ron M. Roth},
  title        = {Higher-Order {MDS} Codes},
  journal      = {{IEEE} Trans. Inf. Theory},
  volume       = {68},
  number       = {12},
  pages        = {7798--7816},
  year         = {2022},
  url          = {https://doi.org/10.1109/TIT.2022.3194521},
  doi          = {10.1109/TIT.2022.3194521},
  bibsource    = {dblp computer science bibliography, https://dblp.org}
}

@article{GST24,
  author       = {Eitan Goldberg and
                  Chong Shangguan and
                  Itzhak Tamo},
  title        = {Singleton-type bounds for list-decoding and list-recovery, and related
                  results},
  journal      = {J. Comb. Theory {A}},
  volume       = {203},
  pages        = {105835},
  year         = {2024},
  url          = {https://doi.org/10.1016/j.jcta.2023.105835},
  doi          = {10.1016/J.JCTA.2023.105835},
  bibsource    = {dblp computer science bibliography, https://dblp.org}
}

@article{Yekhanin_survey,
  author       = {Sergey Yekhanin},
  title        = {Locally Decodable Codes},
  journal      = {Found. Trends Theor. Comput. Sci.},
  volume       = {6},
  number       = {3},
  pages        = {139--255},
  year         = {2012},
  url          = {https://doi.org/10.1561/0400000030},
  doi          = {10.1561/0400000030},
  bibsource    = {dblp computer science bibliography, https://dblp.org}
}

@article{KSY14,
  author       = {Swastik Kopparty and
                  Shubhangi Saraf and
                  Sergey Yekhanin},
  title        = {High-rate codes with sublinear-time decoding},
  journal      = {J. {ACM}},
  volume       = {61},
  number       = {5},
  pages        = {28:1--28:20},
  year         = {2014},
  url          = {https://doi.org/10.1145/2629416},
  doi          = {10.1145/2629416},
  bibsource    = {dblp computer science bibliography, https://dblp.org}
}

@article{RS60,
	Author = {Reed, Irving S. and Solomon, Gustave},
	Journal = {SIAM Journal of the Society for Industrial and Applied Mathematics},
	Number = {2},
	Pages = {300--304},
	Title = {Polynomial Codes over Certain Finite Fields},
	Volume = {8},
	Year = {1960},
	note = {\url{https://doi.org/10.1137/0108018}}
}

@article{RT97,
  title={Codes for the m-metric},
  author={Rosenbloom, Michael and Tsfasman, Michael},
  journal={Problemy Peredachi Informatsii},
  volume={33},
  number={1},
  pages={55--63},
  year={1997},
  publisher={Russian Academy of Sciences, Branch of Informatics, Computer Equipment and Automatization}
}

@ARTICLE{BHKS2024_grid,
  author={Bhandari, Siddharth and Harsha, Prahladh and Kumar, Mrinal and Sudan, Madhu},
  journal={IEEE Transactions on Information Theory}, 
  title={Decoding Multivariate Multiplicity Codes on Product Sets}, 
  year={2024},
  volume={70},
  number={1},
  pages={154-169},
  doi={10.1109/TIT.2023.3306849}
}

@Article{GKORS17,
  title =	"Locally Testable and Locally Correctable Codes
		 approaching the Gilbert-Varshamov Bound",
  author =	"Sivakanth Gopi and Swastik Kopparty and Rafael Oliveira and Noga Ron-Zewi and Shubhangi Saraf",
  journal =	"IEEE Transactions on Information Theory",
  year = 	"2018",
  number =	"8",
  volume =	"64",
  bibdate =	"2018-07-17",
  bibsource =	"DBLP,
		 http://dblp.uni-trier.de/https://doi.org/10.1109/TIT.2018.2809788;
		 DBLP,
		 http://dblp.uni-trier.de/db/journals/tit/tit64.html#GopiKORS18",
  pages =	"5813--5831",
}

@Article{HRW17,
  title =	"Local List Recovery of High-Rate Tensor Codes and
		 Applications",
  author =	"Brett Hemenway and Noga {Ron-Zewi} and Mary Wootters",
  journal =	{SIAM Journal on Computing},
  year = 	"2020",
  number =	"4",
  volume =	"49",
  pages = "157-195",

}

@ARTICLE{PWu04,
  author={Pellikaan, R. and Xin-Wen Wu},
  journal={IEEE Transactions on Information Theory}, 
  title={List decoding of q-ary Reed-Muller codes}, 
  year={2004},
  volume={50},
  number={4},
  pages={679-682},
  doi={10.1109/TIT.2004.825043}
}

@Article{Mul54,
  title =	"Application of Boolean algebra to switching circuit
		 design and to error detection",
  author =	"David Muller",
  journal =	{Transactions of the IRE Professional Group on Electronic Computers},
  year = 	"1954",
  number =	"3",
  volume =	"3",
  bibdate =	"2020-05-20",
  bibsource =	"DBLP,
		 http://dblp.uni-trier.de/https://doi.org/10.1109/IREPGELC.1954.6499441;
		 DBLP,
		 http://dblp.uni-trier.de/https://www.wikidata.org/entity/Q57320281;
		 DBLP,
		 http://dblp.uni-trier.de/db/journals/tc/tc3.html#Muller54",
  pages =	"6--12",
}

@Article{Reed54,
  title =	"A class of multiple-error-correcting codes and the
		 decoding scheme",
  author =	"Irving Reed",
  journal =	{Transactions of the IRE Professional Group on Information Theory},
  year = 	"1954",
  volume =	"4",
  bibdate =	"2020-07-16",
  bibsource =	"DBLP,
		 http://dblp.uni-trier.de/https://doi.org/10.1109/TIT.1954.1057465;
		 DBLP,
		 http://dblp.uni-trier.de/https://www.wikidata.org/entity/Q57320282;
		 DBLP,
		 http://dblp.uni-trier.de/db/journals/tit/tit4.html#Reed54",
  pages =	"38--49",
}

@article{JMST25,
  author       = {Fernando Granha Jeronimo and
                  Tushant Mittal and
                  Shashank Srivastava and
                  Madhur Tulsiani},
  title        = {Explicit Codes approaching Generalized Singleton Bound using Expanders},
  journal      = {CoRR},
  volume       = {abs/2502.07308},
  year         = {2025},
  url          = {https://doi.org/10.48550/arXiv.2502.07308},
  doi          = {10.48550/ARXIV.2502.07308},
  eprinttype    = {arXiv},
  eprint       = {2502.07308},
  bibsource    = {dblp computer science bibliography, https://dblp.org}
}

@article{ST25,
  author       = {Shashank Srivastava and
                  Madhur Tulsiani},
  title        = {List Decoding Expander-Based Codes up to Capacity in Near-Linear Time},
  journal      = {CoRR},
  volume       = {abs/2504.20333},
  year         = {2025},
  url          = {https://doi.org/10.48550/arXiv.2504.20333},
  doi          = {10.48550/ARXIV.2504.20333},
  eprinttype    = {arXiv},
  eprint       = {2504.20333},
  bibsource    = {dblp computer science bibliography, https://dblp.org}
}

@inproceedings{GI02,
  author    = {Venkatesan Guruswami and
               Piotr Indyk},
  title     = {Near-optimal linear-time codes for unique decoding and new
               list-decodable codes over smaller alphabets},
  booktitle = {Proceedings of the 34th Annual {ACM} Symposium on
		 Theory of Computing (STOC)},
		   publisher =	"ACM Press",
  year      = {2002},
  pages     = {812-821},
}

@article{AL96_codes,
  author    = {Noga Alon and
               Michael Luby},
  title     = {A linear time erasure-resilient code with nearly optimal
               recovery},
  journal   = {IEEE Transactions on Information Theory},
  volume    = {42},
  number    = {6},
  year      = {1996},
  pages     = {1732-1736},
}

@InProceedings{AEL95,
  author =	"Noga Alon and Jeff Edmonds and Michael Luby",
  title =	"Linear Time Erasure Codes with Nearly Optimal
		 Recovery",
  pages =	"512--519",
  booktitle =	"proceedings of the 36th Annual IEEE Symposium on Foundations of Computer
		 Science ({FOCS})",
  ISBN = 	"0-8186-7183-1",
  publisher =	"IEEE Computer Society",
  year = 	"1995",
}

@InProceedings{LMS25,
      title={Random Reed-Solomon Codes and Random Linear Codes are Locally Equivalent}, 
      author={Matan Levi and Jonathan Mosheiff and Nikhil Shagrithaya},
      year={2025},
     booktitle =	"proceedings of the 66th Annual IEEE Symposium on Foundations of Computer
		 Science ({FOCS})",
  publisher =	"IEEE Computer Society"
}

@article{Schwartz80,
  author       = {Jacob T. Schwartz},
  title        = {Fast Probabilistic Algorithms for Verification of Polynomial Identities},
  journal      = {J. {ACM}},
  volume       = {27},
  number       = {4},
  pages        = {701--717},
  year         = {1980},
  url          = {https://doi.org/10.1145/322217.322225},
  doi          = {10.1145/322217.322225},
  bibsource    = {dblp computer science bibliography, https://dblp.org}
}

@inproceedings{Zippel79,
  author       = {Richard Zippel},
  editor       = {Edward W. Ng},
  title        = {Probabilistic algorithms for sparse polynomials},
  booktitle    = {Symbolic and Algebraic Computation, {EUROSAM} '79, An International
                  Symposiumon Symbolic and Algebraic Computation, Marseille, France,
                  June 1979, Proceedings},
  series       = {Lecture Notes in Computer Science},
  volume       = {72},
  pages        = {216--226},
  publisher    = {Springer},
  year         = {1979},
  url          = {https://doi.org/10.1007/3-540-09519-5\_73},
  doi          = {10.1007/3-540-09519-5\_73},
  bibsource    = {dblp computer science bibliography, https://dblp.org}
}

@article{DL78,
  author       = {Richard A. DeMillo and
                  Richard J. Lipton},
  title        = {A Probabilistic Remark on Algebraic Program Testing},
  journal      = {Inf. Process. Lett.},
  volume       = {7},
  number       = {4},
  pages        = {193--195},
  year         = {1978},
  url          = {https://doi.org/10.1016/0020-0190(78)90067-4},
  doi          = {10.1016/0020-0190(78)90067-4},
  bibsource    = {dblp computer science bibliography, https://dblp.org}
}

@article{BK95,
  author       = {Manuel Blum and
                  Sampath Kannan},
  title        = {Designing Programs that Check Their Work},
  journal      = {J. {ACM}},
  volume       = {42},
  number       = {1},
  pages        = {269--291},
  year         = {1995},
  url          = {https://doi.org/10.1145/200836.200880},
  doi          = {10.1145/200836.200880},
  bibsource    = {dblp computer science bibliography, https://dblp.org}
}

@inproceedings{GLRSW91,
  author       = {Peter Gemmell and
                  Richard J. Lipton and
                  Ronitt Rubinfeld and
                  Madhu Sudan and
                  Avi Wigderson},
  editor       = {Cris Koutsougeras and
                  Jeffrey Scott Vitter},
  title        = {Self-Testing/Correcting for Polynomials and for Approximate Functions},
  booktitle    = {Proceedings of the 23rd Annual {ACM} Symposium on Theory of Computing,
                  May 5-8, 1991, New Orleans, Louisiana, {USA}},
  pages        = {32--42},
  publisher    = {{ACM}},
  year         = {1991},
  url          = {https://doi.org/10.1145/103418.103429},
  doi          = {10.1145/103418.103429},
  bibsource    = {dblp computer science bibliography, https://dblp.org}
}

@inproceedings{Lip90,
	Author = {Richard J. Lipton},
	Booktitle = {proceedings of the 7th Annual ACM Symposium on Theoretical Aspects of Computer Science (STACS)},
	Pages = {207--215},
	Publisher = {Springer},
	Series = {lncs},
	Title = {Efficient Checking of Computations},
	Volume = {415},
	Year = {1990}}

@article{GK16_lifted,
  title={List-decoding algorithms for lifted codes},
  author={Guo, Alan and Kopparty, Swastik},
  journal={IEEE Transactions on Information Theory},
  volume={62},
  number={5},
  pages={2719--2725},
  year={2016},
  publisher={IEEE}
}

@article{BK09,
	Author = {Kristian Brander and Swastik Koparty},
	Journal = {Manuscript},
	Title = {List-decoding Reed-Muller over large fields upto the
Johnson radius},
	Year = {2009}}

@article{AS03,
  author       = {Sanjeev Arora and
                  Madhu Sudan},
  title        = {Improved Low-Degree Testing and its Applications},
  journal      = {Comb.},
  volume       = {23},
  number       = {3},
  pages        = {365--426},
  year         = {2003},
  url          = {https://doi.org/10.1007/s00493-003-0025-0},
  doi          = {10.1007/S00493-003-0025-0},
  bibsource    = {dblp computer science bibliography, https://dblp.org}
}

@inproceedings{BhandariHKS2023,
  author       = {Siddharth Bhandari and
                  Prahladh Harsha and
                  Mrinal Kumar and
                  Ashutosh Shankar},
  editor       = {Nikhil Bansal and
                  Viswanath Nagarajan},
  title        = {Algorithmizing the Multiplicity Schwartz-Zippel Lemma},
  booktitle    = {Proceedings of the 2023 {ACM-SIAM} Symposium on Discrete Algorithms,
                  {SODA} 2023, Florence, Italy, January 22-25, 2023},
  pages        = {2816--2835},
  publisher    = {{SIAM}},
  year         = {2023},
  url          = {https://doi.org/10.1137/1.9781611977554.ch106},
  doi          = {10.1137/1.9781611977554.CH106},
  bibsource    = {dblp computer science bibliography, https://dblp.org}
}

@inproceedings{BF90,
  author       = {Donald Beaver and
                  Joan Feigenbaum},
  editor       = {Christian Choffrut and
                  Thomas Lengauer},
  title        = {Hiding Instances in Multioracle Queries},
  booktitle    = {{STACS} 90, 7th Annual Symposium on Theoretical Aspects of Computer
                  Science, Rouen, France, February 22-24, 1990, Proceedings},
  series       = {Lecture Notes in Computer Science},
  volume       = {415},
  pages        = {37--48},
  publisher    = {Springer},
  year         = {1990},
  url          = {https://doi.org/10.1007/3-540-52282-4\_30},
  doi          = {10.1007/3-540-52282-4\_30},
  bibsource    = {dblp computer science bibliography, https://dblp.org}
}

@article{KimK2017,
 author = {Kim, John Y. and Kopparty, Swastik},
 title = {Decoding Reed-Muller Codes over Product Sets},
 year = {2017},
 pages = {1--38},
 doi = {10.4086/toc.2017.v013a021},
 publisher = {Theory of Computing},
 journal = {Theory of Computing},
 volume = {13},
 number = {21},
 URL = {https://theoryofcomputing.org/articles/v013a021},
}

@inproceedings{GoyalHKS2024,
  author    = {Rohan Goyal and Prahladh Harsha and Mrinal Kumar and
               Ashutosh Shankar},
  booktitle = {Proc.\ $65$th IEEE Symp.\ on Foundations of Comp.\
               Science (FOCS)},
  editor    = {Santosh Vempala},
  title     = {Fast list-decoding of univariate multiplicity and
               folded {R}eed-{S}olomon codes},
  year      = {2024},
  eccc      = {2023/185}
}

@book{GathenG-MCA,
  author    = {Joachim von zur Gathen and J{\"{u}}rgen Gerhard},
  edition   = {3},
  publisher = {Cambridge University Press},
  title     = {Modern Computer Algebra},
  year      = {2013},
  url       = {https://doi.org/10.1017/CBO9781139856065}
}

@article{RothR2000,
  author  = {Ron M. Roth and Gitit Ruckenstein},
  journal = {IEEE Trans.\ Inform.\ Theory},
  number  = {1},
  pages   = {246--257},
  title   = {Efficient decoding of {R}eed-{S}olomon codes beyond
             half the minimum distance},
  volume  = {46},
  year    = {2000},
  url     = {https://doi.org/10.1109/18.817522}
}

@article{Alekhnovich,
  author  = {Michael Alekhnovich},
  journal = {IEEE Trans.\ Inform.\ Theory},
  number  = {7},
  pages   = {2257--2265},
  title   = {Linear {D}iophantine equations over polynomials and
             soft decoding of Reed-Solomon codes},
  volume  = {51},
  year    = {2005},
  url     = {https://doi.org/10.1109/TIT.2005.850097}
}

@Article{Peterson60,
  title =	"Encoding and error-correction procedures for the
		 {Bose-Chaudhuri} codes",
  author =	"W. Wesley Peterson",
  journal =	"IRE Transactions on Information Theory",
  year = 	"1960",
  number =	"4",
  volume =	"6",
  bibdate =	"2020-05-18",
  bibsource =	"DBLP,
		 http://dblp.uni-trier.de/https://doi.org/10.1109/TIT.1960.1057586;
		 DBLP,
		 http://dblp.uni-trier.de/https://www.wikidata.org/entity/Q62514864;
		 DBLP,
		 http://dblp.uni-trier.de/db/journals/tit/tit6.html#Peterson60",
  pages =	"459--470",
}

@Misc{BW87,
  author =	"E. R. Berlekamp and L. Welch",
  title =	"Error correction of algebraic block codes",
  howpublished = "US Patent Number 4,633,470",
  year = 	"1987",
}

@Misc{GRS_survey,
  author =	"Venkatesan Guruswami and Atri Rudra and Madhu Sudan",
title = {Essential Coding Theory},
    note = {\url{https://cse.buffalo.edu/faculty/atri/courses/coding-theory/book/}. Last accessed: August 2025},
}

@book{Berlekamp_book,
  author       = {Elwyn R. Berlekamp},
  title        = {Algebraic coding theory},
  series       = {McGraw-Hill series in systems science},
  publisher    = {McGraw-Hill},
  year         = {1968},
  url          = {https://www.worldcat.org/oclc/00256659},
  isbn         = {0070049033},
  bibsource    = {dblp computer science bibliography, https://dblp.org}
}

@article{Massey69,
  author       = {James L. Massey},
  title        = {Shift-register synthesis and {BCH} decoding},
  journal      = {{IEEE} Trans. Inf. Theory},
  volume       = {15},
  number       = {1},
  pages        = {122--127},
  year         = {1969},
  url          = {https://doi.org/10.1109/TIT.1969.1054260},
  doi          = {10.1109/TIT.1969.1054260},
  bibsource    = {dblp computer science bibliography, https://dblp.org}
}

@article{GZ61,
  author       = {Daniel Gorenstein and Neal Zierler},
  title        = {A class of error-correcting codes in pm symbols},
  journal      = {Journal of the Society for Industrial and Applied Mathematics},
  volume       = {9},
  number       = {2},
  pages        = {207--214},
  year         = {1961},
}

@article{MRRSW24,
  author       = {Jonathan Mosheiff and
                  Nicolas Resch and
                  Noga Ron{-}Zewi and
                  Shashwat Silas and
                  Mary Wootters},
  title        = {Low-Density Parity-Check Codes Achieve List-Decoding Capacity},
  journal      = {{SIAM} J. Comput.},
  volume       = {53},
  number       = {6},
  pages        = {S20--38},
  year         = {2024},
  url          = {https://doi.org/10.1137/20m1365934},
  doi          = {10.1137/20M1365934},
  bibsource    = {dblp computer science bibliography, https://dblp.org}
}

@article{GR22,
  author       = {Zeyu Guo and
                  Noga Ron{-}Zewi},
  title        = {Efficient List-Decoding With Constant Alphabet and List Sizes},
  journal      = {{IEEE} Trans. Inf. Theory},
  volume       = {68},
  number       = {3},
  pages        = {1663--1682},
  year         = {2022},
  url          = {https://doi.org/10.1109/TIT.2021.3131992},
  doi          = {10.1109/TIT.2021.3131992},
  bibsource    = {dblp computer science bibliography, https://dblp.org}
}

@inproceedings{GZ23,
  author       = {Zeyu Guo and
                  Zihan Zhang},
  title        = {Randomly Punctured Reed-Solomon Codes Achieve the List Decoding Capacity
                  over Polynomial-Size Alphabets},
  booktitle    = {64th {IEEE} Annual Symposium on Foundations of Computer Science, {FOCS}
                  2023, Santa Cruz, CA, USA, November 6-9, 2023},
  pages        = {164--176},
  publisher    = {{IEEE}},
  year         = {2023},
  url          = {https://doi.org/10.1109/FOCS57990.2023.00019},
  doi          = {10.1109/FOCS57990.2023.00019},
  bibsource    = {dblp computer science bibliography, https://dblp.org}
}

@article{Tamo24,
  author       = {Itzhak Tamo},
  title        = {Tighter List-Size Bounds for List-Decoding and Recovery of Folded
                  Reed-Solomon and Multiplicity Codes},
  journal      = {{IEEE} Trans. Inf. Theory},
  volume       = {70},
  number       = {12},
  pages        = {8659--8668},
  year         = {2024},
  url          = {https://doi.org/10.1109/TIT.2024.3402171},
  doi          = {10.1109/TIT.2024.3402171},
  bibsource    = {dblp computer science bibliography, https://dblp.org}
}

@article{BDGZ25,
  author       = {Joshua Brakensiek and
                  Manik Dhar and
                  Sivakanth Gopi and
                  Zihan Zhang},
  title        = {{AG} Codes Achieve List-Decoding Capacity Over Constant-Sized Fields},
  journal      = {{IEEE} Trans. Inf. Theory},
  volume       = {71},
  number       = {8},
  pages        = {5935--5956},
  year         = {2025},
  url          = {https://doi.org/10.1109/TIT.2025.3577506},
  doi          = {10.1109/TIT.2025.3577506},
  bibsource    = {dblp computer science bibliography, https://dblp.org}
}

@article{Kop15,
	  title={List-decoding multiplicity codes},
	    author={Kopparty, Swastik},
	      journal={Theory of Computing},
	        volume={11},
		  number={5},
		    pages={149--182},
		      year={2015}
}

@Article{GW13,
  title =	"Linear-Algebraic List Decoding for Variants of
		 {Reed-Solomon} Codes",
  author =	"Venkatesan Guruswami and Carol Wang",
  journal =	"IEEE Transactions on Information Theory",
  year = 	"2013",
  number =	"6",
  volume =	"59",
    pages =	"3257--3268",
}

@article{Joshi58,
  author       = {D. D. Joshi},
  title        = {A Note on Upper Bounds for Minimum Distance Codes},
  journal      = {Inf. Control.},
  volume       = {1},
  number       = {3},
  pages        = {289--295},
  year         = {1958},
  url          = {https://doi.org/10.1016/S0019-9958(58)80006-6},
  doi          = {10.1016/S0019-9958(58)80006-6},
  bibsource    = {dblp computer science bibliography, https://dblp.org}
}

@article{Singleton64,
  author       = {Richard C. Singleton},
  title        = {Maximum distance q -nary codes},
  journal      = {{IEEE} Trans. Inf. Theory},
  volume       = {10},
  number       = {2},
  pages        = {116--118},
  year         = {1964},
  url          = {https://doi.org/10.1109/TIT.1964.1053661},
  doi          = {10.1109/TIT.1964.1053661},
  bibsource    = {dblp computer science bibliography, https://dblp.org}
}

@article{GX22,
  author       = {Venkatesan Guruswami and
                  Chaoping Xing},
  title        = {Optimal Rate List Decoding over Bounded Alphabets Using Algebraic-geometric
                  Codes},
  journal      = {J. {ACM}},
  volume       = {69},
  number       = {2},
  pages        = {10:1--10:48},
  year         = {2022},
  url          = {https://doi.org/10.1145/3506668},
  doi          = {10.1145/3506668},
  bibsource    = {dblp computer science bibliography, https://dblp.org}
}

@misc{BCDZ25,
      title={From Random to Explicit via Subspace Designs With Applications to Local Properties and Matroids}, 
      author={Joshua Brakensiek and Yeyuan Chen and Manik Dhar and Zihan Zhang},
      year={2025},
      eprint={2510.13777},
      archivePrefix={arXiv},
      primaryClass={cs.IT},
      url={https://arxiv.org/abs/2510.13777}, 
}

@Article{GK16,
	author="Guruswami, Venkatesan
		and Kopparty, Swastik",
	title="Explicit subspace designs",
	journal="Combinatorica",
	year="2016",
	volume="36",
	number="2",
	pages="161--185",
}

@phdthesis{Kiraly-thesis,
    author = {  Tam{\'{a}}s Kir{\'{a}}ly},
    school = {E\"{o}tv\"{o}s Lor{\'{a}}nd University},
    title = {{Edge-connectivity of undirected and directed hypergraphs}},
    url = {https://tkiraly.web.elte.hu/pub/tkiraly_thesis.pdf},
    year = {2003}
}

@article{Elias57,
  author       = {Peter Elias},
  title        = {List decoding for noisy channels},
  journal      = {Research Laboratory
of Electronics, MIT},
  volume       = {Technical Report 335},
  year         = {1957},
}

@article{Wozencraft58,
  author       = {John Wozencraft},
  title        = {List decoding},
  journal      = {Quarterly Progress Report, Research Laboratory
of Electronics, MIT},
  volume       = {48},
pages = {90–95}, 
  year         = {1957},
}

@article{GRS00,
  author       = {Oded Goldreich and
                  Ronitt Rubinfeld and
                  Madhu Sudan},
  title        = {Learning Polynomials with Queries: The Highly Noisy Case},
  journal      = {{SIAM} J. Discret. Math.},
  volume       = {13},
  number       = {4},
  pages        = {535--570},
  year         = {2000},
  url          = {https://doi.org/10.1137/S0895480198344540},
  doi          = {10.1137/S0895480198344540},
  bibsource    = {dblp computer science bibliography, https://dblp.org}
}

@article{GR08_folded_RS,
	Author = {Venkatesan Guruswami and Atri Rudra},
	Journal = {IEEE Transactions on Information Theory},
	Number = {1},
	Pages = {135-150},
	Title = {Explicit Codes Achieving List Decoding Capacity: Error-Correction With Optimal Redundancy},
	Volume = {54},
	Year = {2008}}

@article{Gur-survey,
	Author = {Venkatesan Guruswami},
	Bibdate = {2008-05-20},
	Bibsource = {DBLP, http://dblp.uni-trier.de/db/journals/fttcs/fttcs2.html#Guruswami06},
	Journal = {Foundations and Trends in Theoretical Computer Science},
	Number = {2},
	Title = {Algorithmic Results in List Decoding},
	Url = {http://dx.doi.org/10.1561/0400000007},
	Volume = {2},
	Year = {2006}}

@article{GS92,
  author       = {Peter Gemmell and
                  Madhu Sudan},
  title        = {Highly Resilient Correctors for Polynomials},
  journal      = {Inf. Process. Lett.},
  volume       = {43},
  number       = {4},
  pages        = {169--174},
  year         = {1992},
  url          = {https://doi.org/10.1016/0020-0190(92)90195-2},
  doi          = {10.1016/0020-0190(92)90195-2},
  bibsource    = {dblp computer science bibliography, https://dblp.org}
}

@book{Guruswami-Thesis,
  title={List decoding of error-correcting codes: winning thesis of the 2002 ACM doctoral dissertation competition},
  author={Guruswami, Venkatesan},
  volume={3282},
  year={2004},
  publisher={Springer Science \& Business Media}
}

@Article{Sudan97,
  title =	"Decoding of {Reed Solomon} Codes beyond the
		 Error-Correction Bound",
  author =	"Madhu Sudan",
  journal =	"Journal of Complexity",
  year = 	"1997",
  number =	"1",
  volume =	"13",
  bibdate =	"2017-05-26",
  bibsource =	"DBLP,
		 http://dblp.uni-trier.de/https://doi.org/10.1006/jcom.1997.0439;
		 DBLP,
		 http://dblp.uni-trier.de/db/journals/jc/jc13.html#Sudan97",
  pages =	"180--193",
}

@article{GS-list-dec,
	Author = {Venkatesan Guruswami and Madhu Sudan},
	Bibdate = {2016-03-09},
	Bibsource = {DBLP, http://dblp.uni-trier.de/db/journals/tit/tit45.html#GuruswamiS99},
	Journal = {IEEE Transactions on Information Theory},
	Number = {6},
	Pages = {1757--1767},
	Title = {Improved decoding of {Reed-Solomon} and algebraic-geometry codes},
	Url = {http://dx.doi.org/10.1109/18.782097},
	Volume = {45},
	Year = {1999}}

@article{BFNW93,
	Author = {L{\'a}szl{\'o} Babai and Lance Fortnow and Noam Nisan and Avi Wigderson},
	Journal = {Computational Complexity},
	Number = {4},
	Pages = {307--318},
	Title = {{BPP} Has Subexponential Time Simulations Unless {EXPTIME} has Publishable Proofs},
	Volume = {3},
	Year = {1993}}

@article{STV01,
	Author = {Madhu Sudan and Luca Trevisan and Salil Vadhan},
	Journal = {Journal of Computer and System Sciences},
	Number = {2},
	Pages = {236-266},
	Title = {Pseudorandom Generators without the {XOR} Lemma},
	Volume = {62},
	Year = {2001}}

@article{Vadhan_survey,
  author       = {Salil P. Vadhan},
  title        = {Pseudorandomness},
  journal      = {Found. Trends Theor. Comput. Sci.},
  volume       = {7},
  number       = {1-3},
  pages        = {1--336},
  year         = {2012},
  url          = {https://doi.org/10.1561/0400000010},
  doi          = {10.1561/0400000010},
  bibsource    = {dblp computer science bibliography, https://dblp.org}
}

@InProceedings{BFLS91,
  title =	"Checking Computations in Polylogarithmic Time",
  author =	"L{\'a}szl{\'o} Babai and Lance Fortnow and Leonid A.
		 Levin and Mario Szegedy",
  bibdate =	"2011-10-17",
  bibsource =	"DBLP,
		 http://dblp.uni-trier.de/db/conf/stoc/stoc91.html#BabaiFLS91",
  booktitle =	"Proceedings of the 23rd Annual {ACM} Symposium on
		 Theory of Computing (STOC)",
  publisher =	"ACM Press",
  year = 	"1991",
  ISBN = 	"0-89791-397-3",
  pages =	"21--31",
  URL =  	"http://doi.acm.org/10.1145/103418.103428",
}

@InProceedings{KT00,
  author =	"Jonathan Katz and Luca Trevisan",
  title =	"On the efficiency of local decoding procedures for
		 error-correcting codes",
  booktitle =	"Proceedings of the 32nd Annual ACM Symposium on Theory
		 of Computing (STOC)",
  pages =	"80--86",
  year = 	"2000",
  publisher =	"ACM Press",
  cdate =	"1970-01-01",
  mdate =	"2005-08-18",
}

@article{BHKS24_ideal,
  author       = {Siddharth Bhandari and
                  Prahladh Harsha and
                  Mrinal Kumar and
                  Madhu Sudan},
  title        = {Ideal-Theoretic Explanation of Capacity-Achieving Decoding},
  journal      = {{IEEE} Trans. Inf. Theory},
  volume       = {70},
  number       = {2},
  pages        = {1107--1123},
  year         = {2024},
  url          = {https://doi.org/10.1109/TIT.2023.3345890},
  doi          = {10.1109/TIT.2023.3345890},
  bibsource    = {dblp computer science bibliography, https://dblp.org}
}

@Article{KMRS,
  title =	"High-Rate Locally Correctable and Locally Testable
		 Codes with Sub-Polynomial Query Complexity",
  author =	"Swastik Kopparty and Or Meir and Noga {Ron-Zewi} and
		 Shubhangi Saraf",
  journal =	"Journal of ACM",
  year = 	"2017",
  number =	"2",
  volume =	"64",
  bibdate =	"2017-06-07",
  bibsource =	"DBLP,
		 http://dblp.uni-trier.de/db/journals/jacm/jacm64.html#KoppartyMRS17",
  pages =	"11:1--11:42",
  URL =  	"http://doi.acm.org/10.1145/3051093",
}

@inproceedings{GL89,
	  title={A hard-core predicate for all one-way functions},
	    author={Goldreich, Oded and Levin, Leonid A},
            booktitle={Proceedings of the 21st Annual ACM Symposium on Theory of Computing (STOC)},
	        pages={25--32},
		  year={1989},
		    organization={ACM}
}

@article{KM93,
	  title={Learning decision trees using the {Fourier} spectrum},
	    author={Kushilevitz, Eyal and Mansour, Yishay},
	      journal={SIAM Journal on Computing},
	        volume={22},
		  number={6},
		    pages={1331--1348},
		      year={1993},
		        publisher={SIAM}
}

@article{GLSTW24,
  author       = {Zeyu Guo and
                  Ray Li and
                  Chong Shangguan and
                  Itzhak Tamo and
                  Mary Wootters},
  title        = {Improved List-Decodability and List-Recoverability of Reed-Solomon
                  Codes via Tree Packings},
  journal      = {{SIAM} J. Comput.},
  volume       = {53},
  number       = {2},
  pages        = {389--430},
  year         = {2024},
  url          = {https://doi.org/10.1137/21m1463707},
  doi          = {10.1137/21M1463707},
  bibsource    = {dblp computer science bibliography, https://dblp.org}
}

@article{Menger27,
author = {Menger, Karl},
journal = {Fundamenta Mathematicae},
language = {ger},
number = {1},
pages = {96-115},
title = {Zur allgemeinen Kurventheorie},
url = {http://eudml.org/doc/211191},
volume = {10},
year = {1927},
}

@article{FKK03a,
  author       = {Andr{\'{a}}s Frank and
                  Tam{\'{a}}s Kir{\'{a}}ly and
                  Matthias Kriesell},
  title        = {On decomposing a hypergraph into k connected sub-hypergraphs},
  journal      = {Discret. Appl. Math.},
  volume       = {131},
  number       = {2},
  pages        = {373--383},
  year         = {2003},
  url          = {https://doi.org/10.1016/S0166-218X(02)00463-8},
  doi          = {10.1016/S0166-218X(02)00463-8},
  bibsource    = {dblp computer science bibliography, https://dblp.org}
}

@article{FKK03b,
title = {On the orientation of graphs and hypergraphs},
journal = {Discrete Applied Mathematics},
volume = {131},
number = {2},
pages = {385-400},
year = {2003},
issn = {0166-218X},
doi = {https://doi.org/10.1016/S0166-218X(02)00462-6},
url = {https://www.sciencedirect.com/science/article/pii/S0166218X02004626},
author = {András Frank and Tamás Király and Zoltán Király},
}

@inproceedings{JMS03,
  author       = {Kamal Jain and
                  Mohammad Mahdian and
                  Mohammad R. Salavatipour},
  title        = {Packing Steiner trees},
  booktitle    = {Proceedings of the Fourteenth Annual {ACM-SIAM} Symposium on Discrete
                  Algorithms, January 12-14, 2003, Baltimore, Maryland, {USA}},
  pages        = {266--274},
  publisher    = {{ACM/SIAM}},
  year         = {2003},
  url          = {http://dl.acm.org/citation.cfm?id=644108.644154},
  bibsource    = {dblp computer science bibliography, https://dblp.org}
}

@inproceedings{FK08,
  author       = {Andr{\'{a}}s Frank and
                  Tam{\'{a}}s Kir{\'{a}}ly},
  editor       = {William J. Cook and
                  L{\'{a}}szl{\'{o}} Lov{\'{a}}sz and
                  Jens Vygen},
  title        = {A Survey on Covering Supermodular Functions},
  booktitle    = {Research Trends in Combinatorial Optimization, Bonn Workshop on Combinatorial
                  Optimization, November 3-7, 2008, Bonn, Germany},
  pages        = {87--126},
  publisher    = {Springer},
  year         = {2008},
  url          = {https://doi.org/10.1007/978-3-540-76796-1\_6},
  doi          = {10.1007/978-3-540-76796-1\_6},
  bibsource    = {dblp computer science bibliography, https://dblp.org}
}

@article{KRSW23,
  author       = {Swastik Kopparty and
                  Noga Ron{-}Zewi and
                  Shubhangi Saraf and
                  Mary Wootters},
  title        = {Improved List Decoding of Folded Reed-Solomon and Multiplicity Codes},
  journal      = {{SIAM} J. Comput.},
  volume       = {52},
  number       = {3},
  pages        = {794--840},
  year         = {2023},
  url          = {https://doi.org/10.1137/20m1370215},
  doi          = {10.1137/20M1370215},
  bibsource    = {dblp computer science bibliography, https://dblp.org}
}

@inproceedings{CZ25,
  author       = {Yeyuan Chen and
                  Zihan Zhang},
  editor       = {Michal Kouck{\'{y}} and
                  Nikhil Bansal},
  title        = {Explicit Folded Reed-Solomon and Multiplicity Codes Achieve Relaxed
                  Generalized Singleton Bounds},
  booktitle    = {Proceedings of the 57th Annual {ACM} Symposium on Theory of Computing,
                  {STOC} 2025, Prague, Czechia, June 23-27, 2025},
  pages        = {1--12},
  publisher    = {{ACM}},
  year         = {2025},
  url          = {https://doi.org/10.1145/3717823.3718114},
  doi          = {10.1145/3717823.3718114},
  bibsource    = {dblp computer science bibliography, https://dblp.org}
}

@Article{DKSS13,
  title =	"Extensions to the Method of Multiplicities, with
		 Applications to {Kakeya} Sets and Mergers",
  author =	"Zeev Dvir and Swastik Kopparty and Shubhangi Saraf and
		 Madhu Sudan",
  journal =	"SIAM Journal on Computing",
  year = 	"2013",
  number =	"6",
  volume =	"42",
  bibdate =	"2017-05-27",
  bibsource =	"DBLP,
		 http://dblp.uni-trier.de/https://doi.org/10.1137/100783704;
		 DBLP,
		 http://dblp.uni-trier.de/db/journals/siamcomp/siamcomp42.html#DvirKSS13",
  pages =	"2305--2328",
}

@book{Frank11,
  title={Connections in combinatorial optimization},
  author={Andr{\'{a}}s Frank},
  volume={38},
  year={2011},
  publisher={Oxford University Press}
}

@article{CX18,
  author       = {Chandra Chekuri and
                  Chao Xu},
  title        = {Minimum Cuts and Sparsification in Hypergraphs},
  journal      = {{SIAM} J. Comput.},
  volume       = {47},
  number       = {6},
  pages        = {2118--2156},
  year         = {2018},
  url          = {https://doi.org/10.1137/18M1163865},
  doi          = {10.1137/18M1163865},
  bibsource    = {dblp computer science bibliography, https://dblp.org}
}

\end{document}